%% file: main.tex
\documentclass[twoside,11pt]{article}

\usepackage{amsfonts}
\usepackage{amsmath}
\usepackage{amsthm}
\usepackage{amssymb}
\usepackage{graphicx}
\usepackage{subfigure}
\usepackage{paralist}

\usepackage{algorithm}
\usepackage{algorithmic}

\usepackage{paralist}

\newtheorem{definition}{Definition}

\newtheorem{theorem}{Theorem}
\newtheorem{lemma}{Lemma}
\newtheorem{claim}{Claim}
\newtheorem{question}{Question}

\newtheorem{corollary}{Corollary}
\newenvironment{Proof}{\noindent {\em Proof:}}{\\\hspace*{\fill}\mbox{$\diamond$}}

\newcommand{\Trace }[1]{\mbox{}{\bf{Tr}}\left(#1\right)}

\DeclareMathOperator*{\argmin}{arg\!min}

\newcommand{\vol}{\mathrm{vol}}
\newcommand{\defeq}{\stackrel{\textup{def}}{=}}

\renewcommand{\sectionmark}[1]{\markboth{\textsc{M. W. Mahoney}}
{\textsc{Lecture Notes on Spectral Graph Methods}}}
\renewcommand{\subsectionmark}[1]{\markboth{\textsc{M. W. Mahoney}}
{\textsc{Lecture Notes on Spectral Graph Methods}}}

\newlength{\defbaselineskip}
\begin{document}

\title{
Lecture Notes on Spectral Graph Methods 
}

\author{
Michael W. Mahoney%
\thanks{
International Computer Science Institute 
and 
Department of Statistics, 
University of California at Berkeley, 
Berkeley, CA. 
E-mail: mmahoney@stat.berkeley.edu.
}
}

\date{}


\clearpage\maketitle
\thispagestyle{empty}


\newpage
\thispagestyle{empty}

\begin{abstract}

\noindent
These are lecture notes that are based on the lectures from a class I taught on the topic of Spectral Graph Methods at UC Berkeley during the Spring 2015 semester. 

Spectral graph techniques are remarkably broad, they are widely-applicable and often very useful, and they also come with a rich underlying theory, some of which provides a very good guide to practice.
Spectral graph theory is often described as the area that studies properties of graphs by studying properties of eigenvalues and eigenvectors of matrices associated with the graph.
While that is true, especially of ``classical'' spectral graph theory, many of the most interesting recent developments in the area go well beyond that and instead involve early-stopped (as well as asymptotic) random walks and diffusions, metric embeddings, optimization and relaxations, local and locally-biased algorithms, statistical and inferential considerations, etc.
Indeed, one of the things that makes spectral graph methods so interesting and challenging is that researchers from different areas, e.g., computer scientists, statisticians, machine learners, and applied mathematicians (not to mention people who actually use these techniques), all come to the topic from very different perspectives.
This leads them 
to parameterize problems quite differently, 
to formulate basic problems as well as variants and extensions of basic problems in very different ways, and
to think that very different things are ``obvious'' or ``natural'' or ``trivial.''
These lectures will provide an overview of the theory of spectral graph methods, including many of the more recent developments in the area, with an emphasis on some of these complementary perspectives, and with an emphasis on those methods that are useful in practice.

I have drawn liberally from the lectures, notes, and papers of others, often without detailed attribution in each lecture.
Here are the sources upon which I drew most heavily, in rough order of appearance over the semester.
\begin{compactitem}
\item
``Lecture notes,'' from Spielman's Spectral Graph Theory class, Fall 2009 and 2012
\item
``Survey: Graph clustering,'' in Computer Science Review, by Schaeffer
\item
``Geometry, Flows, and Graph-Partitioning Algorithms,'' in CACM, by Arora, Rao, and Vazirani 
\item
``Lecture Notes on Expansion, Sparsest Cut, and Spectral Graph Theory,'' by Trevisan
\item
``Expander graphs and their applications,'' in Bull. Amer. Math. Soc., by Hoory, Linial, and Wigderson
\item
``Multicommodity max-flow min-cut theorems and their use in designing approximation algorithms,'' in JACM, by Leighton and Rao 
\item
``Efficient Maximum Flow Algorithms,'' in CACM, by Goldberg and Tarjan 
\item
``A Tutorial on Spectral Clustering,'' in Statistics and Computing, by von Luxburg
\item
``A kernel view of the dimensionality reduction of manifolds,'' in ICML, by Ham, et al.
\item
``Laplacian Eigenmaps for dimensionality reduction and data representation,'' in Neural Computation, by Belkin and Niyogi
\item
``Diffusion maps and coarse-graining: a unified framework for dimensionality reduction, graph partitioning, and data set parameterization,'' in IEEE-PAMI, by Lafon and Lee 
\item
``Transductive learning via spectral graph partitioning,'' in ICML, by Joachims
\item
``Semi-Supervised Learning Using Gaussian Fields and Harmonic Functions,'' in ICML, by Zhu, Ghahramani, and Lafferty 
\item
``Learning with local and global consistency,'' in NIPS, by Zhou et al.
\item
``Random Walks and Electric Networks,'' in arXiv, by Doyle and Snell
\item
``Implementing regularization implicitly via approximate eigenvector computation,'' in ICML, by Mahoney and Orecchia
\item
``Regularized Laplacian Estimation and Fast Eigenvector Approximation,'' in NIPS, by Perry and Mahoney
\item
``Spectral Ranking'', in arXiv, by Vigna
\item
``PageRank beyond the Web,'' in SIAM Review, by Gleich
\item
``The Push Algorithm for Spectral Ranking'', in arXiv, by Boldi and Vigna
\item
``Local Graph Partitioning using PageRank Vectors,'' in FOCS, by Andersen, Chung, and Lang
\item
``A Local Spectral Method for Graphs: with Applications to Improving Graph Partitions and Exploring Data Graphs Locally,'' in JMLR, by Mahoney, Orecchia, and Vishnoi
\item
``Anti-differentiating Approximation Algorithms: A case study with Min-cuts, Spectral, and Flow,'' in ICML, by Gleich and Mahoney
\item
``Think Locally, Act Locally: The Detection of Small, Medium-Sized, and Large Communities in Large Networks,'' in PRE, by Jeub, Balachandran, Porter, Mucha, and Mahoney
\item
``Towards a theoretical foundation for Laplacian-based manifold methods,'' in JCSS, by Belkin and Niyogi
\item
``Consistency of spectral clustering,'' in Annals of Statistics, by von Luxburg, Belkin, and Bousquet
\item
``Spectral clustering and the high-dimensional stochastic blockmodel,'' in The Annals of Statistics, by Rohe, Chatterjee, and Yu
\item
``Regularized Spectral Clustering under the Degree-Corrected Stochastic Blockmodel,'' in NIPS, by Qin and Rohe
\item
``Effective Resistances, Statistical Leverage, and Applications to Linear Equation Solving,'' in arXiv, by Drineas and Mahoney
\item
``A fast solver for a class of linear systems,'' in CACM, by Koutis, Miller, and Peng
\item
``Spectral Sparsification of Graphs: Theory and Algorithms,'' in CACM, by Batson, Spielman, Srivastava, and Teng 
\end{compactitem}

Finally, I should note that these notes are unchanged, relative to the notes that have been available on my web page since the class completed; but, in response to a number of requests, I decided to put them all together as a single file and post them on the arXiv. 
They are still very rough, and they likely contain typos and errors.
Thus, feedback and comments---both in terms of specific technical issues as well as general scope---are most welcome.

\vspace{5mm}

\hfill \hfill Michael W. Mahoney

\hfill \hfill August 2016
 
\end{abstract}

\newpage
\rule{6in}{0.1mm}
\tableofcontents
\rule{6in}{0.1mm}
\newpage


\input{lect01}

\newpage
\input{lect02}

\newpage
\input{lect03}

\newpage
\input{lect04}

\newpage
\input{lect05}

\newpage
\input{lect06}

\newpage
\input{lect07}

\newpage
\input{lect08}

\newpage
\input{lect09}

\newpage
\input{lect10}

\newpage
\input{lect11}

\newpage
\input{lect12}

\newpage
\input{lect13}

\newpage
\input{lect14}

\newpage
\input{lect15}

\newpage
\input{lect16}

\newpage
\input{lect17}

\newpage
\input{lect18}

\newpage
\input{lect19}

\newpage
\input{lect20}

\newpage
\input{lect21}

\newpage
\input{lect22}

\newpage
\input{lect23}

\newpage
\input{lect24}

\newpage
\input{lect25}

\newpage
\input{lect26}

%
%

%
%

\end{document}

%% file: lect01.tex
\section{%
(01/22/2015):
Introduction and Overview}

\subsection{Basic motivating background}

The course will cover several topics in \emph{spectral graph methods}.
By that, I mean that it will cover not spectral graph theory per se, nor will 
it cover the application of spectral graph methods per se.
In addition, spectral methods is a more general topic, and graph methods is 
a more general topic.
Spectral graph theory uses eigenvectors and eigenvalues (and related 
quantities) of matrices associated with graphs to say things about those 
graphs.
It is a topic which has been studied from a wide range of perspectives, 
e.g., theoretical computer science, scientific computing, machine 
learning, statistics, etc., and as such it is a topic which can be viewed 
from a wide range of approaches.

The reason for the focus on spectral graph methods is that a wide range of 
problems are obviously spectral graph methods, and thus they are useful in 
practice as well as interesting in theory; but, in addition, many other 
methods that are not obviously spectral graph methods really are spectral 
graph methods under the hood.  
We'll get to what I mean by that, but for now think of running some 
procedure that seems to work, and if one were to perform a somewhat 
rigorous algorithmic or statistical analysis, it would turn out that that 
method essentially boiled down to a spectral graph method.
As an example, consider the problem of viral propagation on a social 
network, which is usually described in terms of some sort of infectious 
agent, but which has strong connections with spectral graph methods.

Our goal will be to understand---by drawing strength from each of the wide
range of approaches that have been brought to bear on these 
problems---when/why spectral graph methods are useful in practical machine 
learning and data analysis applications, and when/why they (or a vanilla
variant of them) are not useful.
In the latter case, of course, we'll be interested in understanding whether
a better understanding of spectral graph methods can lead to the development
of improved algorithmic and statistical techniques---both for large-scale
data as well as for small-scale data.
Relatedly, we will be interested in whether other methods can perform better
or whether the data are just ``bad'' in some sense.

\subsection{Types of data and types of problems}

Data comes from all sorts of places, and it can be a challenge to find a good
way to represent the data in order to obtain some sort of meaningful insight
from the data.
Two popular ways to model data are as \emph{matrices} and as \emph{graphs}.
\begin{itemize}
\item
Matrices often arise when there are $n$ things, each of which is described 
by $m$ features.
In this case, we have an $m \times n$ matrix $A$, where each column is a 
data point described by a bunch of features (or vice versa) and where each 
row is a vector describing the value of that feature at each data point.
Alternatively, matrices can arise when there are $n$ things and we have 
information about the correlations (or other relationships) between them.
\item
Graphs often arise when there are $n$ things and the pairwise relationships
between them are though to be particularly important.
Let's specify a graph by $G=(V,E)$, where $V$ is the set of vertices and $E$ 
is the set of edges, which are pairs of vertices.
(Later they can be weighted, etc., but for now let's say they are undirected, 
unweighted, etc.)
Examples of data graphs include the following.
\begin{itemize}
\item
Discretizations of partial differential equations and other physical 
operators give rise to graphs, where the nodes are points in a physical 
medium and edges correspond to some sort of physical interaction.
\item
Social networks and other internet applications give rise to graphs, where 
the nodes are individuals and there is an edge between two people if they 
are friends or have some other sort of interaction.
\item
Non-social networks give rise to graphs, where, e.g., devised, routers, or 
computers are nodes and where there is an edge between two nodes if they are
connected and/or have traffic between them.
\item
Graphs arise more generally in machine learning and data analysis 
applications.
For example, given a bunch of data points, each of which is a feature vector, 
we could construct a graph, where the nodes correspond to data points and 
there is an edge between two data points if they are close in some sense (or 
a soft version of this, which is what rbf kernels do).
\end{itemize}
\end{itemize}

\noindent
In the same way as we can construct graphs from matrices, we can also 
construct matrices from graphs.
We will see several examples below (e.g., adjacency matrices, Laplacians, 
low-rank embedding matrices, etc.).
Spectral graph methods involve using eigenvectors and eigenvalues of matrices
associated with graphs to do stuff.

In order to do stuff, one runs some sort of algorithmic or statistical 
methods, but it is good to keep an eye on the types of problems that might
want to be solved.
Here are several canonical examples.

\begin{itemize}
\item
Graph partitioning: finding clusters/communities.
Here, the data might be a bunch of data points (put a picture: sprinkled into 
a left half and a right half) or it might be a graph (put a picture: two 
things connected by an edge).
There are a million ways to do this, but one very popular one boils down 
to computing an eigenvector of the so-called Laplacian matrix and using that
to partition the data.
Why does such a method work?  
One answer, from TCS, is that is works since it is a relaxation of a 
combinatorial optimization problem for which there are worst-case 
quality-of-approximation guarantees.
Another answer, from statistics and machine learning, is that is can be used
to recover hypothesized clusters, say from a stochastic blockmodel (where the
graph consists of several random graphs put together) or from a 
low-dimensional manifold (upon which the data points sit).
\item
Prediction: e.g., regression and classification.
Here, there is a similar picture, and one popular procedure is to run the same
algorithm, compute the same vector and use it to classify the data.
In this case, one can also ask: why does such a method work?
One answer that is commonly given is that if there are meaningful clusters in 
the data, say drawn from a manifold, then the boundaries between the clusters
correspond to low-density regions; or relatedly that class labels are smooth 
in the graph topology, or some notion of distance in $\mathbb{R}^{n}$.
But, what if the data are from a discrete place? 
Then, there is the out-of-sample extension question and a bunch of other 
issues.
\item
Centrality and ranking.
These are two (different but sometimes conflated) notions from sociology and 
social networks having to do with how ``important'' or ``central'' is an 
individual/node and relatedly how to rank individuals/nodes.
One way to do this is to choose the highest degree node, but this is 
relatively easy to spam and might not be ``real'' in other ways, and so 
there are several related things that go by the name of spectral ranking, 
eigenvector centrality, and so on.
The basic idea is that a node is important if important nodes thing it is
important.
This suggests looking at loops and triangles in a graph, and when this 
process is iterated you get random walks and diffusions on the graph.
It's not obvious that this has very strong connections with the clustering, 
classification, etc. problems described above, but it does.
Basically, you compute the same eigenvector and use it to rank.
\item
Encouraging or discouraging ``viral propagation.''
Here, one is given, say, a social network, and one is interested in some 
sort of iterative process, and one wants to understand its properties.  
Two examples are the following:
there might be a virus or other infectious agent that goes from node to node
making other agents sick; or
there might be some sort of ``buzz'' about a new movie or new tennis shoes, 
and this goes from individual to individual.
Both of these are sometimes called viral propagation, but there are important
differences, not the least of which is that in the former people typically 
want to stop the spread of the virus, while in the latter people want to 
encourage the spread of the virus to sell more tennis shoes.
\end{itemize}

\subsection{Examples of graphs}

When algorithms are run on data graphs---some of which might be fairly nice
but some of which might not, it can be difficult to know why the algorithm 
performs as it does.
For example, would it perform that way on every possible graph?
Similarly, if we are not in some asymptotic limit or if worst-case analysis
is somewhat too coarse, then what if anything does the method reveal about
the graph?
To help address these and related questions, it helps to have several 
examples of graphs in mind and to see how algorithms perform on those graphs.
Here are several good examples.
\begin{itemize}
\item
A discretization of nice low-dimensional space, e.g., the integers/lattice in 
some fixed dimension:
$\mathbb{Z}_d$, $\mathbb{Z}^{2}_d$, and $\mathbb{Z}^{3}_d$.
\item
A star, meaning a central node to which all of the other nodes are attached.
\item
A binary tree.
\item
A complete graph or clique.
\item
A constant-degree expander, which is basically a very sparse graph that has 
no good partitions.  
Alternatively, it can be viewed as a sparse version of the complete graph.
\item
A hypercube on $2^n$ vertices.
\item
A graph consisting of two complete graphs or two expanders or two copies of 
$\mathbb{Z}^{2}_d$ that are weakly connected by, say, a line graph.
\item
A lollipop, meaning a complete graph of expander, with a line graph 
attached, where the line/stem can have different lengths.
\end{itemize}
Those are common constructions when thinking about graphs.
The following are examples of constructions that are more common in certain 
network applications.
\begin{itemize}
\item
An Erdos-Renyi random graph, $G_{np}$, for $p=3/n$ or $p \gtrsim \log(n)/n$
\item
A ``small world'' graph, which is basically a ring plus a $3$ regular random 
graph.
\item
A heavy-tailed random graph, with or without min degree assumption, or one
constructed from a preferential attachment process.
\end{itemize}
In addition to things that are explicitly graphs, it also helps to have 
several examples of matrix-based data to have in mind from which graphs can 
be constructed.
These are often constructed from some sort of nearest-neighbor process.
Here are several common examples.
\begin{itemize}
\item
A nice region of a low-dimensional subspace of $\mathbb{R}^{n}$ or of a 
nice low-dimensional manifold embedded in $\mathbb{R}^{n}$.
\item
A full-dimensional Gaussian in $\mathbb{R}^{n}$. 
Here, most of the mass is on the shell, but what does the graph 
corresponding to this ``look like''?
\item
Two low-dimensional Gaussians in $\mathbb{R}^{n}$.
This looks like a dumbbell, with two complete graphs at the two ends or two 
copies of $\mathbb{Z}^{n}_{d}$ at the ends, depending on how parameters are 
set.
\end{itemize}

\subsection{Some questions to consider}

Here are a few questions to consider.
\begin{itemize}
\item
If the original data are vectors that form a matrix, how sensitive are 
these methods to the details of the graph construction?
(Answer: in theory, no; in practice, often yes.)
\item
If the original data are represented by a graph, how sensitive are these 
methods to a bit of noise in the graph?
(Answer: in theory, no; in practice, often yes.)
\item
How good a guide is worst case cs and asymptotic statistical theory?
(Answer: in theory, good; in practice, often not, but it depends on what is
the reference state, e.g., manifold versus stochastic blockmodel.)
\item
What if you are interested in a small part of a very large graph?
E.g., you and your $100$ closest friends on a social network, as opposed to 
you and your $10^9$ closest friends on that social network.
Do you get the same results if you run some sort of local algorithm on a 
small part of the graph as you do if you run a global algorithm on a 
subgraph that is cut out?
(Typically no, unless you are very careful.)
\end{itemize}

\subsection{Matrices for graphs}

Let $G=(V,E,W)$ be an undirected, possibly weighted, graph.
There are many matrices that one can associate with a graph.
Two of the most basic are the \emph{adjacency matrix} and the \emph{diagonal 
degree matrix}.
\begin{definition}
Given $G=(V,E,W)$,
the adjacency matrix $A\in\mathbb{R}^{n \times n}$ is defined to be
\[ 
A_{ij} = \left\{ \begin{array}{l l}
                    W_{ij} & \quad \text{if $(ij)\in E$}\\
                    0      & \quad \text{otherwise}
                 \end{array} 
         \right.  ,
\]
and the diagonal degree matrix $D\in\mathbb{R}^{n \times n}$ is defined to be
\[ 
D_{ij} = \left\{ \begin{array}{l l}
                    \sum_k W_{ik} & \quad \text{if $i=j$}\\
                    0             & \quad \text{otherwise}
                 \end{array} 
         \right.  .
\]
\end{definition}
Note that for undirected graphs, i.e., when $W_{ij}$ equals $1$ or $0$ 
depending on whether or not there is an edge between nodes $i$ and $j$, 
the adjacency matrix specifies which edges are connected and the diagonal 
degree matrix gives the degree of $i^{th}$ node at the $(ii)^{th}$ 
diagonal position.

(Given this setup, it shouldn't be surprising that most spectral graph 
methods generalize in nice ways from unweighted to weighted graphs.
Of interest also are things like time-evolving graphs, directed graphs, 
etc.  In those cases, the situation is more subtle/complex.  Typically, 
methods for those more complex graphs boil down to methods for simpler 
undirected, static graphs.)

Much of what we will discuss has strong connections with spectral graph 
theory, which is an area that uses eigenvectors and eigenvalues of matrices 
associated with the graph to understand properties of the graph.
To begin, though, we should note that it shouldn't be obvious that 
eigenstuff should reveal interesting graph properties---after all, graphs 
by themselves are essentially combinatorial things and most traditional 
graph problems and algorithms don't mention anything having to do with 
eigenvectors.
In spite of this, we will see that eigenstuff reveals a lot about graphs 
that are useful in machine learning and data analysis applications, and we
will want to understand why this is the case and how we can take advantage
of it in interesting ways.

Such an approach of using eigenvectors and eigenvalues is most useful when 
used to understand a natural operator of natural quadratic form associated
with the graphs. 
Perhaps surprisingly, adjacency matrices and diagonal degree matrices are not 
so useful in that sense---but they can be used to construct other matrices 
that are more useful in that sense.

One natural and very useful operator to associate with a graph $G$ is the 
following.
\begin{definition}
Given $G=(V,E,W)$, the \emph{diffusion operator} is 
\[
W = D^{-1}A  \quad \text{(or $M=AD^{-1}$, if you multiply from the other side)} .
\]
\end{definition}
This matrix describes the behavior of diffusions and random walks on $G$.
In particular, if $x \in \mathbb{R}^{n}$ is a row vector that gives the 
probability that a particle is at each vertex of $G$, and if the particle
then moves to a random neighbor, then $pW$ is the new probability 
distribution of the particle.
If the graph $G$ is regular, meaning that it is degree-homogeneous, then 
$W$ is a rescaling of $A$, but otherwise it can be very different.
Although we won't go into too much detail right now, note that applying this 
operator to a vector can be interpreted as doing one step of a diffusion or 
random walk process.
In this case, one might want to know what happens if we iteratively apply an 
operator like $W$.
We will bet back to this.

One natural and very useful quadratic form to associate with a graph $G$ is 
the following.
\begin{definition}
Given $G=(V,E,W)$, 
the \emph{Laplacian matrix} (or combinatorial Laplacian matrix) is
\[
L=D-A  .
\]
\end{definition}
Although we won't go into detail, the Laplacian has an interpretation in 
terms of derivatives.
(This is most common/obvious in continuous applications, where it can be 
used to measure the smoothness of the Laplacian and/or of some continuous 
place from which the gradient was constructed---if it was---in a nice way, 
which is often \emph{not} the case.)
Given a function or a vector $x \in \mathbb{R}^{n}$, the Laplacian quadratic 
form is
\[
x^TLx = \sum_{(ij)\in E} (x_i-x_j)^2   .
\]
This is a measure of smoothness of the vector/function $x$---smoothness of 
$x$, in some sense, conditioned on the graph structure. 
(That is, it is a statement about the graph itself, independent of how it was 
constructed.  This is of interest by itself but also for machine learning and 
data analysis application, e.g., since labels associated with the nodes that 
correspond to a classification function might be expected to be smooth.)

Alternatively, we can define the \emph{normalized Laplacian matrix}.
\begin{definition}
Given $G=(V,E,W)$, the \emph{normalized Laplacian matrix} is
\[
\mathcal{L} = L
            = D^{-1/2} L D^{-1/2}
            = I - D^{-1/2} L D^{-1/2}   .
\]
\end{definition}
(Note that I have already started to be a little sloppy, by using the same 
letter to mean two different things.  I'll point out as we go where this 
matters.)
As we will see, for degree-homogeneous graphs, these two Laplacians are 
essentially the same, but for degree-heterogeneous graphs, they are quite 
different. 
As a general rule, the latter is more appropriate for realistic 
degree-heterogeneous graphs, but it is worth keeping the two in mind, since 
there are strong connections between them and how they are computed.
Similar smoothness, etc. interpretations hold for the normalized Laplacian, 
and this is important in many applications.

\subsection{An overview of some ideas}

Here is a vanilla version of a spectral graph algorithm that will be central 
to a lot of what we do.
We'll be more precise and go into a lot more detail later.
The algorithm takes as input a graph, as specified by a Laplacian matrix, 
$L$ or $\mathcal{L}$.
\begin{enumerate}
\item
Compute, exactly or approximately, the leading nontrivial eigenvector of $L$ 
or $\mathcal{L}$.
\item
Use that vector to split the nodes of the graph into a left half and a right 
half.
\end{enumerate}
Those two pieces can be the two clusters, in which case this algorithm is a 
vanilla version of spectral graph partitioning; or with some labels that can 
be used to make predictions for classification or regression; or we can rank 
starting from the left and going to the right; or we can use the details of 
the approximate eigenvector calculation, e.g., random walks and related 
diffusion-based methods, to understand viral diffusion problems.
But in all those cases, we are interested in the leading nontrivial 
eigenvector of the Laplacian.
We'll have a lot more to say about that later, but for now think of it just 
as some vector that in some sense describes important directions in the graph, 
in which case what this vanilla spectral algorithm does is putting or 
``embedding'' the nodes of the graph on this line and cuts the nodes into 
two pieces, a left half and a right half.
(The embedding has big distortion, in general, for some points at least; the 
two halves can be very unbalanced, etc.; but at least in very nice cases, 
that informal intuition is true, and it is true more generally if the two
halves can be unbalanced, etc.  Understanding these issues will be important 
for what we do.)
(Make a picture on the board.)

We can ask: what is the optimization problem that this algorithm solves?
As we will see eventually, in some sense, what makes a spectral graph 
algorithm a spectral graph algorithm is the first step, and so let's focus 
on that.
Here is a basic spectral optimization problem that this problem solves.

\begin{alignat*}{4}
  &\text{min} & x^T  L_{G} x \\
                     &\text{s.t.} & x^T D_{G} x = 1  \\
                     &            & x^T D_{G} \vec{1} = 0  \\
		  &            & x \in \mathbb{R}^V
\end{alignat*}

\noindent
That is, find the vector that minimizes the quadratic form $x^TLx$ subject
to the constraints that $x$ sits on a (degree-weighted) unit ball and that
$x$ is perpendicular (in a degree-weighted norm) to the ``trivial'' all-ones
vector.
The solution to this problem is a vector, and it is the leading nontrivial 
eigenvector of $L$ or $\mathcal{L}$.

Importantly, this is \emph{not} a convex optimization problem; but rather 
remarkably it can be solved in low-degree polytime by computing an 
eigenvector of $L$. 
How can it be that this problem is solvable if it isn't convex?
After all, the usual rule of thumb is that convex things are good and 
non-convex things are bad.
There are two (related, certainly not inconsistent) answers to this.
\begin{itemize}
\item
One reason is that this is an eigenvector (or generalized eigenvector) 
problem.
In fact, it involves computing the leading nontrivial eigenvector of $L$, 
and so it is a particularly nice eigenvalue problem.
And computing eigenvectors is a \emph{relatively} easy thing to do---for 
example, with a black box solver, or in this special case with random walks.
But more on this later.
\item
Another reason is that is is secretly convex, in that it is convex in a 
different place.
Importantly, that different place there are better duality properties for 
this problem, and so it can be used to understand this problem and its 
solution better.
\end{itemize}
Both of these will be important, but let's start by focusing on the second
reason.
Consider the following version of the basic spectral optimization problem.
\begin{eqnarray*}
{\textsf{SDP}:} &\min & L \bullet X  \\
                & \mbox{s.t.} & \Trace{X} = I_0 \bullet X=1  \\
                & & X \succeq 0  ,
\end{eqnarray*}
where $\bullet$ stands for the Trace, or matrix inner product, operation, 
\emph{i.e.}, $A \bullet B = \Trace{A B^T}=\sum_{ij}A_{ij}B_{ij}$ for matrices 
$A$ and $B$.
\emph{Note that, both here and below, $I_0$ is sometimes the Identity on the 
subspace perpendicular to the all-ones vector.  This will be made more 
consistent later.}
\textsf{SDP} is a relaxation of the spectral program \textsf{SPECTRAL} from 
an optimization over unit vectors to an optimization over distributions over 
unit vectors, represented by the density matrix $X$.
But, the optimal values for \textsf{SPECTRAL} and \textsf{SDP} are the same, 
in the sense that they are given by the second eigenvector $v$ of $L$ for 
\textsf{SPECTRAL} and by $X = vv^T$ for \textsf{SDP}.  

Thus, this is an SDP.
While solving the vanilla spectral optimization problem with a black-box SDP
solver is possible, it is not advisable, since one can just can a black-box
eigenvalue solver (or run some non-black-box method that approximates the
eigenvalue).
Nevertheless, the SDP can be used to understand spectral graph methods.
For example, we can consider the dual of this SDP:
\begin{alignat*}{4}
\quad&  &\text{maximize} \quad && \alpha \\
  &  &\text{s.t.} \quad &&  L_{G}  \succeq  \alpha  L_{K_{n}} \\
  & & & & \alpha \in \mathbb{R}
\end{alignat*}
This is a standard dual construction, the only nontrivial thing is we write out explicitly $I_0=L_{K_n}$, as the identity matrix on the subspace perpendicular to $\vec{1}$ is $I_0=I-\vec{1}^T\vec{1}=L_{K_n}$ where $K_n$ is the complete graph on $n$ vertices.

We will get into more detail later what this means, but informally this means 
that we are in some sense ``embedding'' the Laplacian of the graph in the 
Laplacian of a complete graph.
Slightly less informally, the $\succeq$ is an inequality over graphs (we will 
define this in more detail later) which says that the Laplacian quadratic 
form of one graph is above or below that of another graph (in the sense of 
SPSD matrices, if you know what that means).
So, in this dual, we want to choose the largest $\alpha$ such that that 
inequality is true.

\subsection{Connections with random walks and diffusions}

A final thing to note is that the vector $x^*$ that solves these problems
has a natural interpretation in terms of diffusions and random walks.
This shouldn't be surprising, since one of the ways this vector is to 
partition a graph into two pieces that captures a qualitative notion of 
connectivity.
The interpretation is that is you run a random walk---either a vanilla 
random walk defined by the matrix $W=D^{-1}A$ above, meaning that at each
step you go to one of your neighbors with equal probability, or a fancier
random walk---then $x^{*}$ defines the slowest direction to mixing, i.e., 
the direction that is at the pre-asymptotic state before you get to the 
asymptotic uniform distribution.

So, that spectral graph methods are useful in these and other applications is 
largely due to two things.
\begin{itemize}
\item
Eigenvectors tend to be ``global'' things, in that they optimize a global 
objective over the entire graph.
\item
Random walks and diffusions optimize almost the same things, but they often 
do it in a very different way.
\end{itemize}
One of the themes of what we will discuss is the connection between random 
walks and diffusions and eigenvector-based spectral graph methods on 
different types of graph-based data.
Among other things, this will help us to address local-global issues, e.g., 
the global objective that defines eigenvectors versus the local nature of 
diffusion updates.

Two things should be noted about diffusions.
\begin{itemize}
\item
Diffusions are robust/regularized notions of eigenvectors
\item
The behavior of diffusions is very different on $K_n$ or expander-like metric
spaces than it is on line-like or low-dimensional metric spaces.
\end{itemize}
An important subtlety is that most data have some sort of degree heterogeneity, 
and so the extremal properties of expanders are mitigated since it is 
constant-degree expanders that are most unlike low-dimensional metric 
spaces.
(In the limit, you have a star, and there it is trivial why you don't have good
partitions, but we don't want to go to that limit.)

\subsection{Low-dimensional and non-low-dimensional data}

Now, complete graphs are very different than line graphs.
So, the vanilla spectral graph algorithm is also putting the data in a 
complete graph.
To get a bit more intuition as to what is going on and to how these methods
will perform in real applications, consider a different type of graph known
as an expander.
I won't give the exact definition now---we will later---but expanders are 
very important, both in theory and in practice.
(The former may be obvious, while the latter may be less obvious.)
For now, there are three things you need to know about expanders, either 
constant-degree expanders and/or degree-heterogeneous expanders.
\begin{itemize}
\item
Expanders are extremely sparse graphs that do not have any good 
clusters/partitions, in a precise sense that we will define later.
\item
Expanders are also the metric spaces that are least like low-dimensional 
spaces, e.g., a line graph, a two-dimensional grid, etc.
That is, if your intuition comes from low-dimensional places like 1D or 2D
places, then expanders are metric spaces that are most different than that.
\item
Expanders are sparse versions of the complete graph, in the sense that 
there are graph inequalities of the form $\succeq$ that relate the Laplacian
quadratic forms of expanders and complete graphs.
\end{itemize}
So, in a certain precise sense, the vanilla spectral graph method above (as
well as other non-vanilla spectral methods we will get to) put or embed the 
input data in two extremely different places, a line as well as a dense 
expander, i.e., a complete graph.

Now, real data have low-dimensional properties, e.g., sometimes you can 
visualize them in a two-dimensional piece of paper and see something 
meaningful, and since they often have noise, they also have expander-like 
properties.
(If that connection isn't obvious, it soon will be.)
We will see that the properties of spectral graph methods when applied to real 
data sometimes depend on one interpretation and sometimes depend on the other 
interpretation.
Indeed, many of the properties---both the good properties as well as the 
bugs/features---of spectral graph methods can be understood in light of 
this tension between embedding the data in a low-dimensional place and 
embedding the data in an expander-like place.

There are some similarities between this---which is a statement about 
different types of graphs and metric spaces---and analogous statements about 
random vectors in $\mathbb{R}^{n}$, e.g., from a full-dimensional Gaussian 
distribution in $\mathbb{R}^{n}$.
Some of these will be explored.

\subsection{Small world and heavy-tailed examples}

There are several types or classes of generative models that people 
consider, and different communities tend to adopt one or the other class.
Spectral graph methods are applied to all of these, although they can be
applied in somewhat different ways.
\begin{itemize}
\item
Discretization or random geometric graph of some continuous low-dimensional 
place, e.g., a linear low-dimensional space, a low-dimensional curved 
manifold, etc.
In this case, there is a natural low-dimensional geometry.
(Put picture on board.)
\item
Stochastic blockmodels, where there are several different types of 
individuals, and each type interacts with individuals in the same group 
versus different groups with different probabilities.
(Put picture on board, with different connection probabilities.)
\item
Small-world and heavy-tailed models.
These are generative graph models, and they attempt to capture some local 
aspect of the data (in one case, that there is some low-dimensional 
geometry, and in the other case, that there is big variability in the local 
neighborhoods of individuals, as captured by degree or some other simple 
statistic) and some global aspect to the data (typically that there is a
small diameter).
\end{itemize}

We will talk about all three of these in due course, but for now let's say
just a bit about the small-world models and what spectral methods might 
reveal about them in light of the above comments.
Small-world models start with a one-dimensional or two-dimensional geometry 
and add random edges in one of several ways.
(Put picture here.)
The idea here is that you reproduce local clustering and small diameters, 
which is a property that is observed empirically in many real networks.
Importantly, for algorithm and statistical design, we have intuition about 
low-dimensional geometries; so let's talk about the second part: random graphs.

Consider $G_{np}$ and $G_{nm}$, which are the simplest random graph models,
and which have an expected or an exact number of edges, respectively.
In particular, start with $n$ isolated vertices/nodes; then:
\begin{itemize}
\item
For $G_{np}$,  insert each of the ${n \choose 2}$ possible edges, 
independently, each with probability $p$.
\item
For $G_{nm}$, among all 
${ { n \choose 2 } \choose m}$
subsets of $m$ edges, select one, independently at random.
\end{itemize}
In addition to being of theoretical interest, these models are used in all
sorts of places.
For example, they are the building blocks of stochastic block models, in 
addition to providing the foundation for common generative network models.
(That is, there are heavy-tailed versions of this basic model as well as other 
extensions, some of which we will consider, but for now let's stick with 
this.)
These vanilla ER graphs are often presented as strawmen, which in some sense
they are; but when taken with a grain of salt they can reveal a lot about 
data and algorithms and the relationship between the two.

First, let's focus on $G_{np}$. 
There are four regimes of particular interest.
\begin{itemize}
\item
$p < \frac{1}{n}$.
Here, the graph $G$ is not fully-connected, and it doesn't even have a 
giant component, so it consists of just a bunch of small things.
\item
$ \frac{1}{n} \lesssim p \lesssim \frac{\log(n)}{n}$.
Here there is a giant component, i.e., set of $\Omega(n)$ nodes that 
are connected, that has a small $O(\log(n))$ diameter.
In addition, random walks mix in $O(\log^2(n))$ steps, and the graph is 
locally tree-like.
\item
$\frac{\log(n)}{n} \lesssim p$.
Here the graph is fully-connected.
In addition, it has a small $O(\log(n))$ diameter and random walks mix in 
$O(\log^2(n))$ steps (but for a slightly different reason that we will get 
back to later).
\item
$\frac{\log(n)}{n} \ll p$.
Here the graph is pretty dense, and methods that are applicable to pretty 
dense graphs are appropriate.
\end{itemize}
If $p \gtrsim \log(n)/n$, then $G_{np}$ and $G_{nm}$ are ``equivalent'' in a 
certain sense.  
But if they are extremely sparse, e.g., $p=3/n$ or $p=10/n$, and the 
corresponding values of $m$, then they are different.

In particular, if $p=3/n$, then the graph is not fully-connected, but we can 
ask for a random $r$-regular graph, where $r$ is some fixed small integer.
That is, fix the number of edges to be $r$, so we have a total of $nr$ edges
which is \emph{almost} a member of $G_{nm}$, and look at a random such graph.
\begin{itemize}
\item
A random $1$-regular graph is a matching.
\item
A random $2$-regular graph is a disjoint union of cycles.
\item
A random $3$-regular graph: 
it is fully-connected and had a small $O(\log(n))$ diameter; 
it is an expander;
it contains a perfect matching (a matching, i.e., a set of pairwise 
non-adjacent edges that matches all vertices of the graph) and a Hamiltonian 
cycle (a closed loop that visits each vertex exactly once).
\item
A random $4$-regular graph is more complicated to analyze
\end{itemize}

How does this relate to small world models?
Well, let's start with a ring graph (a very simple version of a 
low-dimensional lattice, in which each node is connected to neighbors 
within a distance $k=1$) and add a matching (which is a bunch of random 
edges in a nice analyzable way).
Recall that a random $3$-regular graph has both a Hamiltonian cycle and
a perfect matching; well it's also the case that the union of an 
$n$ cycle and a random machine is contiguous to a random $3$ regular 
random graphs.
(This is a type of graph decomposition we won't go into.)

This is a particular theoretical form to say that small world models have a 
local geometry but globally are also expanders in a strong sense of the word.
Thus, in particular, when one runs things like diffusions on them, or 
relatedly when one runs spectral graph algorithms (which have strong 
connections under the hood to diffusions) on them, what one gets will depend 
sensitively on the interplay between the line/low-dimensional properties and 
the noise/expander-like properties.

It is well known that similar results hold for heavy-tailed network models 
such as PA models or PLRG models or many real-world networks.
There there is degree heterogeneity, and this can give a lack of measure 
concentration that is analogous to the extremely sparse Erdos-Renyi graphs, 
unless one does things like make minimum degree assumptions.
It it less well known that similar things also hold for various types of
constructed graphs.
Clearly, this might happen if one constructs stochastic blockmodels, since
then each piece is a random graph and we are interested in the interactions
between different pieces.
But what if construct a manifold method, but there is a bit of noise?
This is an empirical question; but noise, if it is noise, can be thought of 
as a random process, and in the same way as the low-dimensional geometry of 
the vanilla small world model is not too robust to adding noise, similarly 
geometric manifold-based methods are also not too robust.

In all of this, there algorithm questions, as well as statistical and 
machine learning questions such as model selection questions and questions 
about how to do inference with vector-based of graph-based data, as well as
mathematical questions, as well as questions about how these methods perform
in practice.
We will revisit many of these over the course of the semester.

\subsection{Outline of class}

In light of all that, here is an outline of some representative topics that
we will cover.
\begin{enumerate}
\item
Basics of graph partitioning, including spectral, flow, etc., degree 
heterogeneity, and other related objectives.
\item
Connections with diffusions and random walks, including connections with 
resistor network, diffusion-based distances, expanders, etc.
\item
Clustering, prediction, ranking/centrality, and communities, i.e., solving 
a range of statistics and data analysis methods with variants of spectral 
methods.
\item
Graph construction and empirical properties, i.e., different ways graphs 
can be constructed and empirical aspects of ``given'' and ``constructed'' 
graphs.
\item
Machine learning and statistical approaches and uses, e.g., stochastic 
blockmodels, manifold methods, regularized Laplacian methods, etc.
\item
Computations, e.g., nearly linear time Laplacian solvers and graph 
algorithms in the language of linear algebra.
\end{enumerate}

%% file: lect02.tex
\section{%
(01/27/2015): 
Basic Matrix Results (1 of 3)}

Reading for today.
\begin{compactitem}
\item
``Lecture notes,'' from Spielman's Spectral Graph Theory class, Fall 2009 and 2012
\end{compactitem}

\subsection{Introduction}

Today and next time, we will start with some basic results about matrices, 
and in particular the eigenvalues and eigenvectors of matrices, that will 
underlie a lot of what we will do in this class.
The context is that eigenvalues and eigenvectors are complex (no pun intended, 
but true nonetheless) things and---in general---in many ways not so ``nice.''
For example, they can change arbitrarily as the coefficients of the matrix 
change, they may or may not exist, real matrices may have complex 
eigenvectors and eigenvalues, a matrix may or may not have a full set of $n$ 
eigenvectors, etc.
Given those and related instabilities, it is an initial challenge is to 
understand what we can determine from the spectra of a matrix.
As it turns out, for many matrices, and in particular many matrices that 
underlie spectral graph methods, the situation is much nicer; and, in 
addition, in some cases they can be related to even nicer things like random 
walks and diffusions.

So, let's start by explaining ``why'' this is the case.
To do so, let's get some context for how/why matrices that are useful for 
spectral graph methods are nicer and also how these nicer matrices sit in 
the larger universe of arbitrary matrices. 
This will involve establishing a few basic linear algebraic results; then we
will use them to form a basis for a lot of the rest of what we will discuss.
This is good to know in general; but it is also good to know for more 
practical reasons.
For example, it will help clarify when vanilla spectral graph methods can be 
extended, \emph{e.g.}, to weighted graphs or directed graphs or time-varying 
graph or other types of normalizations, etc.

\subsection{Some basics}

To start, recall that we are interested in the Adjacency matrix of a graph 
$G=(V,E)$ (or $G=(V,E,W)$ if the graph is weighted) and other matrices that 
are related to the Adjacency matrix.
Recall that the $n \times n$ Adjacency matrix is defined to be
\[ 
A_{ij} = \left\{ \begin{array}{l l}
                    W_{ij} & \quad \text{if $(ij)\in E$}\\
                    0      & \quad \text{otherwise}
                 \end{array} 
         \right.  ,
\]
where $W_{ij}=1$, for all $(i,j)\in E$ if the graph is unweighted.
Later, we will talk about directed graphs, in which case the Adjacency matrix 
is not symmetric, but note here it is symmetric.
So, let's talk about symmetric matrices:
a symmetric matrix is a matrix $A$ for which $A=A^T$, i.e., for which 
$A_{ij}=A_{ji}$.

Almost all of what we will talk about will be real-valued matrices.
But, for a moment, we will start with complex-valued matrices.
To do so, recall that if $x=\alpha+i\beta \in \mathbb{C}$ is a complex 
number, then $\bar{x}=\alpha-i\beta \in \mathbb{C}$ is the \emph{complex 
conjugate} of $x$.
Then, if $M\in\mathbb{C}^{m \times n}$ is a complex-valued matrix, i.e., 
an $m \times n$ matrix each entry of which is a complex number, then the 
\emph{conjugate transpose} of $M$, which is denoted $M^{*}$, is the matrix 
defined as
\[
\left(M^{*}\right)_{ij} = \bar{M}_{ji} .
\]
Note that if $M$ happens to be a real-valued $m \times n$ matrix, then this 
is just the transpose.

If $x,y \in \mathbb{C}^{n}$ are two complex-valued vectors, then we can define
their inner product to be 
\[
\left<x,y\right> = x^{*}y = \sum_{i=1}^{n} \bar{x}_i y_i  .
\]
Note that from this we can also get a norm in the usual way, i.e., 
$\left<x,x\right> = \|x\|_2^2 \in \mathbb{R}$.
Given all this, we have the following definition.

\begin{definition}
If $M\in\mathbb{C}^{n \times n}$ is a square complex matrix, 
 $\lambda \in \mathbb{C}$ is a scalar, and
 $x \in \mathbb{C}^{n} \backslash \{ 0 \}$ is a non-zero vector such that 
\begin{equation}
\label{eqn:eigensystem1}
M x = \lambda x
\end{equation}
then $\lambda$ is an \emph{eigenvalue} of $M$ and
     $x$ is the corresponding \emph{eigenvector} of $\lambda$.
\end{definition}

\noindent
Note that when Eqn.~(\ref{eqn:eigensystem1}) is satisfied, then this is equivalent
to 
\begin{equation}\label{eq:det}
\left( M - \lambda I \right)x = 0 , \mbox{ for } x \ne 0  ,
\end{equation}
where $I$ is an $n \times n$ Identity matrix.
In particular, this means that we have at least one eigenvalue/eigenvector pair.
Since ~\eqref{eq:det} means $M-\lambda I$ is rank deficient, this in turn is equivalent to 
\[
\mbox{det}\left( M - \lambda I \right) = 0  .
\]
Note that this latter expression is a polynomial with $\lambda$ as the variable. 
That is, if we fix $M$, then the function given by
$\lambda \rightarrow \mbox{det}\left( M - \lambda I \right)$ 
is a univariate polynomial of degree $n$ in $\lambda$.
Now, it is a basic fact that every non-zero, single-variable, degree 
polynomial of degree $n$ with complex coefficients has---counted with 
multiplicity---exactly $n$ roots.
(This counting multiplicity thing might seem pedantic, but it will be 
important latter, since this will correspond to potentially degenerate
eigenvalues, and we will be interested in how the corresponding 
eigenvectors behave.)
In particular, any square complex matrix $M$ has $n$ eigenvectors, counting 
multiplicities, and there is at least one eigenvalue.

As an aside, someone asked in class if this fact about complex polynomials
having $n$ complex roots is obvious or intuitive.
It is sufficiently basic/important to be given the name the fundamental theorem 
of algebra, but its proof isn't immediate or trivial.
We can provide some intuition though.
Note that related formulations of this state that every non-constant 
single-variable polynomial with complex coefficients has at least one complex 
root, etc. (e.g., complex roots come in pairs); and that the field of complex 
numbers is algebraically closed.
In particular, the statements about having complex roots applies to 
real-valued polynomials, i.e., since real numbers are complex numbers
polynomials in them have complex roots; but it is false that real-valued 
polynomials always have real roots.
Equivalently, the real numbers are not algebraically closed.
To see this, recall that the equation $x^2-1=0$, viewed as an equation over 
the reals has two real roots, $x=\pm 1$; but the equation $x^2+1=0$ does 
not have any real roots.
Both of these equations have roots over the complex plane: the former 
having the real roots $x=\pm 1$, and the latter having imaginary roots 
$x=\pm i$.

\subsection{Two results for Hermitian/symmetric matrices}

Now, let's define a special class of matrices that we already mentioned.
\begin{definition}
A matrix $M\in\mathbb{C}^{n \times n}$ is \emph{Hermitian} if $M=M^{*}$. 
In addition, a matrix $M\in\mathbb{R}^{n \times n}$ is \emph{symmetric} 
if $M=M^{*}=M^{T}$.
\end{definition}

For complex-valued Hermitian matrices, we can prove the following two 
lemmas.

\begin{lemma}
Let $M$ be a Hermitian matrix.
Then, all of the eigenvalues of $M$ are real.
\end{lemma}
\begin{Proof}
Let $M$ be Hermitian and $\lambda\in\mathcal{C}$ and $x$ non-zero be s.t.
$Mx = \lambda x$.
Then it suffices to show that $\lambda=\lambda^{*}$, since that means 
that $\lambda\in\mathbb{R}$.
To see this, observe that
\begin{eqnarray}
\nonumber
\left<Mx,x\right> &=& \sum_i \sum_j \bar{M_{ij}}\bar{x_j} x_i \\
\label{eqn:herm1}
                  &=& \sum_i \sum_j M_{ji}     x_i \bar{x_j} \\
\nonumber
                  &=& \left< x,Mx\right>
\end{eqnarray}
where Eqn.~(\ref{eqn:herm1}) follows since $M$ is Hermitian.
But we have
\[
\left<Mx,x\right> = \left< \lambda x, x\right> = \bar{\lambda} \left<x,x\right> = \bar{\lambda} \|x\|_2^2
\] 
and also that 
\[
\left<x,Mx\right> = \left< x, \lambda x\right> = \lambda     \left<x,x\right> = \lambda   \|x\|_2^2  .
\] 
Thus, $\lambda = \bar{\lambda}$, and the lemma follows.
\end{Proof}

\begin{lemma}
Let $M$ be a Hermitian matrix; and 
let $x$ and $y$ be eigenvectors corresponding to different eigenvalues.
Then $x$ and $y$ are orthogonal.
\end{lemma}
\begin{Proof}
Let $Mx = \lambda x$ and $My = \lambda^{\prime} y$.
Then, 
\[
\left<Mx,y\right> = \left(Mx\right)^{*}y = x^{*}M^{*}y = x^{*}My = \left< x, My\right>  .
\] 
But, 
\[
\left<Mx,y\right> = \lambda \left<x,y\right>
\] 
and 
\[
\left<x,My\right> = \lambda^{\prime} \left<x,y\right> .
\] 
Thus
\[
\left(\lambda-\lambda^{\prime}\right) \left<x,y\right> = 0 .
\]
Since $\lambda \ne \lambda^{\prime}$, by assumption, it follows that 
$\left<x,y\right> = 0$, from which the lemma follows.
\end{Proof}

So, Hermitian and in particular real symmetric matrices have real eigenvalues
and the eigenvectors corresponding to to different eigenvalues are orthogonal.
We won't talk about complex numbers and complex matrices for the rest of the 
term.
(Actually, with one exception since we need to establish that the entries of the
eigenvectors are not complex-valued.)

\subsection{Consequences of these two results}

So far, we haven't said anything about a full set of orthogonal 
eigenvectors, etc., since, e.g., all of the eigenvectors could be the same 
or something funny like that.
In fact, we will give a few counterexamples to show how the niceness results
we establish in this class and the next class fail to hold for general 
matrices.
Far from being pathologies, these examples will point to interesting ways 
that spectral methods and/or variants of spectral method ideas do or do not 
work more generally (e.g., periodicity, irreducibility etc.)

Now, let's restrict ourselves to real-valued matrices, in which case Hermitian
matrices are just symmetric matrices.
With the exception of some results next time on positive and non-negative 
matrices, where we will consider complex-valued things, the rest of the 
semester will consider real-valued matrices.
Today and next time, we are only talking about complex-valued matrices to 
set the results that underlie spectral methods in a more general context.
So, let's specialize to real-values matrices.

First, let's use the above results to show that we can get a full set of 
(orthogonalizable) eigenvectors.
This is a strong ``niceness'' result, for two reasons:
(1) there is a full set of eigenvectors; and 
(2) that the full set of eigenvectors can be chosen to be orthogonal.
Of course, you can always get a full set of orthogonal vectors for 
$\mathbb{R}^{n}$---just work with the canonical vectors or some other set of
vectors like that.
But what these results say is that for symmetric matrices we can also get a 
full set of orthogonal vectors that in some sense have something to do with 
the symmetric matrix under consideration.
Clearly, this could be of interest if we want to work with vectors/functions
that are in some sense adapted to the data.

Let's start with the following result, which says that given several (i.e., 
at least one) eigenvector, then we can find another eigenvector that is 
orthogonal to it/them.
Note that the existence of at least one eigenvector follows from the 
existence of at least one eigenvalue, which we already established.

\begin{lemma}
Let $M\in\mathbb{R}^{n \times n}$ be a real symmetric matrix, and let 
$x_1,\ldots,x_k$, where $1 \le k < n$, be orthogonal eigenvectors of $M$.
Then, there is an eigenvector $x_{k+1}$ of $M$ that is orthogonal to 
$x_1,\ldots,x_k$.
\end{lemma}
\begin{Proof}
Let $V$ be the $(n-k)$-dimensional subspace of $\mathbb{R}^{n}$ that 
contains all vectors orthogonal to $x_1,\ldots,x_k$.
Then, we claim that: for all $x \in V$, we have that $Mx \in V$.
To prove the claim, note that for all $i\in[k]$, we have that
\[
\left<x_i,Mx\right> 
 = x_i^TMx 
 = \left(Mx_i\right)^Tx 
 = \lambda_i x_i x 
 = \lambda_i \left<x_i,x\right>
 = 0   ,
\]
where $x_i$ is one of the eigenvectors assumed to be given.

Next, let
\begin{itemize}
\item 
$B \in \mathbb{R}^{n \times (n-k)}$ be a matrix consisting of the vectors 
$b_1,\ldots,b_{n-k}$ that form an orthonormal basis for $V$.
(This takes advantage of the fact that $\mathbb{R}^{n}$ has a full set of 
exactly $n$ orthogonal vectors that span it---that are, of course, not 
necessarily eigenvectors.)
\item
$B^{\prime} = B^T$.
(If $B$ is any matrix, then $B^{\prime}$ is a matrix such that, for all 
$y \in V$, we have that $B^{\prime}y$ is an $(n-k)$-dimensional vector 
such that $BB^{\prime}y=y$.
I \emph{think} we don't loose any generality by taking $B$ to be orthogonal.)
\item
$\lambda$ be a real eigenvalue of the real symmetric matrix 
\[
M^{\prime} = B^{\prime} M B \in \mathbb{R}^{(n-k)\times(n-k)} ,
\]
with $y$ a corresponding real eigenvector of $M$.
I.e., $M^{\prime}y=\lambda y$.
\end{itemize}
Then, 
\[
B^{\prime} M B y = \lambda y  ,
\]
and so
\[
B B^{\prime} M B y = \lambda B y  ,
\]
from which if follows that 
\[
M B y = \lambda B y .
\]
The last equation follows from the second-to-last since 
$By \perp \{ x_1,\ldots,x_k\}$, from which it follows that
$MBy \perp \{ x_1,\ldots,x_k\}$, by the above claim, and thus
$BB^{\prime} MBy = MBy$. 
I.e., this doesn't change anything since $BB^{\prime} \xi = \xi$, for 
$\xi$ in that space.

So, we can now construct that eigenvector.
In particular, we can choose $x_{k+1} = By$, and we have that 
$Mx_{k+1} = \lambda x_{k+1}$, from which the lemma follows.
\end{Proof}

Clearly, we can apply the above lemma multiple times.
Thus, as an important aside, the following ``spectral theorem'' is basically 
a corollary of the above lemma.

\begin{theorem}[Spectral Theorem]
Let $M\in\mathbb{R}^{n \times n}$ be a real symmetric matrix, and let 
$\lambda_1,\ldots,\lambda_n$ be its real eigenvalues, including 
multiplicities.
Then, there are $n$ orthonormal vectors $x_1,\ldots,x_n$, with 
$x_i\in\mathbb{R}^{n}$, such that $x_i$ is an eigenvector corresponding 
to $\lambda_i$, i.e., $M x_i = \lambda_i x_i$.
\end{theorem}

A few comments about this spectral theorem.
\begin{itemize}
\item
This theorem and theorems like this are very important and many 
generalizations and variations of it exist.
\item
Note the wording: there are $n$ vectors ``such that $x_i$ is an eigenvector 
corresponding to $\lambda_i$.''
In particular, there is no claim (yet) about uniqueness, etc.
We still have to be careful about that.
\item
From this we can derive several other things, some of which we will mention
below.
\end{itemize}

\noindent
Someone asked in class about the connection with the SVD.
The equations $M x_i = \lambda_i x_i$, for all $\lambda_i$, can be written
as $MX = X \Lambda$, or as $M = X \Lambda X^T$, since $X$ is orthogonal.
The SVD writes an arbitrary $m \times n$ matrix $A$ a $A=U \Sigma V^T$, 
where $U$ and $V$ are orthogonal and $\Sigma$ is diagonal and non-negative.
So, the SVD is a generalization or variant of this spectral theorem for 
real-valued square matrices to general $m \times n$ matrices.
It is \emph{not} true, however, that the SVD of even a symmetric matrix gives
the above theorem.
It is true by the above theorem that you can write a symmetric matrix as
$M = X \Lambda X^T$, where the eigenvectors $\Lambda$ are real. 
But they might be negative.
For those matrices, you also have the SVD, but there is no immediate 
connection.
On the other hand, some matrices have all $\Lambda$ positive/nonnegative.
They are called SPD/SPSD matrices, and form them the eigenvalue decomposition
of the above theorem essentially gives the SVD.
(In fact, this is sometimes how the SVD is proven---take a matrix $A$ and 
write the eigenvalue decomposition of the SPSD matrices $AA^T$ and $A^TA$.)
SPD/SPSD matrices are important, since they are basically covariance or 
correlation matrices; and several matrices we will encounter, e.g., 
Laplacian matrices, are SPD/SPSD matrices.

We can use the above lemma to provide the following variational 
characterization of eigenvalues, which will be very important for us.

\begin{theorem}[Variational Characterization of Eigenvalues]
Let $M\in\mathbb{R}^{n \times n}$ be a real symmetric matrix; 
let $\lambda_l \le \cdots \le \lambda_n$ be its real eigenvalues, containing
multiplicity and sorted in nondecreasing order; and let $x_1,\ldots,x_k$, 
for $k < n$ be orthonormal vectors such that $M x_i = \lambda_i x_i$, for 
$i\in[k]$.
Then
\[
\lambda_{k+1} 
   = \min_{ \substack{x\in\mathbb{R}^{n}\diagdown\{\vec{0}\}\\x \perp x_i \quad \forall i\in[k]}  } 
     \frac{x^TMx}{x^Tx} ,
\]
and any minimizer of this is an eigenvector of $\lambda_{k+1}$.
\end{theorem}
\begin{Proof}
First, by repeatedly applying the above lemma, then we get $n-k$ orthogonal 
eigenvectors that are also orthogonal to $x_1,\ldots,x_k$.
Next, we claim that the eigenvalues of this system of $n$ orthogonal 
eigenvectors include all eigenvalues of $M$.
The proof is that if there were any other eigenvalues, then its eigenvector 
would be orthogonal to the other $n$ eigenvectors, which isn't possible, 
since we already have $n$ orthogonal vectors in $\mathbb{R}^{n}$.

Call the additional $n-k$ vectors $x_{k+1},\ldots,x_n$, where $x_i$ is an
eigenvector of $\lambda_i$.
(Note that we are inconsistent on whether subscripts mean elements of a 
vectors or different vectors themselves; but it should be clear from context.)
Then, consider the minimization problem
\[
\min_{ \substack{x\in\mathbb{R}^{n}\diagdown\{\vec{0}\}\\x \perp x_i \quad \forall i\in[k]}  } 
     \frac{x^TMx}{x^Tx}
\]
The solution $x \equiv x_{k+1}$ is feasible, and it has cost $\lambda_{k+1}$, 
and so $\mbox{min} \le \lambda_{k+1}$.

Now, consider any arbitrary feasible solution $x$, and write it as 
\[
x = \sum_{i={k+1}}^{n} \alpha_i x_i  .
\]
The cost of this solution is 
\[
\frac{\sum_{i=k+1}^{n}\lambda_i\alpha_i^2}{\sum_{i=k+1}^{n}\alpha_i^2}
\ge \lambda_{k+1} \frac{\sum_{i=k+1}^{n}\alpha_i^2}{\sum_{i=k+1}^{n}\alpha_i^2}
  = \lambda_{k+1}  , 
\]
and so $\mbox{min} \ge \lambda_{k+1}$.
By combining the above, we have that $\mbox{min} = \lambda_{k+1}$.

Note that is $x$ is a minimizer of this expression, i.e., if the cost of $x$
equals $\lambda_{k+1}$, then $a_i=0$ for all $i$ such that 
$\lambda_i > \lambda_{k+1}$, and so $x$ is a linear combination of 
eigenvectors of $\lambda_{k+1}$, and so it itself is an eigenvector of 
$\lambda_{k+1}$.
\end{Proof}

Two special cases of the above theorem are worth mentioning.
\begin{itemize}
\item
The leading eigenvector.
\[
\lambda_{1} 
   = \min_{ x\in\mathbb{R}^{n}\diagdown\{\vec{0}\} }
     \frac{x^TMx}{x^Tx} 
\]
\item
The next eigenvector.
\[
\lambda_{2} 
   = \min_{ x\in\mathbb{R}^{n}\diagdown\{\vec{0}\},x \perp x_1 }
     \frac{x^TMx}{x^Tx} ,
\]
where $x_1$ is a minimizer of the previous expression.
\end{itemize}

\subsection{Some things that were skipped}

Spielman and Trevisan
give two slightly different versions of the variational 
characterization and Courant-Fischer theorem, i.e., a min-max result, which 
might be of interest to present.

From wikipedia, there is the following discussion of the min-max theorem 
which is nice.
\begin{itemize}
\item
Let $A \in \mathbb{R}^{n \times n}$ be a Hermitian/symmetric matrix, then the
Rayleigh quotient $R_A:\mathbb{R}^{n}\diagdown\{0\}\rightarrow\mathbb{R}$ 
is $R_A(x) = \frac{\left<Ax,x\right>}{\left<x,x\right>}$, or equivalently
$f_A(x) = \left<Ax,x\right> : \|x\|_2=1$.
\item
Fact: for Hermitian matrices, the range of the continuous function $R_A(x)$ 
or $f_A(x)$ is a compact subset $[a,b]$ of $\mathbb{R}$.
The max $b$ and min $a$ are also the largest and smallest eigenvalue of $A$, 
respectively.
The max-min theorem can be viewed as a refinement of this fact.
\item
\begin{theorem}
If $A\in\mathbb{R}^{n \times n}$ is Hermitian with eigenvalues 
$\lambda_1 \ge \dots \ge \lambda_k \ge \cdots$, then
\[
\lambda_k = \max\{ \min\{ R_A(x) : x \in U, x \ne 0 \} , \mbox{dim}(U) = k \} ,
\]
and also 
\[
\lambda_k = \min\{ \max\{ R_A(x) : x \in U, x \ne 0 \} , \mbox{dim}(U) = n-k+1 \} .
\]
\end{theorem}
\item
In particular, 
\[
\lambda_n \le R_A(x) \le \lambda_1 ,
\]
for all $x\in\mathbb{R}^{n}\diagdown\{0\}$.
\item
A simpler formulation for the max and min is
\begin{eqnarray*}
\lambda_1 &=& \max \{ R_A(x) : x \ne 0 \} \\
\lambda_n &=& \min \{ R_A(x) : x \ne 0 \} 
\end{eqnarray*}
\end{itemize}

Another thing that follows from the min-max theorem is the Cauchy Interlacing
Theorem.
See Spielman's 9/16/09 notes and Wikipedia for two different forms of this.
This can be used to control eigenvalues as you make changes to the matrix.
It is useful, and we may revisit this later.

And, finally, here is counterexample to these results in general.
Lest one thinks that these niceness results always hold, here is a simple 
non-symmetric matrix.
\[
A = \left( \begin{array}{cc} 0 & 1 \\
                             0 & 0 
           \end{array} 
    \right)
\]

(This is an example of a nilpotent matrix.)
\begin{definition}
A \emph{nilpotent matrix} is a square matrix $A$ such that $A^k=0$ for some 
$k \in \mathbb{Z}^{+}$.
\end{definition}
More generally, any triangular matrix with all zeros on the diagonal; but it
could also be a dense matrix.)

For this matrix $A$, we can define $R_A(x)$ as with the Rayleigh quotient.
Then, 
\begin{itemize}
\item
The only eigenvalue of $A$ equals $0$.
\item
The maximum value of $R_A(x)$ is equal to $\frac{1}{2}$, which is larger 
that $0$.
\end{itemize}
So, in particular, the Rayleigh quotient doesn't say much about the spectrum.

\subsection{Summary} 

Today we showed that any symmetric matrix (e.g., adjacency matrix $A$ of an 
undirected graph, Laplacian matrix, but more generally) is nice in that it 
has a full set of $n$ real eigenvalues and a full set of $n$ orthonormal
eigenvectors.

Next time, we will ask what those eigenvectors look like, since spectral 
methods make crucial use of that.
To do so, we will consider a different class of matrices, namely positive
or nonnegative (not PSD or SPSD, but element-wise positive or nonnegative)
and we will look at the extremal, i.e., top or bottom, eigenvectors.

%% file: lect03.tex
\section{%
(01/29/2015): 
Basic Matrix Results (2 of 3)}

Reading for today.
\begin{compactitem}
\item
Same as last class.
\end{compactitem}

\subsection{Review and overview}

Last time, we considered symmetric matrices, and we showed that is $M$ is an 
$n \times n$ real-valued matrix, then the following hold.
\begin{itemize}
\item
There are $n$ eigenvalues, counting multiplicity, that are all real.
\item
The eigenvectors corresponding to different eigenvalues are orthogonal.
\item
Given $k$ orthogonal eigenvectors, we can construct one more that is 
orthogonal to those $k$, and thus we can iterate this process to get a full
set of $n$ orthogonal eigenvectors
\item
This spectral theorem leads to a variational characterization of 
eigenvalues/eigenvectors and other useful characterizations.
\end{itemize}
These results say that symmetric matrices have several ``nice'' properties, 
and we will see that spectral methods will use these extensively.

Today, we will consider a different class of matrices and establish a 
different type of ``niceness'' result, which will also be used extensively 
by spectral methods.
In particular, we want to say something about how eigenvectors, and in 
particular the extremal eigenvectors, e.g., the largest one or few or the 
smallest one of few ``look like.''
The reason is that spectral methods---both vanilla and non-vanilla 
variants---will rely crucially on this; thus, understanding when and why 
this is true will be helpful to see how spectral methods sit with respect 
to other types of methods, to understand when they can be generalized, or 
not, and so on.

The class of matrices we will consider are \emph{positive matrices} as well 
as related \emph{non-negative matrices}.
By positive/non-negative, we mean that this holds element-wise.
Matrices of this form could be, e.g., the symmetric adjacency matrix of an 
undirected graph, but they could also be the non-symmetric adjacency matrix 
of a directed graph.
(In the latter case, of course, it is not a symmetric matrix, and so the 
results of the last class don't apply directly.)
In addition, the undirected/directed graphs could be weighted, assuming in 
both cases that weights are non-negative.
In addition, it could apply more generally to any positive/non-negative 
matrix (although, in fact, we will be able to take a positive/non-negative 
matrix and interpret it as the adjacency matrix of a graph).
The main theory that is used to make statements in this context and that we
will discuss today and next time is something called Perron-Frobenius theory.

\subsection{Some initial examples}

Perron-Frobenius theory deals with positive/non-negative vectors and 
matrices, i.e., vectors and matrices that are entry-wise positive/nonnegative.
Before proceeding with the main results of Perron-Frobenius theory, let us
see a few examples of why it might be of interest and when it doesn't hold.

\textbf{Example.}
Non-symmetric and not non-negative matrix.
Let's start with the following matrix, which is neither positive/non-negative 
nor symmetric.
\[
A = \left( \begin{array}{cc} 0 & -1 \\
                             2 &  3 
           \end{array} 
    \right)  .
\]
The characteristic polynomial of this matrix is 
\begin{eqnarray*}
\mbox{det}\left( A - \lambda I \right) 
   &=& \left| \begin{array}{cc} -\lambda & -1         \\
                                       2 &  3-\lambda
              \end{array}
       \right|  \\
   &=& -\lambda(3-\lambda)+2 \\
   &=& \lambda^2-3\lambda+2  \\
   &=& \left(\lambda-1\right)\left(\lambda-2\right)  ,
\end{eqnarray*}
from which if follows that the eigenvalues are $1$ and $2$.
Plugging in $\lambda=1$, we get $x_1+x_2=0$, and so the eigenvector 
corresponding to $\lambda=1$ is
\[
x_{\lambda=1} = \frac{1}{\sqrt{2}} 
                \left( \begin{array}{c} 1 \\ -1 \end{array} \right).
\]
Plugging in $\lambda=2$, we get $2x_1+x_2=0$, and so the eigenvector 
corresponding to $\lambda=2$ is
\[
x_{\lambda=1} = \frac{1}{\sqrt{5}} 
                \left( \begin{array}{c} 1 \\ -2 \end{array} \right).
\]
So, this matrix has two eigenvalues and two eigenvectors, but they are not
orthogonal, which is ok, since $A$ is not symmetric.

\textbf{Example.}
Defective matrix.
Consider the following matrix, which is an example of a ``defective'' matrix.
\[
A = \left( \begin{array}{cc} 1 & 1 \\
                             0 & 1 
           \end{array} 
    \right)  .
\]
The characteristic polynomial is
\begin{equation*}
\mbox{det}\left( A - \lambda I \right) 
   = \left| \begin{array}{cc} 1-\lambda & 1         \\
                                      0 & 1-\lambda
            \end{array}
     \right|  
   = \left(1-\lambda\right)^{2} , 
\end{equation*}
and so $1$ is a double root.
If we plug this in, then we get the system of equations
\begin{equation*}
\left( \begin{array}{cc} 0 & 1 \\
                         0 & 0 
       \end{array}
\right)
\left( \begin{array}{c} x_1 \\ x_2 \end{array} \right)
 = 
\left( \begin{array}{c} 0 \\ 0 \end{array} \right)  ,
\end{equation*}
meaning that $x_2=0$ and $x_1$ is arbitrary.
(BTW, note that the matrix that appears in that system of equations is a
nilpotent matrix.
See below.
From the last class, this has a value of the Rayleigh quotient that is not 
in the closed interval defined by the min to max eigenvalue.)
Thus, there is only one linearly independent eigenvector corresponding
to the double eigenvalue $\lambda=1$ and it is 
\[
x_{\lambda=1} = \left( \begin{array}{c} 1 \\ 0 \end{array} \right).
\]

\textbf{Example.}
Nilpotent matrix.
Consider the following matrix,
\[
A = \left( \begin{array}{cc} 0 & 1 \\
                             0 & 0 
           \end{array} 
    \right)  .
\]
The only eigenvalue of this equals zero. The eigenvector is the same as in the above example.
But this matrix has the property that if you raise it to some finite power
then it equals the all-zeros matrix.

\textbf{Example.}
Identity.
The problem above with having only one linearly independent eigenvector is 
\emph{not} due to the multiplicity in eigenvalues.
For example, consider the following identity matrix,
\[
A = \left( \begin{array}{cc} 1 & 0 \\
                             0 & 1 
           \end{array} 
    \right)  .
\]
which has characteristic polynomial $\lambda^2-1=0$, and so which has 
$\lambda=1$ as a repeated root.
Although it has a repeated root, it has two linearly independent eigenvectors.
For example, 
\[
x_{1} = \left( \begin{array}{c} 1 \\ 0 \end{array} \right)
\quad
\mbox
\quad
x_{2} = \left( \begin{array}{c} 0 \\ 1 \end{array} \right)  ,
\]
or, alternatively, 
\[
x_{1} = \frac{1}{\sqrt{2}} \left( \begin{array}{c} 1 \\  1 \end{array} \right)
\quad
\mbox
\quad
x_{2} = \frac{1}{\sqrt{2}} \left( \begin{array}{c} -1 \\ 1 \end{array} \right) .
\]

This distinction as to whether there are multiple eigenvectors associated
with a degenerate eigenvalue is an important distinction, and so we 
introduce the following definitions.

\begin{definition}
Given a matrix $A$, for an eigenvalue $\lambda_i$
\begin{itemize}
\item
it's \emph{algebraic multiplicity}, denoted $\mu_A(\lambda_i)$, is the 
multiplicity of $\lambda$ as a root of the characteristic polynomial; and
\item
it's \emph{geometric multiplicity}, denoted $\gamma_A(\lambda_i)$ is the 
maximum number of linearly independent eigenvectors associated with it.
\end{itemize}
\end{definition}

Here are some facts (and terminology) concerning the relationship between 
the algebraic multiplicity and the geometric multiplicity of an eigenvalue.
\begin{itemize}
\item
$ 1 \le \gamma_A(\lambda_i) \le \mu_A(\lambda_i)$.
\item
If $\mu_A(\lambda_i)=1$, then $\lambda_i$ is a simple eigenvalue.
\item
If $\gamma_A(\lambda_i)=\mu_A(\lambda_i)$, then $\lambda_i$ is a semi-simple 
eigenvalue.
\item
If $\gamma_A(\lambda_i) < \mu_A(\lambda_i)$, for some $i$, then the matrix 
$A$ is defective.
Defective matrices are more complicated since you need things like Jordan 
forms, and so they are messier.
\item
If $\sum_i \gamma_A(\lambda_i) = n$, then $A$ has $n$ linearly independent eigenvectors.
In this case, $A$ is diagonalizable.
I.e., we can write $AQ = Q\Lambda$, and so $Q^{-1}AQ = \Lambda$.
And conversely.
\end{itemize}

\subsection{Basic ideas behind Perron-Frobenius theory}

The basic idea of Perron-Frobenius theory is that if you have a matrix $A$
with all positive entries (think of it as the adjacency matrix of a 
general, i.e., possibly directed, graph) then it is ``nice'' in several ways:
\begin{itemize}
\item
there is one simple real eigenvalue of $A$ that has magnitude larger than 
all other eigenvalues; 
\item
the eigenvector associated with this eigenvalue has all positive entires;
\item
if you increase/decrease the magnitude of the entries of $A$, then that 
maximum eigenvalue increases/decreases; and 
\item
a few other related properties.
\end{itemize}
These results generalize to non-negative matrices (and slightly more 
generally, but that is of less interest in general).
There are a few gotchas that you have to watch out for, and those typically 
have an intuitive meaning.
So, it will be important to understand not only how to establish the above
statements, but also what the gotchas mean and how to avoid them.

These are quite strong claims, and they are certainly false in general, 
even for non-negative matrices, without those additional assumptions.
About that, note that every nonnegative matrix is the limit of positive
matrices, and so there exists an eigenvector with nonnegative components.
Clearly, the corresponding eigenvalue is nonnegative and greater or equal 
in absolute value.
Consider the following examples.

\textbf{Example.}
Symmetric matrix.
Consider the following matrix.
\[
A = \left( \begin{array}{cc} 0 & 1 \\
                             1 & 0 
           \end{array} 
    \right)  .
\]
This is a non-negative matrix, and there is an eigenvalue equal to $1$.
However, there exist other eigenvalues of the same absolute value (and not 
strictly less) as this maximal one.
The eigenvalues are $-1$ and $1$, both of which have absolute value $1$.

\textbf{Example.}
Non-symmetric matrix.
Consider the following matrix.
\[
A = \left( \begin{array}{cc} 0 & 1 \\
                             0 & 0 
           \end{array} 
    \right)  .
\]
This is a matrix in which the maximum eigenvalue is not simple.
The only root of the characteristic polynomial is $0$, and the corresponding
eigenvector, i.e., $\left( \begin{array}{c} 1 \\ 0 \end{array} \right) $, is 
not strictly positive.

These two counter examples contain the basic ideas underlying the two main 
gotchas that must be dealt with when generalizing Perron-Frobenius theory
to only non-negative matrices.

(As an aside, it is worth wondering what is unusual about that latter matrix 
and how it can be generalized.
One is 
\[
A = \left( \begin{array}{cccc} 0 & 1 & 0 & 0 \\
                               0 & 0 & 1 & 0 \\
                               0 & 0 & 0 & 1 \\
                               0 & 0 & 0 & 0 \\
           \end{array} 
    \right)  ,
\]
and there are others.
These simple examples might seem trivial, but they contain several key ideas
we will see later.)

One point of these examples is that the requirement that the entries of the
matrix $A$ be strictly positive is important for Perron-Frobenius theory
to hold. If instead we only have non-negativity, we need further assumption on $A$ which we will see below (and in the special case of matrices associated with graphs, the reducibility property of the matrix is equivalent to the connectedness of the graph).

\subsection{Reducibility and types of connectedness}

We get a non-trivial generalization of Peron-Frobenius theory from 
all-positive matrices to non-negative matrices, if we work with the class 
of irreducible matrices.
(We will get an even cleaner statement if we work with the class of 
irreducible aperiodic matrices.
We will start with the former first, and then we will bet to the latter.)

We start with the following definition, which applied to an $n \times n$ 
matrix $A$.
For those readers familiar with Markov chains and related topics, there is 
an obvious interpretation we will get to, but for now we just provide the
linear algebraic definition.

\begin{definition}
A matrix $A \in \mathbb{R}^{n \times n}$ is \emph{reducible} if there exist
a permutation matrix $P$ such that 
\[
C = PAP^T 
  = \left( \begin{array}{cc} A_{11} & A_{12} \\
                                  0 & A_{22} 
           \end{array} 
    \right)  ,
\]
with $A_{11} \in \mathbb{R}^{r \times r}$ and
$A_{22} \in \mathbb{R}^{(n-r) \times (n-r)}$, where $0 < r < n$.
(Note that the off-diagonal matrices, $0$ and $A_{12}$, will in general 
be rectangular.)
A matrix $A\in\mathbb{R}^{n \times n}$ is \emph{irreducible} if it is not
reducible.
\end{definition}

As an aside, here is another definition that you may come across and that 
we may point to later.
\begin{definition}
A nonnegative matrix $A\in\mathbb{R}^{n \times n}$ is \emph{irreducible} if
$\forall i,j \in[n]^2 , \exists t \in \mathbb{N} : A_{ij}^t >0$.
And it is \emph{primitive} if
$ \exists t \in \mathbb{N} , \forall i,j \in[n]^2 : A_{ij}^t >0$.
\end{definition}

This is less intuitive, but I'm mentioning it since these are algebraic and 
linear algebraic ideas, and we haven't yet connected it with random walks.
But later we will understand this in terms of things like lazy random walks 
(which is more intuitive for most people than the gcd definition of 
aperiodicity/primitiveness).

Fact: If $A$, a non-negative square matrix, is nilpotent (i.e., s.t. $A^k=0$, for 
some $k\in\mathbb{Z}^{+}$, then it is reducible.
\begin{Proof} By contradiction, suppose $A$ is irreducible, and nilpotent. Let $k$ be the smallest $k$ such that $A^k=0$. Then we know $A^{k-1}\neq 0$. Suppose $A^{k-1}_{ij}> 0$ for some $i,j$, since $A$ irreducible, we now there exist $t\ge 1$ such that $A^t_{ji}> 0$. Note all powers of $A$ are non-negative, then $A^{k-1+t}_{ii}=A^{k-1}_{i,\cdot}A^t_{\cdot,i}\ge A^{k-1}_{ij}A^t_{ji} > 0$ which gives a contradiction, since we have $A^k=0 \Rightarrow A^{k'}=0 \quad \forall k'\ge k$, while $k-1+t\ge k$, but $A^{k-1+t}\neq 0$.
\end{Proof}

We start with a lemma that, when viewed the right way, i.e., in a way that
is formal but not intuitive, is trivial to prove.
\begin{lemma}
Let $A\in\mathbb{R}^{n \times n}$ be a non-negative square matrix.
If $A$ is primitive, then $A$ is reducible.
\end{lemma}
\begin{Proof}
$\exists\forall\rightarrow\forall\exists$
\end{Proof}

It can be shown that the converse is false.
But we can establish a sort of converse in the following lemma.
(It is a sort of converse since $A$ and $I+A$ are related, and in particular
in our applications to spectral graph theory the latter will essentially 
have an interpretation in terms of a lazy random walk associated with the 
former.)
\begin{lemma}
Let $A\in\mathbb{R}^{n \times n}$ be a non-negative square matrix.
If $A$ is irreducible, then $I+A$ is primitive.
\end{lemma}
\begin{Proof}
Write out the binomial expansion
\[
\left(I+A\right)^n = \sum_{k=0}^{n} {n \choose k} A^k .
\]
This has all positive entries since $A$ is irreducible, i.e., it eventually 
has all positive entries if $k$ is large enough.
\end{Proof}

Note that a positive matrix may be viewed as the adjacency matrix of a 
weighted complete graph.

Let's be more precise about directed and undirected graphs.

\begin{definition}
A \emph{directed graph} $G(A)$ associated with an $n \times n$ nonnegative 
matrix $A$ consists of $n$ nodes/vertices $P_1,\ldots,P_n$, where an edge 
leads from $P_i$ to $P_j$ iff $A_{ij} \ne 0$.
\end{definition}

Since directed graphs are directed, the connectivity properties are a little
more subtle than for undirected graphs.
Here, we need the following.
We will probably at least mention other variants later.

\begin{definition}
A directed graph $G$ is \emph{strongly connected} if $\forall$ ordered 
pairs $(P_i,P_j)$ of vertices of $G$, $\exists$ a path, i.e., a sequence of 
edges, $(P_i,P_{l_1}), (P_{l_1},P_{l_2}),\ldots, (P_{l_{r-1}},P_{j})$, which  
leads from $P_i$ to $P_j$.
The \emph{length} of the path is $r$.
\end{definition}

Fact: The graph $G(A^k)$ of a nonnegative matrix $A$ consists of all paths 
of $G(A)$ of length $k$ (i.e. there is an edge from $i$ to $j$ in $G(A^k)$ iff there is a path of length $k$ from $i$ to $j$ in $G$).

Keep this fact in mind since different variants of spectral methods involve
weighting paths of different lengths in different ways.

Here is a theorem that connects the linear algebraic idea of irreducibility 
with the graph theoretic idea of connectedness.
Like many things that tie together notions from two different areas, it can 
seem trivial when it is presented in such a way that it looks obvious; but 
it really is connecting two quite different ideas.
We will see more of this later.

\begin{theorem}
An $n \times n$ matrix $A$ is irreducible iff the corresponding directed 
graph $G(A)$ is strongly connected.
\end{theorem}
\begin{Proof}
Let $A$ be an irreducible matrix.
Assume, for contradiction, that $G(A)$ is \emph{not} strongly connected.
Then, there exists an ordered pair of nodes, call them $(P_i,P_j)$, s.t. 
there does not exist a connection from $P_i$ to $P_j$.
In this case, let $S_1$ be the set of nodes connected to $P_i$, and let 
$S_2$ be the remainder of the nodes.
Note that there is no connection between any nodes $P_{\ell}\in S_2$ and
any node $P_q \in S_1$, since otherwise we sould have $P_{\ell}\in S_1$.
And note that both sets are nonempty, since $P_j\in S_1$ and $P_i\in S_2$.
Let $r = |S_1|$ and $n-r = |S_2|$.
Consider a permutation transformation $C=PAP^T$ that reorders the nodes of 
$G(A)$ such that
\[
\begin{cases}
    P_1,P_2,\cdots,P_r \in S_1 \\
    P_{r+1},P_{r+2},\cdots,P_n \in S_2
\end{cases}
\]
That is 
\[
C_{k\ell} = 0 \quad\forall
\begin{cases}
    k = r+1,r+2,\ldots,n \\
    \ell = 1,2,\ldots,r   .
\end{cases}
\]
But this is a contradiction, since $A$ is irreducible.

Conversely, assume that $G(A)$ is strongly connected, and assume for 
contradiction that $A$ is not irreducible.
Reverse the order of the above argument, and we arrive at the conclusion 
that $G(A)$ is not strongly connected, which is a contradiction.
\end{Proof}

We conclude by noting that, informally, there are two types of 
irreducibility.
To see this, recall that in the definition of reducibility/irreducibility, 
we have the following matrix:
\[
C = PAP^T 
  = \left( \begin{array}{cc} A_{11} & A_{12} \\
                                  0 & A_{22} 
           \end{array} 
    \right)  .
\]
In one type, $A_{12} \ne 0$: in this case, we can go from the first set to 
the second set and get stuck in some sort of sink.
(We haven't made that precise, in terms of random walk interpretations, but
there is some sort of interaction between the two groups.)
In the other type, $A_{12} = 0$: in this case, there are two parts that 
don't talk with each other, and so essentially there are two separate 
graphs/matrices.

\subsection{Basics of Perron-Frobenius theory}

Let's start with the following definition.
(Note here that we are using subscripts to refer to elements of a vector, 
which is inconsistent with what we did in the last class.)
\begin{definition}
A vector $x\in\mathbb{R}^{n}$ is \emph{positive} (resp, \emph{non-negative}) 
if all of the entries of the vector are positive (resp, non-negative), i.e., 
if $x_i > 0$ for all $i\in[n]$ (resp if $x_i \ge 0$ for all $i\in[n]$).
\end{definition}
A similar definition holds for $m \times n$ matrices.
Note that this is \emph{not} the same as SPD/SPSD matrices.

Let's also provide the following definition.
\begin{definition}
Let $\lambda_1,\ldots,\lambda_n$ be the (real or complex) eigenvalues of a 
matrix $A\in\mathbb{C}^{n \times n}$.
Then the \emph{spectral radius} 
$ \rho_A = \rho(A) = \max_i \left( \left| \lambda_i \right| \right) $
\end{definition}

Here is a basic statement of the Perron-Frobenius theorem.

\begin{theorem}[Perron-Frobenius]
Let $A\in\mathbb{R}^{n \times n}$ be an irreducible non-negative matrix.
Then,
\begin{enumerate}
\item
$A$ has a positive real eigenvalue equal to its spectral radium.
\item
That eigenvalue $\rho_A$ has algebraic and geometric multiplicity equal to 
one.
\item
The one eigenvector $x$ associated with the eigenvalue $\rho_A$ has all 
positive entries.
\item
$\rho_A$ increases when any entry of $A$ increases.
\item
There is no other non-negative eigenvector of $A$ different than $x$.
\item
If, in addition, $A$ is primitive, then each other eigenvalue $\lambda$ of 
$A$ satisfies $\left| \lambda \right| < \rho_A$.
\end{enumerate}
\end{theorem}

Before giving the proof, which we will do next class, let's first start with 
some ideas that will suggest how to do the proof.

Let $P=\left(I+A\right)^{n}$.
Since $P$ is positive, it is true that for every non-negative and 
non-null vector $v$, that we have that $Pv >0$ element-wise.
Relatedly, if $v \le w$ element-wise, and $v \ne w$, then $Pv < Pw$.

Let 
\[
Q = \left\{ x \in \mathbb{R}^{n} \mbox{ s.t. } x \geq 0 , x \neq 0 \right\}
\]
be the nonnegative orthant, excluding the origin.
In addition, let
\[
C = \left\{ x \in \mathbb{R}^{n} \mbox{ s.t. } x \geq 0 , ||x||=1 \right\} ,
\]
where $||\cdot||$ is any vector norm.
Clearly, $C$ is compact, i.e., closed and bounded.

Then, for all $z \in Q$, we can define the following function:
let 
\[
f(z) = \max \left\{ s\in\mathbb{R} : sz \le Az \right\}
     = \min_{1 \le i \le n,z_i \ne 0} \frac{\left(Az\right)_{i}}{z_i}
\]

Here are facts about $f$.
\begin{itemize}
\item
$f(rz) = f(z)$, for all $r > 0$.
\item
If $Az = \lambda z$, i.e., if $(\lambda,z)$ is an eigenpair, then 
$f(z) = \lambda$.
\item
If $sz \le Az$, then $sPz \le PAz = APz$, where the latter follows since 
$A$ and $P$ clearly commute.
So, 
\[
f(z) \le f(Pz) .
\]
In addition, if $z$ is \emph{not} an eigenvector of $A$, then $sz \ne Az$, 
for all $s$; and $sPz < APz$.
From the second expression for $f(z)$ above, we have that in this case that
\[
f(z) < f(Pz)  ,
\]
i.e., an inequality in general but a strict inequality if not an eigenvector.
\end{itemize}

This \emph{suggests} an idea for the proof: look for a positive vector that
maximizes the function $f$; show it is an eigenvector we want in the theorem; 
and show that it established the properties stated in the theorem.

%% file: lect04.tex
\section{%
(02/03/2015): 
Basic Matrix Results (3 of 3)}

Reading for today.
\begin{compactitem}
\item
Same as last class.
\end{compactitem}

\subsection{Review and overview}

Recall the basic statement of the Perron-Frobenius theorem from last class.

\begin{theorem}[Perron-Frobenius]
Let $A\in\mathbb{R}^{n \times n}$ be an irreducible non-negative matrix.
Then,
\begin{enumerate}
\item
$A$ has a positive real eigenvalue $\lambda_{max}$;
which is equal to the spectral radius; and $\lambda_{max}$
has an associated eigenvector $x$ with all positive entries.
\item
If $0 \le B \le A$, with $B \ne A$, then every eigenvalue $\sigma$ of $B$ 
satisfies $|\sigma| < \lambda_{max} = \rho_A$.
(Note that $B$ does not need to be irreducible.)
In particular, $B$ can be obtained from $A$ by zeroing out entries; and also
all of the diagonal minors $A_{(i)}$ obtained from $A$ by deleting the 
$i^{th}$ row/column have eigenvalues with absolute value strictly less than 
$\lambda_{max} = \rho_A$.
Informally, this says: $\rho_A$ increases when any entry of $A$ increases.
\item
That eigenvalue $\rho_A$ has algebraic and geometric multiplicity equal to 
one.
\item
If $y \ge 0$, $y \ne 0$ is a vector and $\mu$ is a number such that 
$A y \le \mu y$, then $y > 0$ and $\mu \ge \lambda_{max}$;
with $\mu = \lambda_{max}$ iff $y$ is a multiple of $x$.
Informally, this says: there is no other non-negative eigenvector of $A$ 
different than $x$.
\item
If, in addition, $A$ is primitive/aperiodic, then each other eigenvalue 
$\lambda$ of $A$ satisfies $\left| \lambda \right| < \rho_A$.
\item
If, in addition, $A$ is primitive/aperiodic, then
\[
\lim_{t \rightarrow \infty} \left( \frac{1}{\rho_A} A \right)^{t} = xy^T ,
\]
where $x$ and $y$ are positive eigenvectors of $A$ and $A^T$ with eigenvalue 
$\rho_A$, i.e., $Ax = \rho_A x$ and $A ^T y = \rho_A y$ (i.e., 
$y^TA = \rho_Ay^T$), normalized such that $x^Ty=1$.
\end{enumerate}
\end{theorem}

Today, we will do three things:
(1) we will prove this theorem;
(2) we will also discuss periodicity/aperiodicity issues;
(3) we will also briefly discuss the first connectivity/non-connectivity result for 
Adjacency and Laplacian matrices of graphs that will use the ideas we have
developed in the last few~classes.

Before proceeding, one note: 
an interpretation of a matrix $B$ generated from $A$ by zeroing out an 
entry or an entire row/column is that you can remove an edge from a graph 
or you can remove a node and all of the associated edges from a graph.
(The monotonicity provided by that part of this theorem will be important 
for making claims about how the spectral radius behaves when such changes 
are made to a graph.)
This obviously holds true for Adjacency matrices, and a similar statement also
holds true for Laplacian matrices.

\subsection{Proof of the Perron-Frobenius theorem}

We start with some general notation and definitions; then we prove each 
part of the theorem in~turn.

Recall from last time that we let $P=\left(I+A\right)^{n}$ and thus $P$ is positive.
Thus, for every non-negative and non-null vector $v$, then we have that $Pv >0$ 
element-wise; and (equivalently) if $v \le w$ element-wise, and $v \ne w$, then 
we have that $Pv < Pw$.
Recall also that we defined
\begin{eqnarray*}
Q &=& \left\{ x \in \mathbb{R}^{n} \mbox{ s.t. } x \geq 0 , x \neq 0 \right\} \\
C &=& \left\{ x \in \mathbb{R}^{n} \mbox{ s.t. } x \geq 0 , ||x||=1 \right\} ,
\end{eqnarray*}
where $||\cdot||$ is any vector norm.
Note in particular that this means that $C$ is compact, i.e., closed and bounded.
Recall also that, for all $z \in Q$, we defined the following function:
let 
\[
f(z) = \max \left\{ s\in\mathbb{R} : sz \le Az \right\}
     = \min_{1 \le i \le n,z_i \ne 0} \frac{\left(Az\right)_{i}}{z_i}
\]

Finally, recall several facts about the function $f$.
\begin{itemize}
\item
$f(rz) = f(z)$, for all $r > 0$.
\item
If $Az = \lambda z$, i.e., if $(\lambda,z)$ is an eigenpair, then 
$f(z) = \lambda$.
\item
In general, $ f(z) \le f(Pz) $; and if $z$ is \emph{not} an eigenvector of 
$A$, then $ f(z) < f(Pz) $.
(The reason for the former is that if $sz \le Az$, then $sPz \le PAz = APz$.
The reason for the latter is that in this case $sz \ne Az$, for all $s$, and 
$sPz < APz$, and by considering the second expression for $f(z)$ above.)
\end{itemize}

We will prove the theorem in several steps.

\subsection{Positive eigenvalue with positive eigenvector.}

Here, we will show that there is a positive eigenvalue $\lambda^{*}$ and that the associated eigenvector $x^{*}$ is a positive vector.

To do so, consider $P(C)$, the image of $C$ under the action of the operator $P$.
This is a compact set, and all vectors in $P(C)$ are positive.
By the second expression in definition of $f(\cdot)$ above, we have that $f$ is continuous of $P(C)$.
Thus, $f$ achieves its maximum value of $P(C)$, i.e., there exists a vector $x \in P(C)$ such that 
\[
f(x) = \sup_{z \in C} f(Pz) .
\]
Since $f(z) \le f(Pz)$, the vector $x$ realizes the maximum value $f_{max}$ of $f$ on $Q$.
So, 
\[
f_{max} = f(x) \le f(Px) \le f_{max}  .
\]
Thus, from the third property of $f$ above, $x$ is an eigenvector of $A$ with eigenvalue $f_{max}$.
Since $x \in P(C)$, then $x$ is a positive vector; and since $Ax > 0$ and $Ax = f_{max}x$, it follows that $f_{max}>0$.

(Note that this result shows that $f_{max}=\lambda^{*}$ is achieved on an eigenvector $x=x^{*}$, but it doesn't show yet that it is equal to the spectral radius.)

\subsection{That eigenvalue equals the spectral radius.}

Here, we will show that $f_{max}= \rho_A$, i.e., $f_{max}$ equals the spectral radius.

To do so, let $z\in\mathbb{C}^{n}$ be an eigenvector of $A$ with eigenvalue $\lambda\in\mathbb{C}$; and let $|z|$ be a vector, each entry of which equals $|z_i|$.
Then, $|z| \in Q$.

We claim that $ |\lambda| |z| \le A |z| $.
To establish the claim, rewrite it as $ |\lambda| |z_i| \le \sum_{k=1}^{n} A_{ik} |z_k| $.
Then, since $Az = \lambda z$, i.e., $\lambda z_i = \sum_{k=1}^{n} A_{ik} z_k $, and since $A_{ik}\ge0$, we have that 
\[
|\lambda| |z| = \left| \sum_{k=1}^{n} A_{ik} z_k \right| \le \sum_{k=1}^{n} A_{ik} |z_k| ,
\]
from which the claim follows.

Thus, by the definition of $f$ (i.e., since $f(z)=\min\frac{(Az)_{i}}{(z)_{i}}$, we have that $|\lambda| \le f(|z|)$.
Hence, $|\lambda| \le f_{max}$, and thus $\rho_A \le f_{max}$ (where $\rho_A$ is the spectral radius).
Conversely, from the above, i.e., since $f_{max}$ is an eigenvalue it must be $\le$ the maximum eigenvalue, we have that $f_{max} \le \rho_A$.
Thus, $f_{max}=\rho_A$.

\subsection{An extra claim to make.}
\label{sxn:extra-claim}

We would like to establish the following result:
\[
f(z) = f_{max} \Rightarrow \left( Az = f_{max}z \mbox{ and } z > 0 \right)  .
\]

To establish this result, observe that above it is shown that:
if $f(z)=f_{max}$, then $f(z)=f(Pz)$.
Thus, $z$ is an eigenvector of $A$ for eigenvalue $f_{max}$.
It follows that $Pz = \lambda z$, i.e., that $z$ is also an eigenvector of $P$.
Since $P$ is positive, we have that $Pz > 0$, and so $z$ is positive.

\subsection{Monotonicity of spectral radius.}

Here, we would like to show that $0 \le B \le A$ and $B \ne A$ implies that $\rho_B < \rho_A$.
(Recall that $B$ need \emph{not} be irreducible, but $A$ is.)

To do so, suppose that $Bz = \lambda z$, with $z \in \mathbb{C}^{n}$ and with $\lambda\in\mathbb{C}$.
Then, 
\[
|\lambda| |z| \le B |z| \le A |z| , 
\]
from which it follows that 
\[
|\lambda| \le f_A(|z|) \le \rho_A   , 
\]
and thus $\rho_B \le \rho_A$.

Next, assume for contradiction that $|\lambda| = \rho_A$.
Then from the above claim (in Section~\ref{sxn:extra-claim}), we have that $f_A(z) = \rho_A$.
Thus from above it follows that
$|z|$ is an eigenvector of $A$ for the eigenvalue $\rho_A$ and also that 
$z$ is positive.
Hence, $B |z| = A |z|$, with $z > 0$; but this is impossible unless $A=B$.

\textbf{Remark.}
Replacing the $i^{th}$ row/column of $A$ by zeros gives a non-negative matrix $A_{(i)}$ such that $0 \le A_{(i)} \le A$.
Moreover, $A_{(i)} \ne A$, since the irreducibility of $A$ precludes the possibility that all entries in a row are equal to zero.
Thus, for all matrices $A_{(i)}$ that are obtained by eliminating the $i^{th}$ row/column of $A$, the eigenvalues of $A_{(i)} < \rho$.

\subsection{Algebraic/geometric multiplicities equal one.}

Here, we will show that the algebraic and geometric multiplicity of $\lambda_{max}$ equal $1$.
Recall that the geometric multiplicity is less than or equal to the algebraic multiplicity, and that both are at least equal to one, so it suffices to prove this for the algebraic multiplicity.

Before proceeding, also define the following: given a square matrix $A$: 
\begin{itemize}
\item
Let $A_{(i)}$ be the matrix obtained by eliminating the $i^{th}$ row/column.
In particular, this is a smaller matrix, with one dimension less along each column/row.
\item
Let $A_{i}$ be the matrix obtained by zeroing out the $i^{th}$ row/column.
In particular, this is a matrix of the same size, with all the entries in one full row/column zeroed out.
\end{itemize}

To establish this result, here is a lemma that we will use; its proof (which we won't provide) boils down to expanding $\mbox{det}\left(\Lambda-A\right)$ along the $i^{th}$ row.
\begin{lemma}
Let $A$ be a square matrix, and let $\Lambda$ be a diagonal matrix of the same size with $\lambda_1,\ldots,\lambda_n$ (as variables) along the diagonal.
Then, 
\[
\frac{\partial}{\partial \lambda_i} \mbox{det}\left( \Lambda-A \right)
   = \mbox{det}\left( \Lambda_{(i)} - A_{(i)} \right)  ,
\]
where the subscript $(i)$ means the matrix obtained by eliminating the $i^{th}$ row/column from each matrix.
\end{lemma}

Next, set $\lambda_i = \lambda$ and apply the chain rule from calculus to get 
\[
\frac{d}{d\lambda} \mbox{det}\left( \lambda I - A \right) = \sum_{i=1}^{n} \mbox{det}\left( \lambda I - A_{(i)} \right)  .
\]

Finally, note that
\[
\mbox{det}\left( \lambda I - A_{i} \right) = \lambda \mbox{det} \left( \lambda I - A_{(i)} \right)   .
\]
But by what we just proved (in the Remark at the end of last page), we have that $\mbox{det}\left( \rho_A I - A_{(i)} \right) > 0 $.
Thus, the derivative of the characteristic polynomial of $A$ is nonzero at $\rho_A$, and so the 
algegraic multiplicity equals $1$.

\subsection{No other non-negative eigenvectors, etc.}

Here, we will prove the claim about other non-negative vectors, including that there are no other non-negative eigenvectors.

To start, we claim that: $0 \le B \le A \Rightarrow f_{max}(B) \le f_{max}(A)$.
(This is related to but a little different than the similar result we had above.)
To establish the claim, note that if $z \in Q$ is s.t. $sz \le Bz$, then $sz \le Az$ (since $Bz \le Az$), and so $f_B(z) \le f_A(z)$, for all $z$.

We can apply that claim to $A^T$, from which it follows that $A^T$ has a positive eigenvalue, call it $\eta$. 
So, there exists a row vector, $w > 0$ s.t. $w^TA = \eta w^T$.
Recall that $x > 0$ is an eigenvector of $A$ with maximum eigenvalue $\lambda_{max}$.
Thus, 
\[
w^TAx = \eta w^Tx = \lambda_{max} w^Tx , 
\]
and thus $\eta = \lambda_{max}$ (since $w^Tx > 0$).

Next, suppose that $y \in Q$ and $Ay \le \mu y$.
Then, 
\[
\lambda_{max} w^Ty = w^T A y \le \mu w^Ty , 
\]
from which it follows that $\lambda_{max} \le \mu$.
(This is since all components of $w$ are positive and some components of $y$ is positive, and so $w^Ty > 0$).

In particular, if $Ay = \mu y$, then $\mu = \lambda_{max}$.

Further, if $y \in Q$ and $Ay \le \mu y$, then $\mu \ge 0$ and $y > 0$.
(This is since $0 < Py = \left(I+A\right)^{n-1}y \le \left( 1+\mu\right)^{n-1}y$.)

This proves the first two parts of the result; now, let's prove the last part of the result.

If $\mu = \lambda_{max}$, then $w^T (Ay - \lambda_{max}y ) = 0$.
But, $Ay - \lambda_{max} y \le 0$.
So, given this, from $w^T\left(Ay-\lambda_{max}y\right) = 0$, it follows that $Ay = \lambda_{max}y$.
Since $y$ must be an eigenvector with eigenvalue $\lambda_{max}$, the last result (i.e., that $y$ is a scalar multiple of $x$) follows since $\lambda_{max}$ has multiplicity $1$.

To establish the converse direction march through these steps in the other direction.

\subsection{Strict inequality for aperiodic matrices}

Here, we would like to establish the result that the eigenvalue we have been talking about is strictly larger in magnitude than the other eigenvalues, under the aperiodicity assumption.

To do so, recall that the $t^{th}$ powers of the eigenvalues of $A$ are the eigenvalues of $A^t$.
So, if we want to show that there does \emph{not} exist eigenvalues of a primitive matrix with absolute value $=\rho_A$, other than $\rho_A$, then it suffices to prove this for a positive matrix $A$.

Let $A$ be a positive matrix, and suppose that $Az = \lambda z$, with $z\in\mathbb{C}^{n}$, $\lambda\in\mathbb{C}$, and $|\lambda|=\rho_A$, in which case the goal is to show $\lambda < \rho_A$.

(We will do this by showing that any eigenvector with eigenvalue equal in magnitude to $\rho_A$ is the top eigenvalue.)
(I.e., we will show that such a $z$ equals $|z|$ and thus there is no other one with $\rho_A$.)

Then, 
\[
\rho_A |z| = |Az| \le A |z|  ,
\]
from which it follows that
\[
\rho_A \le f(|z|) \le \rho_A  ,
\]
which implies that $f(|z|) = \rho_A  $.
From a result above, this implies that $|z|$ is an eigenvector of $A$ with eigenvalue $\rho_A$.
Moreover, $|Az| = A|z|$.
In particular, 
\[
\left| \sum_{i=1}^{n} A_{1i} z_{i} \right|  = \sum_{i=1}^{n} A_{1i} | z_i |  .
\]
Since all of the entries of $A$ are positive, this implies that there exists a number $u\in\mathbb{C}$ (with $|u|=1$) s.t. for all $i\in[n]$, we have that $z_i = u|z_i|$.
Hence, $z$ and $|z|$ are collinear eigenvectors of $A$.
So, the corresponding eigenvalues of $\lambda$ and $\rho$ are equal, as required.

\subsection{Limit for aperiodic matrices}

Here, we would like to establish the limiting result.

To do so, note that $A^T$ has the same spectrum (including multiplicities) as $A$; and in particular the spectral radius of $A^T$ equals $\rho_A$.

Moreover, since $A^T$ is irreducible (a consequence of being primitive), we can apply the Perron-Frobenius theorem to it to get $yA = \rho_A y$.
Here $y$ is determined up to a scalar multiple, and so let's choose it s.t. $x^Ty = \sum_{i=1}^{n} x_iy_i = 1$.

Next, observe that we can decompose the $n$-dimensional vector space $\mathbb{R}^{n}$ into two parts, 
\[
\mathbb{R}^{n} = R \oplus N , 
\]
where both $R$ and $N$ are invariant under the action of $A$.
To do this, define the rank-one matrix $H=xy^T$, and: 
\begin{itemize}
\item
let $R$ be the \emph{image space} of $H$; and 
\item
let $N$ be the \emph{null space} of $H$.
\end{itemize}
Note that $H$ is a projection matrix (in particular, $H^2=H$), and thus $I-H$ is also a projection matrix, and the image space of $I-H$ is $N$.
Also, 
\[
AH = Axy^T = \rho_A xy^T = x \rho_A y^T = x y^T A = HA .
\]
So, we have a direct sum decomposition of the space $\mathbb{R}^{n}$ into $R \oplus N$, and this decomposition is invariant under the action of $A$.

Given this, observe that the restriction of $A$ to $N$ has all of its eigenvalues strictly less that $\rho_A$ in absolute value, while the restriction of $A$ to the one-dimensional space $R$ is simply a multiplication/scaling by $\rho_A$.
So, if $P$ is defined to be $P = \frac{1}{\rho_A}A$, then the restriction of $P$ to $N$ has its eigenvalues $<1$ in absolute value.
This decomposition is also invariant under all positive integral powers of $P$.
So, the restriction of $P^k$  to $N$ tends to zero as $k \rightarrow \infty$, while the restriction of $P$ to $R$ is the identity.
So, $\lim_{t \rightarrow \infty} \left( \frac{1}{\rho_A}A \right)^{t} = H = xy^T$.

\subsection{Additional discussion form periodicity/aperiodic and cyclicity/primitiveness}

%
%

Let's switch gears and discuss the periodicity/aperiodic and cyclicity/primitiveness issues.

(This is an algebraic characterization, and it holds for general non-negative matrices.  I think that most people find this less intuitive that the characterization in terms of connected components, but it's worth at least knowing about it.)

Start with the following definition.

\begin{definition}
The \emph{cyclicity} of an irreducible non-negative matrix $A$ is the g.c.d. (greatest common denominator) of the length of the cycles in the associated graph.
\end{definition}

\noindent
Let's let $\mathbb{N}_{ij}$ be a positive subset of the integers s.t.
\[
\{ t \in \mathbb{N} \mbox{ s.t. } (A^t)_{ij} > 0 \}  ,
\]
that is, it is the values of $ t \in \mathbb{N} $ s.t. the matrix $A^t$'s $(i,j)$ entry is positive (i.e. exists a path from $i$ to $j$ of length $t$) .
Then, to define $\gamma$ to be the cyclicity of $A$, first define $\gamma_i = \mbox{gcd}\left( \mathbb{N}_{ii} \right)$, and 
then clearly $\gamma = \mbox{gcd} \left( \{ \gamma_i \mbox{ s.t. } i \in V \} \right) $.
Note that each $\mathbb{N}_{ii}$ is closed under addition, and so it is a semi-group.

Here is a lemma from number theory (that we won't prove).

\begin{lemma}
A set $\mathbb{N}$ of positive integers that is closed under addition contains all but a finite number of multiples of its g.c.d.
\end{lemma}

\noindent
From this it follows that $\forall i \in [n] , \gamma_i  = \gamma$.

The following theorem (which we state but won't prove) provides several related conditions for an irreducible matrix to be primitive.

\begin{theorem}
Let $A$ be an irreducible matrix.
Then, the following are equivalent.
\begin{enumerate}
\item
The matrix $A$ is primitive.
\item
All of the eigenvalues of $A$ different from its spectral radius $\rho_A$ satisfy $|\lambda| < \rho_A$.
\item
The sequence of matrices $\left( \frac{1}{\rho_A}A \right)^{t}$ converges to a positive matrix.
\item
There exists an $i \in [n]$ s.t., $\gamma_i = 1$.
\item
The cyclicity of $A$ equals $1$.
\end{enumerate}
\end{theorem}

For completeness, note that sometimes one comes across the following definition.

\begin{definition}
Let $A$ be an irreducible non-negative square matrix.
The \emph{period} of $A$ is the g.c.d. of all natural numbers $m$ s.t. $\left(A^m\right)_{ii} > 0$ for some $i$.
Equivalently, the g.c.d. of the lengths of closed directed paths of the directed graph $G_A$ associated with $A$.
\end{definition}

\textbf{Fact.}
All of the statements of the Perron-Frobenius theorem for positive matrices remain true for irreducible aperiodic matrices.
In addition, all of those statements generalize to periodic matrices. 
The the main difference in this generalization is that for periodic matrices the ``top'' eigenvalue isn't ``top'' any more, in the sense that there are other eigenvalues with equal absolute value that are different: they equal the $p^{th}$ roots of unity, where $p$ is the periodicity. 

Here is an example of a generalization.

\begin{theorem}
\label{thm:pf-generalization}
Let $A$ be an irreducible non-negative $n \times n$ matrix, with period equal to $h$ and spectral radius equal to $\rho_A = r$.
Then, 
\begin{enumerate}
\item
$r > 0$, and it is an eigenvalue of $A$.
\item
$r$ is a simple eigenvalue, and both its left and right eigenspace are one-dimensional.
\item
$A$ has left/right eigenvectors $v$/$w$ with eigenvalue $r$, each of which has all positive entries.
\item
$A$ has exactly $h$ complex eigenvalues with absolute value $=r$; and each is a simple root of the characteristic polynomial and equals the $r \cdot h^{th}$ root of unity.
\item
If $h > 0$, then there exists a permutation matrix $P$ s.t.
\begin{equation}
PAP^T = \left( \begin{array}{ccccc} 0 & A_{12} &  &  & 0 \\
                                      & 0 & A_{23} &   &   \\
                                      &   & \ddots &  \ddots &   \\
                                     0 &  &   & 0 & A_{h-1,h} \\
                                    A_{h1} & 0  &   &   & 0 \\
               \end{array}
        \right)  .
\label{eqn:block-matrix-periodic}
\end{equation}
\end{enumerate}
\end{theorem}

\subsection{Additional discussion of directness, periodicity, etc.}

Today, we have been describing Perron-Frobenius theory for non-negative matrices.  There are a lot of connections with graphs, but the theory can be developed algebraically and linear-algebraically, i.e., without any mention of graphs. 
(We saw a hint of this with the g.c.d. definitions.)
In particular, Theorem~\ref{thm:pf-generalization} is a statement about matrices, and it's fair to ask what this might say about graphs we will encounter.  
So, before concluding, let's look at it and in particular at Eqn.~(\ref{eqn:block-matrix-periodic}) and ask what that might say about graphs---and in particular undirected graphs---we will consider.

To do so, recall that the Adjacency Matrix of an undirected graph is symmetric; and, informally, there are several different ways (up to permutations, etc.) it can ``look like.''
In particular:
\begin{itemize}
\item
It can look like this:
\begin{equation}
A = \left( \begin{array}{cc} A_{11} & A_{12} \\
                             A_{12}^{T} & A_{22}
           \end{array}
    \right)   ,
\label{eqn:block-matrix-vanilla}
\end{equation}
where let's assume that all-zeros blocks are represented as $0$ and so each $A_{ij}$ is not all-zeros.
This corresponds to a vanilla graph you would probably write down if you were asked to write down a graph.
\item
It can look like this:
\begin{equation}
A = \left( \begin{array}{cc} A_{11} & 0 \\
                             0 & A_{22}
           \end{array}
    \right)  ,
\label{eqn:block-matrix-disconnected}
\end{equation}
in which case the corresponding graph is not connected.
\item
It can even look like this:
\begin{equation}
A = \left( \begin{array}{cc} 0 & A_{12} \\
                             A_{21} & 0
           \end{array}
    \right)  ,
\label{eqn:block-matrix-bipartite}
\end{equation}
which has the interpretation of having two sets of nodes, each of which has edges to only the other set, and which will correspond to a bipartite graph.
\item
Of course, it could be a line-like graph, which would look like a tridiagonal banded matrix, which is harder for me to draw in latex, or it can look like all sorts of other things.
\item
But it can\emph{not} look like this:
\begin{equation}
A = \left( \begin{array}{cc} A_{11} & A_{12} \\
                             0 & A_{22}
           \end{array}
    \right)  ,
\label{eqn:block-matrix-reducible1}
\end{equation}
and it can\emph{not} look like this:
\begin{equation}
A = \left( \begin{array}{cc} 0 & A_{12} \\
                             0 & 0
           \end{array}
    \right)  ,
\label{eqn:block-matrix-reducible2}
\end{equation}
where recall  we are assuming that each $A_{ij}$ is not all-zeros.
In both of these cases, these matrices are not symmetric.
\end{itemize}

\noindent
In light of today's results and looking forward, it's worth commenting for a moment on the relationship between Eqns.~(\ref{eqn:block-matrix-periodic}) and Eqns~(\ref{eqn:block-matrix-vanilla}) through (\ref{eqn:block-matrix-reducible2}).


Here are a few things to note.
\begin{itemize}
\item
One might think from Eqns.~(\ref{eqn:block-matrix-periodic}) that periodicity means that that the graph is directed and so if we work with undirected graphs we can ignore it.
That's true if the periodicity is $3$ or more, but note that the matrix of Eqn~(\ref{eqn:block-matrix-bipartite}) is periodic with period equal to $2$.
In particular,  Eqn~(\ref{eqn:block-matrix-bipartite}) is of the form of Eqn.~(\ref{eqn:block-matrix-periodic}) if the period $h=2$.
(It's eigenvalues are real, which they need to be since the matrix is symmetric, since the complex ``$2^{th}$ roots of unity,'' which equal $\pm1$, are both real.)
\item
You can think of Eqn.~(\ref{eqn:block-matrix-disconnected}) as a special case of Eqn.~(\ref{eqn:block-matrix-reducible1}), with the $A_{12}$ block equal to $0$, but it is not so helpful to do so, since its behavior is very different than for an irreducible matrix with $A_{12} \ne 0$.
\item
For directed graphs, e.g., the graph that would correspond to Eqn.~(\ref{eqn:block-matrix-reducible1})
(or Eqn.~(\ref{eqn:block-matrix-reducible2})), there is very little spectral theory.
It is of interest in practice since edges are often directed.
But, most spectral graph methods for directed graphs basically come up---either explicitly or implicitly---with some sort of symmetrized version of the directed graph and then apply undirected spectral graph methods to that symmetrized graph.
(Time permitting, we'll see an example of this at some point this semester.)
\item
You can think of Eqn.~(\ref{eqn:block-matrix-reducible1}) as corresponding to a ``bow tie'' picture (that I drew on the board and that is a popular model for the directed web graph and other directed graphs).
Although this is directed, it can be made irreducible by adding a rank-one update of the form $11^T$ to the adjacency matrix.
E.g., $A \rightarrow A+ \epsilon 11^T$.
This has a very natural interpretation in terms of random walkers, it is the basis for a lot of so-called ``spectral ranking'' methods, and it is a very popular way to deal with directed (and undirected) graphs.
In addition, for reasons we will point out later, we can get spectral methods to work in a very natural way in this particular case, even if the initial graph is undirected.
\end{itemize}

%% file: lect05.tex
\section{%
(02/05/2015): 
Overview of Graph Partitioning}

Reading for today.
\begin{compactitem}
\item
``Survey: Graph clustering,'' in Computer Science Review, by Schaeffer
\item
``Geometry, Flows, and Graph-Partitioning Algorithms,'' in CACM, by Arora, Rao, and Vazirani 
\end{compactitem}

The problem of \emph{graph partitioning} or \emph{graph clustering} refers
to a general class of problems that deals with the following task:
given a graph $G=(V,E)$, group the vertices of a graph into groups or 
clusters or communities.
(One might be interested in cases where this graph is weighted, directed, 
etc., but for now let's consider non-directed, possibly weighted, graphs.
Dealing with weighted graphs is straightforward, but extensions to directed 
graphs are more problematic.)
The graphs might be given or constructed, and there may or may not be extra
information on the nodes/edges that are available, but insofar as the black 
box algorithm that actually does the graph partitioning is concerned, all 
there is is the information in the graph, i.e., the nodes and edges or 
weighted edges.
Thus, the graph partitioning algorithm takes into account the node and edge 
properties, and thus it typically relies on some sort of ``edge counting''
metric to optimize.
Typically, the goal is to group nodes in such a manner that nodes within a 
cluster are more similar to each other than to nodes in different clusters, 
\emph{e.g.}, more and/or better edges within clusters and relatively few 
edges between clusters.


\subsection{Some general comments}

Two immediate questions arise.
\begin{itemize}
\item
A first question is to settle on an objective that captures this bicriteria. 
There are several ways to quantify this bicriteria which we will describe, 
but each tries to cut a data graph into $2$ or more ``good'' or ``nice'' 
pieces. 
\item
A second question to address is how to compute the optimal solution to that
objective.  
In some cases, it is ``easy,'' e.g., it is computable in low-degree 
polynomial time, while in other cases it is ``hard,'' e.g., it is 
intractable in the sense that the corresponding decision problem is NP-hard 
or NP-complete.
\end{itemize}
In the case of an intractable objective, people are often interested in 
computing some sort of approximate solution to optimize the objective that 
has been decided upon.
Alternatively, people may run a procedure without a well-defined objective 
stated and decided upon beforehand, and in some cases this procedure 
returns answers that are useful.
Moreover, the procedures often bear some sort of resemblance to the steps
of algorithms that solve well-defined objectives exactly.
Clearly, there is potential interest in understanding the relationship 
between these two complementary approaches: this will help people who run 
procedures know what they are optimizing; this can feed back and help to 
develop statistically-principled and more-scalable procedures; and so on.

Here, we will focus on several different methods (i.e., classes of 
algorithms, e.g., ``spectral graph algorithms'' as well as other classes 
of methods) that are very widespread in practice and that can be analyzed 
to prove strong bounds on the quality of the partitions found.
The methods are the following.
\begin{enumerate}
\item 
Spectral-based methods. 
This could include either global or local methods, both of which come with
some sort of Cheeger Inequality.
\item 
Flow-based methods. 
These have connections with the min-cut/max-flow theorem, and they can be 
viewed in terms of embeddings via their LP formulation, and here too there 
is a local improvement version.
\end{enumerate}
In addition, we will also probably consider methods that combine spectral 
and flow in various ways.
Note that most or all of the theoretically-principles methods people use 
have steps that boil down to one of these.
Of course, we will also make connections with methods such as local 
improvement heuristics that are less theoretically-principled but that are 
often important in practice.

Before doing that, we should point out something that has been implicit in
the discussion so far.
That is, while computer scientists (and in particular TCS) often draw a 
strong distinction between problems and algorithms, researchers in other 
areas (in particular machine learning and data analysis as well as 
quantitatively-inclined people in nearly every other applied area) often do 
not.
For the latter people, one might run some sort of procedure that solves 
something insofar as, e.g., it finds clusters that are useful by a 
downstream metric.
As you can imagine, there is a proliferation of such methods.
One of the questions we will address is when we can understand those 
procedures in terms of the above theoretically-principled methods.
In many cases, we can; and that can help to understand when/why these
algorithms work and when/why they don't.

Also, while we will mostly focus on a particular objective (called expansion 
or conductance) that probably is the combinatorial objective that most 
closely captures the bicriteria of being well-connected intra-cluster and
not well-connected inter-cluster, we will probably talk about some other 
related methods.
For example, finding dense subgraphs, and finding so-called good-modularity 
partitions.
Those are also of widespread interest; they will illustrate other ways that 
spectral methods can be used; and understanding the relationship between 
those objectives and expansion/conductance is important.

Before proceeding, a word of caution: 
For a given objective quantifying how ``good'' is a partition, it is 
\emph{not} the case that all graphs have good partitions---but all graph 
partitioning algorithms (as will other algorithms) will return some 
answer, \emph{i.e.}, they will give you some output clustering.  
In particular, there is a class of graphs called \emph{expanders} that do 
not have good clusters with respect to the so-called expansion/conductance
objective function.
(Many real data graphs have strong expander-like properties.)

In this case, i.e., when there are no good clusters, the simple answer is 
just to say don't do clustering.
Of course, it can sometimes in practice be difficult to tell if you are in 
that case.
(For example, with a thousand graphs and a thousand methods---that may or 
may not be related but that have different knobs and so are at least 
minorly-different---you are bound to find things look like clusters, and
controlling false discovery, etc., is tricky in general but in particular
for graph-based data.)
Alternatively, especially in practice, you might have a graph that has both
expander-like properties and non-expander-like properties, e.g., in different
parts of the graph.
A toy example of this could be given by the lollipop graph.
In that case, it might be good to know how algorithms behave on different 
classes of graphs and/or different parts of the graph.

Question (raised by this): Can we certify that there are no good clusters 
in a graph?
Or certify the nonexistence of hypothesized things more generally?
We will get back to this later.

Let's go back to finding an objective we want to consider.

As a general rule of thumb, when most people talk about clusters or 
communities (for some reason, in network and especially in social graph 
applications clusters are often called communities---they may have a 
different downstream, e.g., sociological motivation, but operationally they
are typically found with some sort of graph clustering algorithm) 
``desirable'' or ``good'' clusters tend to have the following properties:
\begin{enumerate}
\item  
Internally (intra) - well connected with other members of the cluster. 
Minimally, this means that it should be connected---but it is a challenge 
to guarantee this in a statistically and algorithmically meaningful manner.
More generally, this might mean that it is ``morally 
connected''---\emph{e.g.}, that there are several paths between vertices in 
intra-clusters and that these paths should be internal to the cluster.
(Note: this takes advantage of the fact that we can classify edges incident 
to $v \in C$ as internal (connected to other members of $C$) and 
external (connected to $\bar{C}$).
\item 
Externally (inter) - relatively poor connections between members of a cluster
and members of a different cluster.
For example, this might mean that there are very few edges with one 
endpoint in one cluster and the other endpoint in the other cluster.
\end{enumerate}
Note that this implies that we can classify edges, i.e., pairwise 
connections, incident to a vertex $v \in C$ into edges that are internal 
(connected to other members of $C$) and edges that are external (connected 
to members of $\bar{C}$).
This technically is well-defined; and, informally, it makes sense, since if 
we are modeling the data as a graph, then we are saying that things and
pairwise relationships between things are of primary importance. 

So, we want a relatively dense or well-connected (very informally, those two
notions are similar, but they are often different when one focuses on a 
particular quantification of the informal notion) induced subgraph with 
relatively few inter-connections between pieces.
Here are extreme cases to consider:
\begin{itemize}
\item
Connected component, \emph{i.e.}, the ``entire graph,'' if the graph is 
connected, or one connected component if the graph is not connected.
\item
Clique or maximal clique, \emph{i.e.}, complete subgraph or a maximal 
complete subgraph, \emph{i.e.}, subgraph in which no other vertices can be 
added without loss of the clique property.
\end{itemize}
But how do we quantify this more generally?


\subsection{A first try with min-cuts}

Here we will describe an objective that has been used to partition graphs.
Although it is widely-used for certain applications, it will have certain 
aspects that are undesirable for many other applications.
In particular, we cover it for a few reasons: 
first, as a starter objective before we get to a better objective; 
second, since the dual is related to a non-spectral way to partition graphs; 
and third, although it doesn't take into account the bi-criteria we have 
outlined, understanding it will be a basis for a lot of the stuff later.

\subsubsection{Min-cuts and the Min-cut problem}

We start with the following definition.

\begin{definition}
Let $G=(V,E)$ be a graph.
A \emph{cut} $C=(S,T)$ is a partition of the vertex set $V$ of $G$.
An \emph{$s$-$t$-cut} $C=(S,T)$ of $G=(V,E)$ is a cut $C$ s.t. $s \in S$ and
$t \in T$, where $s,t \in V$ are pre-specified source and sink vertices/nodes.
A \emph{cut set} is $\{(u,v) \in E : u \in S, v \in T \}$, i.e., the edges
with one endpoint on each side of the cut.
\end{definition}
The above definition applies to both directed and undirected graphs. Notice in the directed case, the cut set contains the edges from node in $S$ to nodes in $T$, but not those from $T$ to $S$.


Given this set-up the \emph{min-cut problem} is: find the ``smallest'' 
cut, \emph{i.e.}, find the cut with the ``smallest''  cut set, \emph{i.e.}
the smallest boundary (or sum of edge weights, more generally).
That is:
\begin{definition}
The \emph{capacity} of an $s$-$t$-cut is 
$c(S,T) = \sum_{(u,v)\in(S,\bar{S})} c_{uv}$.
In this case, the \emph{Min-Cut Problem} is to solve 
\[
\min_{s \in S,t \in \bar{S}} c(S,\bar{S})   .
\]
\end{definition}
That is, the problem is to find the ``smallest'' cut, where by smallest we
mean the cut with the smallest total edge capacity across it, i.e., with the smallest ``boundary.''

Things to note about this formalization:
\begin{enumerate}
\item 
Good: Solvable in low-degree polynomial time by a polynomial time algorithm. 
(As we will see, $\mbox{min-cut} = \mbox{max-flow}$ is related.)
\item 
Bad:   Often get very unbalanced cut.
(This is not \emph{necessarily} a problem, as maybe there are no good cuts, 
but for this formalization, this happens even when it is known that there
are good large cuts.
This objective tends to nibble off small things, even when there are bigger
partitions of interest.)
This is problematic for several reasons:
\begin{itemize}
   \item 
   \textbf{In theory.} 
   Cut algorithms are used as a sub-routine in divide and 
   conquer algorithm, and if we keep nibbling off small pieces then the
   recursion depth is very deep; alternatively, control over inference is 
   often obtained by drawing strength over a bunch of data that are
   well-separated from other data, and so if that bunch is very small
   then the inference control is weak.
   \item   
   \textbf{In practice.} 
   Often, we want to ``interpret'' the clusters or partitions, and 
   it is not nice if the sets returned are uninteresting or trivial. 
   Alternatively, one might want to do bucket testing or something 
   related, and when the clusters are very small, it might not be worth 
   the time.
\end{itemize}
(As a forward-looking pointer, we will see that an `improvement'' of the 
idea of cut and min-cut may also get very imbalanced partitions, but it does 
so for a more subtle/non-trivial reason.
So, this is a bug or a feature, but since the reason is somewhat trivial 
people typically view this as a bug associated with the choice of this 
particular objective in many applications.)
\end{enumerate}


\subsubsection{A slight detour: the Max-Flow Problem}

Here is a slight detour (w.r.t. spectral methods per se), but it is one that
we will get back to, and it is related to our first try objective.

Here is a seemingly-different problem called the Max-Flow problem.
\begin{definition}
The \emph{capacity} of an edge $e \in E$ is a mapping 
$c:E\rightarrow\mathbb{R}^{+}$, denoted $c_e$ or $c_{uv}$ (which will be a
constraint on the maximum amount of flow we allow on that edge).
\end{definition}
\begin{definition}
A \emph{flow} in a directed graph is a mapping $f:E\rightarrow\mathbb{R}$, denoted $f_e$ or
$f_{uv}$ s.t.:
\begin{itemize}
\item $f_{uv} \le c_{uv}, \quad \forall (u,v) \in E$ (Capacity constraint.)
\item
$\sum_{v:(u,v) \in E} f_{uv} = \sum_{v:(v,u) \in E} f_{vu}, \quad\forall u \in V\backslash \{s,t\}$ (Conservation of flow, except at source and sink.)
\item $f_e\geq 0 \quad \forall e\in E$ (obey directions)
\end{itemize}
A \emph{flow} in a undirected graph is a mapping $f:E\rightarrow\mathbb{R}$, denoted $f_e$. We arbitrarily assign directions to each edge, say $e=(u,v)$, and when we write $f_{(v,u)}$, it is just a notation for $-f_{(u,v)}$
\begin{itemize}
\item $|f_{uv}| \le c_{uv}, \quad \forall (u,v) \in E$ (Capacity constraint.)
\item
$\sum_{v:(u,v) \in E} f_{uv} = 0, \quad\forall u \in V\backslash \{s,t\}$ (Conservation of flow, except at source and sink.)
\end{itemize}
\end{definition}

\begin{definition}
The \emph{value of the flow} $|f| = \sum_{v \in V} |f_{sv}|$, where $s$ is the
source.  
(This is the amount of flow flowing out of $s$. It is easy to see that as all the nodes other than $s,t$ obey the flow conservation constraint, the flow out of $s$ is the same as the flow into $t$.
This is the amount of flow flowing from $s$ to $t$.)
In this case, the \emph{Max-Flow Problem} is 
\[
\max |f|   .
\]
\end{definition}

Note: what we have just defined is really the ``single commodity flow 
problem'' since there exists only $1$ commodity that we are routing and 
thus only $1$ source/sink pair $s$ and $t$ that we are routing from/to.  
(We will soon see an important generalization of this to something called
\emph{multi-commodity flow}, and this will be very related
to a non-spectral method for graph partitioning.)


Here is an important result that we won't prove.

\begin{theorem}[Max-Flow-Min-Cut Theorem]
The max value of an $s-t$ flow is equal to the min capacity of an $s-t$ cut.
\end{theorem}

Here, we state the Max-Flow problem and the Min-Cut problem, in terms of the 
primal and dual optimization problems.

\textsc{Primal: (Max-Flow)}: 
\begin{align*} 
\max          & \quad |f|    \\
\mbox{s.t.~~} & \quad f_{uv} \le C_{uv}   ,\quad (uv)\in E \\
              & \sum_{v:(vu)\in E} f_{vu} - \sum_{v:(uv)\in E} f_{uv}  \le 0 , \quad u \in V \\
              & f_{uv} \ge 0
\end{align*} 

\textsc{Dual: (Min-Cut)}:
\begin{align*}
\min          & \quad \sum_{(i, j) \in E} c_{ij} d_{ij}   \\
\mbox{s.t.~~} & \quad d_{ij} - p_i + p_j \ge 0, (ij)\in E   \\
              & \quad p_s=1, p_t=0,    \\
              & \quad p_i \geq 0, i \in V    \\
              & \quad d_{ij} \geq 0, ij \in E 
\end{align*}

There are two ideas here that are important that we will revisit.
\begin{itemize}
\item
Weak duality: for any instance and any feasible flows and cuts, 
$\max flow \le \min cut$.
\item
Strong duality: for any instance, $\exists$ feasible flow and feasible cut
s.t. the objective functions are equal, \emph{i.e.}, s.t. 
$\max flow = \min cut$.
\end{itemize}

We are not going to go into these details here---for people who have seen 
it, it is just to set the context, and for people who haven't seen it, it is 
to give an important fyi.
But we will note the following.
\begin{itemize}
\item
Weak duality generalizes to many settings, and in particular to 
multi-commodity flow; but strong duality does not.
The next question is: does there exist a cut s.t. equality is achieved.
\item
We can get an approximate version of strong duality, \emph{i.e.}, an 
approximate Min-Cut-Max-Flow theorem in that multi-commodity case.
That we can get such a bound will have numerous algorithmic implications, 
in particular for graph partitioning.
\item
We can translate this (in particular, the all-pairs multi-commodity flow 
problem) into $2$-way graph partitioning problems (this should not be 
immediately obvious, but we will cover it later) and get nontrivial 
approximation guarantees.
\end{itemize}
About the last point: for flows/cuts we have introduced 
special source and sink nodes, $s$ and $t$, but when we apply it back to 
graph partitioning there won't be any special source/sink nodes, basically
since we will relate it to the all-pairs multi-commodity flow problem, 
i.e., where we consider all ${n \choose 2}$ possible source-sink pairs.


\subsection{Beyond simple min-cut to ``better'' quotient cut objectives}

The way we described what is a ``good'' clustering above was in terms of 
an intra-connectivity versus intra-connectivity bi-criterion.
So, let's revisit and push on that.
A related thing or a different way (that gives the same result in many 
cases, but sometimes does not) is to say that a bi-criterion is:
\begin{itemize}
\item
We want a good ``cut value''---not too many crossing edges---where 
\emph{cut value} is $E(S, \bar{S})$ or a weighted version of that.
I.e., what we just considered with min-cut.
\item
We want good ``balance'' properties---\emph{i.e.}, both sides of the cut should
be roughly the same size---so both $S, \bar{S}$ are the same size or
approximately the same size.
\end{itemize}


There are several ways to impose a balance condition.
Some are richer or more fruitful (in theory and/or in practice) than others.
Here are several.
First, we can add ``hard'' or \emph{explicit} balance conditions:
\begin{itemize}
\item
Graph bisection---find a min cut s.t. $|S| = |\bar{S}| = n/2$, \emph{i.e.}, 
ask for exactly $50$-$50$ balance.
\item
$\beta$-balanced cut---find a min cut s.t $|S| = \beta n $, 
$ |\bar{S}| = (1-\beta) n$, \emph{i.e.}, give a bit of wiggle room and ask 
for exactly (or more generally no worse than), say, a $70$-$30$ balance.
\end{itemize}
Second, there are also ``soft'' or \emph{implicit} balance conditions, where
there is a penalty and separated nodes ``pay'' for edges in the cut.
(Actually, these are not ``implicit'' in the way we will use the word later;
here it is more like ``hoped for, and in certain intuitive cases it is true.''
And they are not quite soft, in that they can still lead to imbalance; but
when they do it is for much more subtle and interesting reasons.)
Most of these are usually formalized as quotient-cut-style objectives:
\begin{itemize}
\item 
Expansion: 
$\frac{E(S,\bar{S})}{\frac{|S|}{n}}$ or
$\frac{E(S,\bar{S})}{\min\{|S|,|\bar{S}|\}}$ 
(def this as :h(S) )
(a.k.a. q-cut)
\item 
Sparsity: 
$\frac{E(S, \bar S)}{|S| |\bar S|}$   
(def this as :sp(S) )
(a.k.a. approximate-expansion)
\item 
Conductance: 
$\frac{E(S, \bar S)}{\frac{\mbox{Vol}\left(S\right)}{n}}$ 
or $\frac{E(S,\bar{S})}{\min(\mbox{Vol}(|S|),\mbox{Vol}(|\bar{S}|))}$
(a.k.a. Normalized cut)
\item 
Normalized-cut: 
$\frac{E(S,\bar{S})}{\mbox{Vol}(|S|)\cdot\mbox{Vol}(|\bar{S}|)}$
(a.k.a. approximate conductance)
\end{itemize}
Here, $E(S,\bar{S})$ is the number of edges between $S$ and $\bar{S}$, or 
more generally for a weighted graph the sum of the edge weights between $S$ 
and $\bar{S}$, and $\mbox{Vol}(S) = \sum_{ij\in E}{ \mbox{deg}(V_i)}$.
In addition, the denominator in all four cases correspond to different 
volume notions: the first two are based on the number of nodes in $S$, and
the last two are based on the number of edges in $S$ (i.e., the sum of the
degrees of the nodes in $S$.)

Before proceeding, it is worth asking what if we had taken a difference
rather than a ratio, e.g., $SA-VOL$, rather than $SA/VOL$.
At the high level we are discussing now, that would do the same thing---but, 
quantitatively, using an additive objective will generally give 
very different results than using a ratio objective, in particular when one 
is interested in fairly small clusters.

(As an FYI, the first two, i.e., expansion and sparsity, are typically used 
in the theory algorithms algorithms, since they tend to highlight the 
essential points; while the latter two, i.e., conductance and normalized 
cuts, are more often used in data analysis, machine learning, and other 
applications, since issues of normalization are dealt with better.)


Here are several things to note:
\begin{itemize}
\item
Expansion provides a \emph{slightly} stronger bias toward being 
well-balanced than sparsity (i.e. between a $10-90$ cut and a $40-60$ cut, the advantage in the denominator for the more balanced $40-60$ cut in expansion is $4:1$, while it is $2400:900<4:1$ in sparsity)
(
and there \emph{might}
be some cases where this is important.
That is, the product variants have a ``factor of $2$'' weaker preference
for balance than the min variants.
Similarly for normalized cuts versus conductance.
\item
That being said, that difference is swamped by the following.
Expansion and sparsity are the ``same''  (in the following sense:)
\[\min h(S) \approx \min  sp(S)\]
Similarly for normalized cuts versus conductance.
\item
Somewhat more precisely, although the expansion of any particular set 
isn't in general close to the sparsity of that set, The expansion problem and the sparsity problem are equivalent in the following sense:\\

It is clear that 
\[\argmin_S \Phi'(G)=\argmin_S n\Phi'(G)=\argmin_S \frac{C(S,\bar{S})}{\min\{|S|,|\bar{S}|\}}\frac{n}{\max\{|S|,|\bar{S}|\}}
\]
As $1<\frac{n}{\max\{|S|,|\bar{S}|\}}\leq 2$
, the min partition we find by optimizing sparsity will also give, off by a multiplicative factor of $2$, the optimal expansion.
As we will see this is small compared to $O(\log n)$ approximation from flow 
or the quadratic factor with Cheeger, and so is not worth worrying about from 
an optimization perspective.
Thus, we will be mostly cavalier about going back and forth.
\item
Of course, the sets achieving the optimal may be very different.
An analogous thing was seen in vector spaces---the optimal may rotate by 
$90$ degrees, but for many things you only need that the Rayleigh quotient 
is approximately optimal. 
Here, however, the situation is worse.
Asking for the certificates achieving the optimum is a more difficult 
thing---in the vector space case, this means making awkward ``gap'' 
assumptions, and in the graph case it means making strong and awkward 
combinatorial statements.
\item
Expansion $\ne$ Conductance, in general, except for regular graphs.
(Similarly, Sparsity $\ne$ Normalized Cuts, in general, except for regular 
graphs.)
The latter is in general preferable for heterogeneous graphs, \emph{i.e.}, 
very irregular graphs.
The reason is that there are closer connections with random walks and we get 
tighter versions of Cheeger's inequality if we take the weights into account.
\item
Quotient cuts capture exactly the surface area to volume bicriteria that
we wanted.
(As a forward pointer, a question is the following: what does this mean if 
the data come from a low dimensional space versus a high dimensional space;
or if the data are more or less expander-like; and what is the relationship 
between the original data being low or high dimensional versus the graph 
being expander-like or not?
\item
For ``space-like'' graphs, 
these two bicriteria as ``synergistic,'' in that they work together;
for expanders, they are ``uncoupled,'' in that the best cuts don't depend
on size, as they are all bad;
and for many ``real-world'' heavy-tailed informatics graphs, they are
``anti-correlated,'' in that better balance means worse cuts.
\item
An obvious question: are there other notions, e.g., of ``volume'' that 
might be useful and that will lead to similar results we can show about
this?
(In some cases, the answer to this is yes: we may revisit this later.)
Moreover, one might want to choose a different reweighting for 
statistical or robustness reasons.
\end{itemize}


(We will get back to the issues raised by that second-to-last point later 
when we discuss ``local'' partitioning methods.
We simply note that ``space-like'' graphs include, \emph{e.g.}, 
$\mathcal{Z}^{2}$ or random geometric graphs or ``nice'' planar graphs or 
graphs that ``live'' on the earth. 
More generally, there is a trade-off and we might get very imbalanced
clusters or even disconnected clusters. 
For example, for the $G(n, p)$ random graph model 
if $p \ge \log n^2 / n $ then we have an expander,
while for extremely sparse random graphs, \emph{i.e.}, $ p <~ \log n/n  $, 
then due to lack of concentration we can have deep small cuts but be 
expander-like at larger size scales.)


\subsection{Overview Graph Partition Algorithms}

Here, we will briefly describe the ``lay of the land'' when it comes to 
graph partitioning algorithms---in the next few classes, we will go into a 
lot more details about these methods.
There are three basic ideas you need to know for graph partitioning, in 
that nearly all methods can be understood in terms of some combination of 
these methods.
\begin{itemize}
\item
Local Improvement (and multi-resolution).
\item
Spectral Methods.
\item
Flow-based Methods.
\end{itemize}
As we will see, in addition to being of interest in clustering data 
graphs, graph partitioning is a nice test-case since it has been very 
well-studied in theory and in practice and there exists a large number of 
very different algorithms, the respective strengths and weaknesses of which
are well-known for dealing with it.

\subsubsection{Local Improvement}

\emph{Local improvement} methods refer to a class of methods that take an 
input partition and do more-or-less naive steps to get a better partition:
\begin{itemize}
\item
$70$s Kernighan-Lin.
\item
$80$s Fiduccia-Mattheyses---FM and KL start with a partition and improve the 
cuts by flipping nodes back and forth.
Local minima can be a big problem for these methods.
But they can be useful as a post-processing step---can give a big difference
in practice.
FM is better than KL since it runs in linear time, and it is still commonly 
used, often in packages.
\item
$90$s Chaco, Metis, etc.
In particular, METIS algorithm from Karypis and Kumar, works very well in 
practice, especially on low dimensional graphs.
\end{itemize}
The methods of the $90$s used the idea of local improvement, coupled with 
the basically linear algebraic idea of \emph{Multiresolution} get algorithms
that are designed to work well on space-like graphs and that can perform 
very well in practice.
The basic idea is:
\begin{itemize}
\item
Contract edges to get a smaller graph.
\item
Cut the resulting graph.
\item
Unfold back up to the original graph.
\end{itemize}
Informally, the basic idea is that if there is some sort of geometry, say 
the graph being partitioned is the road network of the US, i.e., that lives 
on a two-dimensional surface, then we can ``coarse grain'' over the 
geometry, to get effective nodes and edges, and then partition over the 
coarsely-defined graph.
The algorithm will, of course, work for any graph, and one of the 
difficulties people have when applying algorithms as a black box is that 
the coarse graining follows rules that can behave in funny ways when applied
to a graph that doesn't have the underlying geometry.

Here are several things to note:
\begin{itemize}
\item
These methods grew out of scientific computing and parallel processing, 
so they tend to work on ``space-like'' graphs, where there are nice
homogeneity properties---even if the matrices aren't low-rank, they might
be diagonal plus low-rank off-diagonal blocks for physical reasons, or
whatever.
\item
The idea used previously to speed up convergence of iterative methods.
\item
Multiresolution allows globally coherent solutions, so it avoids some of 
the local minima problems.
\item
$90$s: Karger showed that one can compute min-cut by randomly contracting 
edges, and so multiresolution may not be just changing the resolution at 
which one views the graph, but it may be taking advantage of this property 
also.
\end{itemize}

An important point is that local improvement (and even multiresolution 
methods) can easily get stuck in local optima.
Thus, they are of limited interest by themselves.
But they can be very useful to ``clean up'' or ``improve'' the output of 
other methods, e.g., spectral methods, that in a principled way lead to a 
good solution, but where the solution can be improved a bit by doing 
some sort or moderately-greedy local improvement.


\subsubsection{Spectral methods}

\emph{Spectral methods} refer to a class of methods that, at root, are a 
relaxation or rounding method derived from an NP-hard QIP and that involves
eigenvector computations.
In this case that we are discussing, it is the QIP formulation of the graph 
bisection problem that is relaxed.
Here is a bit of incomplete history.
\begin{itemize}
\item
Donath and Hoffman (ca. 72,73) introduced the idea of using the leading 
eigenvector of the Adjacency Matrix $A_G$ as a heuristic to find good 
partitions.
\item
Fiedler (ca. 73) associated the second smallest eigenvalue of the 
Laplacian $L_G$ with graph connectivity and suggested splitting the graph by 
the value along the associated eigenvector.
\item
Barnes and Hoffman (82,83) and Bopanna (87)
used LP, SDP, and convex programming methods to look at the leading
nontrivial eigenvector.
\item
Cheeger (ca. 70) established connections with isoperimetric relationships 
on continuous manifolds, establishing what is now known ad Cheeger's 
Inequality.
\item
$80$s: saw performance guarantees from Alon-Milman, Jerrum-Sinclair, etc.,
connecting $\lambda_2$ to expanders and rapidly mixing Markov chains.
\item
80s: saw improvements to approximate eigenvector computation, e.g., Lanczos
methods, which made computing eigenvectors more practical and easier.
\item
80s/90s: saw algorithms to find separators in certain classes of graphs, e.g., 
planar graphs, bounds on degree, genus, etc.
\item
Early 90s: saw lots of empirical work showing that spectral partitioning works
for ``real'' graphs such as those arising in scientific computing 
applications
\item
Spielman and Teng (96) 
showed that ``spectral partitioning works'' on bounded degree planar graphs 
and well-shaped meshed, i.e., in the application where it is usually applied.
\item
Guattery and Miller (95, 97)
showed ``spectral partitioning doesn't work'' on certain classes of graphs, 
e.g., the cockroach graph, in the sense that there are graphs for which the
quadratic factor is achieved.
That particular result holds for vanilla spectral, but similar constructions 
hold for non-vanilla spectral partitioning methods.
\item
Leighton and Rao (87, 98)
established a bound on the duality gap for multi-commodity flow problems, 
and used multi-commodity flow methods to get an $O(\log n)$ approximation to
the graph partitioning problem.
\item
LLR (95)
considered the geometry of graphs and algorithmic applications, and 
interpreted LR as embedding $G$ in a metric space, making the connection 
with the $O(\log n)$ approximation guarantee via Bourgain's embedding 
lemma.
\item
90s: saw lots of work in TCS on LP/SDP relaxations of IPs and randomized 
rounding to get $\{\pm 1\}$ solutions from fractional solutions.
\item
Chung (97)
focused on the normalized Laplacian for degree irregular graphs and the 
associated metric of conductance.
\item
Shi and Malik (99)
used normalized cuts for computer vision applications, which is essentially 
a version of conductance.
\item
Early 00s:
saw lots of work in ML inventing and reinventing and reinterpreting spectral 
partitioning methods, including relating it to other problems like 
semi-supervised learning and prediction (with, e.g., boundaries between 
classes being given by low-density regions).
\item
Early 00s:
saw lots of work in ML on manifold learning, etc., where one constructs a 
graph and recovers an hypothesized manifold; constructs graphs for 
semi-supervised learning applications; and where the diffusion/resistance 
coordinates are better or more useful/robust than geodesic distances.
\item
ARV (05)
got an SDP-based embedding to get an $O(\sqrt{\log n})$ approximation, which
combined ideas from spectral and flow; and there was related follow-up work.
\item
00s: saw local/locally-biased spectral methods and improvements to flow 
improve methods.
\item
00s: saw lots of spectral-like methods like viral diffusions with 
social/complex networks. 
\end{itemize}

For the moment and for simplicity, say that we are working with unweighted 
graphs.
The graph partitioning QIP is:
\begin{align*}
\min          & \quad x^TLx    \\
\mbox{s.t.~~} & \quad x^T1=0   \\
              & x_i\in\{-1,+1\}
\end{align*}
and the spectral relaxation is:
\begin{align*}
\min          & \quad x^TLx    \\
\mbox{s.t.~~} & \quad x^T1=0   \\
              & x_i\in\mathbb{R},\quad x^Tx=n
\end{align*}
That is, we relax $x$ from being in $\{-1,1\}$, which is a discrete/combinatorial 
constraint, to being a real continuous number that is $1$ ``on average.''
(One could relax in other ways---\emph{e.g.}, we could relax to say that 
it's magnitude is equal to $1$, but that it sits on a higher-dimensional
sphere.  We will see an example of this later.  Or other things, like 
relaxing to a metric.)
This spectral relaxation is not obviously a nice problem, e.g., it is not
even convex; but it can be shown 
that the solution to this relaxation can be computed as the second smallest 
eigenvector of $L$, the Fiedler vector, so we can use an eigensolver to get 
the eigenvector.

Given that vector, we then have to perform a \emph{rounding} to get an actual 
cut.
That is, we need to take the real-valued/fractional solution obtained from 
the continuous relaxation and round it back to $\{-1,+1\}$.
There are different ways to do that.

So, here is the basic spectral partitioning method.
\begin{itemize}
\item 
Compute an eigenvector of the above program. 
\item
Cut according to some rules, \emph{e.g.}, do a hyperplane rounding, or 
perform some other more complex rounding rule.
\item
Post process with a local improvements method.
\end{itemize}
The hyperplane rounding is the easiest to analyze, and we will do it here, 
but not surprisingly factors of $2$ can matter in practice; and 
so---\emph{when spectral is an appropriate thing to do}---other rounding 
rules often do better in practice.
(But that is a ``tweak'' on the larger question of spectral versus flow 
approximations.)
In particular, we can do local improvements here to make the output slightly
better in practice.
Also, there is the issue of what exactly is a rounding, e.g., if one 
performs a sophisticated flow-based rounding then one may obtain a better 
objective function but a worse cut value.
Hyperplane rounding involves:
\begin{itemize}
\item
Choose a split point $\hat{x}$ along the vector $x$
\item
Partition nodes into $2$ sets: $\{x_i<\hat{x}\}$ and $\{x_i > \hat{x} \}$
\end{itemize}
By \emph{Vanilla spectral}, we refer to spectral with hyperplane rounding 
of the Fiedler vector embedding.

Given this setup of spectral-based partitioning, what can go ``wrong'' with
this approach.
\begin{itemize}
\item
We can choose the wrong direction for the cut:
\begin{itemize}
\item
Example---Guattery and Miller construct an example that is ``quadratically
bad'' by taking advantage of the confusion that spectral has between 
``long paths'' and ``deep cuts.''
\item
Random walk interpretation---long paths can also cause slow mixing since 
the expected progress of a $t$-step random walk is $O(\sqrt{t})$.
\end{itemize}
\item
The hyperplane rounding can hide good cuts:
\begin{itemize}
\item
In practice, it is often better to post-process with FM to improve the 
solution, especially if want good cuts, i.e., cuts with good objective 
function value.
\end{itemize}
\end{itemize}
An important point to emphasize is that, although both of these examples of
``wrong'' mean that the task one is trying to accomplish might not work, 
i.e., one might not find the best partition, sometimes that is not all bad.
For example, the fact that spectral methods ``strip off'' long stringy 
pieces might be ok if, e.g., one obtains partitions that are ``nice'' in 
other ways.
That is, the direction chosen by spectral partitioning might be nice or 
regularized relative to the optimal direction.
We will see examples of this, and in fact it is often for this reason that
spectral performs well in practice.
Similarly, the rounding step can also potentially give an implicit 
regularization, compared to more sophisticated rounding methods, and we will 
return to discuss this.


\subsubsection{Flow-based methods}

There is another class of methods that uses very different ideas to 
partition graphs.
Although this will not be our main focus, since they are not spectral 
methods, we will spend a few classes on it, and it will be good to know 
about them since in many ways they provide a strong contrast with spectral 
methods.

This class of flow-based methods uses the ``all pairs'' multicommodity flow 
procedure to reveal bottlenecks in the graph.
Intuitively, flow should be ``perpendicular'' to the cut (i.e. in the sense of complementary slackness for LPs, and similar relationship between primal/dual variables to dual/primal constraints in general).
The idea is to route a large number of commodities \emph{simultaneously}
between random pairs of nodes and then choose the cut with the most edges
congested---the idea being that a bottleneck in the flow computation 
corresponds to a good~cut.

Recall that the single commodity max-flow-min-cut procedure has zero 
duality gap, but that is not the case for multi-commodity problem.
On the other hand, the $k$-multicommodity has $O(\log k)$ duality 
gap---this result is due to LR and LLR, and it says that there is an 
\emph{approximate} min-flow-max-cut.
Also, it implies an $O(\log n)$ gap for the all pairs problem.

The following is an important point to note.
\begin{claim}
The $O(\log n)$ is tight on expanders.
\end{claim}


For flow, there are connections to embedding and linear programming, so
as we will see, we can think of the algorithm as being:
\begin{itemize}
\item 
Relax flow to LP, and solve the LP.
\item 
Embed solution in the $\ell_1$ metric space.
\item 
Round solution to $\{0,1\}$.
\end{itemize}

\subsection{Advanced material and general comments}

We will conclude with a brief discussion of these results in a broader 
context.
Some of these issues we may return to later.

\subsubsection{Extensions of the basic spectral/flow ideas}

Given the basic setup of spectral and flow methods, both of which come 
with strong theory, here are some extensions of the basic ideas.
\begin{itemize}
\item
Huge graphs.
Here want to do computations depending on the size of the sets and not the
size of the graph, \emph{i.e.}, we don't even want to touch all the nodes 
in the graph, and we want to return a cut that is nearby an input seed set 
of nodes.
This includes ``local'' spectral methods---that take advantage of diffusion 
to approximate eigenvectors and get Cheeger-like guarantees.
\item
Improvement Methods.
Here we want to ``improve'' an input partition---there are both spectral 
and flow versions.
\item
Combining Spectral and Flow.
\begin{itemize}
\item
ARV solves an SDP, which takes time like $O(n^{4.5})$ or so; 
but we can do it faster (\emph{e.g.}, on graphs with $\approx 10^{5}$ 
nodes) using ideas related to approximate multiplicative weights.
\item
There are strong connections here to online learning---roughly since we can 
view ``worst case'' analysis as a ``game'' between a cut player and a 
matching player.
\item
Similarly, there are strong connections to boosting, which suggest that 
these combinations might have interesting statistical properties.
\end{itemize}
\end{itemize}

A final word to reemphasize: at least as important for what we will be doing 
as understanding when these methods work is understanding when these methods 
``fail''---that is, when they achieve their worst case 
quality-of-approximation guarantees:
\begin{itemize}
\item
Spectral methods ``fail'' on graphs with ``long stringy'' pieces, like that 
constructed by Guattery and Miller.
\item
Flow-based methods ``fail'' on expander graphs (and, more generally, on 
graphs where most of the $\binom{n}{2}$ pairs but most pairs are far 
$\log n$ apart).
\end{itemize}
Importantly, a lot of real data have ``stringy'' pieces, as well as 
expander-like parts; and so it is not hard to see artifacts of spectral 
and flow based approximation algorithms when they are run on real data.

\subsubsection{Additional comments on these methods}

Here are some other comments on spectral versus flow.
\begin{itemize}
\item
The SVD gives good ``global'' but not good ``local'' guarantees.
For example, it provides global reconstruction error, and going to the 
low-dimensional space might help to speed up all sorts of algorithms; but 
any pair of distances might be changed a lot in the low-dimensional space, 
since the distance constraints are only satisfied on average.
This should be contrasted with flow-based embedding methods and all sorts 
of other embedding methods that are used in TCS and related areas, where 
one obtains very strong local or pairwise guarantees.
There are two important (but not immediately obvious) consequences of this.
\begin{itemize}
\item
The lack of local guarantees makes it hard to exploit these embeddings 
algorithmically (in worst-case), whereas the pair-wise guarantees provided 
by other types of embeddings means that you can get worst-case bounds and
show that the solution to the subproblem approximates in worst case the 
solution to the original problem.
\item
That being said, the global guarantee means that one obtains results that are 
more robust to noise and not very sensitive to a few ``bad'' distances, which 
explains why spectral methods are more popular in many machine learning and
data analysis applications.
\item
That local guarantees hold for all pair-wise interactions to get worst-case
bounds in non-spectral embeddings essentially means that we are 
``overfitting'' or ``most sensitive to'' data points that are most far 
apart.
This is counter to a common design principle, e.g., exploited by Gaussian 
rbf kernels and other NN methods, that the most reliable information in the 
data is given by nearby points rather than far away points.
\end{itemize}
\end{itemize}


\subsection{References}
\begin{enumerate}
\item Schaeffer, "Graph Clustering", Computer Science Review 1(1): 27-64, 2007
\item Kernighan, B. W.; Lin, Shen (1970). "An efficient heuristic procedure for partitioning graphs". Bell Systems Technical Journal 49: 291-307.
\item CM Fiduccia, RM Mattheyses. "A Linear-Time Heuristic for Improving Network Partitions". Design Automation Conference. 
\item G Karypis, V Kumar (1999). "A Fast and High Quality Multilevel Scheme for Partitioning Irregular Graphs". Siam Journal on Scientific Computing.
\end{enumerate}

%% file: lect06.tex
\section{%
(02/10/2015): 
Spectral Methods for Partitioning Graphs (1 of 2):
Introduction to spectral partitioning and Cheeger's Inequality}

Reading for today.
\begin{compactitem}
\item
``Lecture Notes on Expansion, Sparsest Cut, and Spectral Graph Theory,'' by Trevisan
\end{compactitem}

Today and next time, we will cover what is known as \emph{spectral graph 
partitioning}, and in particular we will discuss and prove Cheeger's 
Inequality.
This result is central to all of spectral graph theory as well as a wide 
range of other related spectral graph methods.
(For example, the isoperimetric ``capacity control'' that it provides 
underlies a lot of classification, etc. methods in machine learning that are 
not explicitly formulated as partitioning problem.)
Cheeger's Inequality relates the quality of the cluster found with spectral 
graph partitioning to the best possible (but intractable to compute) 
cluster, formulated in terms of the combinatorial objectives of 
expansion/conductance. 
Before describing it, we will cover a few things to relate what we have done 
in the last few classes with how similar results are sometimes presented 
elsewhere.

\subsection{Other ways to define the Laplacian}

Recall that $L=D-A$ is the graph Laplacian, or we could work with the 
normalized Laplacian, in which case $L=I-D^{-1/2}AD^{-1/2}$.
While these definition might not make it obvious, the Laplacian actually has
several very intuitive properties (that could alternatively be used as 
definitions).
Here, we go over two of these.

\subsubsection{As a sum of simpler Laplacians}

Again, let's consider $d$-regular graphs. 
(Much of the theory is easier for this case, and expanders are more extremal in this case; but the theory goes through to degree-heterogeneous 
graphs, and this will be more natural in many applications, and so we will 
get back to this later.)

Recall the definition of the Adjacency Matrix of an unweighted graph 
$G=(V,E)$:
\[ 
A_{ij} = \left\{ \begin{array}{l l}
                    1 & \quad \text{if $(ij)\in E$}\\
                    0 & \quad \text{otherwise}
                 \end{array} 
         \right.  ,
\]
In this case, we can define the Laplacian as $L=D-A$ or the normalized 
Laplacian as $L=I-\frac{1}{d}A$.

Here is an alternate definition for the Laplacian $L=D-A$.
Let $G_{12}$ be a graph on two vertices with one edge between those two 
vertices, and define 
\[
L_{G_{12}} = \left( \begin{array}{cc} 1  & -1       \\
                                      -1 &  1 
                    \end{array} 
             \right)
\]
Then, given a graph on $n$ vertices with just one edge between vertex $u$ 
and $v$, we can define $L$ to be the all-zeros matrix, except for the 
intersection between the $u^{th}$ and $v^{th}$ column and row, where we 
define that intersection to be 
$$ L_{G_{uv}} = \left( \begin{array}{cc} 1  & -1       \\
                                         -1 &  1 
                       \end{array} 
                \right)  .$$
Then, for a general graph $G=(V,E)$, one we can define
\[
L_G = \sum_{(u,v) \in E} L_{G_{uv}}   .
\]
This provides a simpler way to think about the Laplacian and in particular 
changes in the Laplacian, e.g., when one adds or removes edges.
In addition, note also that this generalizes in a natural way to 
$ \sum_{(u,v) \in E} w_{uv}L_{G_{uv}}$ if the graph $G=(V,E,W)$ is weighted.

\textbf{Fact.}
This is identical to the definition $L=D-A$.
It is simple to prove this.

From this characterization, several things follow easily.
For example, 
\[
x^TLx = \sum_{(u,v) \in E} w_{uv} \left(x_u-x_v\right)^2  ,
\]
from which it follows that if $v$ is an eigenvector of $L$ with eigenvalue 
$\lambda$, then $v^TLv = \lambda v^Tv \ge 0$.
This means that every eigenvalue is nonnegative, i.e., $L$ is SPSD.

\subsubsection{In terms of discrete derivatives}

Here are some notes that I didn't cover in class that relate the Laplacian 
matrix to a discrete notion of a derivative.

In classical vector analysis, the Laplace operator is a differential 
operator given by the divergence of the gradient of a function in Euclidean 
space.  It is denoted:
\[
\nabla\cdot\nabla 
\mbox{ or }
\nabla^2
\mbox{ or }
\triangle
\]
In the cartesian coordinate system, it takes the form:
\[
\nabla = \left( \frac{\partial}{\partial x_1},\cdots,\frac{\partial}{\partial x_n} \right)  ,
\]
and so 
\[
\triangle f = \sum_{i=1}^{n} \frac{\partial^2 f}{\partial x_i^2} .
\]
This expression arises in the analysis of differential equations of many 
physical phenomena, e.g., electromagnetic/gravitational potentials, 
diffusion equations for heat/fluid flow, wave propagation, quantum 
mechanics, etc.

The \emph{discrete Laplacian} is defined in an analogous manner.
To do so, somewhat more pedantically, let's introduce a discrete analogue 
of the gradient and divergence operators in graphs.

Given an undirected graph $G=(V,E)$ (which for simplicity we take as 
unweighted), fix an \emph{arbitrary} orientation of the edges.
Then, let $K\in\mathbb{R}^{V \times E}$ be the edge-incidence matrix of 
$G$, defined as 
\[ 
K_{ue} = \left\{ \begin{array}{l l}
                    +1 & \quad \text{if edge $e$ exits vertex $u$}\\
                    -1 & \quad \text{if edge $e$ enters vertex $u$}\\
                    0  & \quad \text{otherwise}
                 \end{array} 
         \right.  .
\]
Then, 
\begin{itemize}
\item
define the \emph{gradient} as follows:
let $f:V\rightarrow\mathbb{R}$ be a function on vertices, viewed as a row
vector indexed by $V$;
then $K$ maps $f \rightarrow fK$, a vector indexed by $E$, measures the 
change of $f$ along edges of the graph; and if $e$ is an edge from $u$ to 
$v$, then $\left(fK\right)_{e} = f_u - f_v$.
\item
define the \emph{divergence} as follows:
let $g:E\rightarrow\mathbb{R}$ be a function on edges, viewed as a column 
vector indexed by $E$;
then $K$ maps $g \rightarrow Kg$, a vector indexed by $V$; if we think of 
$g$ as describing flow, then its divergence at vertex is the net outbound 
flow: $\left(Kg\right)_{v} = \sum_{e \mbox{ exits } v} g_e 
                           - \sum_{e \mbox{ enters } v} g_v $
\item
define the \emph{Laplacian} as follows:
it should map $f$ to $KK^Tf$, where $f:V\rightarrow\mathbb{R}$.
So, $L = L_G = KK^T$ is the discrete Laplacian.
\end{itemize}
Note that it is easy to show that 
\[ 
L_{uv} = \left\{ \begin{array}{l l}
                    -1            & \quad \text{if $(u,v) \in E$} \\
                    \mbox{deg}(u) & \quad \text{if $u=v$}  \\
                    0             & \quad \text{otherwise}
                 \end{array} 
         \right.  ,
\]
which is in agreement with the previous definition.
Note also that 
\[
fLf^T = fKK^Tf = \|fK\|_2^2 = \sum_{(u,v)\in E} \left( f_u-f_v \right)^2 ,
\]
which we will later interpret as a smoothness condition for functions on 
the vertices of the graph.

\subsection{Characterizing graph connectivity}

Here, we provide a characterization in terms of eigenvalues of the 
Laplacian of whether or not a graph is connected.
Cheeger's Inequality may be viewed as a ``soft'' version of this result.

\subsubsection{A Perron-Frobenius style result for the Laplacian}

What does the Laplacian tell us about the graph?
A lot of things.
Here is a start.
This is a Perron-Frobenius style result for the Laplacian.

\begin{theorem}
\label{thm:perrof-frob-lapl}
Let $G$ be a $d$-regular undirected graph, let $L=I-\frac{1}{d}A$ be the 
normalized Laplacian; and 
let $\lambda_1 \le \lambda_2 \le \cdots \le \lambda_n$ be the real 
eigenvalues, including multiplicity.
Then: 
\begin{enumerate}
\item
$\lambda_1 = 0$, and the associated eigenvector
$x_{1} = \frac{\vec{1}}{\sqrt{n}} 
               = \left( \frac{1}{\sqrt{n}}, \ldots, \frac{1}{\sqrt{n}} \right)$.
\item
$\lambda_2 \le 2$.
\item
$\lambda_k = 0$ iff $G$ has at least $k$ connected components.
(In particular, $\lambda_2>0$ iff $G$ is connected.)
\item
$\lambda_n = 2$ iff at least one connected component is bipartite.
\end{enumerate}
\end{theorem}
\begin{Proof}
Note that if $x\in\mathbb{R}^{n}$, then 
$x^TLx= \frac{1}{d}\sum_{(u,v)\in E} \left( x_u - x_v \right)^{2}$ and also
\[
\lambda_1 =   \min_{x\in\mathbb{R}^{n}\diagdown\{0\}} 
              \frac{x^TLx}{x^Tx} 
          \ge 0  .
\]
Take $\vec{1} = \left(1,\ldots,1\right)$, in which case 
$\vec{1}^TL\vec{1}=0$, and so $0$ is the smallest eigenvalue, and $\vec{1}$ 
is one of the eigenvectors in the eigenspace of this eigenvalue. This proves part $1$.

We also have the following formulation of $\lambda_k$ by Courant-Fischer:
\[
\lambda_k = \min_{\substack{S\subseteq \mathbb{R}^n \\  dim(S)=k}}\max_{x\in S\diagdown\{\vec{0}\}}\frac{x^TAx}{x^Tx}
            \frac{\sum_{(u,v)\in E} \left(x_u-x_v\right)^{2}}{d \sum_u x_u^2}
\]

So, if $\lambda_k = 0$, then $\exists$ a $k$-dimensional subspace $S$ such 
that $\forall x \in S$, we have 
$\sum_{(u,v)\in E} \left(x_u-x_v\right)^{2}=0$.
But this means that $\forall x\in S$, we have $x_u=x_v \quad \forall$ edges $(u,v)\in E$ with 
positive weight, and so $x_u = x_v$, for any $u,v$ 
in the same connected component.
This means that $x \in S$ is constant within each connected component of $G$.
So, $k=\mbox{dim}(S) \le \Xi$, where $\Xi$ is the number of connected 
components.

Conversely, if $G$ has $\ge k$ connected components, then we can let $S$ be 
the space of vectors that are constant on each component; and this $S$ has dimension $\ge k$. Furthermore, $\forall x \in S$, we have that 
$\sum_{(u,v) \in E} \left( x_u-x_v \right)^{2} = 0$. Thus $\max_{x\in S_k\diagdown\{\vec{0}\}}\frac{x^TAx}{x^Tx}=0$ for any dimension $k$ subspace $S_k$ of the $S$ we choose. Then it is clear from Courant-Fischer $\lambda_k=0$ as any $S_k$ provides an upperbound.

Finally, to study $\lambda_n = 2$, note that 
\[
\lambda_n = \max_{x\in\mathbb{R}^{n}\diagdown\{\vec{0}\}} \frac{x^TLx}{x^Tx}
\]
This follows by using the variational characterization of the eigenvalues of 
$-L$ and noting that $-\lambda_n$ is the smallest eigenvalue of $-L$.
Then, observe that $\forall x \in \mathbb{R}^{n}$, we have that
\[
2 - \frac{x^TLx}{x^Tx} = \frac{\sum_{(u,v) \in E} \left( x_u + x_v \right)^{2}}{d\sum_u x_u^2}\geq 0 ,
\]
from which it follows that $\lambda_n \le 2$ (also $\lambda_k\le 2$ for all $k=2,\ldots,n$).

In addition, if $\lambda_n = 2$, then $\exists x\ne 0$ s.t. 
$\sum_{(u,v)\in E} \left(x_u + x_v \right)^{2} = 0$.
This means that $x_u = -x_v$, for all edges $(u,v)\in E$.

Let $v$ be a vertex s.t. $x_v = a \ne 0$.
Define sets
\begin{eqnarray*}
A &=& \{ v: x_v = a \} \\
B &=& \{ v: x_v = -a \} \\
R &=& \{ v: x_v \ne \pm a \}  .
\end{eqnarray*}
Then, the set $A \cup B$ is disconnected from the rest of the graph, since 
otherwise an edge with an endpoint in $A \cup B$ and the other endpoint in $R$ 
would give a positive contribution to 
$\sum_{ij} A_{ij} \left( x_i + x_j \right)^{2} $.
Also, every edge incident on $A$ has other endpoint in $B$, and vice versa.
So $A \cup B$ is a bipartite connected component (or a collection of connected 
components) of $G$, with bipartition $A,B$.
\end{Proof}

(Here is an aside.
That proof was from Trevisan; Spielman has a somewhat easier proof, but it 
is only for two components.
I need to decide how much I want to emphasize the possibility of using $k$ 
eigenvectors for soft partitioning---I'm leaning toward it, since several 
students asked about it---and if I do I should probably go with the version 
of here that mentions $k$ components.)
As an FYI, here is Spielman's proof of $\lambda_2=0$ iff $G$ is 
disconnected; or, equivalently, that 
\[
\lambda_2 >0 \Leftrightarrow G \mbox{ is connected}  .
\]
Start with proving the first direction: if $G$ is disconnected, then 
$\lambda_2 = 0$.
If $G$ is disconnected, then $G$ is the union of (at least) $2$ graphs, call 
then $G_1$ and $G_2$.
Then, we can renumber the vertces so that we can write the Laplacian of $G$ 
as
\[
L_G = \left( \begin{array}{cc} L_{G_1} & 0       \\
                               0       & L_{G_2} 
             \end{array} 
      \right)
\]
So, $L_G$ has at least $2$ orthogonal eigenvectors with eigenvalue $0$, 
i.e., $\left( \begin{array}{c} 1 \\ 0 \end{array} \right)$ and 
$\left( \begin{array}{c} 0 \\ 1 \end{array} \right)$, where the two vectors 
are given with the same renumbering as in the Laplacians.
Conversely, if $G$ is connected and $x$ is an eigenvector such that 
$L_Gx = 0 x$, then, $L_G x = 0$, and 
$x^TL_Gx = \sum_{(ij)\in E} \left(x_i - x_j\right)^{2} = 0$.
So, for all $(u,v)$ connected by an edge, we have that $x_u = x_u$.
Apply this iteratively, from which it follows that $x$ is a constant 
vector, i.e., $x_u = x_v$, forall $u,v$.
So, the eigenspace of eigenvalue $0$ has dimension $1$.
This is the end of the aside.)

\subsubsection{Relationship with previous Perron-Frobenius results}

Theorem~\ref{thm:perrof-frob-lapl} is an important result, and it has 
several important extensions and variations.
In particular, the ``$\lambda_2>0$ iff $G$ is connected'' result is a 
``hard'' connectivity statement.
We will be interested in how this result can be extended to a ``soft'' connectivity, e.g., ``$\lambda_2$ is far from $0$ iff the graph is 
well-connected,'' and the associated Cheeger Inequality. 
That will come soon enough.
First, however, we will describe how this result relates to the previous 
things we discussed in the last several weeks, e.g., to the Perron-Frobenius 
result which was formulated in terms of non-negative matrices.

To do so, here is a similar result, formulated slightly differently.

\begin{lemma}
\label{thm:perrof-frob-adj}
Let $A_G$ be the Adjacency Matrix of a $d$ regular graph, and recall that it 
has $n$ real eigenvalues $\alpha_1 \ge \cdots \ge \alpha_n$ and $n$ 
associated orthogonal eigenvectors $v_i$ s.t. $A v_i = \lambda_i v_i$.
Then,
\begin{itemize}
\item
$\alpha_1 = d$, with 
$v_{1} = \frac{1}{\sqrt{n}} = \left( \frac{1}{\sqrt{n}}, \ldots, \frac{1}{\sqrt{n}} \right) $.
\item
$\alpha_n \ge -d$.
\item
The graph is connected iff $\alpha_1 > \alpha_2$.
\item
The graph is bipartite iff $\alpha_1 = -\alpha_n$, i.e., if $\alpha_n = -d$.
\end{itemize}
\end{lemma}

Lemma~\ref{thm:perrof-frob-adj} has two changes, relative to 
Theorem~\ref{thm:perrof-frob-lapl}.
\begin{itemize}
\item
The first is that it is a statement about the Adjacency Matrix, rather than 
the Laplacian.
\item
The second is that it is stated in terms of a ``scale,'' i.e., the 
eigenvalues depend on $d$.
\end{itemize}
When we are dealing with degree-regular graphs, then $A$ is trivially 
related to $L=D-A$ (we will see this below) and also trivially related to 
$L=I-\frac{1}{d}A$ (since this just rescales the previous $L$ by $1/d$).
We could have removed the scale from Lemma~\ref{thm:perrof-frob-adj} by
multiplying the Adjacency Matrix by $1/d$ (in which case, e.g., the 
eigenvalues would be in $[-1,1]$, rather than $[-d,d]$), but it is more 
common to remove the scale from the Laplacian.
Indeed, if we had worked with $L=D-A$, then we would have had the scale 
there too; we will see that below.

(When we are dealing with degree-heterogeneous graphs, the situation is more 
complicated.
The reason is basically since the eigenvectors of the Adjacency matrix and
unnormalized Laplacian don't have to be related to the diagonal degree 
matrix $D$, which defined the weighted norm which relates the normalized and 
unnormalized Laplacian.
In the degree-heterogeneous case, working with the normalized Laplacian will 
be more natural due to connections with random walks.
That can be interpreted as working with an unnormalized Laplacian, with 
an appropriate degree-weighted norm, but then the trivial connection with 
the eigen-information of the Adjacency matrix is lost.
We will revisit this below too.)

In the above, $A\in\mathbb{R}^{n \times n}$ is the Adjacency Matrix of an 
undirected graph $G=(V,E)$.
This will provide the most direct connection with the Perro-Frobenius results
we talked about last week.
Here are a few questions about the Adjacency Matrix.
\begin{itemize}
\item
Question: Is it symmetric? 
Answer: Yes, so there are real eigenvalues and a full set of orthonormal 
eigenvectors.
\item
Question: Is it positive?
Answer: No, unless it is a complete graph.
In the weighted case, it could be positive, if there were all the edges but 
they had different weights; but in general it is not positive, since some 
edges might be missing.
\item
Question: Is it nonnegative?
Answer: Yes.
\item
Question: Is it irreducible?
Answer: If no, i.e., if it is reducible, then 
\[
A = \left( \begin{array}{cc} A_{11} & A_{12}       \\
                                  0 & A_{22} 
           \end{array} 
    \right)
\]
must also have $A_{12}=0$ by symmetry, meaning that the graph is disconnected, in which case 
we should think of it as two graphs.
So, if the graph is connected then it is irreducible.
\item
Question: Is is aperiodic?
Answer: If no, then since it must be symmetric, and so it must look like
\[
A = \left( \begin{array}{cc}          0 & A_{12}       \\
                             A_{12}^{T} & 0
           \end{array} 
    \right) ,
\]
meaning that it is period equal to $2$, and so the ``second'' large 
eigenvalue, i.e., the one on the complex circle equal to a root of unity, 
is real and equal to $-1$.
\end{itemize}

How do we know that the trivial eigenvector is uniform?
Well, we know that there is only one all-positive eigenvector.
Let's try the all-ones vector $\vec{1}$.
In this case, we get 
\[
A \vec{1} = d \vec{1} ,
\]
which means that $\alpha_1 = d$ and 
$v_{1} = \frac{\vec{1}}{\sqrt{n}} = \left( \frac{1}{\sqrt{n}}, \ldots, \frac{1}{\sqrt{n}} \right) $.
So, the graph is connected if $\alpha_1 > \alpha_2$, 
and the graph is bipartite if $\alpha_1 = -\alpha_n$.

For the Laplacian $L=D-A$, there exists a close relationship between the 
spectrum of $A$ and $L$. 
(Recall, we are still considering the $d$-regular case.)
To see this, let $d = \alpha_1 \ge \ldots \alpha_n$ be the eigenvalues of 
$A$ with associated orthonormal eigenvectors $v_1,\ldots,v_n$.
(We know they are orthonormal, since $A$ is symmetric.)
In addition, let $ 0 \le \lambda_1 \le \cdots \le \lambda_n$ be the 
eigenvalues of $L$.
(We know they are all real and in fact all positive from the above 
alternative definition.)
Then, 
\[
\alpha_i = d-\lambda_i
\]
and
\[
A_G v_i = \left(dI-L_G\right) v_i = \left(d-\lambda_i\right)v_i  .
\]
So, $L$ ``inherits'' eigen-stuff from $A$.
So, even though $L$ isn't positive or non-negative, we get Perron-Frobenius 
style results for it, in addition to the results we get for it since it is a 
symmetric matrix.
In addition, if $L \rightarrow D^{-1/2}L D^{-1/2}$, then the eigenvalues of 
$L$ become in $[0,2]$, and so on.
This can be viewed as changing variables $y \leftarrow D^{-1/2}x$, and then 
defining Laplacian (above) and the Rayleigh quotient in the degree-weighted 
dot product.
(So, many of the results we will discuss today and next time go through to 
degree-heterogeneous graphs, for this reason.  
But some of the results, in particular the result having to do with expanders
being least like low-dimensional Euclidean space, do not.)

\subsection{Statement of the basic Cheeger Inequality}

We know the $\lambda_2$ captures a ``hard'' notion of connectivity, since 
the above result in Theorem~\ref{thm:perrof-frob-lapl} states that 
$\lambda_2 = 0 \Leftrightarrow G \mbox{ is disconnected}$.
Can we get a ``soft'' version of this?

To do so, let's go back to $d$-regular graphs, and recall the definition.

\begin{definition}
Let $G=(V,E)$ be a $d$-regular graph, and let $\left(S,\bar{S}\right)$ be a 
cut, i.e., a partition of the vertex set.
Then, 
\begin{itemize}
\item
the \emph{sparsity of $S$} is: 
$\sigma\left(S\right) = \frac{E\left(S,\bar{S}\right)}{\frac{d}{|V|}|S|\cdot|\bar{S}|}$
\item
the \emph{edge expansion of $S$} is:
$\phi\left(S\right) = \frac{E\left(S,\bar{S}\right)}{d|S|}$
\end{itemize}
\end{definition}
This definition holds for sets of nodes $S \subset V$, and we can extend 
them to hold for the graph $G$.
\begin{definition}
Let $G=(V,E)$ be a $d$-regular graph.
Then, 
\begin{itemize}
\item
the \emph{sparsity of $G$} is: 
$\sigma\left(G\right) = \min_{S \subset V: S \neq 0, S \neq V} \sigma\left(S\right)  $.
\item
the \emph{edge expansion of $G$} is:
$\phi\left(G\right) = \min_{S \subset V: |S| \le \frac{|V|}{2}} \phi\left(S\right)  $.
\end{itemize}
\end{definition}

For $d$-regular graphs, the graph partitioning problem is to find the 
sparsity or edge expansion of $G$.
Note that this means finding a number, i.e., the value of the objective 
function at the optimum, but people often want to find the corresponding 
set of nodes, and algorithms can do that, but the ``quality of 
approximation'' is that number.

\textbf{Fact.}
For all $d$ regular graphs $G$, and for all $S \subset V$ s.t. $|S| \le \frac{|V|}{2}$, 
we have that 
\[
\frac{1}{2}\sigma\left(S\right) \le \phi\left(S\right) \le \sigma\left(S\right) .
\]
Thus, since $\sigma\left(S\right) = \sigma\left(\bar{S}\right)$, we have that
\[
\frac{1}{2}\sigma\left(G\right) \le \phi\left(G\right) \le \sigma\left(G\right) .
\]
BTW, this is what we mean when we say that these two objectives are 
``equivalent'' or ``almost equivalent,'' since that factor of $2$ ``doesn't 
matter.'' 
By this we mean:
\begin{itemize}
\item
If one is interested in theory, then this factor of $2$ is well below the 
guidance that theory can provide.
That is, this objective is intractable to compute exactly, and the only 
approximation algorithms give quadratic or logarithmic (or square root of 
log) approximations.
If they could provide $1\pm\epsilon$ approximations, then this would matter, 
but they can't and they are much coarser than this factor of $2$.
\item
If one is interested in practice, then we can often do much better than this 
factor-of-$2$ improvement with various local improvement heuristics.
\item
In many cases, people actually write one and optimize the other.
\item
Typically in theory one is most interested in the number, i.e., the value of
the objective, and so we are ok by the above comment.
On the other hand, typically in practice, one is interested in using that 
vector to do things, e.g., make statements that the two clusters are close; 
but that requires stronger assumptions to say nontrivial about the actual 
cluster.
\end{itemize}

Given all that, here is the basic statement of Cheeger's inequality.

\begin{theorem}[Cheeger's Inequality]
Recall that
\[
\lambda_2 = \min_{x:x\perp\vec{1}}\max_{x:x\perp\vec{1}}\frac{x^TLx}{x^Tx}
\]
where $L=I-\frac{1}{d}A$.
Then, 
\[
\frac{\lambda_2}{2} \le \phi(G) \le \sqrt{ 2 \lambda_2 }  .
\]
\end{theorem}

\subsection{Comments on the basic Cheeger Inequality}

Here are some notes about the basic Cheeter Inequality.
\begin{itemize}
\item
This result ``sandwiches'' $\lambda_2$ and $\phi$ close to each other on 
both sides.
Clearly, from this result it immediatly follows that
\[
\frac{\phi(G)^{2}}{2} \le \lambda_2 \le 2 \phi(G)   .
\]
\item
Later, we will see that $\phi(G)$ is large, i.e., is bounded away from $0$, 
if the graph is well-connected.  
In addition, other related things, e.g., that random walks will mix 
rapidly, will also hold.
So, this result says that $\lambda_2$ is large if the graph is 
well-connected and small if the graph is not well-connected.
So, it is a soft version of the hard connectivity statement that we had 
before.
\item
The inequality $\frac{\lambda_2}{2} \le \phi(G)$ is sometimes known as the 
``easy direction'' of Cheeger's Inequality.
The reason is that the proof is more straightforward and boils down to 
showing one of two related things:
that you can present a test vector, which is roughly the indicator vector 
for a set of interest, and since $\lambda_2$ is a min of a Rayleigh 
quotient, then it is lower than the Rayleigh quotient of the test vector; 
or that the Rayleigh quotient is a relaxation of the sparsest cut problem, 
i.e., it is minimizing the same objective over a larger set.
\item
The inequality $\phi(G) \le \sqrt{ 2 \lambda_2 }$ is sometimes known as the 
``hard direction'' of Cheeger's Inequality.
The reason is that the proof is constructive and is basically a vanilla 
spectral partitioning algorithm.
Again, there are two related proofs for the ``hard'' direction of 
Cheeger.
One way uses a notion of inequalities over graphs; 
the other way can be formulated as a randomized rounding argument.
\item
Before dismissing the easy direction, note that 
it gives a polynomial-time certificate that a graph is expander-like,
\emph{i.e.}, that $\forall$ cuts (and there are $2^n$ of them to check) at
least a certain number of edges cross that cut.
(So the fact that is holds is actually pretty strong---we have a 
polynomial-time computable certificate of having no sparse cuts, which 
you can imagine is of interest since the naive way to check is to check
everything.)
\end{itemize}

Before proceeding, a question came up in the class about whether the upper 
or lower bound was more interesting or useful in applications.
It really depend on on what you want.
\begin{itemize}
\item
For example, if you are in a networking application where you are routing 
bits and you are interested in making sure that your network is 
well-connected, then you are most interested in the easy direction, since 
that gives you a quick-to-compute certificate that the graph is 
well-connected and that your bits won't get stuck in a bottleneck.
\item
Alternatively, if you want to run a divide and conquer algorithm or you want
to do some sort of statistical inference, both of which might require 
showing  that you have clusters in your graph that are well-separated from 
the rest of the data, then you might be more interested in the hard 
direction which provides an upper bound on the intractable-to-compute 
expansion and so is a certificate that there are some well-separated 
clusters.
\end{itemize}

%% file: lect07.tex
\section{%
(02/12/2015): 
Spectral Methods for Partitioning Graphs (2 of 2):
Proof of Cheeger's Inequality}

Reading for today.
\begin{compactitem}
\item
Same as last class.
\end{compactitem}

Here, we will prove the easy direction and the hard direction of 
Cheeger's Inequality.
Recall that what we want to show is that
\[
\frac{\lambda_2}{2} \le \phi(G) \le \sqrt{ 2 \lambda_2 }  .
\]

\subsection{Proof of the easy direction of Cheeger's Inequality}

For the easy direction, recall that what we want to prove is that 
\[
\lambda_2 \le \sigma(G) \le 2 \phi(G)  .
\]
To do this, we will show that the Rayleigh quotient is a relaxation of the 
sparsest cut problem.

Let's start by restating the sparsest cut problem:
\begin{eqnarray}
\nonumber
\sigma(G) &=& \min_{S \subset V : S \neq 0, S \neq V} \frac{E\left(S,\bar{S}\right)}{\frac{d}{|V|}|S|\cdot|\bar{S}|} \\
\nonumber
          &=& \min_{ x\in\{0,1\}^{n} \diagdown \{\vec{0},\vec{1} \} } \frac{ \sum_{{\{u,v\}}\in E} | x_u - x_v |   }{ \frac{d}{n}\sum_{\substack{ \{u,v\} \in V \times V }} | x_u - x_v |    } \\
\label{eqn:sparsity-quadratic}
          &=& \min_{ x\in\{0,1\}^{n} \diagdown \{\vec{0},\vec{1} \} } \frac{ \sum_{{\{u,v\}}\in E} | x_u - x_v |^2 }{ \frac{d}{n}\sum_{\substack{ \{u,v\} \in V \times V}}  | x_u - x_v |^2  }  ,
\end{eqnarray}
where the last equality follows since $x_u$ and $x_v$ are Boolean values, 
which means that $|x_u-x_v|$ is also a Boolean value.

Next, recall that 
\begin{equation}
\lambda_2 = \min_{x\in\mathbb{R}^{n}\diagdown\{\vec{0}\},x\perp\vec{1} } \frac{ \sum_{\{u,v\}\in E} |x_u-x_v|^2 }{ d\sum_v x_v^2 }  .
\label{eqn:lambda-quadratic-almost}
\end{equation}
Given that, we claim the following.
\begin{claim}
\begin{equation}
\lambda_2 = \min_{x\in\mathbb{R}^{n}\diagdown Span\{\vec{1}\} } \frac{ \sum_{\{u,v\}\in E} |x_u-x_v|^2 }{ \frac{d}{n} \sum_{\{u,v\}} |x_u-x_v|^2 }  .
\label{eqn:lambda-quadratic}
\end{equation}
\end{claim}
\begin{Proof}
Note that
\begin{eqnarray*}
\sum_{u,v} |x_u-x_v|^2 &=& \sum_{uv} x_u^2 + \sum_{uv} x_v^2 - 2 \sum_{uv} x_u x_v \\
                       &=& 2n \sum_v x_v^2 -2 \left( \sum_v x_v \right)^2 .
\end{eqnarray*}
Note that for all $x\in\mathbb{R}^{n}\diagdown\{\vec{0}\}$ s.t. $x \perp \vec{1}$, we have that $\sum_vx_v=0$, so
\begin{eqnarray*}
\sum_v x_v^2 &=& \frac{1}{2n} \sum_{u,v} |x_u - x_v |^2 \\
             &=& \frac{1}{n} \sum_{ \{u,v\} } |x_u-x_v |^2  ,
\end{eqnarray*}
where the first sum is over unordered pairs $u,v$ of nodes, and where the second sum 
of over ordered pairs $\{u,v\}$ (i.e. we double count $(u,v)$ and $(v,u)$ in first sum, but not in second sum).
So, 
\[
\min_{x\in\mathbb{R}^{n}\diagdown\{\vec{0}\},x\perp\vec{1} } \frac{ \sum_{\{u,v\}\in E} |x_u-x_v|^2 }{ d\sum_v x_v^2 } =
\min_{x\in\mathbb{R}^{n}\diagdown\{0\},x\perp\vec{1} } \frac{ \sum_{\{u,v\}\in E} |x_u-x_v|^2 }{ \frac{d}{n} \sum_{\{u,v\}} |x_u-x_v|^2 }  .
\]
Next, we need to remove the part along the all-ones vector, since the claim 
doesn't have that.

To do so, let's choose an $x^{*}$ that maximizes 
Eqn.~(\ref{eqn:lambda-quadratic}).
Observe the following.
If we shift every coordinate of that vector $x^{*}$ by the same constant, 
then we obtain another optimal solution, since the shift will cancel in all 
the expressions in the numerator and denominator.
(This works for any shift, and we will choose a particular shift to get what
we want.)

So, we can define $x^{\prime}$ s.t. $x_{v}^{\prime} = x_v - \frac{1}{n} \sum_u x_u$, 
and note that the entries of $x^{\prime}$ sum to zero.
Thus $x^{\prime} \perp \vec{1}$. Note we need $x\not\in Span(\vec{1})$ to have $x^{\prime}\neq \vec{0}$
So, 
\[
\min_{x\in\mathbb{R}^{n}\diagdown\{0\},x\perp\vec{1} } \frac{ \sum_{\{u,v\}\in E} |x_u-x_v|^2 }{ \frac{d}{n} \sum_{\{u,v\}} |x_u-x_v|^2 } = 
\min_{x\in\mathbb{R}^{n}\diagdown Span\{\vec{1}\} } \frac{ \sum_{\{u,v\}\in E} |x_u-x_v|^2 }{ \frac{d}{n} \sum_{\{u,v\}} |x_u-x_v|^2 }  .
\]
This establishes the claim.
\end{Proof}

So, from Eqn.~(\ref{eqn:sparsity-quadratic}) and 
Eqn.~(\ref{eqn:lambda-quadratic}), it follows that $\lambda$ is a
continuous relaxation of $\sigma(G)$, and so 
$\lambda_2 \le \sigma(G)$, from which the easy direction of Cheeger's 
Inequality follows.

\subsection{Some additional comments}

Here are some things to note.
\begin{itemize}
\item
There is nothing required or forced on us about the use of this relaxation, 
and in fact there are other relaxations.  
We will get to them later.
Some of them lead to traditional algorithms, and one of them provides the 
basis for flow-based graph partitioning algorithms.
\item
Informally, this relaxation says that we can replace $x \in \{0,1\}^{n}$ or
$x\in \{ -1,1\}^{n}$ constraint with the constraint that $x$ satisfies this 
``on average.''
By that, we mean that $x$ in the relaxed problem is on the unit ball, but any 
particular value of $x$ might get distorted a lot, relative to its 
``original'' $\{0,1\}$ or $\{ -1,1\}$ value.
In particular, note that this is a very ``global'' constraint.
As we will see, that has some good features, e.g., many of the well-known 
good statistical properties; but, as we will see, it has the consequence 
that any particular local pairwise metric information gets distorted, and 
thus it doesn't lead to the usual worst-case bounds that are given only in 
terms of $n$ the size of the graph (that are popular in TCS).
\item
While providing the ``easy'' direction, this lemma gives a quick low-degree 
polynomial time (whatever time it takes to compute an exact or approximate 
leading nonrtivial eigenvector) certificate that a given graph is 
expander-like, in the sense that for all cuts, at least a certain number of 
edges cross it.
\item
There has been a lot of work in recent years using approaches like this one; 
I don't know the exact history in terms of who did it first, but it was
explained by Trevisan very cleanly in course notes he has had, and this and 
the proof of the other direction is taken from that.
In particular, he describes the randomized rounding method for the other 
direction.
Spielman has slightly different proofs.
These proofs here are a combination of results from them.
\item
We could have proven this ``easy direction'' by just providing a test vector.
E.g., a test vector that is related to an indicator vector or a partition.
We went with this approach to highlight similarities and differences with 
flow-based methods in a week or two.
\item
The other reason to describe $\lambda_2$ as a relaxation of $\sigma(G)$ is 
that the proof of the other direction that 
$\phi(G) \le \sqrt{ 2 \lambda_2 }$ can be structured as a randomized 
rounding algorithm, i.e., given a real-valued solution to 
Eqn.~(\ref{eqn:lambda-quadratic}), find a similarly good solution to 
Eqn.~(\ref{eqn:sparsity-quadratic}).
This is what we will do next time.
\end{itemize}

\subsection{A more general result for the hard direction}

For the hard direction, recall that what we want to prove is that 
\[
\phi(G) \le \sqrt{ 2 \lambda_2 }  .
\]
Here, we will state---and then we will prove---a more general result.
For the proof, we will use the randomized rounding method.
The proof of this result is algorithmic/constructive, and it can be seen 
as an analysis for the following algorithm.

\textsc{VanillaSpectralPartitioning}.
Given as input a graph $G=(V,E)$, a vector $x\in\mathbb{R}^{n}$, 
\begin{enumerate}
\item
Sort the vertices of $V$ in non-decreasing order of values of entries of 
$x$, i.e., let $V = \{ v_1,\cdots,v_n\}$, where $x_{v_1} \le \cdots\le x_{v_n}$.
\item
Let $i \in [n-1]$ be s.t. 
\[
\max \{ 
        \phi \left(\left\{ v_1    ,\cdots,v_i \right\}\right) , 
        \phi \left(\left\{ v_{i+1},\cdots,v_n \right\}\right) 
      \}, 
\]
is minimal.
\item
Output $S = \{ v_1,\ldots,v_i\}$ and $\bar{S} = \{ v_{i+1},\ldots v_n\}$.
\end{enumerate}
This is called a ``sweep cut,'' since it involves sweeping over the 
input vector and looking at $n$ (rather than $2^n$ partitions) to find a 
good partition.

We have formulated this algorithm as taking as input a graph $G$ and any 
vector $x$.
You might be more familiar with the version that takes as input a graph 
$G$ that first compute the leading nontrivial eigenvector and then 
performs a sweep cut.
We have formulated it the way we did for two reasons.
\begin{itemize}
\item
We will want to separate out the spectral partitioning question from the 
question about how to compute the leading eigenvector or some approximation
to it.
For example, say that we don't run an iteration ``forever,'' i.e., to the 
asymptotic state to get an ``exact'' answer to machine precision.
Then we have a vector that only approximates the leading nontrivial 
eigenvector. 
Can we still use that vector and get nontrivial results?
There are several interesting results here, and we will get back to this.
\item
We will want to separate out the issue of global eigenvector to something
about the structure of the relaxation.
We will see that we can use this result to get local and locally-biased 
partitions, using both optimization and random walk based idea.
In particular, we will use this to generalize to locally-biased spectral 
methods.
\end{itemize}

So, establishing the following lemma is sufficient for what we want.
\begin{lemma}
Let $G=(V,E)$ be a $d$-regular graph, and let $x\in\mathbb{R}^{n}$ be s.t. 
$x\perp\vec{1}$.
Define 
\[
R(x) = \frac{ \sum_{ \{u,v \} \in E } |x_u-x_v |^2 }{ d \sum_v x_v^2 }
\]
and let $S$ be the side with less than $|V|/2$ nodes of the output cut of \textsc{VanillaSpectralPartitioning}.
Then, 
\[
\phi(S) \le \sqrt{ 2 R(x) }
\]
\end{lemma}
\noindent
Before proving this lemma, here are several things to note.
\begin{itemize}
\item
If we apply this lemma to a vector $x$ that is an eigenvector of $\lambda_2$, 
then $R(x) = \lambda_2$, and so we have that $\phi(G)\le \phi(S) \le \sqrt{2\lambda_2}$, 
i.e., we have the difficult direction of Cheeger's Inequality.
\item
On the other hand, any vector whose Rayleigh quotient is close to that of 
$\lambda_2$ also gives a good solution.
This ``rotational ambiguity'' is the usual thing with eigenvectors, and it 
is different than any approximation of the relatation to the original 
expansion IP.
We get ``goodness'' results for such a broad class of vectors to sweep over 
since we are measuring goodness rather modestly: only in terms of objective 
function value.
Clearly, the actual clusters might change a lot and in general will be very 
different if we sweep over two vectors that are orthogonal to each other.
\item
This result actually holds for vectors $x$ more generally, i.e., vectors 
that have nothing to do with the leading eigenvector/eigenvalue, as we will 
see below with locally-biased spectral methods, where we will use it to get 
upper bounds on locally-biased variants of Cheeger's Inequality.
\item
In this case, in ``eigenvector time,'' we have found a set $S$ with expansion
s.t. $\phi(S) \le \sqrt{\lambda_2} \le 2 \sqrt{\phi(G)}$.
\item
This is \emph{not} a constant-factor approximation, or any nontrivial 
approximation factor in terms of $n$; and it is incomparable with other 
methods (e.g., flow-based methods) that do provide such an approximation 
factor.
It is, however, nontrivial in terms of an important structural parameter of 
the graph.
Moreover, it is efficient and useful in many machine learning and data 
analysis applications.
\item
The above algorithm can be implemented in roughly 
$O\left( |V| \log|V| + |E| \right)$ time, assuming arithmetic operations and 
comparisons take constant time.  
This is since once we have computed 
$$
E\left( \{v_1,\ldots,v_i\},\{v_{i+1},\ldots,v_n\} \right)  ,
$$
it only takes 
$O(\mbox{degree}(v_{i+1}))$ time to compute
$$
E\left( \{v_1,\ldots,v_{i+1}\},\{v_{i+2},\ldots,v_n\} \right)  .
$$
\item
As a theoretical point, there exists nearly linear time algorithm to find 
a vector $x$ such that $R(x) \approx \lambda_2$, and so by coupling this 
algorithm with the above algorithm we can find a cut with expansion 
$O\left( \sqrt{\phi(G)} \right)$ in nearly-linear time.
Not surprisingly, there is a lot of work on providing good implementations 
for these nearly linear time algorithms.
We will return to some of these later.
\item
This quadratic factor is ``tight,'' in that there are graphs that are that
bad; we will get to these (rings or Guattery-Miller cockroach, depending on 
exactly how you ask this question) graphs below.
\end{itemize}

\subsection{Proof of the more general lemma implying the hard direction of Cheeger's Inequality}

Note that $\lambda_2$ is a relaxation of $\sigma(G)$ and the lemma provides a
rounding algorithtm for real vectors that are a solution of the relaxation.
So, we will view this in terms of a method from TCS known as randomized 
rounding.
This is a useful thing to know, and other methods, e.g., flow-based methods 
that we will discuss soon, can also be analyzed in a similar manner.

For those who don't know, here is the one-minute summary of randomized 
rounding.
\begin{itemize}
\item
It is a method for designing and analyzing the quality of approximation 
algorithms.
\item
The idea is to use the probabilistic method to convert the optimal solution 
of a relaxation of a problem into an approximately optimal solution of the 
original problem.
(The probabilistic method is a method from combinatorics to prove the 
existence of objects.
It involves randomly choosing objects from some specified class in some 
manner, i.e., according to some probability distribution, and showing that 
the objects can be found with probability $>0$, which implies that the 
object exists.  Note that it is an existential/non-constructive and not 
algorithmic/constructive method.)
\item
The usual approach to use randomized rounding is the following.
\begin{itemize}
\item 
Formulate a problem as an integer program or integer linear program (IP/ILP).
\item
Compute the optimal fractional solution $x$ to the LP relaxation of this IP.
\item
Round the fractional solution $x$ of the LP to an integral solution 
$x^{\prime}$ of the IP.
\end{itemize}
\item
Clearly, if the objective is a min, then 
$ \mbox{cost}(x) \le \mbox{cost}(x^{\prime})$.
The goal is to show that $\mbox{cost}(x^{\prime})$ is not much more that 
$\mbox{cost}(x)$.
\item
Generally, the method involves showing that, given any fractional solution 
$x$ of the LP, w.p. $>0$ the randomized rounding procedure produces an 
integral solution $x^{\prime}$ that approximated $x$ to some factor.
\item
Then, to be computationally efficient, one must show that 
$x^{\prime} \approx x$ w.h.p. (in which case the algorithm can stay 
randomized) or one must use a method like the method of conditional 
probabilities (to derandomize it).
\end{itemize}

Let's simplify notation:
let $V = \{1,\ldots,n\}$; and so $x_1 \le x_2 \le \cdots x_n$.
In this case, the goal is to show that there exists $i\in[n]$ w.t.
\[
\phi\left( \{ 1  ,\ldots,i \} \right) \le \sqrt{ 2 R(x) } 
\quad \mbox{and} \quad
\phi\left( \{ i+1,\ldots,n \} \right) \le \sqrt{ 2 R(x) }  .
\]

We will prove the lemma by showing that there exists a distribution $D$ over
sets $S$ of the form $\{ 1,\ldots,i\}$ s.t.
\begin{equation}
\frac{ \mathbb{E}_{S \sim D}\left\{ E(S,\bar{S}) \right\} }{\mathbb{E}_{S \sim D}\left\{ d\min\{ |S|,|\bar{S}| \} \right\} } \le \sqrt{2 R(x) }  .
\label{eqn:ratio-expectations}
\end{equation}

Before establishing this, note that Eqn.~(\ref{eqn:ratio-expectations}) does
\emph{not} imply the lemma.
Why? 
In general, it is the case that
$\mathbb{E}\left\{ \frac{X}{Y} \right\} \neq \frac{\mathbb{E}\left\{X\right\}}{\mathbb{E}\left\{Y\right\}} $, 
but it suffices to establish something similar.

\textbf{Fact.}
For random variables $X$ and $Y$ over the sample space, even though 
$\mathbb{E}\left\{ \frac{X}{Y} \right\} \neq \frac{\mathbb{E}\left\{X\right\}}{\mathbb{E}\left\{Y\right\}} $,
it is the case that 
\[
\mathbb{P}\left\{  \frac{X}{Y} \le \frac{\mathbb{E}\left\{X\right\}}{\mathbb{E}\left\{Y\right\}} \right\} > 0 , 
\]
provided that $Y > 0$ over the entire sample space.

But, by linearity of expectation, from 
Eqn.~(\ref{eqn:ratio-expectations}) it follows that
\[
\mathbb{E}_{S \sim D}\left[ E(S,\bar{S}) - d\sqrt{2R(x)}\min\{ |S|,|\bar{S}| \}\right] \le 0  .
\]
So, there exists a set $S$ in the sample space s.t.
\[
E(S,\bar{S}) - d\sqrt{2R(x)}\min\{ |S|,|\bar{S}| \} \le 0  .
\]
That is, for $S$ and $\bar{S}$, at least on of which has size $\le \frac{n}{2}$, 
\[
\phi(S) \le \sqrt{ 2 R(x) }  ,
\]
from which the lemma will follow.

So, because of this, it will suffice to establish Eqn.~(\ref{eqn:ratio-expectations}).
So, let's do that.

Assume, WLOG, that $x_{\lceil\frac{n}{2}\rceil} = 0$, i.e., the median of the entires 
of $x$ equals $0$; and $x_1^2+x_n^2=1$.
This is WLOG since, if $x\perp\vec{1}$, then adding a fixed constant $c$ to all 
entries of $x$ can only decrease the Rayleigh quotient:
\begin{eqnarray*} 
R\left(x+(c,\ldots,c)\right) 
   &=& \frac{ \sum_{ \{(u,v)\} \in E} |(x_u+c) - (x_v+c)|^2 }{ d \sum_v (x_v+c)^2 } \\
   &=& \frac{ \sum_{ \{(u,v)\} \in E} | x_u    -  x_v   |^2 }{ d \sum_v x_v^2 - 2dc\sum_v x_v + nc^2 } \\
   &=& \frac{ \sum_{ \{(u,v)\} \in E} | x_u    -  x_v   |^2 }{ d \sum_v x_v^2 + nc^2 } \\
   &\le& R(x) .
\end{eqnarray*} 
Also, 
multiplying all entries by fixed constant does \emph{not} change the value
of $R(x)$, nor does it change the property that $x_1 \le \cdots \le x_n$.

We have made these choices since they will allow us to define a distribution
$D$ over sets $S$ s.t.
\begin{equation}
\mathbb{E}_{S \sim D} \min\left\{ |S|,|\bar{S}| \right\} = \sum_i x_i^2
\label{eqn:expect2}
\end{equation}

Define a distribution $D$ over sets $\{1,\ldots,i\}$, $1\le i\le n-1$, as the 
outcome of the following probabilistic process.
\begin{enumerate}
\item
Choose a $t \in [ x_1,x_n ] \in \mathbb{R}$ with probability density function
equal to $f(t) = 2 |t|$, i.e., for $x_1 \le a \le b \le x_n$, let 
\[
\mathbb{P}\left[ a \le t \le b \right] 
   = \int_a^b 2|t| dt  
   = \left\{ \begin{array}{l l}
                 | a^2-b^2 | & \quad \text{if $a$ and $b$ have the same sign}\\
                 a^2+b^2     & \quad \text{if $a$ and $b$ have different signs}
              \end{array} 
      \right.  ,
\]
\item
Let $S = \{ u : x_i \le t \}$
\end{enumerate}

From this definition
\begin{itemize}
\item
The probability that an element $i \le \frac{n}{2}$ belongs to the smaller of
the sets $S,\bar{S}$ equals the probability of $i\in S$ and $|S|\le |\bar{S}|$, which 
equals the probability that the threshold $t$ is in the range $[x_i,0]$, 
which equals $x_i^2$.
\item
The probability that an element $i > \frac{n}{2}$ belongs to the smaller of
the sets $S,\bar{S}$ equals the probability of $i\in \bar{S}$ and $|S|\ge |\bar{S}|$, which 
equals the probability that the threshold $t$ is in the range $[0,x_i]$, 
which equals $x_i^2$.
\end{itemize}
So, Eqn.~(\ref{eqn:expect2}) follows from linearity of expectation.

Next, we want to estimate the expected number of edges between $S$ and $\bar{S}$, 
i.e., 
\[
\mathbb{E}\left[ E\left(S,\bar{S}\right)\right]
   = \sum_{(i,j)\in E} \mathbb{P}\left[ \mbox{edge } (i,j) \mbox{ is cut by } (S,\bar{S}) \right] .
\]

To estimate this, note that the event that the edge $(i,j)$ is cut by the 
partition $(S,\bar{S})$ happens when $t$ falls in between $x_i$ and $x_j$.
So, 
\begin{itemize}
\item
if $x_i$ and $x_j$ have the same sign, then
\[
\mathbb{P}\left[ \mbox{edge } (i,j) \mbox{ is cut by } (S,\bar{S}) \right]=|x_i^2-x_j^2|
\]
\item
if $x_i$ and $x_j$ have the different signs, then
\[
\mathbb{P}\left[ \mbox{edge } (i,j) \mbox{ is cut by } (S,\bar{S}) \right] = x_i^2 + x_j^2
\]
\end{itemize}

The following expression is an upper bound that covers both cases:
\[
\mathbb{P}\left[ \mbox{edge } (i,j) \mbox{ is cut by } (S,\bar{S}) \right] 
   \le | x_i - x_j | \cdot \left( |x_i| + |x_j| \right)  .
\]
Plugging into the expressions for the expected number of cut edges, and 
applying the Cauchy-Schwatrz inequality gives
\begin{eqnarray*}
\mathbb{E} E\left(S,\bar{S}\right)
   &\le& \sum_{(i,j) \in E} |x_i - x_j| \left( |x_i|+|x_j|\right)  \\
   &\le& \sqrt{ \sum_{(i,j)\in E} \left(x_i-x_j\right)^2 } \sqrt{ \sum_{(i,j)\in E} \left(|x_i|+|x_j|\right)^{2} }
\end{eqnarray*}

Finally, to deal with the expression $\sum_{(ij)\in E} \left(|x_i|+|x_j|\right)^{2}$,
recall that $(a+b)^2 \le 2a^2+2b^2$.
Thus,
\[
\sum_{(ij)\in E} \left(|x_i|+|x_j|\right)^{2} 
   \le \sum_{(ij) \in E} 2x_i^2 + 2x_j^2 
     = 2d \sum_i x_i^2  .
\]
Putting all of the pieces together, we have that
\[
\frac{ \mathbb{E} \left[E\left(S,\bar{S}\right) \right]}{ \mathbb{E}\left[d \min\{ |S|,|\bar{S}| \}\right] } 
 \le \frac{ \sqrt{ \sum_{(ij)\in E} \left(x_i-x_j\right)^2 } \sqrt{ 2d \sum_i x_i^2  } }{ d \sum_i x_i^2 }
 = \sqrt{ 2R(x) }  ,
\]
from which the result follows.

%% file: lect08.tex
\section{%
(02/17/2015): 
Expanders, in theory and in practice (1 of 2)}

Reading for today.
\begin{compactitem}
\item
``Expander graphs and their applications,'' in Bull. Amer. Math. Soc., by Hoory, Linial, and Wigderson
\end{compactitem}

\subsection{Introduction and Overview}

\emph{Expander graphs}, also called \emph{expanders}, are remarkable 
structures that are widely-used in TCS and discrete mathematics.
They have a wide range of applications:
\begin{itemize} 
\item
They reduce the need for randomness and are useful for derandomizing 
randomized algorithms---so, if random bits are a valuable resource and thus 
you want to derandomized some of the randomized algorithms we discussed 
before, then this is a good place to start.
\item
They can be used to find good error-correcting codes that are efficiently 
encodable and decodable---roughly the reason is that they spread things out.
\item
They can be used to provide a new proof of the so-called PCP theorem, which 
provides a new characterization of the complexity class NP, and applications
to the hardness of approximate computation.
\item
They are a useful concept in data analysis applications, since expanders 
look random, or are empirically quasi-random, and it is often the case that
the data, especially when viewed at large, look pretty noisy.
\end{itemize} 
For such useful things, it is somewhat surprising that (although they are 
very well-known in computer science and TCS in particular due to their 
algorithmic and complexity connections) expanders are almost unknown outside
computer science.
This is unfortunate since:
\begin{itemize}
\item
The world is just a bigger place when you know about expanders.
\item
Expanders have a number of initially counterintuitive properties, like they
are very sparse and very well-connected, that are typical of a lot of data 
and thus that are good to have an intuition about.
\item
They are ``extremal'' in many ways, so they are a good limiting case if you
want to see how far you can push your ideas/algorithms to work.
\item
Expanders are the structures that are ``most unlike'' low-dimensional 
spaces---so if you don't know about them then your understanding of the 
mathematical structures that can be used to describe data, as well as of 
possible ways that data can look will be rather limited, \emph{e.g.}, you 
might think that curved low-dimensional spaces are good ideas.
\end{itemize}
Related to the comment about expanders having extremal properties, if you 
know how your algorithm behaves on, say, expanders, hypercubes (which are 
similar and different in interesting ways), trees (which we won't get to as 
much, but will mention), and low-dimensional spaces, they you probably have 
a pretty good idea of how it will behave on your data.
That is very different than knowing how it will behave in any one of those 
places, which doesn't give you much insight into how it will behave more 
generally; this extremal property is used mostly by TCS people for algorithm
development, but it can be invaluable for understanding how/when your 
algorithm works and when it doesn't on your non-worst-case data.

We will talk about expander graphs.
One issue is that we can define expanders both for degree-homogeneous graphs
as well as for degree-heterogeneous graphs; and, although many of the basic 
ideas are similar in the two cases, there are some important differences 
between the two cases.
After defining them (which can be done via expansion/conductance or the 
leading nontrivial eigenvalue of the combinatorial/normalized Laplacian), we 
will focus on the following aspects of expanders and expander-like graphs.
\begin{itemize}
\item
Expanders are graphs that are very well-connected.
\item
Expanders are graphs that are sparse versions/approximations of a complete
graph.
\item
Expanders are graphs on which diffusions and random walks mix rapidly.
\item
Expanders are the metric spaces that are least like low-dimensional 
Euclidean spaces.
\end{itemize}
Along the way, we might have a chance to mention a few other things, e.g.:
how big $\lambda_2$ could be with Ramanujan graphs and Wigner's semicircle
result;
trivial ways with $d_{max}$ to extend the Cheeger Inequality to 
degree-heterogeneous graphs, as well as non-trivial ways with the normalized
Laplacian;
pseudorandom graphs, converses, and the Expander Mixing Lemma; and maybe 
others.

Before beginning with some definitions, we should note that we can't draw a 
meaningful/interpretable picture of an expander, which is unfortunate since 
people like to visualize things.
The reason for that is that there are no good ``cuts'' in an 
expander---relatedly, they embed poorly in low-dimensional spaces, which is
what you are doing when you visualize on a two-dimensional piece of paper.
The remedy for this is to compute all sorts of other things to try to get a 
non-visual intuition about how they behave.

\subsection{A first definition of expanders}

Let's start by working with $d$-regular graphs---we'll relax this regularity
assumption later.
But many of the most extremal properties of expanders hold for 
degree-regular graphs, so we will consider them first.
\begin{definition}
A graph $G=(V,E)$ is \emph{$d$-regular} if all vertices have the same degree
$d$, \emph{i.e.}, each vertex is incident to exactly $d$ edges.
\end{definition}
Also, it will be useful to have the following notion of the set of edges 
between two sets $S$ and $T$ (or from $S$ to $T$), both of which are subsets 
of the vertex set (which may or may not be the complement of each other).
\begin{definition}
For $S,T \subset V$, denote
\[
E(S,T)=\{(u,v)\in E|\; u\in S,\, v\in T\}   .
\]
\end{definition}

Given this notation, we can define the expansion of a graph.
(This is slightly different from other definitions I have given.)
\begin{definition}
The \emph{expansion} or \emph{edge expansion ratio} of a graph $G$ is
\[
h(G) = \min_{S:|S|\le\frac{n}{2}} \frac{E(S,\bar{S})}{|S|} 
\]
\end{definition}
Note that this is slightly different (just in terms of the scaling) than the 
edge expansion of $G$ which we defined before as:
$$\phi\left(G\right) 
    = \min_{S \subset V: |S| \le \frac{|V|}{2}} 
      \frac{E\left(S,\bar{S}\right)}{d|S|}  .$$
We'll use this today, since I'll be following a proof from HLW, and they use
this, and following their notation should make it easier.
There should be no surprises, except just be aware that there is a factor of 
$d$ difference from what you might expect.

(As an aside, recall that there are a number of extensions of this basic 
idea to measure other or more fine versions of this how well connected is a 
graph:
\begin{itemize}
\item
Different notions of boundary---\emph{e.g.}, vertex expansion.
\item
Consider size-resolved minimum---in Markov chains and how good communities 
are as a function of size.
\item
Different denominators, which measure different notions of the ``size'' of a 
set $S$:
\begin{itemize}
\item
Sparsity or cut ratio: $\min \frac{E(S,\bar{S})}{|S|\cdot|\bar{S}|}$---this 
is equivalent to expansion in a certain sense that we will get to.
\item
Conductance or NCut---this is identical for $d$-regular graphs but is more 
useful in practice and gives tighter bounds in theory if there is degree
heterogeneity.
\end{itemize}
\end{itemize}
We won't deal with these immediately, but we will get back to some later.
This ends the aside.)

In either case above, the expansion is a measure to quantify how 
well-connected is the graph.
Given this, informally we call a $d$-regular graph $G$ an \emph{expander} if 
$h(G) \geq \epsilon$ where $\epsilon$ is a constant. 
More precisely, let's define an expander:
\begin{definition}
A graph $G$ is a $(d,\epsilon)$-expander if it is $d$-regular and 
$h(G) \geq \epsilon$, where $\epsilon$ is a constant, independent of $n$.
\end{definition}
Alternatively, sometimes expansion is defined in terms of a sequence of 
graphs:
\begin{definition}
A sequence of $d$-regular graphs $\{G_i\}_{i \in \mathcal{Z}^{+}}$ is a 
family of \emph{expander graphs} if $\exists \epsilon > 0$ s.t.
$h(G_i)\geq\epsilon , \forall i$.
\end{definition}

\noindent
If we have done the normalization correctly, then $h(G) \in [0,d]$ and 
$\phi(G) \in [0,1]$, where large means more expander-like and small means 
that there are good partitions.
So, think of the constant $\epsilon$ as $d/10$ (and it would be $1/10$, if 
we used $\phi(G)$ normalization).
Of course, there is a theory/practice issue here, e.g., sometimes you are 
given a single graph and sometimes it can be hard to tell a moderately large
constant from a factor of $\log(n)$; we will return to these issues later.

\subsection{Alternative definition via eigenvalues}

Although expanders can be a little tricky and counterintuitive, there are a 
number of ways to deal with them.
One of those ways, but certainly not the only way, is to compute 
eigenvectors and eigenvalues associated with matrices related to the graph.
For example, if we compute the second eigenvalue of the Laplacian, then we 
have Cheeger's Inequality, which says that if the graph $G$ is an expander, 
then we have a (non-tight, due to the quadratic approximation) bound on the 
second eigenvalue, and vice versa.
That is, one way to test if a graph is an expander is to compute that 
eigenvalue and check.

Of central interest to a lot of things is $\lambda_2^{LAP}$, which is the
Fiedler value or second smallest eigenvalue of the Laplacian.
Two things to note:
\begin{itemize}
\item
If we work with Adjacency matrices rather than Laplacians, then we are 
interested in how far $\lambda_2^{ADJ}$ is from $d$.
\item
We often normalized things so as to interpret them in terms of a random 
walk, in which case the top eigenvalue $=1$ with the top eigenvector being
the probability distribution. 
In that case, we are interested in how far $\lambda_2$ is from $1$.
\end{itemize}
(Since I'm drawing notes from several different places, we'll be a little 
inconsistent on what the notation means, but we should be consistent within
each class or section of class.)

Here is Cheeger's Inequality, stated in terms of $h(G)$ above. 
\begin{itemize}
\item
If $ 0=\lambda_1 \leq \lambda_2 \leq \cdots \leq \lambda_n $ are the 
eigenvalues of the Laplacian (not normalized, i.e. $D-A$) of a $d$-regular graph $G$, then:
\[
\frac{\lambda_2}{2} \leq h(G) \leq \sqrt{ 2 d \lambda_2 }
\]
The $\sqrt{d}$ in the upper bound is due to our scaling.
\end{itemize}
Alternatively, here is Cheeger's Inequality, stated in terms of $h(G)$ for 
an Adjacency Matrix.
\begin{itemize}
\item
If $d=\mu_{1}\geq\mu_{2}\geq\ldots\geq\mu_{n}$ are the eigenvalues of the 
Adjacency Matrix $A(G)$ of $d$-regular graph $G$, then:
\[
\frac{d-\mu_{2}}{2}\leq h(G)\leq\sqrt{2d(d-\mu_{2})}
\]
\end{itemize}
Therefore, the expansion of the graph is related to its spectral gap
($d-\mu_{2}$). 
Thus, we can define a graph to be an expander if $\mu_2 \leq d-\epsilon $ or 
$\lambda_2 \geq \epsilon$ where $\lambda_2$ is the second eigenvalue of the 
matrix $L(G) = D - A(G)$ where $D$ is the diagonal degree matrix. 
Slightly more formally, here is the alternate definition of expanders:
\begin{definition}
A sequence of $d$-regular graphs $\{G_n\}_n\in\mathbb{N}$ is a family of expander graphs
if $|\lambda_i^{ADJ}|\le d- \epsilon$, \emph{i.e.} if all the eigenvalues of 
$A$ are bounded away from $d$
\end{definition}

\noindent
\textbf{Remark.}
The last requirement can be written as $\lambda_2^{LAP}\ge c, \forall n$, 
\emph{i.e.}, that all the eigenvalues of the Laplacian bounded below and 
away from $c>0$.

In terms of the edge expansion $\phi(G)$ we defined last time, this definition 
would become the following.

\begin{definition}
A family of constant-degree expanders is a family of graphs $\{ G_n \}_{n\in \mathbb{N}}$ s.t. 
each graph in $G_n$ is $d$-regular graph on $n$ vertices and such that there
exists an absolute constant $\phi^{*}$, independent of $n$, s.t. 
$\phi(G_n) \ge \phi^{*}$, for all $n$.
\end{definition}

\subsection{Expanders and Non-expanders}

A clique or a complete graph is an expander, if we relax the requirement 
that the $d$-regular graph have a fixed $d$, independent of $n$.
Moreover, $G_{n,p}$ (the random graph), for $p \gtrsim \frac{\log(n)}{n}$ is also an expander, 
with $d$ growing only weakly with $n$.
(We may show that later.)
Of greatest interest---at least for theoretical considerations---is the case 
that $d$ is a constant independent of $n$.

\subsubsection{Very sparse expanders}

In this case, the idea of an expander, \emph{i.e.}, an \emph{extremely} 
sparse and \emph{extremely} well-connected graph is nice; but do they exist?
It wasn't obvious until someone proved it, but the answer is YES.
In fact, a typical $d$-regular graph is an expander with high probability 
under certain random graph models.
Here is a theorem that we will not prove.

\begin{theorem}
Fix $d \in \mathbb{Z}^{+} \ge 3$.
Then, a randomly chosen $d$-regular graph is an expander w.h.p.
\end{theorem}

\noindent
\textbf{Remark.}
Clearly, the above theorem is false if $d=1$ (in which case we get a bunch 
of edges) or if $d=2$ (in which case we get a bunch of cycles); but it holds 
even for $d=3$.

\noindent
\textbf{Remark.}
The point of comparison for this should be if $d \gtrsim \frac{\log(n)}{n}$.
In this case, ``measure concentration'' in the asymptotic regime, and so it 
is plausible (and can be proved to be true) that the graph has no good 
partitions.
To understand this, recall that one common random graph model is the 
Erdos-Renyi $\mathcal{G}_{n,p}$ model, where there are $n$ vertices and 
edges are chosen to exist with probability $p$. 
(We will probably describe this ER model as well as some of its basic 
properties later; at a minimum, we will revisit it when we talk about 
stochastic blockmodels.)
The related $\mathcal{G}_{n,m}$ model is another common model where graphs 
with $n$ vertices and $m$ edges are chosen uniformly at random. 
An important fact is that if we set $p$ such that there are on average $m$ 
edges, then $\mathcal{G}_{n,m}$ is very similar (in strong senses of the 
word) to $\mathcal{G}_{n,p}$---if $p \geq \log n /n$. 
(That is the basis for the oft-made observation that $\mathcal{G}_{n,m}$ and $\mathcal{G}_{n,p}$ 
are ``the same.'') 
However, for the above definition of expanders, we require in addition that 
$d$ is a constant.  
Importantly, in that regime, the graphs are sparse enough that measure 
hasn't concentrated, and they are \emph{not} the same.
In particular, if $p = 3/n$, $\mathcal{G}_{n,p}$ usually generates a graph 
that is not connected (and there are other properties that we might return 
to later). 
However, (by the above theorem) $\mathcal{G}_{n,m}$ with corresponding 
parameters usually yields a connected graph with very high expansion. 
We can think of randomized expander construction as a version of 
$\mathcal{G}_{n,m}$, further constrained to $d$-regular graphs. 

\textbf{Remark.}
There are explicit deterministic constructions for expanders---they have 
algorithmic applications.
That is an FYI, but for what we will be doing that won't matter much.
Moreover, later we will see that the basic idea is still useful even when we 
aren't satisfying the basic definition of expanders given above, e.g., when 
there is degree heterogeneity, when a graph has good small but no good large 
cuts, etc.

\subsubsection{Some non-expanders}

It might not be clear how big is big and how small is small---in particular, 
how big can $h$ (or $\lambda$) be.
Relatedly, how ``connected'' can a graph be?
To answer this, let's consider a few graphs.

\begin{itemize}
\item
Path graph.
(For a path graph, $\mu_1 = \Theta(1/n^2)$.
If we remove $1$ edge, then we can cut the graph into two $50$-$50$ pieces.
\item
Two-dimensional $\sqrt{n} \times \sqrt{n}$ grid.
(For a $\sqrt{n} \times \sqrt{n}$ grid, $\mu_1 = \Theta(1/n)$.)
Here, you can't disconnect the graph by removing $1$ edge, and the removal
of a constant number of edges can only remove a constant number of vertices
from the graph.
But, it is possible to remove $\sqrt{n}$ of the edge, i.e., an 
$O(\frac{1}{\sqrt{n}})$ fraction of the total, and split the graph into two 
$50$-$50$ pieces.
\item 
For a 3D grid, $\mu_1 = \Theta(1/n^{2/3})$.
\item
A $k$-dimensional hypercube is still better connected.
But it is possible to remove a very small fraction of the edges (the 
edges of a dimension cut, which are $\frac{1}{k}= \frac{1}{\log(n)}$ fraction
of the total) and split half the vertices from the other half.
\item 
For a binary tree, e.g., a complete binary tree on $n$ vertices, 
$\mu_1 = \Theta(1/n)$.
\item 
For a $K_n - K_n$ dumbbell, (two expanders or complete graphs joined by an edge) $\mu_1 = \Theta(1/n)$.
\item 
For a ring on $n$ vertices, $\mu_1 = \Theta(1/ n)$.
\item
Clique.
Here, to remove a $p$ fraction of the vertices from the rest, you must 
remove $\ge p(1-p)$ fraction of the edges.
That is, it is very well connected.
(While can take a complete graph to be the ``gold standard'' for connectivity,
it does, however, have the problem that it is dense; thus, we will be 
interested in sparse versions of a complete that are similarly well-connected.)
\item 
For an expander, $\mu_1 = \Theta(1)$.
\end{itemize}

\noindent
\textbf{Remark.}
A basic question to ask is whether, say, 
$\mu_1 \sim \Theta(\mbox{poly}(1/ n))$ is 
``good'' or ``bad,'' say, in some applied sense?
The answer is that it can be either:
it can be bad, if you are interested in connectivity, e.g., a network where 
nodes are communication devices or computers and edges correspond to an 
available link; or
it can be good, either for algorithmic reasons if e.g. you are interested in 
divide and conquer algorithms, or for statistical reasons since this can be 
used to quantify conditional independence and inference.

\noindent
\textbf{Remark.}
Recall the quadratic relationship between $d-\lambda_2$ and $h$.
If $d-\lambda_2$ is $\Theta(1)$, then that is not much difference (a topic
which will return to later), but if it is $\Theta(1/n)$ or
$\Theta(1/n^2)$ then it makes a big difference.
A consequence of this is that by TCS standards, spectral partitioning does
a reasonably-good job partitioning expanders (basically since the quadratic 
of a constant is a constant), while everyone else would wonder why it makes
sense to partition expanders; while by TCS standards, spectral partitioning
does \emph{not} do well in general, since it has a worst-case approximation 
factor that depends on $n$, while everyone else would say that it does 
pretty well on their data sets.

\subsubsection{How large can the spectral gap be?}

A question of interest is: how large can the spectral gap be? 
The answer here depends on the relationship between $n$, the number of nodes
in the graph and $d$, the degree of each node (assumed to be the same for 
now.)
In particular, the answer is different if $d$ is fixed as $n$ grows or if 
$d$ grows with $n$ as $n$ grows.
As as extreme example of the latter case, consider the complete graph $K_n$
on $n$ vertices, in which case $d=n-1$.
The adjacency matrix of $K_n$ is $A_{K_n} = J-I$, where $J$ is the all-ones
matrix, and where $I=I_n$ is the diagonal identity matrix.
The spectrum of the adjacency matrix of $K_n$ is $\{n-1,-1,\ldots,-1\}$, and 
$\lambda=1$.
More interesting for us here is the case that $d$ is fixed and $n$ is large, 
in which case $n \gg d$, in which case we have the following theorem (which 
is due to Alon and Boppana).

\begin{theorem}[Alon-Boppana]
Denoting $\lambda=\max(|\mu_{2}|,|\mu_{n}|)$, we
have, for every $d$-regular graph:
\[
\lambda\geq2\sqrt{d-1}-o_{n}(1)
\]
\end{theorem}

\noindent
So, the eigengap $d-\mu_2$ is not larger than $d-2\sqrt{d-1}$.
For those familiar with Wigner's semicircle law, note the similar form.

The next question is: How tight is this?
In fact, it is pretty close to tight in the following sense:
there exists constructions of graphs, called Ramanujan graphs, 
where the second eigenvalue of $L(G)$ is $\lambda_1(G) = d- 2 \sqrt{d-1}$, 
and so the tightness is achieved. 
Note also that this is of the same scale as Wigner's semicircle law; the 
precise statements are somewhat different, but the connection should not be
surprising.

\subsection{Why is $d$ fixed?}

A question that arises is why is $d$ fixed in the definition, since there 
is often degree variability in practice.
Basically that is since it makes things harder, and so it is significant 
that expanders exist even then.
Moreover, for certain theoretical issues that is important.
But, in practice the idea of an expander is still useful, and so we go into
that here.

We can define expanders: i.t.o. boundary expansion; or i.t.o. $\lambda_2$.
The intuition is that it is well-connected and then get lots of nice 
properties:
\begin{itemize}
\item
Well-connected, so random walks converge fast.
\item
Quasi-random, meaning that it is empirically random (although in a fairly 
weak sense).
\end{itemize}

Here are several things to note:
\begin{itemize}
\item
Most theorems in graph theory go through to weighted graphs, if you are
willing to have factors like $\frac{w_{max}}{w_{min}}$---that is a problem 
if there is very significant degree heterogeneity or heterogeneity in 
weights, as is common.
So in that case many of those results are less interesting.
\item
In many applications the data are extremely sparse, like a constant number
of edges on average (although there may be a big variance).
\item
There are several realms of $d$, since it might not be obvious what is 
big and what is small:
\begin{itemize}
\item
$d=n$: complete (or nearly complete) graph.
\item
$d=\Omega(\mbox{polylog}(n))$: still dense, certainly in a theoretical 
sense, as this is basically the asymptotic region.  
\item
$d=\Theta(\mbox{polylog}(n))$:
still sufficiently dense that measure concentrates, \emph{i.e.}, enough 
concentration for applications;
Harr measure is uniform, and there are no ``outliers''
\item
$d=\Theta(1)$: In this regime things are very sparse, $G_{nm} \ne G_{np}$, 
so you have a situation where the graph has a giant component but isn't 
fully connected; so $3$-regular random graphs are different than $G_{np}$ 
with $p=\frac{3}{n}$.
\end{itemize}
You should think in terms of $d=\Theta(\mbox{polylog}(n))$ at most, although 
often can't tell $O(\log n)$ versus a big constant, and comparing trivial 
statistics can hide what you want.
\item
The main properties we will show will generalize to degree variability.
In particular:
\begin{itemize}
\item
High expansion $\rightarrow$ high conductance.
\item
Random walks converge to ``uniform'' distribution $\rightarrow$ 
random walks converge to a distribution that is uniform over the edges, 
meaning proportional to the degree of a node. 
\item
Expander Mixing Property $\rightarrow$ 
Discrepancy and Empirical Quasi-randomness
\end{itemize}
\end{itemize}

So, for theoretical applications, we need $d=\Theta(1)$; but
for data applications, think i.t.o. a graph being expander-like, i.e., think 
of some of the things we are discussing as being relevant for the properties
of that data graph, if:
(1) it has good conductance properties; and
(2) it is empirically quasi-random.
This happens when data are extremely sparse and pretty noisy, both of which 
they often are.

\subsection{Expanders are graphs that are very well-connected}

Here, we will describe several results that quantify the idea that expanders 
are graphs that are very well-connected.

\subsubsection{Robustness of the largest component to the removal of edges}

Here is an example of a lemma characterizing how constant-degree graphs with 
constant expansion are very sparse graphs with extremely good connectivity 
properties.
In words, what the following lemma says is that the removal of $k$ edges cannot 
cause more that $O\left(\frac{k}{d}\right)$ vertices to be disconnected from 
the rest.
(Note that it is always possible to disconnect $\frac{k}{d}$ vertices after 
removing $k$ edges, so the connectivity of an expander is the best possible.)

\begin{lemma}
Let $G=(V,E)$ be a regular graph with expansion $\phi$.
Then, after an $\epsilon < \phi$ fraction of the edges are adversarially 
removed, the graph has a connected component that has at least $1-\frac{\epsilon}{2\phi}$
fraction of the vertices.
\end{lemma}
\begin{Proof}
Let $d$ be the degree of $G$.
Let $E^{\prime} \subseteq E$ be an arbitrary subset of 
$\le \epsilon |E| = \epsilon d \frac{|V|}{2}$ edges.
Let $C_1,\ldots,C_m$ be the connected components of the graph $(V,E\diagdown E^{\prime})$, 
ordered s.t. 
\[
|C_1| \ge |C_2| \ge \cdots |C_m|  .
\]
In this case, we want to prove that 
\[
|C_1| \ge |V|\left( 1-\frac{2\epsilon}{\phi} \right)
\]
To do this, note that 
\[
|E^{\prime}| \ge \frac{1}{2} \sum_{ij} E\left(C_i,C_j\right)
               = \frac{1}{2} \sum_i    E\left(C_i,V\diagdown C_i\right)  .
\]
So, if $|C_1| \le \frac{|V|}{2}$, then 
\[
|E^{\prime}| \ge \frac{1}{2} \sum_i d \phi |C_i| = \frac{1}{2} d \phi |V| ,
\]
which is a contradiction if $\phi > \epsilon$.
On the other hand, if $|C_1| \ge \frac{|V|}{2}$, then let's 
define $S$ to be $S= C_2 \cup \ldots \cup C_m$.
Then, we have 
\[
|E^{\prime}| \ge E(C_1,S) \ge d \phi |S|  ,
\]
which implies that
\[
|S| \le \frac{\epsilon}{2\phi} |V|  ,
\]
and so $|C_1 \ge \left( 1-\frac{\epsilon}{2\phi}\right) |V|$, 
from which the lemma follows.
\end{Proof}

\subsubsection{Relatedly, expanders exhibit quasi-randomness}

In addition to being well-connected in the above sense (and other senses), 
expanders also ``look random'' in certain senses.

\paragraph{One direction}

For example, here I will discuss connections with something I will call 
``empirical quasi-randomness.''
It is a particular notion of things looking random that will be useful for
what we will discuss.
Basically, it says that the number of edges between any two subsets of nodes
is very close to the expected value, which is what you would see in a random 
graph.
Somewhat more precisely, it says that when $\lambda$ below is small, then 
the graph has the following quasi-randomness property: for every two disjoint 
sets of vertices, $S$ and $T$, the number of edges between $S$ and $T$ is 
close to $\frac{d}{n}|S|\cdot|T|$, i.e., what we would expect a random graph
with the same average degree $d$ to have.
(Of course, this could also hide other structures of potential interest, as 
we will discuss later, but it is a reasonable notion of ``looking random'' 
in the large scale.)
Here, I will do it in terms of expansion---we can generalize it and do it 
with conductance and discrepancy, and we may do that later.

We will start with the following theorem, called the ``Expander Mixing 
Lemma,'' which shows that if the spectral gap is large, then the number of 
edges between two subsets of the graph vertices can be approximated by the 
same number for a random graph, \emph{i.e.}, what would be expected on 
average, so it looks empirically random.
Note that $\frac{d}{n}|S|\cdot|T|$ is the average value of the number of 
edges between the two sets of nodes in a random graph; also, note that 
$\lambda\sqrt{|S|\cdot|T|}$ is an ``additive'' scale factor, which might be
very large, e.g., too large for the following lemma to give an interesting 
bound, in particular when one of $S$ or $T$ is very small.

\begin{theorem}[Expander Mixing Lemma]
Let $G=(V,E)$ be a $d$-regular graph, with $|V|=n$ and 
$\lambda=\max(|\mu_{2}|,|\mu_{n}|)$, where $\mu_i$ is the $i$-th largest eigenvalue 
of the (non-normalized) Adjacency Matrix.
Then, for all $S,T\subseteq V$, we have the following:
\[
\left|\left|E(S,T)\right|-\frac{d}{n}|S|\cdot|T|\right|
   \leq\lambda\sqrt{|S|\cdot|T|}   .
\]
\end{theorem}
\begin{proof}
Define $\chi_{S}$ and $\chi_{T}$ to be the characteristic vectors of $S$ 
and $T$. 
Then, if $\{v_{j}\}_{j=1}^{n}$ are orthonormal eigenvectors of $A_{G}$, 
with $v_1 = \frac{1}{\sqrt{n}}(1,\ldots,1)$, then we can write the expansion
of $\chi_{S}$ and $\chi_{T}$ in terms of those eigenvalues as: 
$\chi_{S}=\sum_{i}\alpha_{i}v_{i}$ and $\chi_{T}=\sum_{j}\beta_{j}v_{j}$. 
Thus,
\begin{eqnarray*}
\left|E(S,T)\right|     
   &=& \chi_{S}^{T}A\chi_{T}     \\
   &=& \left( \sum_{i}\alpha_{i}v_{i} \right) A \left( \sum_{j}\beta_{j}v_{j} \right)     \\
   &=& \left( \sum_{i}\alpha_{i}v_{i} \right) \left( \sum_{j}\mu_{j}\beta_{j}v_{j} \right)\\
   &=& \sum_i\mu_i\alpha_i\beta_i \quad\mbox{since the $v_i$'s are orthonormal}  .
\end{eqnarray*}
Thus, 
\begin{eqnarray*}
\left|E(S,T)\right|     
   &=& \sum\mu_{i}\alpha_{i}\beta_{i}     \\
   &=& \mu_{1}\alpha_1\beta_1+\sum_{i\geq 2}\mu_{i}\alpha_{i}\beta_{i}  \\
   &=& d\frac{|S|.|T|}{n}+\sum_{i\geq1}\mu_{i}\alpha_{i}\beta_{i}   ,
\end{eqnarray*}
where the last inequality is because, 
$\alpha_{1}= \langle \chi_{S},\frac{\overrightarrow{1}}{\sqrt{n}} \rangle=\frac{|S|}{\sqrt{n}}$ and (similarly) $\beta_{1}=\frac{|T|}{\sqrt{n}}$, and $\mu_{1}=d$.
Hence,
\begin{eqnarray*}
\left|\left|E(S,T)\right|-\frac{d}{n}|S|\cdot|T|\right|
   &   =& \left|\sum_{i=2}^{n}\mu_{i}\alpha_{i}\beta_{i}\right|   \\
   &\leq& \sum_{i\geq 2}|\mu_{i}\alpha_{i}\beta_{i}|   \\
   &\leq& \lambda\sum_{i\geq1}|\alpha_{i}||\beta_{i}|   \\
   &\leq& \lambda||\alpha||_{2}||\beta||_{2}     
       =  \lambda||\chi_{S}||_{2}||\chi_{T}||_{2}   
       =  \lambda\sqrt{|S|\cdot|T|}
\end{eqnarray*}
\end{proof}

\paragraph{Other direction}

There is also a partial converse to this result:
\begin{theorem}[Bilu and Linial]
Let $G$ be a $d$-regular graph, and suppose that
\[
\left| E(S,T) - \frac{d}{n}|S|\cdot|T| \right|
   \leq \rho \sqrt{|S|\cdot|T|}
\]
holds $\forall$ disjoint $S$,$T$ and for some $\rho > 0$.
Then
\[
\lambda \le O\left( \rho\left( 1+ \log(\frac{d}{\rho}) \right) \right)
\]
\end{theorem}

\subsubsection{Some extra comments}

We have been describing these results in terms of regular and unweighted
graphs for simplicity, especially of analysis since the statements of the 
theorems don't change much under generalization.
Important to note: these results can be generalized to weighted graphs 
with irregular number of edges per nodes using discrepancy.
Informally, think of these characterizations as intuitively defining 
what the interesting properties of an expander are for real data, or what
an expander is more generally, or what it means for a data set to look
expander-like.

Although we won't worry too much about those issues, it is important to 
note that for certain, mostly algorithmic and theoretical applications, 
the fact that $d=\Theta(1)$, etc. are very important.


\subsection{Expanders are graphs that are sparse versions/approximations of a complete graph}

To quantify the idea that constant-degree expanders are sparse 
approximations to the complete graph, we need two steps:
\begin{enumerate}
\item
first, a way to say that two graphs are close; and 
\item
second, a way to show that, with respect to that closeness measure, 
expanders and the complete graph are close.
\end{enumerate}

\subsubsection{A metric of closeness between two graphs}

For the first step, we will view a graph as a Laplacian and vice versa, and
we will consider the partial order over PSD matrices.
In particular, recall that for a symmetric matrix $A$, we can write
\[
A \succeq 0
\] 
to mean that 
\[
A \in PSD
\]
(and, relatedly, $A \succ 0$ to mean that it is PD).
In this case, we can write $A \succeq B$ to mean that $A-B \succeq 0$.
Note that $\succeq$ is a partial order.
Unlike the real numbers, where every pair is comparable, for symmetric 
matrices, some are and some are not.
But for pairs to which it does apply, it acts like a full order, in 
that, e.g., 
\begin{eqnarray*}
& & A \succeq B \mbox{ and } B \succeq C \mbox{ implies } A \succeq C   \\
& & A \succeq B \mbox{ implies that } A + C \succeq B+C  ,
\end{eqnarray*}
for symmetric matrices $A$, $B$, and $C$.

By viewing a graph as its Laplacian, we can use this to define an inequality
over graphs.  In particular, for graphs $G$ and $H$, we can write 
\[
 G \succeq H \mbox{ to mean that } L_G \succeq L_H   .
\]
In particular, from our previous results, we know that if $G=(V,E)$ is 
a graph and $H=(V,F)$ is a subgraph of $G$, then $L_G \succeq L_H$.
This follows since the Laplacian of a graph is the sum of the Laplacians 
of its edges: i.e., since $F \subseteq E$, we have
\[
L_G = \sum_{e in E} L_e 
    = \sum_{e \in F} L_e + \sum_{e in E \diagdown F} L_e
    \preceq \sum_{e \in F} L_e = L_H , 
\]
which follows since $\sum_{e \in E \diagdown F} L_e \succeq 0$.

That last expression uses the additive property of the order; now let's look 
at the multiplicative property that is also respected by that order.

If we have a graph $G=(V,E)$ and a graph $H=(V,E^{\prime})$, let's define the
graph $c \cdot H$ to be the same as the graph $H$, except that every edge is
multiplied by $c$.
Then, we can prove relationships between graphs such as the following.

\begin{lemma}
If $G$ and $H$  are graphs s.t. 
\[
G \succeq c \cdot H
\]
then, for all $k$ we have that
\[
\lambda_k(G) \ge c \lambda_k(H)  .
\]
\label{lem:graphic-ineq-mult}
\end{lemma}
\begin{Proof}
The proof is by the min-max Courant-Fischer variational characterization.
We won't do it in detail.
See DS, 09/10/12.
\end{Proof}

\noindent
From this, we can prove more general relationships, e.g., bounds if edges 
are removed or rewieghted.
In particular, the following two lemmas are almost corollaries of
Lemma~\ref{lem:graphic-ineq-mult}.

\begin{lemma}
If $G$ is a graph and $H$ is obtained by adding an edge to $G$ or 
increasing the weight of an edge in $G$, then, for all $i$, we have
that $\lambda_i(G) \le \lambda_i(H)$.
\end{lemma}

\begin{lemma}
If $G=(V,E,W_1)$ is a graph and $H=(V,E,W_2)$ is a graph that differs from
$G$ only in its weights, then
\[
G \succeq \min_{e \in E} \frac{w_1(e)}{w_2(e)} H  .
\]
\end{lemma}

Given the above discussion, we can use this to define the notion that two 
graphs approximate each other, basically by saying that they are close if 
their Laplacian quadratic forms are close.
In particular, here is the definition.

\begin{definition}
Let $G$ and $H$ be graphs.
We say that $H$ is a $c$-approximation to $H$ if 
\[
cH \succeq G \succeq \frac{1}{c} H  .
\]
\label{def:graph-c-approx}
\end{definition}

\noindent
As a special case, note that if $c=1+\epsilon$, for some $\epsilon\in(0,1)$, 
then we have that the two graphs are very close.

\subsubsection{Expanders and complete graphs are close in that metric}

Given this notion of closeness between two graphs, we can now show that 
constant degree expanders are sparse approximations of the complete 
graph.
The following theorem is one formalization of this idea.
This establishes the closeness; and, since constant-degree expanders are 
very sparse, this result shows that they are sparse approximations of 
the complete graph.
(We note in passing that it is know more generally that every graph can
be approximated by a complete graph; this graph sparsification problem is
of interest in many areas, and we might return to it.)

\begin{theorem}
For every $\epsilon > 0$, there exists a $d > 0$ such that for all 
sufficiently large $n$, there is a $d$ regular graph $G_n$ that is a
$1\pm\epsilon$ approximation of the complete graph $K_n$
\end{theorem}
\begin{Proof}
Recall that a constant-degree expander is a $d$-regular graph whose 
Adjacency Matrix eigenvalues satisfy
\begin{equation}
|\alpha_i| \le \epsilon d ,
\label{eqn:expander-eigenval-bound}
\end{equation}
for all $i \ge 2$, for some $\epsilon < 1$.
We will show that graphs satisfying this condition also satisfy the
condition of Def.~\ref{def:graph-c-approx} (with $c=1+\epsilon$) to be 
a good approximation of the complete graph.

To do so, recall that 
\[
\left(1-\epsilon\right) H \preceq G \preceq \left(1+\epsilon\right) H 
\]
means that 
\[
\left(1-\epsilon\right) x^TL_Hx \le x^TL_Gx \le \left(1+\epsilon\right) x^TL_Hx   .
\]
Let $G$ be the Adjacency Matrix of the graph whose eigenvalues 
satisfy Eqn.~(\ref{eqn:expander-eigenval-bound}).
Given this, recall that the Laplacian eigenvalues satisfy 
$\lambda_i = d-\alpha_i$, and so all of the non-zero eigenvalues of 
$L_G$ are in the interval between $\left(1-\epsilon\right)d$ and
$\left(1+\epsilon\right)d$.
I.e., for all $x$ s.t. $x \perp \vec{1}$, we have that
\[
\left(1-\epsilon\right) x^Tx \le x^TL_Gx \le \left(1+\epsilon\right) x^Tx   .
\]
(This follows from Courant-Fischer or by expanding $x$ is an eigenvalue basis.)

On the other hand, for the complete graph $K_n$, we know that all vectors 
$x$ that are $\perp \vec{1}$ satisfy
\[
x^T L_{K_n}x = n x^Tx .
\]
So, let $H$ be the graph 
\[
H = \frac{d}{n} K_{n}   ,
\]
from which it follows that
\[
x^TL_Hx = d x^Tx .
\]
Thus, the graph $G$ is an $\epsilon$-approximation of the graph $H$,
from which the theorem follows.
\end{Proof}

\noindent
For completeness, consider $G-H$ and let's look at its norm to see that 
it is small.
First note that 
\[
\left(1-\epsilon\right)H \preceq G \preceq \left(1+\epsilon\right)H 
\mbox{ implies that }
-\epsilon H \preceq G-H \preceq \epsilon H .
\]
Since $G$ and $H$ are symmetric, and all of the eigenvalues of 
$\epsilon H$ are either $0$ or $d$, this tells us that 
\[
\|L_G - L_H\|_2 \le \epsilon d  . 
\]

\subsection{Expanders are graphs on which diffusions and random walks mix rapidly}

We will have more to say about different types of diffusions and random 
walks later, so for now we will only work with one variant and establish 
one simple variant of the idea that random walks on expander graphs mix or 
equilibrate quickly to their equilibrium distribution. 

Let $G=(V,E,W)$ be a weighted graph, and we want to understand something 
about how random walks behave on $G$.
One might expect that if, e.g., the graph was a dumbbell graph, then 
random walks that started in the one half would take a very long time to
reach the other half; on the other hand, one might hope that if there 
are no such bottlenecks, e.g., bottlenecks revealed by the expansion 
of second eigenvalue, than random walks would mix relatively quickly.

To see this, let $p_t \in \mathbb{R}^{n}$, where $n$ is the number of 
nodes in the graph, be a probability distribution at time $t$. 
This is just some probability distribution over the nodes, \emph{e.g.}, 
it could be a discrete Dirac $\delta$-function, \emph{i.e.}, the 
indicator of a node, at time $t=0$; it could be the uniform distribution; 
or it could be something else.
Given this distribution at time $t$, the transition rule that governs 
the distribution at time $t+1$ is:
\begin{itemize}
\item
To go to $p_{t+1}$, move to a neighbor with probability $\sim$ the weight 
of the edge.
(In the case of unweighted graphs, this means that move to each 
neighbor with equal probability.)
That is, to get to $p_{t+1}$ from $p_t$, sum over neighbors
\[
p_{t+1}(u) = \sum_{v:(u,v) \in E} \frac{W(u,v)}{d(v)} p_{t}(v)
\]
where $d(v) = \sum_u W(u,v)$ is the weighted degree of $v$.
\end{itemize}

As a technical point, there are going to be bottlenecks, and so we will 
often consider a ``lazy'' random walk, which removed that trivial 
bottleneck that the graph is bipartite thus not mixing (i.e. the stationary distribution doesn't exist) 
and only increases the mixing time by a factor of two (intuitively, on expectation in two steps in the ``lazy'' walk we walk one step as in the simple random walk)---which doesn't 
matter in theory, since there we are interested in polynomial versus 
exponential times, and in practice the issues might be easy to diagnose 
or can be dealt with in less aggressive ways.
Plus it's nicer in theory, since then things are SPSD.

By making a random walk ``lazy,'' we mean the following:
Let
\[
p_{t+1}(u) = \frac{1}{2} p_{t}(u) 
           + \frac{1}{2} \sum_{v:(u,v) \in E} \frac{W(u,v)}{d(v)} p_{t}(v)  .
\]
That is, $p_{t+1} = \frac{1}{2}\left(I+AD^{-1}\right)p_{t}$, and so the
transition matrix $W_{G}=A_{G} D_{G}^{-1}$ is replaced with
$W_{G} = \frac{1}{2}\left(I+A_{G}D_{G}^{-1}\right)$---this is an asymmetric 
matrix that is similar in some sense to the normalized Laplacian.

Then, after $t$ steps, we are basically considering $W_G^t$, in the sense
that
\[
p_0 \rightarrow p_t = W p_{t-1} = W^2 p_{t-2} = \cdots = W^t p_t  .
\]

\textbf{Fact.}
Regardless of the initial distribution, the lazy random walk converges to 
$\pi(i) = \frac{d(i)}{\sum_j d(j)}$, which is the right eigenvector of $W$
with eigenvalue $1$.

\textbf{Fact.}
If $1=\omega_1 \ge \omega_2 \ge \cdots \omega_n \ge 0$ are eigenvalues of
$W$, with $\pi(i) = \frac{d(i)}{\sum_j d(j)}$, then $\omega_2$ governs the
rate of convergence to the stationary distribution.

There are a number of ways to formalize this ``rate of mixing'' result, 
depending on the norm used and other things.
In particular, a very good way is with the total variation distance, which 
is defined as:
\[
\|p-q\|_{TVD} = \max_{S \subseteq V} \left\{ \sum_{v \in S} p_v - \sum_{v \in S} q_v \right\}
              = \frac{1}{2} \|p-q\|_1 .
\]
(There are other measures if you are interested in mixing rates of Markov chains.)
But the basic point is that if $1-\omega_2$ is large, \emph{i.e.}, you 
are an expander, then a random walk converges fast.
For example:
\begin{theorem}
Assume $G=(V,E)$ with $|V|=n$ is $d$-regular, $A$ is the adjacency matrix 
of $G$, and $\hat{A}=\frac{1}{d}A$ is the transition matrix of a random 
walk on $G$, i.e., the normalized Adjacency Matrix. 
Also, assume $\lambda=\max(|\mu_{2}|,|\mu_{n}|)=\alpha d$ (recall $\mu_i$ is the $i$-th largest eigenvalue of $A$, not $\hat{A}$).
Then
\[
||\hat{A}^{t}p-u||_{1}\leq\sqrt{n}\alpha^{t}   , 
\]
where $u$ is the stationary distribution of the random walk, which is the uniform distribution in the undirected $d$-regular graph, and $p$ is
an arbitrary initial distribution on $V$.

In particular, 
if $t \ge \frac{c}{1-\alpha}\log\left(\frac{n}{\epsilon}\right)$, for 
some absolute constant $c$ independent of $n$,
then $\| u - \hat{A}^{t} p \| \le \epsilon$.
\end{theorem}
\begin{proof}
Let us define the matrix $\hat J = \frac{1}{n} \vec{1} \vec{1}^\top $,
where, as before, $\vec{1}$ is the all ones vector of length $n$. Note that,
for any probability vector $p$, we have
\begin{align*}
	\hat J p &= \frac{1}{n} \vec{1} \vec{1}^\top p \\
	&= \frac{1}{n} \vec{1} \cdot 1 \\
	&= u.
\end{align*}
Now, since $\hat A = \frac{1}{d} A $ we have $\hat \mu_i = \mu_i/d$,
where $\hat \mu_i$ denotes the $i$th largest eigenvalue of $\hat A$, and the
eigenvectors of $\hat A$ are equal to those of $A$. Hence, we have
\begin{align*}
	\big \Vert \hat A ^t - \hat J \big \Vert_2 &= \max_{w:\Vert w \Vert_2 \leq 1} \Vert (\hat A^t - \hat J) w \Vert_2 \\
	&= \sigma_{\max} \left( \hat A^t - \hat J \right) \\
	&= \sigma_{\max} \left( \sum_{i=1}^{n} \hat \mu_i^t v_i v_i^\top - \frac{1}{n} \vec{1} \vec{1}^\top  \right) \\
	&\overset{(a)}{=}  \sigma_{\max} \left(  \sum_{i=2}^{n} \hat \mu_i^t v_i v_i^\top  \right) \\
	&= \max\{ \vert \hat \mu_2^t \vert, \vert \hat \mu_n^t \vert\} \\
	&= \alpha^t,
\end{align*}
where $(a)$ follows since $v_1 = \frac{1}{\sqrt{n}} \vec{1}$ and $\hat
\mu_1 = 1$. Then,
\begin{align*}
	\big \Vert \hat A^t p - u \big \Vert_1 
	&\leq \sqrt{n} \big\Vert \hat A^t p - u \big\Vert_2 \\
	&\leq \sqrt{n} \big\Vert \hat A^t p - \hat J p \big\Vert_2 \\
	&\leq \sqrt{n} \big\Vert \hat A^t - \hat J \big\Vert_2 \big \Vert p \big \Vert_2 \\
	&\leq \sqrt{n} \alpha^t,
\end{align*}
which concludes the proof.
\end{proof}

This theorem shows that if the spectral gap is large (i.e. $\alpha$ is small), 
then we the walk mixes rapidly.

This is one example of a large body of work on rapidly mixing Markov 
chains.  
For example, there are extensions of this to degree-heterogeneous graphs 
and all sorts of other thigns
Later, we might revisit this a little, when we see how tight this is; 
in particular, one issue that arises when we discuss local and locally-biased 
spectral methods is that how quickly a random walk mixes depends on not 
only the second eigenvalue but also on the size of the set achieving that
minimum conductance value.

%% file: lect09.tex
\section{%
(02/19/2015): 
Expanders, in theory and in practice (2 of 2)}

Reading for today.
\begin{compactitem}
\item
Same as last class.
\end{compactitem}

Here, we will describe how expanders are the metric spaces that are least 
like low-dimensional Euclidean spaces (or, for that matter, any-dimensional 
Euclidean spaces).
Someone asked at the end of the previous class about what would an expander
``look like'' if we were to draw it.
The point of these characterizations of expanders---that they don't have 
good partitions, that they embed poorly in low dimensional spaces, etc.---is
that \emph{you can't draw them to see what they look like}, or at least you 
can't draw them in any particularly meaningful way.
The reason is that if you could draw them on the board or a two-dimensional
piece of paper, then you would have an embedding into two dimensions.
Relatedly, you would have partitioned the expander into two parts, i.e., 
those nodes on the left half of the page, and those nodes on the right half 
of the page.
Any such picture would have roughly as many edges crossing between the two 
halves as it had on either half, meaning that it would be a 
non-interpretable mess.
This is the reason that we are going through this seemingly-circuitous 
characterizations of the properties of expanders---they are important, but
since they can't be visualized, we can only characterize their properties 
and gain intuition about their behavior via these indirect means.

\subsection{Introduction to Metric Space Perspective on Expanders}

To understand expanders from a metric space perspective, and in particular to
understand how they are the metric spaces that are least like low-dimensional 
Euclidean spaces, let's back up a bit to the seemingly-exotic subject of 
\emph{metric spaces} (although in retrospect it will not seem so exotic or be 
so surprising that it is relevant).
\begin{itemize}
\item
Finite-dimensional Euclidean space, i.e., $\mathbb{R}^{n}$, with $n < \infty$, 
is an example of a metric space that is very nice but that is also quite 
nice/structured or limited.
\item
When you go to infinite-dimensional Hilbert spaces, things get much more 
complex; but $\infty$-dimensional RKHS, as used in ML, are 
$\infty$-dimensional Hilbert spaces that are sufficiently regularized that 
they inherit most of the nice properties of $\mathbb{R}^{n}$.
\item
If we measure distances in $\mathbb{R}^{n}$ w.r.t. other norms, \emph{e.g.}, 
$\ell_1$ or $\ell_{\infty}$, then we step outside the domain of Hilbert 
spaces to the domain of Banach spaces or normed vector spaces.
\item
A graph $G=(V,E)$ is completely characterized by its shortest path or 
geodesic metric; so the metric space is the nodes, with the distance being 
the geodesic distance between the nodes.
Of course, you can modify this metric by adding nonnegative weights to 
edges like with some nonlinear dimensionality reduction methods.
Also, you can assign a vector to vertices and thus view a graph geometrically.
(We will get back to the question of whether there are other distances that 
one can associate with a graph, e.g., resistance of diffusion based distances; 
and we will ask what is the relationship between this and geodesic distance.)
\item
The data may not be obviously a matrix or a graph.
Maybe you just have similarity/dissimilarity information, \emph{e.g.}, 
between DNA sequences, protein sequences, or microarray expression levels.
Of course, you might want to relate these things to matrices or graphs in 
some way, as with RKHS, but let's deal with metrics first.
\end{itemize}
So, let's talk aobut metric spaces more generally.
The goal will be to understand how good/bad things can be when we consider 
metric information about the data.


So, we start with a definition:
\begin{definition}
$(X,d)$ is a \emph{metric space} if
\begin{itemize}
\item
$d:X \times X \rightarrow \mathbb{R}^{+}$ (nonnegativity)
\item
$d(x,y)=0$ iff $x=y$   
\item
$d(x,y)=d(y,x)$ (symmetric)
\item
$d(x,y) \le d(x,z)+d(z,y)$ (triangle inequality)
\end{itemize}
\end{definition}
The idea is that there is a function over the set $X$ that takes as input 
pairs of variables that satisfies a generalization of what our intuition
from Euclidean distances is: namely, nonnegativity, the second condition 
above, symmetry, and the triangle inequality.
Importantly, this metric does not need to come from a dot product, and so
although the intuition about distances from Euclidean spaces is the 
motivation, it is significantly different and more general.
Also, we should note that if various conditions are satisfied, then 
various metric-like things are obtained:
\begin{itemize}
\item
If the second condition above is relaxed, but the other conditions are 
satisfied, then we have a \emph{psuedometric}.
\item
If symmetry is relaxed, but the other conditions are satisfied, 
then we have a \emph{quasimetric}.
\item
If the triangle inequality is relaxed , but the other conditions are 
satisfied, then we have a \emph{semimetric}.
\end{itemize}
We should note that those names are not completely standard, and to confuse
matters further sometimes the relaxed quantities are called metrics---for
example, we will encounter the so-called \emph{cut metric} describing 
distances with respect to cuts in a graph, which is not really a metric 
since the second condition above is not satisfied.

More generally, the distances can be from a Gram matrix, a kernel, or even 
allowing algorithms in an infinite-dimensional space.

Some of these metrics can be a little counterintuitive, and so for a range 
of reasons it is useful to ask how similar or different two metrics are,
\emph{e.g.}, can we think of a metric as a tweak of a low-dimensional space, 
in which case we might hope that some of our previous machinery might apply.
So, we have the following question:
\begin{question}
How well can a given metric space $(X,d)$ be approximated by $\ell_2$, 
where $\ell_2$ is the metric space $(\mathbb{R}^{n},||\cdot||)$, where 
$\forall x,y\in\mathbb{R}^{n}$, we have 
$||x-y||^2=\sum_{i=1}^{n}(x_i-y_i)^2$.
\end{question}
The idea here is that we want to replace the metric $d$ with something 
$d'$ that is ``nicer,'' while still preserving distances---in that case, 
since a lot of algorithms use only distances, we can work with $d'$ in the 
nicer place, and get results that are algorithmically and/or statistically 
better without introducing too much error.
That is, maybe it's faster without too much loss, as we formulated it
before; or maybe it is better, in that the nicer place introduced some 
sort of smoothing.
Of course, we could ask this about metrics other than $\ell_2$; we just 
start with that since we have been talking about it.

There are a number of ways to compare metric spaces.
Here we will start by defining a measure of distortion between two metrics.
\begin{definition}
Given a metric space $(X, d)$ and our old friend the metric space 
$(\mathbb{R}^{n}, \ell_2)$, and a mapping f: $X \rightarrow \mathbb{R}^{n}$:
\begin{itemize}
\item
$expansion(f) = max_{x_1, x_2 \in X} \frac{\left||{f(x_1) - f(x_2)}\right||_2}{d(x_1, x_2)}$
\item
$contraction(f) = max \frac{d(x_1, x_2)}{\left||{f(x_1) - f(x_2)}\right||}$
\item
$distortion(f) = expansion(f) \cdot contraction(f)$
\end{itemize}
\end{definition}

As usual, there are several things we can note:
\begin{itemize}
\item
An embedding with distortion $1$ is an \emph{isometry}.
This is very limiting for most applications of interset, which is OK since 
it is also unnecessarily strong notion of similarity for most applications 
of interest, so we will instead look for low-distortion embeddings.
\item
There is also interest in embedding into $\ell_1$, which we will return to 
below when talking about graph partitioning.
\item
There is also interest in embedding in other ``nice'' places, like trees, 
but we will not be talking about that in this class.
\item
As a side comment, a Theorem of Dvoretzky 
says that any embedding into normed spaces, $\ell_2$ is the hardest.
So, aside from being something we have already seen, this partially 
justifies the use of $\ell_2$ and the central role of $\ell_2$ in embedding 
theory more generally.
\end{itemize}

Here, we should note that we have already seen one example (actually, 
several related examples) of a low-distortion embedding.
Here we will phrase the JL lemma that we saw before in our new nomenclature.

\begin{theorem}[JL Lemma]
Let $X$ be an $n$-point set in Euclidean space, \emph{i.e.}, 
$X \subset \ell_2^n$, and fix $\epsilon \in (0,1]$.
Then $\exists$ a $(1+\epsilon)$-embedding of $X$ into $\ell_2^k$, 
where $k=O\left(\frac{\log n}{\epsilon^2}\right)$.
\end{theorem}
That is, Johnson-Lindenstrauss says that we can map $x_i \rightarrow f(x)$ 
such that distance is within $1 \pm \epsilon$ of the original.

A word of notation and some technical comments:
For $x\in\mathbb{R}^{d}$ and $p \in [1,\infty)$, the $\ell_p$ norm of $x$ is 
defined as $||x||_p = \left( \sum_{i=1}^{d} |x_i|^p \right)^{1/p}$.
Let $\ell_p^d$ denote the space $\mathbb{R}^{d}$ equipped with the $\ell_p$ 
norm.
Sometimes we are interested in embeddings into some space $\ell_p^d$, with 
$p$ given but the dimension $d$ unrestricted, \emph{e.g.}, in \emph{some}
Euclidean space s.t. $X$ embeds well.
Talk about: $\ell_p = $ the space of all sequences $(x_1,x_2, \ldots)$, with 
$||x||_p < \infty$, with $||x||_p$ defined as 
$||x||_p = \left( \sum_{i=1}^{\infty} |x_i|^p \right)^{1/p}$.
In this case, embedding into $\ell_p$ is shorthand for embedding into 
$\ell_p^d$ for some $d$.

Here is an important theorem related to this and that we will return to later.
\begin{theorem}[Bourgain]
Every $n$-point metric space $(X,d)$ can be embedded into Euclidean space 
$\ell_2$ with distortion $\leq O(\log n)$.
\end{theorem}
\begin{proof} [Proof Idea.]
(The proof idea is nifty and used in other contexts, but we won't use it 
much later, except to point out how flow-based methods do something similar.)
The basic idea is given $(X, d)$, map each point $x \rightarrow \phi(x)$ 
in $O(\log^2 n)$-dimensional space with coordinates equal to the distance to $S \subseteq X$ 
where $S$ is chosen randomly.
That is, given $(X,d)$, map every point $x \in X$ to $\phi(x)$, an 
$O(\log^2 n)$-dimensional vector, where coordinates in $\phi(\cdot)$ 
correspond to subsets $S \subseteq X$, and s.t. the $s$-th in $\phi(x)$ is
$d(x,S) = \min_{s \in S} d(x,s)$.
Then, to define the map, specify a collection of subsets we use selected 
carefully but randomly---select $O(\log n)$ subsets of size $1$, $O(\log n)$ 
subsets of size $2$, of size $4$, $8$, $\ldots$, $\frac{n}{2}$.
Using that, it works, \emph{i.e.}, that is the embedding.
\end{proof}
Note that the dimension of the Euclidean space was originally $O(\log^2n)$, 
but it has been improved to $O(\log n)$, which I think is tight.
Note also that the proof is algorithmic in that it gives an efficient 
randomized algorithm.

Several questions arise:
\begin{itemize}
\item
Q: Is this bound tight? 
A: YES, on expanders.
\item
Q: Let $c_2(X,d)$ be the distortion of the embedding of $X$ into $\ell_2$;
can we compute $c_2(X,d)$ for a given metric?
A: YES, with an SDP.
\item
Q: Are there metrics such that $c_2(X,d) \ll \log n$?
A: YES, we saw it with JL, \emph{i.e.}, high-dimensional Euclidean spaces, 
which might be trivial since we allow the dimension to float in the 
embedding, but there are others we won't get to.
\end{itemize}

\subsubsection {Primal}

The problem of whether a given metric space is $\gamma$-embeddable into 
$\ell_2$ is polynomial time solvable.
Note: this does not specify the dimension, just whether there is some 
dimension; asking the same question with dimension constraints or a fixed 
dimension is in general much harder.
Here, the condition that the distortion $\leq \gamma$ can be expressed as a
system of linear inequalities in Gram matrix correspond to vectors 
$\phi(x)$.
So, the computation of $c_2(x)$ is an SDP---which is easy or hard, depending 
on how you view SDPs---actually, given an input metric space $(X,d)$ and an 
$\epsilon > 0$, we can determine $c_2(X,d)$ to relative error $\leq\epsilon$ 
in $\mbox{poly}(n,1/\epsilon)$ time.

Here is a basic theorem in the area:
\begin{theorem}[LLR]
$\exists$ a poly-time algorithm that, given as input a metric space $(X,d)$,
computes $c_2(X,d)$, where $c_2(X,d)$ is the least possible distortion of
any embedding of $(X,d)$ into $(\mathbb{R}^{n},\ell_2)$.
\end{theorem}
\begin{proof}
The proof is from HLW, and it is based on semidefinite programming.
Let $(X,d)$ be the metric space, 
let $|X|=n$, and let $f:X\rightarrow\ell_2$.
WLOG, scale $f$ s.t. $contraction(f)=1$.
Then, $distortion(f)\leq\gamma$ iff
\begin{equation}
d(x_i,x_j)^2 \leq ||f(x_i)-f(x_j)||^2 \leq \gamma^2 d(x_i,x_j)^2   .
\label{eqn:llr1-eq1}
\end{equation}
Then, let $u_i=f(x_i)$ be the $i$-th row of the embedding matrix $U$, and 
let $Z=UU^T$.
Note that $Z\in PSD$, and conversely, if $Z\in PSD$, then $Z=UU^T$, for some
matrix $U$.
Note also:
\begin{eqnarray*}
||f(x_i)-f(x_j)||^2 &=& ||u_i-u_j||^2 \\
                    &=& (u_i-u_j)^T(u_i-u_j) \\
                    &=& u_i^Tu_i + u_j^Tu_j - 2u_i^Tu_j  \\
                    &=& Z_{ii} + Z_{jj} - 2Z_{ij}  .
\end{eqnarray*}
So, instead of finding a $u_i=f(x_i)$ s.t.~(\ref{eqn:llr1-eq1}) holds, we
can find a $Z \in PSD$ s.t. 
\begin{equation}
d(x_i,x_j)^2 \leq Z_{ii} + Z_{jj} - 2Z_{ij} \leq \gamma^2 d(x_i,x_j)^2   .
\label{eqn:llr1-eq2}
\end{equation}
Thus, $c_2\leq\gamma$ iff $\exists Z \in SPSD$ s.t.~(\ref{eqn:llr1-eq2})
holds $\forall ij$.
So, this is an optimization problem, and we can solve this with simplex, 
interior point, ellipsoid, or whatever; and all the usual issues apply.
\end{proof}

\subsubsection {Dual}

The above is a Primal version of the optimization problem.
If we look at the corresponding Dual problem, then this gives a 
characterization of $c_2(X,d)$ that is useful in proving lower bounds.
(This idea will also come up later in graph partitioning, and elsewhere.)
To go from the Primal to the Dual, we must take a nonnegative linear 
combination of constraints.
So we must write $Z \in PSD$ in such a way, since that is the constraint 
causing a problem; the following lemma will do that.

\begin{lemma}
$Z \in PSD$ iff $\sum_{ij} q_{ij} z_{ij} \geq 0, \forall Q \in PSD$.
\end{lemma}
\begin{proof}
First, we will consider rank $1$ matrices;
the general result will follow since general PSD matrices are a linear 
combination of rank-$1$ PSD matrices of the form $qq^T$, \emph{i.e.}, 
$Q=qq^T$.

First, start with the $\Leftarrow$ direction:
for $q\in\mathbb{R}^{n}$, let $Q$ be $PSD$ matrix s.t. $Q_{ij}=q_iq_j$;
then
\[
q^tZq = \sum_{ij} q_iZ_{ij}q_j = \sum_{ij} Q_{ij}z_{ij} \geq 0 ,
\]
where the inequality follows since $Q$ is $PSD$.
Thus, $Z \in PSD$.

For the $\Rightarrow$ direction:
let $Q$ be rank-$1$ PSD matrix; thus, it has the form $Q=qq^T$ or 
$Q_{ij}=q_iq_j$, for $q\in\mathbb{R}^{n}$.
Then, 
\[
\sum_{ij}Q_{ij}z_{ij} = \sum_{ij}q_iZ_{ij}q_j \geq 0   ,
\]
where the inequality follows since $A$ is $PSD$.

Thus, since $Q \in PSD$ implies that 
$Q = \sum_i q_i q_i^T = \sum_i \Omega_i$, with $\Omega_i$ being a rank-$i$ 
PSD matrix, the lemma follows by working through things.
\end{proof}

Now that we have this characterization of $Z \in PSD$ in terms of a set 
of (nonnegative) linear combination of constraints, we are ready to get 
out Dual problem which will give us the nice characterization of $c_2(X,d)$.

Recall finding an embedding $f(x_i) = u_i$ iff finding a matrix $Z$ iff 
$\sum_{ij} q_{ij}z_{ij} \geq 0, \forall Q \in SPSD$.
So, the Primal constraints are:
\renewcommand{\labelenumi}{\Roman{enumi}.}
\begin{enumerate}
\item $\sum q_{ij}z_{ij}  \geq 0 $ for all $Q \epsilon PSD$
\item $z_{ii} + z_{jj} - 2z_{ij} \geq {d(x_i,x_j)}^2 $
\item $\gamma^2{d(x_i,x_j)}^2 \geq z_{ii} + z_{jj} - 2z_{ij}$   ,
\end{enumerate}
which hold $\forall ij$.
Thus, we can get the following theorem.

\begin{theorem}[LLR]
\[
C_2(X,d) = \max_{(P \epsilon PSD, P.1=0)} 
           \sqrt{\frac{\sum_{P_{ij} >0} P_{ij}d(x_i,x_j)^2}{-\sum_{(P_{ij} <0)} P_{ij}d(x_i,x_j)^2}}
\]
\end{theorem}
\begin{proof}
The dual program is the statement that for $\gamma < C_2(X,d)$, thre must
exist a non-negative combination of the constraints of the primal problem 
that yields a contradiction.

So, we will assume $\gamma < C_2(x,d)$ and look for a contradiction, 
\emph{i.e.}, look for a linear combination of constraints such that the 
primal gives a contradiction.
Thus, the goal is to construct a nonnegative linear combination of 
primal constraints  to give a contradiction s.t. 
$Q.Z= \sum q_{ij}z_{ij} \geq 0$.

Recall that the cone of PSD matrices is convex. 

The goal is to zero out the $z_{ij}$.

(I) above says that $Q\cdot Z = \sum_{ij}q_{ij}z_{ij} \geq 0$.
Note that since PSD cone is convex, a nonnegative linear combination of
the form 
$\sum_{k} \alpha_k \langle Q,Z \rangle = P \cdot z$, for some $P \in PSD$.
So, modifying first constraint from the primal, you get 
\renewcommand{\labelenumi}{\Roman{enumi}'.}
\begin{enumerate}
\item 
$\sum_{ij} P_{ij}Z_{ij} = P \cdot Z  \geq 0 $, for some $P \epsilon PSD $ \\
To construct $P$, choose the elements such that you zero out $z_{ij}$ in the
following manner.
\begin{itemize}
\item
If $P_{ij}>0$, multiply second constraint from primal by $ P_{ij}/2$, (\emph{i.e.}, the constraint ${d(x_i,x_j)}^2 \leq z_{ii} + z_{jj} - 2z_{ij} $)
\item 
If $P_{ij}<0$, multiply third constraint from primal by $-P_{ij}/2$, (\emph{i.e.}, the constraint $z_{ii} + z_{jj} - 2z_{ij} \leq \gamma^2{d(x_i,x_j)}^2 $)
\item
If $P_{ij}=0$, multiply by $0$ constraints involving $z_{ij}$.
\end{itemize}
This gives
\begin{eqnarray*}
\frac{P_{ij}}{2}\left( z_{ii} + z_{jj} - 2z_{ij} \right) 
   &\geq& \frac{P_{ij}}{2} {d(x_i,x_j)}^2      \\
-\frac{P_{ij}}{2} \gamma^2{d(x_i,x_j)}^2 
   &\geq& -\frac{P_{ij}}{2} \left( z_{ii} + z_{jj} - 2z_{ij} \right) 
\end{eqnarray*}
from which it follows that you can modify the other constraints from primal
to be:
\item 
$\sum_{{ij}, P_{ij}>0} \frac{P_{ij}}{2} (z_{ii} + z_{jj} - 2z_{ij}) 
   \geq \sum_{{ij}, P_{ij}>0} \frac{P_{ij}}{2} d(x_i,x_j)^{2}$
\item 
$\sum_{{ij}, P_{ij}<0} \frac{P_{ij}}{2} (z_{ii} + z_{jj} - 2z_{ij}) 
   \geq \sum_{{ij}, P_{ij}<0} \frac{P_{ij}}{2} \gamma^2(d(x_i,x_j)^2)$
\end{enumerate}
If we add those two things, then we get, 
\[
\sum_i P_{ii}z_{ii} + \sum_{ij:P_{ij}>0} \frac{P_{ij}}{2}(z_{ii}+z_{jj}) + \sum_{ij} \frac{P_{ij}}{2}(z_{ii}+z_{jj}) \geq \text{RHS Sum} ,
\]
and so
\[
\sum_i P_{ii}z_{ii} + \sum_{ij:P_{ij} \ne 0} \frac{P_{ij}}{2}(z_{ii}+z_{jj}) \geq \text{RHS Sum},
\]
and so, since we choose $P$ s.t. $P\cdot\vec{1} = \vec{0}$, (i.e. $\sum_j P_{ij}=0$ for all $i$, and $\sum_i P_{ij}=0$ for all $j$ by symmetry) we have that
\[
0    = \sum_i\left(P_{ii} + \sum_{j:P_{ij}\ne 0} P_{ij}\right)z_{ij} 
  \geq RHS
     = \sum_{ij:P_{ij}>0} \frac{P_{ij}}{2}d(x_i,x_j)^2
     + \sum_{ij:P_{ij}<0} \frac{P_{ij}}{2}\gamma^2d(x_i,x_j)^2
\]
So, it follows that
\[
0 \geq \sum_{ij:P_{ij}>0} P_{ij}d(x_i,x_j)^2 + \sum_{ij:P_{ij}<0}\gamma^2d(x_i,x_j)^2 .
\]
This last observation is FALSE if
\[
\gamma^2 < \frac{ \sum_{ij:P_{ij}>0} P_{ij} d(x_i,x_j)^2 }{ \sum_{ij:P_{ij}<0} (-P_{ij}) d(x_i,x_j)^2 } 
\]
and so the theorem follows.

(In brief,
adding the second and third constraints above,
\begin{center}
$ 0 \geq \sum_{ij, P_{ij} >0 } P_{ij} d(x_i,x_j)^{2} +  \sum_{ij, P_{ij} <0 }P_{ij}\gamma^2(d(x_i,x_j)^2) $
\end{center}
This will be false if you choose $\gamma$ to be small---in particular, it 
will be false if $\gamma^2 \leq top/bottom$, from which the theorem will 
follow.
\end{proof}

\subsection{Metric Embedding into $\ell_{2}$}

We will show that expanders embed poorly in $\ell_2$---this is the basis 
for the claim that they are the metric spaces that are least like 
low-dimensional spaces in a very strong sense.

It is easy to see that an expander can be embedded into $\ell_{2}$ with
distortion $O(\log\, n)$ (just note that any graph can be embedded with 
distortion equal to its diameter)---in fact, any metric space can be 
embedded with that distortion. 
We will show that this result is tight, and thus that expanders are the
worst.

The basic idea for showing that expanders embed poorly in $\ell_2$ is:
If $G$ is a $d$-regular, $\epsilon$-expander, then $\lambda_2$ of $A_G$ is
$< d-\delta$ for $\delta = \delta(d,\epsilon) \ne \delta(n)$.
The vertices of a bounded degree graph can be paired up s.t. every pair of 
vertices are a distance $\Omega(\log n)$.
We can then let $B$ be a permutation matrix for the pairing, and use the
matrix $P = dI - A + \frac{\delta}{2}(B-I)$.

Note: We can have a simpler proof, using the theorem of Bourgain that 
expanders don't embed well in $\ell_2$, since we can embed in 
$\mathbb{R}^{diameter}$, and $diameter(expander)=\log n$.
But we go through this here to avoid (too much) magic.

Start with the following definitions:
\begin{definition}
A \emph{Hamilton cycle} in a graph $G=(V,E)$ is a cycle that visits every
vertex exactly once (except for the start and end vertices).
\end{definition}
\begin{definition}
A \emph{matching} is a set of pairwise non-adjacent edges, \emph{i.e.}, 
no two edges share a common vertex.
A vertex is \emph{matched} if it is incident to an edge.
A \emph{perfect matching} is a matching that matched all the vertices of 
a graph.
\end{definition}


The following theorem is the only piece of magic we will use here:
\begin{theorem}
A simple graph with $n \ge 3$ edges is Hamiltonian if every vertex has 
degree $\ge \frac{n}{2}$.
\end{theorem}
Note that if every vertex has degree $\ge \frac{n}{2}$, then the graph is
actually quite dense, and so from Szemeredi-type results relating dense
graphs to random graphs it might not be so surprising that there is a lot
of wiggle room.

\textbf{Note:}
A cycle immediately gives a matching.

Thus, we have the following lemma:
\begin{lemma}
Let $G=(V,E)$ be a $d$-regular graph, with $|V|=n$
If $H=(V,E')$ is a graph with the same vertex set as
$G$, in which two vertices $u$ and $v$ are adjacent iff 
$d_{G}(u,v)\geq \lfloor\log_{k}n\rfloor$.
Then, $H$ has a matching with $n/2$ edges.
\end{lemma}
\begin{proof}
Since $G$ is a $d$-regular graph, hence for any vertex $x\in V$ and any 
value $r$, it has $\leq k^{r}$ vertices $y\in V$ can have $d_{G}(x,y)\leq r$,
\emph{i.e.}, only that many vertices are within a distance $r$.
If $r = \lfloor \log_{k} n \rfloor -1 $, then $\exists \leq \frac{n}{2}$ 
vertices within distance $r$; that is, at least half of the nodes of $G$ 
are further than $\log_{k}n-2$ from $x$; this means every node in $H$ has 
at least degree $n/2$.
So $H$ has a Hamilton cycle and thus a perfect matching, and by the 
above theorem the lemma follows.
\end{proof}

Finally we get to the main theorem 
that says that expanders embed poorly in $\ell_2$---note that this 
is a particularly strong statement or notion of nonembedding, as by Bourgain we know any graph (with the graph distance metric) can be embedded into $\ell_2$ with distortion $O(\log n)$, so expander is the worst case in this sense.
\begin{theorem} 
Let $d \geq 3$, and let $\epsilon > 0$.
If $G=(V,E)$ is a $(n,d)$-regular graph with $\lambda_2(A_G)\leq d-\epsilon$
and $|V|=n$, then
\[
C_{2}(G)=\Omega(\log\, n)
\]
where the constant inside the $\Omega$ depends on $d,\epsilon$.
\end{theorem}
\begin{proof}
To prove the lower bound, we use the characterization from the last section
that for the minimum distortion in embedding a metric space $(X,d)$ into 
$l_{2}$, denoted by $C_{2}(X,d)$, is:
\begin{equation}
C_{2}(X,d)=\max_{P\in PSD,P.\overrightarrow{1}=\overrightarrow{0}}\sqrt{\frac{\sum_{p_{ij}>0}p_{ij}d(x_{i},x_{j})^{2}}{-\sum_{p_{ij}<0}p_{ij}d(x_{i},x_{j})^{2}}}
\label{eq:c2_x_d}
\end{equation}
and we will find some $P$ that is feasible that gives the LB.

Assume $B$ is the adjacency matrix of the matching in $H$, whose existence 
is  proved in the previous lemma. 
Then, define 
\[
P=(dI-A_{G})+\frac{\epsilon}{2}(B-I)  . 
\]
Then, we claim that $P\overrightarrow{1}=\overrightarrow{0}$. To see this, notice both $(dI-A_G)$ and $I-B$ are Laplacians (not normalized), as $B$ is the adjacency matrix of a perfect matching (i.e. $1$-regular graph).
Next, we claim that $P\in\mbox{PSD}$.
This proof of this second claim is because, 
for any $x\perp\overrightarrow{1}$, we have
\[
x^{T}(dI-A_{G})x \geq dx^Tx-x^TAx \geq (d-\lambda_2)||x||^2 \geq \epsilon||x||^{2}
\]
(by the assumption on $\lambda_{2}$); and
\begin{eqnarray*}
x^{T}(B-I)x
   &   =& \sum_{(i,j)\in B}2x_{i}x_{j}-\sum_{i}x_{i}^{2}     \\
   &   =& \sum_{(i,j)\in B}(2x_{i}x_{j}-x_{i}^{2}-x_{j}^{2})  \qquad\mbox{since $||x||^2=\sum_{(i,j)\in B} x_i^2+x_j^2$}    \\
   &\geq& -2\sum_{(i,j)\in B} x_{i}^{2}+x_{j}^{2}    \\
   &   =& -2||x||^{2}  
\end{eqnarray*}

\noindent
The last line is since $||x||^2=\sum_{(i,j)\in B} x_i^2+x_j^2$ and since $B$ is a matching so each $i$ shows up in the sum exactly once.

So, we have that
\begin{eqnarray*}
x^TPx &   =& x^T(dI-A_G)x + x^T\frac{\epsilon}{2}(B-I)x   \\
      &\geq& \epsilon||x||^2 -\frac{2||x||^2\epsilon}{2}  \\
      &   =& 0   .
\end{eqnarray*}

Next evaluate the numerator and the denominator.
\begin{eqnarray*}
-\sum_{P_{ij}<0} d(i,j)^2 P_{ij} &   =& dn   \\\
 \sum_{P_{ij}>0} d(i,j)^2P_{ij} &\geq& \frac{\epsilon}{2}n\lfloor\log_{d}n\rfloor^2
\end{eqnarray*}
where the latter follows since the distances of edges in $B$ are at least 
$\lfloor \log_{d}n\rfloor$.
Thus, for this $P$, we have that:
\[
\sqrt{\frac{\sum_{p_{ij}>0}p_{ij}d(x_{i},x_{j})^{2}}{-\sum_{p_{ij}<0}p_{ij}d(x_{i},x_{j})^{2}}} \geq \sqrt{\frac{ \frac{\epsilon}{2}n\lfloor\log_kn\rfloor^2 }{dn}}
             \sim \Theta(\log n)
\]
and thus, 
from~(\ref{eq:c2_x_d}), that $C_2$ is at least this big, \emph{i.e.}, that:
\[
C_{2}(G)\geq\Omega(\log\, n)
\]
\end{proof}

%% file: lect10.tex
\section{%
(02/24/2015): 
Flow-based Methods for Partitioning Graphs (1 of 2)}

Reading for today.
\begin{compactitem}
\item
``Multicommodity max-flow min-cut theorems and their use in designing approximation algorithms,'' in JACM, by Leighton and Rao 
\item
``Efficient Maximum Flow Algorithms,'' in CACM, by Goldberg and Tarjan 
\end{compactitem}

\subsection{Introduction to flow-based methods}

Last time, we described the properties of expander graphs and showed that 
they have several ``extremal'' properties.
Before that, we described a vanilla spectral partitioning algorithms, which 
led to the statement and proof of Cheeger's Inequality.
Recall that one direction viewed $\lambda_2$ as a relaxation of the 
conductance or expansion problem; while the other direction gave a 
``quadratic'' bound as well as a constructive proof of a graph partitioning 
algorithm.
The basic idea was to compute a vector, show that it is a relaxation of the
original problem, and show that one doesn't loose too much in the process.
Later, we will see that there are nice connections between these methods and 
low-dimensional spaces and hypothesized manifolds.

Lest one think that this is the only way to compute partitions, we turn now 
to a \emph{very} different method to partition graphs---it is based on the
ideas of single-commodity and multi-commodity flows.
It is \emph{not} a spectral method, but it is important to know about for 
spectral methods (e.g., when/why spectral methods work and how to diagnose 
things when they don't work as one expects): the reason is basically that
flow-based graph algorithms ``succeed'' and ``fail'' in very different ways 
than spectral methods; and the reason for this is that they too implicitly 
involve embedding the input graph in a metric/geometric place, but one which 
is very different than the line/clique that spectral methods implicitly 
embed the input into.

\subsection{Some high-level comments on spectral versus flow}

Recall that the key idea underlying graph partitioning algorithms is take a 
graph $G=(V,E)$ and spread out the vertices $V$ in some abstract space while 
\emph{not} spreading out edges $E$ too much, and then to partition the 
vertices in that space into tow (or more) sets.
\begin{itemize}
\item
\textbf{Spectral methods} do this by putting nodes on an eigenvector and then 
partitioning based on a sweep cut over the partitions defined by that 
eigenvector. 
For spectral methods, several summary points are worth noting.
\begin{itemize}
\item
They achieve ``quadratic'' worst-case guarantees from Cheeger's Inequality.
This quadratic factor is ``real,'' e.g., it is not an artifact of the 
analysis and there are graphs on which it is achieved.
\item
They are ``good'' when graph has high conductance or expansion, \emph{i.e.}, 
when it is a good expander (in either the degree-weighted or 
degree-unweighted sense).
\item
They are associated with some sort of underlying geometry, e.g., as defined 
by where the nodes get mapped to in the leading eigenvector; but it is not 
really a metric embedding (or at least not a ``good'' metric embedding since 
it only preserves average distances---which is typical of these 
$\ell_2$-based methods---and not the distance between every pair of nodes.
\end{itemize}
\item
\textbf{Flow-based methods} do this by using multi-commodity flows to reveal 
bottlenecks in the graph and then partitioning based on those bottlenecks.
To contrast with spectral methods, note the following summary points for 
flow-based methods.
\begin{itemize}
\item
They achieve $O(\log n)$ worst-case guarantees.
This $O(\log n)$ is ``real,'' in that it too is not an artifact of the 
analysis and there are graphs on which it is achieved.
\item
They are ``good'' when the graph has sparse cuts, \emph{i.e.}, when it is 
not a good expander.
\item
They are also associated with a geometry, but one which is very different 
than that of spectral methods.
In particular, although not immediately-obvious, they can be viewed as 
embedding a finite metric space of graph into $\ell_1$ metric, \emph{i.e.}, 
a very different metric than before, and then partitioning there.
\end{itemize}
\end{itemize}

\noindent
The point here is that these two methods are in many ways complementary, in 
the sense that they succeed and fail in different places.
Relatedly, while ``real'' data might not be exactly one of the idealized 
places where spectral or flow-based methods succeed or fail, a lot of 
graph-based data have some sort of low-dimensional structure, and a lot of 
graph-based data are sufficiently noisy that it is fruitful to view them as 
having expander-like properties. 

Finally, the comments about spectral methods being good on expanders and 
flow methods being good on non-expanders might seem strange.
After all, most people who use spectral methods are not particularly 
interested in partitioning graphs that do not have good partitions (and 
instead they take advantage of results that show that spectral methods find
good partitions in graphs that are ``morally'' low-dimensional, e.g., graphs 
like bounded-degree planar graphs and well-shaped meshes).
The reason for this comment is that the metric we have been using to 
evaluate cluster quality is the objective function and not the quality of the
clusters. 
By this measure, the quadratic of a constant (which is the expansion value of
an expander) is a constant, while the quadratic of, e.g., $1/\sqrt{n}$ is 
very large, i.e., the spectral guarantee is nontrivial in one case and not
in the other.
For flow-based methods by contrast, the $O(\log n)$ is much larger than the 
constant value of an expander and so gives trivial results, while this is 
much smaller than, e.g., $1/\sqrt{n}$, leading to nontrivial results on the
objective for morally low-dimensional graphs.
That these qualitative guides are the opposite of what is commonly done in 
practice is a real challenge for theory, and it is a topic we will return to 
later.

\subsection{Flow-based graph partitioning}

Recall the basic ideas about flow and the single commodity flow problem:
\begin{itemize}
\item
$G=(V,E)$ is a graph, and each edge $e \in E$ has capacity $C(e)$, which is
the maximum amount of flow allowed through it.
Also, we are given a source $s$ and a sink $t$.
\item
(We are going to be applying it to graph partitioning by trying to use
network flow ideas to reveal bottlenecks in the graph, and then cut there.)
\item
So, the goal is to route as much flow $s \rightarrow t$ without violating
any of the constraints.
\item
Max-flow is the maximum amount of such flow.
\item
Min-Cut $=$ the minimum amount of capacity to be removed from a network to 
disconnect $s$ from $t$;
that is,
\[
Min-Cut =\min_{U \subset V: s \in U,t \in \bar{U}} \sum_{e\in(U,\bar{U})}c(e)  .
\]
(Note that there is no ratio here, since we can assume that demand $=1$, 
WLOG, but that we won't be able to do this later.)
\item
Weak duality is one thing and is relatively easy to show:
\begin{claim}
$maxflow \le mincut$
\end{claim}
\begin{proof}
For all $ U \subseteq V$ that has $s$ and $t$ on opposite sides, all flows 
from $s \rightarrow t$ must be routed through edges in $(U,\bar{U})$, so the 
total flow is bounded by the capacity in mincut.
\end{proof}
\item
Strong duality is another thing and is harder:
\begin{claim}
$maxflow = mincut$
\end{claim}
\begin{proof}
Suffices to show that mincut bound is always achievable is harder.
We won't do it here.
\end{proof}
\end{itemize}

\noindent
All of this discussion so far is for $k=1$ single-commodity flow. 

There are also multi-commodity versions of this flow/cut problem that have 
been widely studied. 
Here is the basic definition.

\begin{definition}
Given $k\ge1$, each with a source $s_i$ and sink $t_i$ and also given 
demands $D_i$ for each, i.e., we have $(s_i,t_i,D_i)$ for each $i\in[k]$, 
the \emph{multi-commodity flow problem} is to simultaneously route $D_i$ 
units of flow from $s_i$ 
to $t_i$, for all $i$, while respecting capacity constraints, \emph{i.e.}, 
s.t.  the total amount of all commodities passing through any edge $\le$ 
capacity of that edge.
There are several variants, including:
\begin{itemize}
\item 
\emph{Max throughput flow}: maximize the amount of flow summed over all 
commodities.
\item 
\emph{Max concurrent flow}: specify a demand $D_i$ for each commodity, and
then maximize the \emph{fraction} of demand $D_i$ that can be 
simultaneously shipped or routed by flow:
\[
   \max f \text{~s.t.~}fD_i \text{~units of flow go from $s_i$ to $t_i$.}  
\]
\emph{i.e.}, the maximum $f$ s.t. $fD_i$ units of capacity $i$ are 
simultaneously routed without violating the capacity constraints.
\end{itemize}
\end{definition}

That is, for this multicommodity flow problem, if the flow of commodity $i$ 
along edge $(u,v)$ is $f_i(u,v)$, then an \emph{assignment of flow} 
satisfies:
\begin{itemize}
\item
Capacity constraints.
\[
\sum_{i=1}^{k}|f_i(u,v)| \le c(u,v)
\]
\item
Flow conservation
\begin{eqnarray*}
 & & \sum_{w \in V} f_i(u,w) = 0 \quad \mbox{when } u \ne s_i,t_i \\
 & & \mbox{For all } u,v \quad f_i(u,v) = -f_i(v,u)
\end{eqnarray*}
\item 
Demand satisfaction.
\[
\sum_{w \in V} f_i(s_i,w) = \sum_{w \in V} f_i (w,t) = D_i
\]
\end{itemize}
Then the goal is to find a flow that maximizes one of the above variants of 
the problem.

We can also define a related cut problem.
\begin{definition}
The \emph{MinCut} or \emph{Sparsest Cut} $\Xi$---of an undirected 
multicommodity flow problem---is the mincut over all cuts of the ratio of 
the capacity of cut to the demand of the cut, 
\emph{i.e.}, 
\[
   \Xi=\text{min cut} = \rho = \Xi = \min_{U\subseteq V}\frac{C(U,\bar U)}{D(U,\bar U)}   ,
\]
where
\begin{align*}
    C(U,\bar U)&=\sum_{e\in (U,\bar U)}C(e)\\
    D(U,\bar U)&=\sum_{\substack{i:s_i\in U\\t_i\in \bar U \text{or vice versa}}}D_i
\end{align*}
That is, $C(U,\bar{U})$ is the sum of capacities linking $U$ to $\bar{U}$;
and $D(U,\bar{U})$ is the sum of demands with source and sinks on opposite
sides of the $(U,\bar{U})$ cut.
\end{definition}

Finally, we point out the following two variants (since they will map to the 
expansion and conductance objectives when we consider the application of 
multi-commodity flow to graph partitioning).
There are, of course, other variants of flow problems that we don't consider 
that are of interest if one is primarily interested in flow problems.

\begin{definition}
The \emph{Uniform Multi-Commodity Flow Problem (UMFP)}: here, there is a 
demand for every pair of nodes, and the demand for every commodity is the 
same, \emph{i.e.}, they are uniform, WLOG we take to equal $1$. 
The \emph{Product Multicommodity Flow Problem (PMFP)}: here, there is a 
nonnegative weight $\pi(\cdot)$, for all nodes $u \in V$, \emph{i.e.}, 
$\pi:V\to R^+$.  
The weights of demands between a pair of nodes $u$ and $v$ are 
$\pi(v_i)\pi(v_j)$. 
\end{definition}
Here are several comments about the last definition.
\begin{itemize}
\item 
UMFP is a special case of PMFP with $\pi(i)=1$.
\item
If $\pi(i)=1$, then we will get a way to approximate the expansion objective.
\item
If $\pi(v)=deg(v)$, then we will get a way to approximate the conductance
objective.
\end{itemize}

The MaxFlow-MinCut for single commodity is nice since it relates two 
fundamental graph theoretic entities (that are in some sense dual) via the 
min-max theorem.
But, prior to LR, very little was known about the relationship between 
MaxFlow and MinCut in this more general multi-commodity flow setting for 
general graphs.
For certain special graphs, it was known that there was an equality, i.e., a 
zero duality gap; but for general graphs, it is only known that MaxFlow is 
within a fraction $k$ of the MinCut.
This results can be obtained since each commodity can be optimized 
separately in the obvious trivial way by using $\frac{1}{k}$ of the capacity 
of the edges---this might be ok for $k=\Theta(1)$, but it is is clearly bad 
if $k \sim n$ or $k \sim n^2$.
Somewhat more technically, if we consider the LP formulation of the Max 
Multicommodity Flow Problem, then we can make the following observations.
\begin{itemize}
\item
The dual of this is the LP relaxation of the Min Multicut Problem, 
\emph{i.e.}, the optimal integral solution to the dual is the Mun Multicut.
\item
In general, the vertices of the \emph{dual polytope} are NOT integral.
\item
But, for single commodity flow, they are integral, and so MaxFlow-MinCut
theorem is a consequence of LP duality.
\item
For certain special graphs, it can be shown that they are integral, in which 
case one has zero duality gap for those graphs.
\item
Thus, for multicommodity: MaxFlow = MinFractional, \emph{i.e.}, relaxed, 
Multicut.
\end{itemize}

Here are several facts that we spend some time discussion:
\begin{paragraph}{Fact 1}
If we have $k$ commodities, then one can show that max-flow/min-cut gap 
$\le O(\log k)$.
This can be shown directly, or it can be shown more generally via metric 
embedding methods.
\end{paragraph}
\begin{paragraph}{Fact 2}
If certain conditions are satisfied, then the duality gap = 0.  
If one look at dual polytope, then whether or not this is the case depend on
whether the optimal solution is integral or not.
This, in turn, depends on special structure of the input.
\end{paragraph}
\begin{paragraph}{Fact 3}
For $k$ commodities, 
LR showed that the worst case (over input graph) gap $\Omega(\log k)$.
LLR interpreted this geometrically in terms of embeddings.
The worst case is achieved, and it is achieved on expanders.  
\end{paragraph}

Here, we will spend some time showing these results directly.
Then, next time we will describe it more generally from the metric embedding
perspective, since that will highlight better the similarities and 
differences with spectral methods.

\subsection{Duality gap properties for flow-based methods}

Here, we will show that there is a non-zero duality gap of size 
$\Theta(\log n)$.
We will do it in two steps: 
first, by illustrating a particular graph (any expander) that has a gap at 
least $\Omega(\log n)$; and
second, 
by showing that the gap is $O(\log n)$, i.e., is never worse than that.

Let's start by showing that there is a graph for which the duality gap 
is nonzero by showing a graph for which it is at least $\Theta(\log n)$. 

\begin{theorem}[LR]
$\forall n, \exists$ an $n$-node UMFP with MaxFlow $f$ and MinCut $\Xi$ 
s.t. 
\[
f \le O\left( \frac{\Xi}{\log n} \right)
\]
\end{theorem}
\begin{proof}
Consider any graph with certain expansion properties.
In particular, let $G$ be a $3$-regular $n$-node graph with unit edge 
capacities s.t.  
\[
C(U,\bar{U})    = |(U,\bar{U})|
             \geq C_{onst} \min\{|U|,|\bar{U}|\} ,\quad \forall U,V
\]
\emph{i.e.}, an expander.
(Such graphs exist by Margoulis, zig-zag, etc. constructions; and moreover a
randomly-selected $3$-regular graph satisfies this w.h.p.)
The first claim is that:
\begin{eqnarray*}
\Xi    &=& \min_{U \subseteq V} \frac{(U,\bar{U})|}{|U|\cdot|\bar{U}|}  \\
    &\geq& \min_{U \subseteq V} \frac{C_{onst}}{\max\{|U|,|\bar{U}|\}}  \\
       &=& \frac{C_{onst}}{n-1}  .
\end{eqnarray*}
The second claim is that:
\begin{claim}
The MaxFlow for UMFP is $\leq \frac{6}{(n-1)(\log n -1)}$, which is 
$\Theta(\log n)$ smaller than MinCut.
\end{claim}
\begin{proof}[Proof of claim]
Since the graph is $3$-regular, $\exists \leq \frac{n}{2}$ nodes within 
a distance $\log n - 3$ of $v \in V$.
So, for at least half of the $n \choose 2$ commodities, the shortest 
path connecting $s_i$ and $t_i$ has at least $\log n -2 $ edges.
To sustain a flow of $f$ for such a commodity, at least $f(\log n -2)$ 
capacity must be used by commodity.
To sustain a flow $f$, $\forall {n \choose 2}$ commodities, the capacity
of the network must be $\geq \frac{1}{2}{n \choose 2}f(\log n-2)$.
\end{proof}
Since the graph is $3$-regular with unit capacity, the total capacity is
$\leq \frac{3n}{2}$.
So, 
\[
\frac{1}{2}{n \choose 2}f(\log n -2 ) \leq CAPACITY \leq \frac{3n}{2}  .
\]
So, 
\begin{eqnarray*}
f &\leq& \frac{3n}{{n \choose 2}(\log n-2)}  \\
     &=& \frac{6}{(n-1)(\log n -2)}          \\
  &\leq& \frac{6\Xi}{C_{onst}(\log n -2)} \quad\mbox{since } \Xi \geq \frac{C_{onst}}{n-1} \\
     &=& O\left(\frac{\Xi}{\log n} \right)  .
\end{eqnarray*}
That is, the MaxFlow for the UMFP on an expander $G$ is $\ge \Theta(\log n)$ 
factor smaller than the MinCut.
\end{proof}


An expander has diameter $O(\log n)$, and so for expanders, the gap can't be
worse than $\Theta(\log n)$, but the following shows that this is true more 
generally.

\begin{theorem}[LR]
The MinCut of UMFP can never be more than $\Theta(\log n)$ factor larger
than MaxFlow, \emph{i.e.},
\[
\Omega\left(\frac{\Xi}{\log n}\right) \leq f \leq \Xi  .
\]
\end{theorem}
\begin{proof}
We can do this by showing a polynomial-time algorithm that finds a cut
$(U,\bar{U})$ s.t. 
\[
\frac{C(U,\bar{U})}{|U|\cdot|\bar{U}|} \leq O(f(\log n))
\]
where the LHS is the \emph{ratio cost} of $(U,\bar{U})$;
and recall that 
$\Xi = \min_{S \subseteq V} \frac{C(U,\bar{U})}{|U|\cdot|\bar{U}|}$.
The algorithm is based on LP Dual---the dual of Multicommodity Flow is the 
problem of assigning a fixed weight (where the weights are thought of as 
distances) to edges of $g$ s.t. one maximizes the cumulative distance 
between source and sink pairs.
(We won't actually go through this now
since there is a more elegant and enlightening formulation.)
\end{proof}

In the UMFP, the demand across cut $(U,\bar{U})$ is given by: 
$D(U,\bar U)=|U||\bar U|$.
So, in particular, the mincut is given by:
\begin{align*}
\text{min cut:}\quad\Xi &= \min_{U\subseteq V}\frac{C(U,\bar U)}{|U||\bar U|}\\
    &=\min_{U\subseteq V}\frac{E(U,\bar U)}{|U||\bar U|} \quad\text{if all capacities $=1$}  ,
\end{align*}
where $C(U,\bar{U})= \sum_{e\in(U,\bar{U})} C(e)$.

This is, if the demands are uniform and all the capacities are equal to one, 
then from UMFP we get our old friend, the sparsest cut $\sim$ best expansion.
Among other things, this implies that the $O(\log n)$ approximation for 
general MultiCommodity Flow in ``inherited'' by the algorithm for the 
sparsest cuts problem, and many other related problems as well.
In particular, one way to use flow is the following: 
check all $2^n$ cuts, and use the single-commodity zero-duality gap result 
to show that we can take the one with the best single cwommodity flow to 
get the one with the best mincut.
What these results say is that, instead, we can consider the all-pairs
multi-commodity flow problem and check a lot less and get a result that is 
only a little worse.

\subsection{Algorithmic Applications}
\label{sec:algor-appl}

Let's go back to our old friend the \emph{sparsest cut problem}, and here we 
will make explicit connections with flow based graph partitioning by 
viewing it from an optimization perspective.
This will in turn provide us with an algorithm (that is mentioned in the 
proof of the above theorem) to solve the problem.

Recall that in the sparsest cut problem, we are given: a graph $G = (V,E)$; 
a cost function $c(e),\forall e \in E$, \emph{i.e.}, 
$c: E \rightarrow \mathbb{R}^{+}$; and $k$ pairs of nodes/vertices 
$(s_i, t_i)$. 
We will write the problem as an Integer Program (IP).
To do so, 
\begin{itemize}
\item
Let $x(e)$, $e \in E$, where $x(e) \in \{0,1\}$ to indicate if an edge $e$ 
is cut;
\item
let $y(i)$, $i \in [k]$, where $y(i) \in \{0,1\}$ to indicate if commodity 
$i$ is cut, \emph{i.e.}, is disconnected by this cut; and 
\item
let $\mathcal{P}_i$, $i=1,2,\ldots,k$ be the set of paths between $s_i$ and 
$t_i$.
\end{itemize} 
Then, what we want is to solve:
\begin{align*}
\min          & \frac{\sum_{e \in E}{c(e)x(e)}}{\sum_{i=1}^k{d(i)y(i)}}  \\
\mbox{s.t.~~} & \sum_{e \in P}{x(e)} \geq y(i), \quad \forall P \in \mathcal{P}_i, \quad \forall i=1,2,\ldots,k  \\
              & y(i) \in \{0, 1\}, \quad i \in[k] \\
              & x(e) \in \{0, 1\}, \quad e \in E
\end{align*}


We want to consider the relaxation of this to the case where replace the
last two constraints by $y(i)\geq 0$ and $x(e)\geq 0$
In doing so, note that if $(x,y)$ is a feasible fractional solution, then
$(\alpha x,\alpha y)$ is also a feasible fractional solution with the same 
objective function value.
So, WLOG, we can choose the normalization $\alpha$ s.t.
$\sum_{i=1}^{k} d(i)y(i) = 1$ to get the following LP:
\begin{align*}
\min          & \sum_{e \in E} c(e)x(e)   \\
\mbox{s.t.~~} & \sum_{i=1}^{k} d(i)y(i) = 1 \\
              & \sum_{e \in P}{x(e)} \geq y(i), \quad \forall P \in \mathcal{P}_i, \quad \forall i=1,2,\ldots,k  \\
              & y(i) \geq 0 \quad ,  x(e) \geq 0 .
\end{align*}


Below, we will show that we can compute a cut with sparsity ratio within 
a factor of $O(\log k)$ of this optimal fractional solution.

BTW, before we do that, let's write the LP dual:
\begin{align*}
\max          & \quad \alpha \\
\mbox{s.t.~~} & \sum_{p\in\mathcal{P}_i} f(p) \geq \alpha d(i) ,\quad \forall i \in [k] \\ 
              & \sum_{i=1}^{k} \sum_{p\in\mathcal{P}_i(e)} f(p) \leq c(e) ,\quad \forall e \in E \\
              & f(p) \geq 0, \quad \forall P \in \mathcal{P}_i, i \in [k]  .
\end{align*}
This is the Max Concurrent Flow problem, and $O(\log k)$ approximation gives
an \emph{approximate MaxFlow-MinCut Theorem}.

So, to solve this sparsest cut problem, our strategy will be:
\begin{itemize}
\item 
Solve the LP (either the primal or the dual).
\item 
Round the solution to an integral value.
\end{itemize}
Thus, there are some similarities with spectral---we first solve something
to get a real-valued solution, and then we have to ``round'' to get an 
integral solution, and the ball game will be to show that we don't loose 
too much. 
Of course, solving the LP will \emph{implicitly} embed us in a very different
place than solving an eigenvector problem, so we will expect to see different
artifactual properties between the two approximation algorithms.

The above discussion gives algorithms that run in polynomial time:
\begin{itemize}
\item
Off the shelf LP (due to a connection with LP).
\item
Algorithms for approximately optimal flows of distance functions 
(\emph{i.e.}, take advantage of the structure of the LP).
\item
Fastest update of sparsest cut algorithm is $\tilde{O}(n^2)$, with 
Benczur-Karger sparsification.
\item
Standard algorithm that runs in something like $\tilde{O}(n^{3/2})$ with 
push-relabel methods.
\item
Finally, there is a lot of work on using Laplacian-based solvers to do better, 
and we may return to these later.
\end{itemize}

Important Note: The way you would actually solve this in practice is 
to use some sort of push-relabel code, which is relatively fast, as opposed 
to the general LP procedure just described, which is easier to analyze 
theoretically.

\subsection{Flow Improve}

Here is an aside that with luck we will get back to later.
These algorithms have a running time that large enough that it can be 
challenging to apply to very large graphs---e.g., $O(n^{3/2})$ or especially 
$O(n^2)$ is certainly too large for ``massive'' graphs.
(Implementations of faster algorithms are still very much a topic of research.)
A question arises, can we do something like ``local spectral'' (which, 
recall, consisted roughly of a few steps of a random walk), to do a local 
flow improvement?

The answer is YES---and here is it.
The so-called \texttt{Improve} algorithm of AL as well as the related MQI 
method:
this algorithm is useful by itself for improving partitions from, say, Metis 
or some other very fast procedures; and 
it is useful as a way to speed up spectral and/or as one way to combine 
spectral-based algorithms with flow-based algorithms.
In more detail, the goal of a local improvement algorithm is: 
Given a partition, find a strictly better partition.
A local improvement algorithm is useful in the following contexts:
    \begin{itemize}
    \item METIS -- post process with a flow based improvement heuristic.
    \item Vanilla spectral: post process with improvement method.
    \item Local improvement at one step online iterative algorithm.
    \end{itemize}
MQI and Improve essentially construct and solve a sequence of $s$-$t$ MinCut 
problems on a modified quotient cut objective to add and remove vertices 
from a proposed cut.
(We won't describe them in detail now.)
Here is the basic theorem stating their running time and approximation 
quality bounds.

\begin{theorem}
Let $A\subseteq V$ s.t. $\pi(A)\le \pi(\bar A)$, and let 
$S=\textsc{Improve}(A)$ be the output of the Flow Improve Algorithm.
Then
\begin{enumerate}
\item 
If $C\subseteq A$, (\emph{i.e.}, $\forall C \subseteq A$) then 
$Q(S)\le Q(C)$ (where $Q(S) = \frac{|\partial S|}{\mbox{Vol}(S)}$)
\item 
If $C$ is such that 
\begin{align*}
   \frac{\pi(A\cap C)}{\pi(C)}
      \ge \frac{\pi(A)}{\pi(V)}+\epsilon\frac{\pi(\bar A)}{\pi(V)}  ,
      \quad\mbox{for some $\epsilon$}
\end{align*}
\emph{i.e.}, if $C$ is $\epsilon$-more-correlated with $A$ than random,
\emph{i.e.}, if the fraction of $C$ that is in $A$ is $\epsilon$-better
than random,
then $Q(S)\le \frac{1}{\epsilon}Q(C)$ \emph{i.e.}, bound on nearby cuts.
\item
The algorithm takes time:
(1) $\pi(V)^2$ iterations if vertex weights $\pi(V) \in \mathcal{Z}$; and
(2) $m$ iterations if the edges are unweighted.
\end{enumerate}
\end{theorem}

%% file: lect11.tex
\section{%
(02/26/2015): 
Flow-based Methods for Partitioning Graphs (2 of 2)}

Reading for today.
\begin{compactitem}
\item
Same as last class.
\end{compactitem}

Recall from last time that we are looking at flow-based graph partitioning
algorithm.
Last time, we covered the basics of flow-based methods, and we showed how 
they are very different than spectral methods.
This time, we will discuss flow-based graph partitioning from an 
embedding perspective.
We will see that flow-based algorithms implicitly embed the data in a metric 
space, but one that is very different than the place where spectral-based 
algorithms embed the data.  
Thus, not only do they run different steps operationally and get 
incomparable quality of approximation bounds, but also they implicitly put 
the data in a very different place---thus ``explaining'' many of the 
empirical results that are empirically observed.

(BTW, I have made some comments that spectral methods can scale up to 
much larger input graphs by using diffusion and random walk ideas, a topic 
we will get back to later.  For the moment, let me just note that the way 
flow is used is not \emph{immediately} relevant for such ``massive'' data.  
For example, the running time of a typical flow-based algorithm will be 
$O(n^2)$, since it involves a multi-commodity variant of flow; single 
commodity variants of flow-based algorithms run in $O(n^{3/2})$ time; and 
more recent work has focused on using Laplacian solver ideas to do even 
better, i.e., to run in time that is nearly-linear in the size of the input.  
There is a lot of interest, mostly from within TCS so far, in this area; and
these fast solvers hold the promise to make these methods applicable to much 
larger graphs.  
I'm hoping to return to some of these Laplacian solver topics at the 
end of the semester, and I'm also planning on giving at least give a brief 
introduction to some of the ideas about how to couple spectral methods later.)

\subsection{Review some things about $\ell_1$ and $\ell_2$}

Let's review a few things about $\ell_1$ and $\ell_2$ and related topics.

\begin{definition}
The \emph{$\ell_p$-norm} on $\mathbb{R}^{k}$ is 
$||x||_p = \left( \sum_{i=1}^{k} |x_i| \right)^{1/p}$.
A finite metric space $(X,\rho)$ is \emph{realized} in $\ell_p^k$ if 
$\exists f:X \rightarrow \mathbb{R}^{k}$ s.t. $\rho(x,y) = ||f(x)-f(y)||_2$,
and it is an \emph{$\ell_p$-metric} if it can be realized  by $\ell_p^k$ for some $k$.
The \emph{metric induced by $\ell_p$} is: 
$d(x,y) = ||x-y||_p ,\quad \forall x,y \in \ell_p$.
\end{definition}
For flow-based methods, we will be most interested in $\ell_1$, while 
for spectral-based methods, we are interested in $\ell_2$.
The $\ell_p$ norm (except for $p=\infty$, which we won't discuss here) is 
usually of at most more theoretical interest.

The spaces $\ell_1$ and $\ell_2$ are very different.
(Think $\ell_1$ regression, i.e., least absolute deviations, versus $\ell_2$
regression, i.e., least-squares regression; or think $\ell_1$ regularized
$\ell_2$ regression, i.e., the lasso, versus $\ell_2$ regularized $\ell_2$
regression, i.e., ridge regression.
These differences are often ``explained'' in terms of differences between the
unit ball in the $\ell_1$ norm versus the unit ball in the $\ell_2$ norm, 
with the former being ``pointy'' and the latter being ``smooth.''
In particular, note that in those situations $\ell_1$ often has some sort of 
connection with sparsity and sparse solutions.)

Here is a comparison between $\ell_1$ and $\ell_2$ with respect to the 
spectral/flow algorithms for the graph partitioning problem we have been 
considering.
\begin{itemize}
\item 
$\ell_2$ norm:
\begin{itemize}
\item
Good for dimension reduction.
\item
Efficient polynomial time algorithm to compute embedding of any finite 
metric space.
\item
Connections to low-dimensional manifolds, diffusions, etc.
\end{itemize}
\item 
$\ell_1$ norm:
\begin{itemize}
\item
No good for dimension reduction.
\item
NP-hard to compute the optimal embedding.
\item
Connections to partitions/cuts, multicommodity flow, etc.
\end{itemize}
\end{itemize}

The following is a fundamental result in the area that is also central to
understanding why flow-based graph partitioning algorithms work.

\begin{theorem}[Bourgain]
Every $n$-point metric space $d$ admits an $\alpha$-distortion 
embedding into $\ell_p$, $\forall p$, with $\alpha = O(\log n)$.
\end{theorem}
\begin{proof}[Proof idea]
The proof is similar to but more general than the proof for the 
corresponding embedding claim for $\ell_2$.
The idea is: to use the so-called Frechet embedding method, where the 
embedding is given by the distance from points to carefully randomly chosen 
subsets of nodes.
\end{proof}

\noindent
Note that we saw the $\ell_2$ version of this before.
Note also that the original embedding had $2^n$ dimensions, but LLR proved 
that it can be done with $O(\log^2 n)$ dimensions.

\subsection{Connection between $\ell_1$ metrics and cut metrics}


First, recall what a metric is.
\begin{definition}
A \emph{metric} is a function $d:V \times V \rightarrow \mathbb{R}$ s.t.:
(1) $d(x,y)=0 \mbox{ if } x=0$ (and sometimes $=0 \mbox{ iff } x=y$);
(2) $d(x,y) = d(y,x)$; and
(3) $d(x,y) + d(x,z) \leq d(x,y)$
\end{definition}
(Recall also that sometimes the word ``metric'' if one or more of those 
conditions is/are relaxed.)

Next, recall the definition of the ``cut metric,'' and recall that this is 
really not a metric but is instead just a semi-metric.
\begin{definition}
Given $G=(V,E)$ and a set $S \subset V$, $\delta_S$ is the ``cut metric'' 
for $S$ if
\[
\delta_S(i,j) = | \chi_S(i) - \chi_S(j) | ,
\]
where 
\[
\chi_S(i) = \left\{ \begin{array}{ll}
                           0 \text{ if } i \in S   \\
                           1 \text{ otherwise}       .
                    \end{array}
                \right.
\]
Thus
\[
\delta_S(i,j) = \left\{ \begin{array}{ll}
                           0 \text{ if } i,j \in S,\mbox{ or } i,j\in\bar{S} \\
                           1 \text{ otherwise}     .
                   \end{array}
                \right.
\]
(That is, if $\delta_s(x,y)$ is the indicator of $x$ and $y$ being on 
different sides of $S$.)
\end{definition}

There are important connection between $\ell_1$ metrics and Cut metrics.
In particular:
\begin{itemize}
\item
There exists a representation of the $\ell_1$ metrics as a conical 
combination of cut metrics.
\item
Cut metrics are the extreme rays of the $\ell_1$ cone.
\item
For these reasons, instead of minimizing the ratio of linear functions over 
the convex cone, we can minimize the ratio over the extreme rays of the 
convex cone.
Minimum ratio function over cone $\iff $ minimum over extreme rays. 
\end{itemize}

We'll spend most of today going over these results.

\textbf{Fact:}
An $n$-point metric space can be associated with a vector in 
$\mathbb{R}^{n \choose 2}$, with each coordinate corresponding to a pair
of vertices.

\textbf{Fact:}
Given a metric $d$, we will refer to the corresponding vector as $\bar{d}$.
Then, $\alpha\bar{d}+(1-\alpha)\bar{d}^{'}$ is a metric, 
$\forall\alpha\in[0,1]$.
In addition, $\forall\alpha\geq0$, $\forall\bar{d}$, $\alpha\bar{d}$ is a 
metric.
So, the set of all metrics forms a convex cone in $\mathbb{R}^{n \choose 2}$.
In somewhat more detail, we have the following result:
\begin{claim}
The set of $\ell_1$ metrics is a convex cone, \emph{i.e.}, if $d_1$ and 
$d_2$ $\in \ell_1 \mbox{ metrics}$, and if 
$\lambda_1,\lambda_2 \in \mathbb{R}^{+}$, then 
$ \lambda_1 d_1 + \lambda_2 d_2 \in \ell_1\quad\mbox{metrics}$.
\end{claim}
\begin{proof}
Recall that the line metric is an $\ell_1$ metric.
Let $d^{(i)}(x,y) = |x_i - y_i|$, for $x,y \in \mathbb{R}^{c}$.
If $d \in \ell_1$ metric, then it is the sum of line metrics.
\end{proof}

\noindent
\textbf{Fact.}
The analogous result does \emph{not} hold for $\ell_2$.

Next, we have the following theorem:

\begin{theorem}
Let $d$ be a finite $\ell_1$ metric.
Then, $d$ can be written as
\[
d = \sum_{S \subset [n]} \alpha_S \delta_S   ,
\]
for some constant $\alpha_S$ and cut metrics $\delta_S$.
That is,
\begin{eqnarray*}
{CUT}_{n} 
   &=& \{ d : d = \sum_{S \subseteq V} \alpha_S \delta_S , \alpha \geq 0\}  \\
   &=& \mbox{Positive cone generated by cut metrics} \\
   &=& \mbox{All $n$-point subsets of $\mathbb{R}^{n}$, under the $\ell_1$ norm.}
\end{eqnarray*}
\end{theorem}
\begin{proof}
Consider any metric $d \in {CUT}_{n}$.
Then, $\forall S$ with $\alpha_S > 0$, we have a dimension, and in that 
dimension, we can put
\[
  \begin{cases}
    \alpha_S & \mbox{if } x \in \bar{S} \\
    0        & \mbox{if } x \in S .
  \end{cases}
\]
So, ${CUT}_{n} \subseteq \ell_1\quad\mbox{Metrics}$.

For the other direction, consider a set of $n$ points from $\mathbb{R}^{n}$.
Take one dimension $d$ and sort the points in increasing values along that 
dimension.
Say that we get $v_1,\ldots,v_k$ as the set of distinct values;
then define $k-1$ cut metrics: $S_i = \{ x : x_d < v_{i+1} \}$, and 
let $\alpha_i = v_{i+1} - v_{i}$, \emph{i.e.}, $k-1$ coeff.
So, along this dimension, we have that 
\[
|x_d-y_d|=\sum_{i=1}^{k}\alpha_{i}\delta_{S_i}  ,
\]
But, one can construct cut metrics for every dimension.
So, we have cut metrics in ${CUT}_{n}$, $\forall$ $n$-point metrics $\ell_1$;
thus, $\ell_1 \subseteq CUT$.
\end{proof}

\subsection{Relating this to a graph partitioning objective}

Why is this result above useful?
The usefulness of this characterization is that we are going to want to 
to optimize functions, and rather than optimize functions over all cut 
metrics, \emph{i.e.}, over the extreme rays, we will optimize over the full convex cone, \emph{i.e.}, over $\ell_1$ metrics.

This leads us to the following lemma:
\begin{lemma}
Let $C \subset \mathbb{R}^{n}$ be a convex cone, and let 
$f,g:\mathbb{R}^{n,+}\rightarrow\mathbb{R}^{+}$ be linear functions.
And assume that $\min_{x \in C} \frac{f(x)}{g(x)}$ exists.
Then
\[
\min_{x \in C} \frac{f(x)}{g(x)}
   = \min_{x\mbox{ in extreme rays of } C} \frac{f(x)}{g(x)}  .
\]
\end{lemma}
\begin{proof}
Let $x_0$ be the optimum.
Since $x_0 \in C$, we have that $x_0 = \sum_i a_i y_i$, where 
$a_i \in \mathbb{R}^{+}$ and $y_i \in\quad\mbox{extreme rays of}\quad C$.
Thus,
\begin{eqnarray*}
\frac{f(x_0)}{g(x_0)} = \frac{f(\sum_i a_i y_i)}{g(\sum_i a_i y_i)} 
   &=& \frac{  \sum_i f(a_i y_i) }{ \sum_i g(a_i y_i)  }  \\
   &\geq& \frac{ f(a_j y_j) }{ g(a_j y_j) }   \qquad \mbox{where $j$ is the min value} \\
   &=& \frac{ f(y_i) }{ g(y_i) }  ,
\end{eqnarray*}
where the first and third line follow by linearity, 
and where the second line follows since 
\[
\frac{ \sum_i \alpha_i }{ \sum_i \beta_i } \geq \min_j \frac{ \alpha_j }{ \beta_j }
\] 
in general.
\end{proof}

To see the connections of all of this to sparsest cut problem, recall that 
given a graph $G=(V,E)$ we define the conductance $h_G$ and sparsity 
$\phi_G$ as follows: 
\begin{align*}
   h_G & := \min_{S \subseteq V} \frac{E(S, \bar{S})}{\min \{ |S|, |\bar{S}| \}} \\
\phi_G & := \min_{S \subseteq V} \frac{E(S, \bar{S})}{\frac{1}{n} |S| |\bar{S}|}   ,
\end{align*} 
and also that:
\[
h_G \leq \phi_G \leq 2 h_G .
\]
(This normalization might be different than what we had a few classes ago.)

Given this, we can write sparsest cut as the following optimization problem:
\begin{lemma}
Solving
\[
\phi_G = \min_{S \subseteq V} \frac{E(S, \bar{S})}{\frac{1}{n} |S| |\bar{S}|}
\]
is equivalent to solving:
\begin{align*}
\min          & \sum_{(ij)\in E} d_{ij}  \\
\mbox{s.t.~~} & \sum_{ij \in V} d_{ij} = 1  \\
              & d \in \ell_1\mbox{metric}
\end{align*}
\end{lemma}
\begin{proof}
Let $\delta_S = $ the cut metric for $S$.
Then, 
\[
 \frac{ |E(S,\bar{S})| }{ |S|\cdot|\bar{S}| }
   = \frac{ \sum_{ij \in E} \delta_S(i,j) }{ \sum_{\forall ij} \delta_S(i,j) }
\]
So, 
\[
\phi_G = \min_S \frac{ \sum_{ij \in E} \delta_S(i,j) }{ \sum_{\forall ij} \delta_S(i,j) }
\]
Since $\ell_1$-metrics are linear combinations of cut metrics, and cut 
metrics are extreme rays of $\ell_1$ from the above lemma, this ratio is 
minimized at one of the extreme rays of the cone.
So, 
\[
\phi_G = \min_{d\in\ell_1} \frac{ \sum_{ij \in E} d_S(ij) }{ \sum_{\forall ij} d_S(ij) }   .
\]
Since this is invariant to scaling, WLOG we can assume 
$\sum_{\forall ij} d_{ij} = 1$; and we get the lemma.
\end{proof}

\subsection{Turning this into an algorithm}

It is important to note that the above formulation is still intractable---we 
have just changed the notation/characterization.
But, the new notation/characterization suggests that we might be able to 
\emph{relax} (optimize the same objective function over a larger set) the 
optimization problem---as we did with spectral, if you recall.


So, the relaxation we will consider is the following:
relax the requirement that $d \in \ell_1 \mbox{ Metric}$ to 
$d \in \mbox{ Any Metric}$.
We can do this by adding $3{n \choose 3}$ triangle inequalities to get the 
following LP:
\begin{align*}
\lambda^{*} = \min & \sum_{ij \in E} d_{ij}              \\
\mbox{s.t.~~}      & \sum_{\forall ij \in V} d_{ij} = 1  \\
                   & d_{ij} \geq 0                       \\
                   & d_{ij} = d_{ji}                     \\
                   & d_{ij} \leq d_{ik} + d_{jk} \quad\forall i,j,k \quad\mbox{triples}
\end{align*}
(Obviously, since there are a lot of constraints, a naive solution won't be 
good for big data, but we will see that we can be a bit smarter.)
Clearly, 
\[
\lambda^{*} \leq \phi^{*}  = \mbox{ Solution with $d \in \ell_1\quad\mbox{Metric constraint}$}
\]
(basically since we are minimizing over a larger set).
So, our goal is to show that we don't loose too much, \emph{i.e.}, that:
\[
 \phi^{*} \leq O(\log n) \lambda^{*}  .
\] 



Here is the \textsc{Algorithm}.
Given as input a graph $G$, do the following:
\begin{itemize} 
\item 
Solve the above LP to get a metric/distance 
$d:V \times V \rightarrow \mathbb{R}^{+}$.
\item 
Use the (constructive) Bourgain embedding result to embed $d$ into an 
$\ell_1$ metric (with, of course an associated $O(\log n)$ distortion).
\item 
Round the $\ell_1$ metric (the solution) to get a cut. 
\begin{itemize} 
	\item For each dimension/direction, covert the $\ell_1$ 
              embedding/metric along that to a cut metric representation. 
	\item Choose the best. 
\end{itemize} 
\end{itemize} 
Of course, this is what is going on under the hood---if you were actually 
going to do it on systems of any size you would use something more 
specialized, like specialized flow or push-relabel code.

Here are several things to note.
\begin{itemize}
\item
If we have $\ell_1$ embedding with distortion factor $\xi$ then can 
approximate the cut up to $\xi$. 
\item
Everything above is polynomial time, as we will show in the next theorem.
\item
In practice, we can solve this with specialized code to solve the dual of 
corresponding multicommodity flow.
\item
Recall that one can ``localize'' spectral by running random walks from 
a seed node.
Flow is hard to localize, but recall the \textsc{Improve} algorithm, 
but which is still $\tilde{O}(n^{3/2})$.
\item
We can combine spectral and flow, as we will discuss, in various ways.
\end{itemize}

\begin{theorem}
The algorithm above is a polynomial time algorithm to provide an 
$O(\log n)$ approximation to the sparsest cut problem.
\end{theorem}
\begin{proof}
First, note that solving the LP is a polynomial time computation to get a 
metric $d^{*}$.
Then, note that the Bourgain embedding lemma is constructive.
Finding an embedding of $d^{*}$ to $d \in \ell_{1}^{O(\log^2 n)}$ with 
distortion $O(\log n)$.
So, we can write $d$ as a linear combination of $O(n \log^2 n)$ cut metrics
$d = \sum_{S \in \mathcal{S}} \alpha_S\delta_S$, where 
$|\mathcal{S}| = O( n \log^2 n)$.
Note:
\begin{eqnarray*}
\min_{S\in\mathcal{S}} \frac{ \sum_{ij \in E} \delta_S(ij) }{ \sum_{\forall ij} \delta_S(ij) }   
   &\leq& \frac{ \sum_{ij \in E} d_{ij} }{ \sum_{\forall ij} d_{ij} }   \\
   &\leq& O(\log n)\frac{ \sum_{ij \in E} d^{*}_{ij} }{ \sum_{\forall ij} d^{*}_{ij} }   ,
\end{eqnarray*}
where the first inequality follows since $d$ is in the cone of cut metrics, 
and where the second inequality follows from Bourgain's theorem.
But, 
\[
\frac{ \sum_{ij \in E} d^{*}_{ij} }{ \sum_{\forall ij} d^{*}_{ij} }
   = \min_{d^{'} \mbox{ is metric}} \frac{ \sum_{ij \in E} d^{'}_{ij} }{ \sum_{\forall ij} d^{'}_{ij} }   \\
   \leq \min_{\forall S} \frac{ \sum_{ij \in E} d_{S}(ij) }{ \sum_{\forall ij} d_{S}(ij) }   ,
\]
where the equality follows from the LP solution and the inequality follows 
since LP is a relaxation of a cut metric.
Thus, 
\[
\min_{S\in\mathcal{S}} \frac{ \sum_{ij \in E} \delta_S(ij) }{ \sum_{\forall ij} \delta_S(ij) }   
\leq 
O(\log n) 
\min_{\forall S} \frac{ \sum_{ij \in E} d_{S}(ij) }{ \sum_{\forall ij} d_{S}(ij) }   .
\]
This establishes the theorem.
\end{proof}

So, we can also approximate the value of the objective---how do we actually
find a cut from this? 
(Note that sometimes in the theory of approximation algorithms you 
\emph{don't} get anything more than an approximation to the optimal number, 
but that is somewhat dissatisfying if you want you use the output of the 
approximation algorithm for some downstream data application.)

To see this:
\begin{itemize}
\item
Any $\ell_1$ metric can be written as a conic combination of cut 
metrics---in our case, with $O(n \log^n)$ 
nonzeros---$d^{\sigma} = \sum_S \alpha_S \delta_S$.
\item
So, pick the best cut from among the ones with nonzero $\alpha$ in the 
cut decomposition of $d^{\sigma}$.
\end{itemize}

\subsection{Summary of where we are}

Above we showed that 
\begin{eqnarray*}
\phi_G &=& \min_{S \subset V} \frac{E(S,\bar{S})}{|S||\bar{S}|} \\
       &=& \min  \sum_{ij \in E} d_{ij}   \\ 
       & & \mbox{s.t.~~}  \sum_{ij \in V} d_{ij}=1 \\
       & & d \in \ell_1\mbox{ metric}
\end{eqnarray*}
can be approximated by relaxing it to
\begin{align*}
\min          & \sum_{ij \in E} d_{ij}     \\
\mbox{s.t.~~} & \sum_{ij \in V} d_{ij} = 1 \\
              & d \in \mbox{Metric}
\end{align*}

This relaxation is different than 
the relaxation associated 
with spectral, where we showed that
\[
\phi = \min_{x \in \{0,1\}^V} \frac{ A_{ij} |x_i-x_j|^2 }{ \frac{1}{n}\sum_{ij} |x_i-x_j|^2 }
\]
can be relaxed to 
\[
d-\lambda_2 = \min_{x\perp\vec{1}} \frac{ A_{ij} (x_i-x_j)^2 }{ \frac{1}{n}\sum_{ij} (x_i-x_j)^2 }
\]
which can be solved with the second eigenvector of the Laplacian.

Note that these two relaxations are very different and incomparable, in the 
sense that one is not uniformly better than the other.
This is related to them succeeding and failing in different places, and it 
is related to them parametrizing problems differently, and it can be used
to diagnose the properties of how each class of algorithms performs on 
real data.
Later, we will show how to generalize this previous flow-based result and 
combine it with spectral.

Here are several questions that the above discussion raises.
\begin{itemize}
\item
What else can you relax to? 
\item
In particular, can we relax to something else and improve the $O(\log n)$
factor?
\item
Can we combine these two incomparable ideas to get better bounds in 
worst-case and/or in practice?
\item
Can we combine these ideas to develop algorithms that smooth or regularize
better in applications for different classes of graphs?
\item
Can we use these ideas to do better learning, e.g., semi-supervised 
learning on graphs?
\end{itemize}
We will address some of these questions later in the term, as there is a lot 
of interest in these and related questions.

%% file: lect12.tex
\section{%
(03/03/2015): 
Some Practical Considerations (1 of 4):
How spectral clustering is typically done in practice}

Reading for today.
\begin{compactitem}
\item
``A Tutorial on Spectral Clustering,'' in Statistics and Computing, by von Luxburg
\end{compactitem}

Today, we will shift gears.
So far, we have gone over the theory of graph partitioning, including 
spectral (and non-spectral) methods, focusing on \emph{why} and \emph{when}
they work.
Now, we will describe a little about \emph{how} and \emph{where} these 
methods are used.
In particular, for the next few classes, we will talk somewhat informally 
about some practical issues, e.g., how spectral clustering is done in 
practice, how people construct graphs to analyze their data, connections 
with linear and kernel dimensionality reduction methods.
Rather than aiming to be comprehensive, the goal will be to provide
illustrative examples (to place these results in a broader context and also 
to help people define the scope of their projects).
Then, after that, we will get back to some theoretical questions making 
precise the how and where.
In particular, we will then shift to talk about how diffusions and random walks 
provide a robust notion of an eigenvector and how they can be used to extend 
many of the vanilla spectral methods we have been discussing to very 
non-vanilla settings.
This will then lead to how we can use spectral graph methods for other 
related problems like manifold modeling, stochastic blockmodeling, Laplacian
solvers, etc.

Today, we will follow the von Luxburg review.
This review was written from a machine learning perspective, and in many 
ways it is a very good overview of spectral clustering methods; but beware: 
it also makes some claims (e.g., about the quality-of-approximation 
guarantees that can be proven about the output of spectral graph methods) 
that---given what we have covered so far---you should immediately see are 
not correct.

\subsection{Motivation and general approach}

The motivation here is two-fold.
\begin{itemize}
\item
Clustering is an extremely common method for what is often called 
\emph{exploratory data analysis}. 
For example, it is very common for a person, when confronted with a new 
data set, to try to get a first view of the data by identifying subsets of 
it that have similar behavior or properties.
\item
Spectral clustering methods in particular are a very popular class of 
clustering methods.
They are usually very simple to implement with standard linear algebra 
libraries; and they often outperform other methods such as $k$-means, 
hierarchical clustering, etc.
\end{itemize}

The first thing to note regarding general approaches is that Section 2 of 
the von Luxburg review starts by saying ``Given a set of data points 
$x_1,\ldots,x_n$ and some notion of similarity $s_{ij} \ge 0$ between all 
pairs of data points $x_i$ and $x_j$ ...''
That is, the data are vectors.
Thus, any graphs that might be constructed by algorithms are constructed from
primary data that are vectors and are useful as intermediate steps only.
This will have several obvious and non-obvious consequences.
This is a very common way to view the data (and thus spectral graph methods), 
especially in areas such as statistics, machine learning, and areas that are
not computer science algorithms.
That perspective is not good or bad per se, but it is worth emphasizing that 
difference.
In particular, the approach we will now discuss will be very different than what 
we have been discussing so far, which is more common in CS and TCS and 
where the data were a graph $G=(V,E)$, e.g., the single web graph out there, 
and thus in some sense a single data point.
Many of the differences between more algorithmic and more machine learning
or statistical approaches can be understood in terms of this difference.
We will revisit it later when we talk about manifold modeling, stochastic 
blockmodeling, Laplacian solvers, and related topics.

\subsection{Constructing graphs from data}

If the data are vectors with associated similarity information, then an obvious 
thing to do is to represent that data as a graph $G=(V,E)$, where each vertex 
$v \in V$ is associated with a data point $x_i$ an edge $e = (v_i v_j) \in E$ is 
defined if $s_{ij}$ is larger than some threshold.
Here, the threshold could perhaps equals zero, and the edges might be 
weighted by $s_{ij}$.
In this case, an obvious idea to cluster the vector data is to cluster the nodes 
of the corresponding graph.

Now, let's consider how to specify the similarity information $s_{ij}$.
There are many ways to construct a similarity graph, given the vectors 
$\{x_i\}_{i=1}^{n}$ data points as well as pairwise similarity (or distance) 
information $s_{ij}$ (of $d_{ij}$).  
Here we describe several of the most popular.
\begin{itemize}
\item
\textbf{$\epsilon$-NN graphs.}
Here, we connect all pairs of points with distance $d_{ij} \le \epsilon$.
Since the distance ``scale'' is set (by $\le \epsilon$), it is common not
to including the weights.  
The justification is that, in certainly idealized situations, including weights 
would not incorporate more information.
\item
\textbf{$k$-NN graphs.}
Here, we connect vertex $i$ with vertex $j$ if $v_i$ is among the $k$-NN of 
$v_i$, where NNs are given by the distance $d_{ij}$.
Note that this is a directed graph.
There are two common ways to make it undirected.
First, ignore directions; and
second, include an edge if 
($v_i$ connects to $v_j$ AND $v_j$ connects to $v_i$) or if 
($v_i$ connects to $v_j$ OR $v_j$ connects to $v_i$).
In either of those cases, the number of edges doesn't equal $k$; sometimes 
people filter it back to exactly $k$ edges per node and sometimes not.
In either case, weights are typically included.
\item
\textbf{Fully-connected weighted graphs.}
Here, we connect all points with a positive similarity to each other.  
Often, we want the similarity function to represent local neighborhoods, 
and so $s_{ij}$ is either transformed into another form or constructed
to represent this.
A popular choice is the Gaussian similarity kernel
\[
s(x_i,x_j) = \exp \left ( \frac{1}{2\sigma^2} \| x_i-x_j \|_2^2 \right)  ,
\]
where $\sigma$ is a parameter that, informally, acts like a width.
This gives a matrix that has a number of nice properties, e.g., it is 
positive and it is SPSD, and so it is good for MLers who like kernel-based 
methods.
Moreover, it has a strong mathematical basis, e.g., in scientific computing.
(Of course, people sometimes use this $s_{ij}=s(x_i,x_j)$ information to 
construct $\epsilon$-NN or $k$-NN graphs.)
\end{itemize}
Note that in describing those various ways to construct a graph from the 
vector data, we are already starting to see a bunch of knobs that can be 
played with, and this is typical of these graph construction methods.

Here are some comments about that graph construction approach.
\begin{itemize}
\item
Choosing the similarity function is basically an art.
One of the criteria is that typically one is not interested in resolving
differences that are  large, i.e., between moderately large and very large 
distances, since the goal is simply to ensure that those points are not close 
and/or since (for domain-specific reasons) that is the least reliable similarity 
information.
\item
Sometimes this approach is of interest in semi-supervised and transductive
learning.
In this case, one often has a lot of unlabeled data and only a little bit of 
labeled data; and one wants to use the unlabeled data to help define some 
sort of geometry to act as a prior to maximize the usefulness of the labeled 
data in making predictions for the unlabeled data.
Although this is often thought of as defining a non-linear manifold, you should 
think of it at using unlabeled data to specify a data-dependent model class to 
learn with respect to.
(That makes sense especially if the labeled and unlabeled data come from the
same distribution, since in that case looking at the unlabeled data is akin
to looking at more training data.)
As we will see, these methods often have an interpretation in terms of a 
kernel, and so they are used to learn linear functions in implicitly-defined
feature spaces anyway.
\item
$k$-NN, $\epsilon$-NN, and fully-connected weighted graphs are all the same 
in certain very idealized situations, but they can be very different in practice.
$k$-NN  often homogenizes more, which people often like, and/or it connects
points of different ``size scales,'' which people often find~useful.
\item
Choosing $k$, $\epsilon$, and $\sigma$ large can easily ``short circuit''
nice local structure, unless (and sometimes even if) the local structure is 
extremely nice (e.g., one-dimensional).
This essentially injects large-scale noise and expander-like structure; and
in that case one should expect very different properties of the constructed 
graph (and thus very different results when one runs algorithms).
\item
The fully-connected weighted graph case goes from being a rank-one complete 
graph to being a diagonal matrix, as one varies $\sigma$.
An important question (that is rarely studied) is how does that graph look 
like as one does a ``filtration'' from no edges to a complete graph.
\item
Informally, it is often thought that mutual-$k$-NN is between $\epsilon$-NN
and $k$-NN: it connects points within regions of constant density, but it 
doesn't connect regions of very different density.
(For $\epsilon$-NN, points on different scales don't get connected.)
In particular, this means that it is good for connecting clusters of 
different densities.
\item
If one uses a fully-connected graph and then sparsifies it, it is often hoped
that the ``fundamental structure'' is revealed and is nontrivial.
This is true in some case, some of which we will return to later, but it is also
very non-robust.
\item
As a rule of thumb, people often choose parameters s.t. $\epsilon$-NN and 
$k$-NN graphs are at least ``connected.''
While this seems reasonable, there is an important question of whether it 
homogenizes too much, in particular if there are interesting 
heterogeneities in the graph.
\end{itemize}

\subsection{Connections between different Laplacian and random walk matrices}

Recall the combinatorial or non-normalized Laplacian 
\[
L=D-W  , 
\]
and the normalized Laplacian 
\[
L_{sym} = D^{-1/2}LD^{-1/2} = I - D^{-1/2}WD^{-1/2}   .
\]
There is also a random walk matrix that we will get to more detail on in a 
few classes and that for today we will call the (somewhat non-standard name)
random walk Laplacian 
\[
L_{rw} = D^{-1}L = I-D^{-1}W  = D^{-1/2}L_{sym}D^{1/2} .
\]
Here is a lemma connecting them.
\begin{lemma}
Given the above definitions of $L$, $L_{sym}$, and $L_{rw}$, we have
the following.
\begin{enumerate}
\item
For all $x\in\mathbb{R}^{n}$, 
\[
x^TL_{sym}x = \frac{1}{2}
              \sum_{ij} W_{ij} 
                        \left( \frac{x_i}{\sqrt{d_i}} 
                             - \frac{x_j}{\sqrt{d_j}} \right)^{2}  .
\]
\item
$\lambda$ is an eigenvalue of $L_{rw}$ with eigenvector $u$ 
iff 
$\lambda$ is an eigenvalue of $L_{sym}$ with eigenvector $w=D^{1/2}u$ 
\item
$\lambda$ is an eigenvalue of $L_{rw}$ with eigenvector $u$ 
iff 
$\lambda$ and $u$ solve the generalized eigenvalue problem $Lu = \lambda D u$.
\item
$0$ is an eigenvalue of $L_{rw}$ with eigenvector $\vec{1}$ 
iff 
$0$ is an eigenvalue of $L_{sym}$ with eigenvector $D^{1/2}\vec{1}$ 
\item
$L_{sym}$ and $L_{rw}$ are PSD and have $n$ non-negative real-valued 
eigenvalues $0 = \lambda_1 \le \cdots \le \lambda_n$.
\end{enumerate}
\end{lemma}

\noindent
Hopefully none of these claims are surprising by now, but they do make 
explicit some of the connections between different vectors and different 
things that could be computed, e.g., one might solve the generalized
eigenvalue problem $Lu = \lambda D u$ or run a random walk to approximate 
$u$ and then from that rescale it to get a vector for $L_{sym}$.

\subsection{Using constructed data graphs}

Spectral clustering, as it is often used in practice, often involves first 
computing several eigenvectors (or running some sort of procedures that 
compute some sort of approximate eigenvectors) and then performing 
$k$-means in a low-dimensional space defined by them.
Here are several things to note.
\begin{itemize}
\item
This is harder to analyze than the vanilla spectral clustering we have so far 
been considering.
The reason is that one must analyze the $k$ means algorithm also.
In this context, $k$-means is essentially used as a rounding algorithm.
\item
A partial justification of this is provided by the theoretical result on using 
the leading $k$ eigenvectors that you considered on the first homework.
\item
A partial justification is also given by a result we will get to below that shows
that it works in very idealized situations.
\end{itemize}

We can use different Laplacians in different ways, as well as different 
clustering, $k$-means, etc. algorithms in different ways to get spectral-like 
clustering algorithms.
Here, we describe $3$ canonical algorithms (that use $L$, $L_{rw}$, and 
$L_{sym}$) to give an example of several related approaches.

Assume that we have $n$ points, $x_1,\ldots,x_n$, that we measure pairwise
similarities $s_{ij} = s(x_i,x_j)$ with symmetric nonnegative similarity 
function, and that we denote the similarity matrix by 
$S = \left(S_{ij}\right)_{i,j=1,\ldots,n}$.
The following algorithm, let's call it \textsc{PopularSpectralClustering}, 
takes as input a similarity matrix $S\in\mathbb{R}^{n \times n}$ and a 
positive integer $k\in\mathbb{Z}^{+}$ which is the number of clusters; and 
it returns $k$ clusters.
It does the following steps.
\begin{enumerate}
\item
Construct a similarity graph (e.g., with $\epsilon$-NN, $k$-NN, 
fully-connected graph, etc.)
\item
Compute the unnormalized Laplacian $L=D-A$.
\item
\begin{itemize}
\item
If (use $L$) \\
then compute the first $k$ eigenvectors $u_1,\ldots,u_k$ of $L$,
\item
else if (use $L_{rw}$) \\
then compute first $k$ generalized eigenvectors $u_1,\ldots,u_k$ of the 
generalized eigenvalue problem $Lu = \lambda D u$.
(Note by the above that these are eigenvectors of $L_{rw}$.)
\item
else if (use $L_{sym}$) \\
then compute the first $k$ eigenvectors $u_1,\ldots,u_k$ of $L_{sym}$.
\end{itemize}
\item
Let $U\in\mathbb{R}^{n \times k}$ be the matrix containing the vectors 
$u_1,\ldots,u_k$ as columns.
\item
\begin{itemize}
\item
If (use $L_{sym}$) \\
then $u_{ij} \leftarrow u_{ij} / \left( \sum_k u_{ik}^{2} \right)^{1/2}$, 
i.e., normalize $U$ row-wise.
\end{itemize}
\item
For $i=\{1,\ldots,n\}$, let $y_i\in\mathbb{R}^{k}$ be a vector containing 
the $i^{th}$ row of $U$.
\item
Cluster points $\left( y_i \right)_{i\in[n]}$ in $\mathbb{R}^{k}$ with a 
$k$-means algorithm into clusters, call them $C_1,\ldots,C_k$.
\item
Return:
clusters $A_1,\ldots,A_k$, with $A_i = \{ j : y_j \in C_i \} $.
\end{enumerate}

Here are some comments about the \textsc{PopularSpectralClustering} algorithm.
\begin{itemize}
\item
The first step is to construct a graph, and we discussed above that 
there are a lot of knobs.
In particular, the \textsc{PopularSpectralClustering} algorithm is not 
``well-specified'' or ``well-defined,'' in the sense that the algorithms
we have been talking about thus far are.
It might be better to think of this as an algorithmic approach, with 
several knobs that can be played with, that comes with suggestive 
but weaker theory than what we have been describing so~far.
\item
$k$-means is often used in the last step, but it is not necessary, and it is
not particularly principled (although it is often reasonable if the data
tend to cluster well in the space defined by $U$).
Other methods have been used but are less popular, presumably since $k$-means
is good enough and there are enough knobs earlier in the pipeline that the last step 
isn't the bottleneck to getting good results.
Ultimately, to get quality-of-approximation guarantees for an algorithm 
like this, you need to resort to a Cheeger-like bound or a heuristic 
justification or weaker theory that provides justification in idealized cases.
\item
In this context, $k$-means is essentially a rounding step to take a continuous 
embedding, provided by the continuous vectors $\{ y_i \}$, where 
$y_i \in \mathbb{R}^{k}$, and put them into one of $k$ discrete values.
This is analogous to what the sweep cut did.
But we will also see that this embedding, given by $\{y_i\}_{i=1}^{n}$ can be 
used for all sorts of other things.
\item
Remark: If one considers the $k$-means objective function, written as an IP 
and then relaxes it (from having the constraint that each data point goes into 
one of $k$ clusters, written as an orthogonal matrix with one nonzero per 
column, to being a general orthogonal matrix), then you get an objective function, 
the solution to which can be computed by computing a truncated SVD, i.e., 
the top $k$ singular vectors.
This provides a $2$-approximation to the $k$-means objective.
There are better approximation algorithms for the $k$-means objective, when 
measured by the quality-of-approximation, but this does provide an interesting 
connection.
\item
The rescaling done in the ``If (use $L_{sym}$) then'' is typical of many 
spectral algorithms, and it can be the source of confusion.
(Note that the rescaling is done with respect to 
$\left(P_U\right)_{ii} = \left(UU^T\right)_{ii}$, i.e., the statistical leverage scores 
of $U$, and this means that more ``outlying'' points get down-weighted.)
From what we have discussed before, it should not be surprising that we need 
to do this to get the ``right'' vector to work with, e.g., for the Cheeger theory 
we talked about before to be as tight as possible.
On the other hand, if you are approaching this from the perspective of 
engineering an algorithm that returns clusters when you expect them, it 
can seem somewhat ad hoc.
There are many other similar ad hoc and seemingly ad hoc decisions that are 
made when engineering implementations of spectral graph methods, and this 
lead to a large proliferation of spectral-based methods, many of which are
very similar ``under the hood.''
\end{itemize}

All of these algorithms take the input data $x_i\in\mathbb{R}^{n}$ and change 
the representation in a lossy way to get data points $y_i\in\mathbb{R}^{k}$.
Because of the properties of the Laplacian (some of which we have been 
discussing, and some of which we will get back to), this \emph{often} 
enhances the cluster properties of the data.

In idealized cases, this approach works as expected, as the following example 
provides.
Say that we sample data points from $\mathbb{R}$ from four equally-spaced
Gaussians, and from that we construct a NN graph.
(Depending on the rbf width of that graph, we might have an essentially
complete graph or an essentially disconnected graph, but let's say we choose 
parameters as the pedagogical example suggests.)
Then $\lambda_1=0$; $\lambda_2$, $\lambda_3$, and $\lambda_4$ are small; and 
$\lambda_5$ and up are larger.
In addition, $v_1$ is flat; and higher eigenfunctions are sinusoids of 
increasing frequency.
The first few eigenvectors can be used to split the data into the four 
natural clusters (they can be chosen to be worse linear combinations, but 
they can be chosen to split the clusters as the pedagogical example suggests).
But this idealized case is chosen to be ``almost disconnected,'' and so it 
shouldn't be surprising that the eigenvectors can be chosen to be almost 
cluster indicator vectors.
Two things: the situation gets much messier for real data, if you consider 
more eigenvectors; and the situation gets much messier for real data, if the 
clusters are, say, 2D or 3D with realistic noise.

\subsection{Connections with graph cuts and other objectives}

Here, we will briefly relate what we have been discussing today with what 
we discussed over the last month.
In particular, we describe the graph cut point of view to this spectral 
clustering algorithm.
I'll follow the notation of the von Luxburg review, so you can go back to
that, even though this is very different than what we used before.
The point here is not to be detailed/precise, but instead to remind you what 
we have been covering in another notation that is common, especially in ML, 
and also to derive an objective that we haven't covered but that is a popular 
one to which to add constraints.

To make connections with the  \textsc{PopularSpectralClustering} algorithm 
and MinCut, RatioCut, and NormalizedCut, recall that 
\[
\mbox{RatioCut}(A_1,\ldots,A_k) 
   = \frac{1}{2} \sum_{i=1}^{k} \frac{W\left(A_i,\bar{A}_i\right)}{|A_i|}
   = \sum_{i=1}^{k} \frac{ \mbox{cut}\left( A_i,\bar{A}_i \right) }{ |A_i| } ,
\]
where 
$\mbox{cut}(A_1,\ldots,A_k) = \frac{1}{2}\sum_{i=1}^{k}\left( A_i,\bar{A}_i \right)$.

First, let's consider the case $k=2$ (which is what we discussed before).
In this case, we want to solve the following problem:
\begin{equation}
\min_{A \subset V} \mbox{RatioCut}\left(A,\bar{A}\right)  .
\label{eqn:ratio-cut-1}
\end{equation}
Given $A \subset V$, we can define a function 
$f = \left( f_1,\ldots,f_n \right)^T \in \mathbb{R}^{n}$ s.t.
\begin{equation}
f_{i} = \left\{ \begin{array}{l l}
                   \sqrt{ |\bar{A}|/|A| } & \quad \text{if $v_i \in A$}  \\
                  -\sqrt{ |A|/|\bar{A}| } & \quad \text{if $v_i \in \bar{A}$}
                \end{array} 
        \right.  .
\label{eqn:f-indicator-2}
\end{equation}
In this case, we can write Eqn.~(\ref{eqn:ratio-cut-1}) as follows:
\begin{align*}
\min_{A \subset V} & f^TLf   \\
\mbox{s.t.~~}      & f \perp \vec{1}  \\
                   & f \mbox{ defined as in Eqn.~(\ref{eqn:f-indicator-2})}  \\
                   & \|f\|=\sqrt{n}   .
\end{align*}
In this case, we can relax this objective to obtain
\begin{align*}
\min_{A \subset V} & f^TLf   \\
\mbox{s.t.~~}      & f \perp \vec{1}  \\
                   & \|f\|=\sqrt{n}   ,
\end{align*}
which can then be solved by computing the leading eigenvectors of $L$.

Next, let's consider the case $k > 2$ (which is more common in practice).
In this case, given a partition of the vertex set $V$ into $k$ sets 
$A_i,\ldots,A_k$, we can define $k$ indicator vectors 
$h_j = \left( h_{ij},\ldots,h_{nj}\right)^T$ by 
\begin{equation}
h_{ij} = \left\{ \begin{array}{l l}
                     1/ \sqrt{ |A_j| } & \quad v_i \in A_j, \quad i\in[n],j\in[k]  \\ 
                     0                 & \quad \text{otherwise}
                  \end{array} 
          \right.  .
\label{eqn:f-indicator-k}
\end{equation}
Then, we can set the matrix $H \in \mathbb{R}^{n \times k}$ as the matrix 
containing those $k$ indicator vectors as columns, and observe that 
$H^TH=I$, i.e., $H$ is an orthogonal matrix (but a rather special one, since
it has only one nonzero per row).

We note the following observation; this is a particular way to write the 
RatioCut problem as a Trace problem that appears in many places.
\begin{claim}
$\mbox{RatioCut}(A_1,\ldots,A_k) = \mbox{Tr}(H^TLH)$
\end{claim}
\begin{proof}
Observe that 
\[
h_i^TLh_i = \frac{ \mbox{cut}\left( A_i,\bar{A}_i \right)}{ |A_i| }
\]
and also that
\[
h_i^TLh_i = \left(H^TLH\right)_{ii} .
\]
Thus, we can write
\[
\mbox{RatioCut}\left(A_1,\ldots,A_k\right)
   = \sum_{i=1}^{k} h_i L h_i 
   = \sum_{i=1}^{k} \left( H^T L H \right)_{ii}
   = \mbox{Tr}\left(H^TLH\right)   .
\]
\end{proof}

So, we can write the problem of 
\[
\min \mbox{RatioCut}\left(A_1,\ldots,A_k\right)  
\]
as follows:
\begin{align*}
\min_{A_1,\ldots,A_k} & \mbox{Tr}\left(H^TLH\right)  \\
\mbox{s.t.~~}         & H^TH=I  \\
                      & H \mbox{ defined as in Eqn.~(\ref{eqn:f-indicator-k})}   .
\end{align*}
We can relax this by letting the entries of $H$ be arbitrary elements of 
$\mathbb{R}$ (still subject to the overall orthogonality constraint on $H$) 
to get 
\begin{align*}
\min_{H\in\mathbb{R}^{n \times k}} & \mbox{Tr}\left(H^TLH\right)  \\
\mbox{s.t.~~}                      & H^TH=I    ,
\end{align*}
and the solution to this is obtained by computing the first $k$ eigenvectors 
of $L$.

Of course, similar derivations could be provided for the NormalizedCut 
objective, in which case we get similar results, except that we deal with 
degree weights, degree-weighted constraints, etc.
In particular, for $k>2$, if we define indicator vectors
$h_j = \left( h_{ij},\ldots,h_{nj}\right)^T$ by 
\begin{equation}
h_{ij} = \left\{ \begin{array}{l l}
                     1/ \sqrt{ \mbox{Vol}(A_j) } & \quad v_i \in A_j, \quad i\in[n],j\in[k]  \\ 
                     0                 & \quad \text{otherwise}
                  \end{array} 
          \right.  .
\label{eqn:f-indicator-kgt2}
\end{equation}
then the problem of minimizing NormalizedCut is 
\begin{align*}
\min_{A_1,\ldots,A_k} & \mbox{Tr}\left(H^TLH\right)  \\
\mbox{s.t.~~}         & H^TDH=I  \\
                      & H \mbox{ defined as in Eqn.~(\ref{eqn:f-indicator-kgt2})}   ,
\end{align*}
and if we let $T = D^{1/2}H$, then the spectral relaxation is 
\begin{align*}
\min_{T\in\mathbb{R}^{n \times k}} & \mbox{Tr}\left(T^TD^{-1/2}LD^{-1/2}T\right)  \\
\mbox{s.t.~~}                      & T^TT=I    ,
\end{align*}
and the solution $T$ to this trace minimization problem is given by the leading 
eigenvectors of $L_{sym}$.
Then $H=D^{-1/2}T$, in which case $H$ consists of the first $k$ eigenvectors of
$L_{rw}$, or the first $k$ generalized eigenvectors of $Lu=\lambda D u$.

Trace optimization problems of this general for arise in many related 
applications.
For example:
\begin{itemize}
\item
One often uses this objective as a starting point, e.g., to add sparsity or other 
constraints, as in one variation of ``sparse PCA.''
\item
Some of the methods we will discuss next time, i.e., LE/LLE/etc. do 
something very similar but from a different motivation, and this provides
other ways to model the data.
\item
As noted above, the $k$-means objective can actually be written as an 
objective with a similar constraint matrix, i.e., if $H$ is the cluster indicator 
vector for the points, then $H^TH=I$ and $H$ has one non-zero per row.
If we relax that constraint to be any orthogonal matrix such that 
$H^TH=I$, then we get an objective function, the solution to which is the
truncated SVD; and this provides a $2$ approximation to the $k$-means 
problem.
\end{itemize}

%% file: lect13.tex
\section{%
(03/05/2015): 
Some Practical Considerations (2 of 4):
Basic perturbation theory and basic dimensionality reduction }

Reading for today.
\begin{compactitem}
\item
``A kernel view of the dimensionality reduction of manifolds,'' in ICML, by Ham, et al.
\end{compactitem}

Today, we will cover two topics: the Davis-Kahan-$\sin\left(\theta\right)$ theorem, which is a basic result from matrix perturbation theory that can be used to understand the robustness of spectral clustering in idealized cases; and basic linear dimensionality reduction methods that, while not spectral graph methods by themselves, have close connections and are often used with spectral graph methods.

\subsection{Basic perturbation theory}

One way to analyze spectral graph methods---as well as matrix algorithms
much more generally---is via matrix perturbation theory.
Matrix perturbation theory asks: how do the eigenvalues and eigenvectors of
a matrix $A$ change if we add a (small) perturbation $E$, i.e., if we are
working the the matrix $\tilde{A}=A+E$?
Depending on the situation, this can be useful in one or more of several~ways.
\begin{itemize}
\item
\textbf{Statistically.}
There is often noise in the input data, and we might want to make claims 
about the unobserved processes that generate the observed data.
In this case, $A$ might be the hypothesized data, e.g., that has some nice
structure; we observe and are working with $\tilde{A}=A+E$, where $E$ might 
be Gaussian noise, Gaussian plus spiked noise, or whatever; and we want to 
make claims that algorithms we run on $\tilde{A}$ say something about the 
unobserved $A$.
\item
\textbf{Algorithmically.}
Here, one has the observed matrix $A$, and one wants to make claims about 
$A$, but for algorithmic reasons (or other reasons, but typically algorithmic 
reasons if randomness is being exploited as a computational resource), one 
performs random sampling or random projections and computes on the
sample/projection.
This amounts to constructing a sketch $\tilde{A}=A+E$ of the full input matrix 
$A$, where $E$ is whatever is lost in the construction of the sketch, and one 
wants to provide guarantees about $A$ by computing on $\tilde{A}$.
\item
\textbf{Numerically.}
This arises since computers can't represent real numbers exactly, i.e., there is
round-off error, even if it is at the level of machine precision, and thus it
is of interest to know the sensitivity of problems and/or algorithms to such
round-off errors.
In this case, $A$ is the answer that would have been computed in exact 
arithmetic, while $\tilde{A}$ is the answer that is computed in the presence
of round-off error.
(E.g., inverting a non-invertible matrix is very sensitive, but inverting
an orthogonal matrix is not, as quantified by the condition number of the
input matrix.)
\end{itemize}
The usual reference for matrix perturbation theory is the book of 
Stewart and Sun, which was written primarily with numerical issues in mind.

Most perturbation theorems say that some notion of distance between 
eigenstuff, e.g., eigenvalues of subspaces defined by eigenvectors of $A$ 
and $\tilde{A}$, depend on the norm of the error/perturbation $E$, often 
times something like a condition number that quantifies the robustness of 
problems.
(E.g., it is easier to estimate extremal eigenvalues than eigenvectors that 
are buried deep in the spectrum of $A$, and it is easier if $E$ is smaller.)

For spectral graph methods, certain forms of matrix perturbation theory can 
provide some intuition and qualitative guidance as to when spectral 
clustering works.
We will cover one such results that is particularly simple to state and 
think about. 
When it works, it works well; but since we are only going to describe a 
particular case of it, when it fails, it might fail ungracefully. 
In some cases, more sophisticated variants of this result can provide 
guidance.

When applied to spectral graph methods, matrix perturbation theory is usually
used in the following way.
Recall that if a graph has $k$ disconnected components, then 
$0 = \lambda_1 = \lambda_2 = \ldots = \lambda_k = < \lambda_{k+1} $, and
the corresponding eigenvectors $v_1,v_2,\ldots,v_k$ can be chosen to be 
the indicator vectors of the connected components.
In this case, the connected components are a reasonable notion of clusters, 
and the $k$-means algorithm should trivially find the correct clustering.
If we let $A$ be the Adjacency Matrix for this graph, then recall that it 
splits into $k$ pieces.
Let's assume that this is the idealized unobserved case, and the data that
we observe, i.e., the graph we are given or the graph that we construct is 
a noisy version of this, call it $\tilde{A} = A+E$, where $E$ is the 
noise/error.
Among other things $E$ will introduce ``cross talk'' between the clusters, 
so they are no longer disconnected.
In this case, if $E$ is small, then we might hope that perturbation theory
would show that only the top $k$ eigenvectors are small, well-separated 
from the rest, and that the $k$ eigenvectors of $\tilde{A}$ are perturbed 
versions of the original indicator vectors.

As stated, this is not true, and the main reason for this is that 
$\lambda_{k+1}$ (and others) could be very small.
(We saw a version of this before, when we showed that we don't actually need
to compute the leading eigenvector, but instead any vector whose Rayleigh
quotient was similar would give similar results---where by similar results we 
mean results on the objective function, as opposed to the actual clustering.)
But, if we account for this, then we can get an interesting perturbation 
bound.
(While interesting, in the context of spectral graph methods, this bound is 
somewhat weak, in the sense that the perturbations are often much larger 
and the spectral gap are often much larger than the theorem permits.)

This result is known as the Davis-Kahan theorem; and it is used to bound the
distance between the eigenspaces of symmetric matrices under symmetric
perturbations.
(We saw before that symmetric matrices are much ``nicer'' than general 
matrices. 
Fortunately, they are very common in machine learning and data analysis, 
even if it means considering correlation matrices $XX^T$ or $X^TX$.
Note that if we relaxed this requirement here, then this result would be 
false, and to get generalizations, we would have to consider all sorts of 
other messier things like pseudo-spectra.)

To bound the distance between the eigenspaces, let's define the notion of 
an angle (a canonical or principal angle) between two subspaces.

\begin{definition}
Let $\mathcal{V}_1$ and $\mathcal{V}_2$ be two $p$-dimensional subspaces of
$\mathbb{R}^{d}$, and let $V_1$ and $V_2$ be two orthogonal matrices (i.e., 
$V_1^TV_1=I$ and $V_2^TV_2=I$) spanning $\mathcal{V}_1$ and $\mathcal{V}_2$.
Then the \emph{principal angles} $\{\theta_i\}_{i=1}^{d}$ are s.t.
$\cos(\theta_i)$ are the singular values of $V_1^TV_2$.
\end{definition}

Several things to note.
First, for $d=1$, this is the usual definition of an angle between two 
vectors/lines.
Second, one can define angles between subspaces of different dimensions, 
which is of interest if there is a chance that the perturbation introduces rank
deficiency, but we won't need that here.
Third, this is actually a full vector of angles, and one could choose the
largest to be \emph{the} angle between the subspaces, if one wanted.

\begin{definition}
Let $\sin\left(\theta\left(\mathcal{V}_1,\mathcal{V}_2\right)\right)$ be the
diagonal matrix with the sine of the canonical angles along the diagonal.
\end{definition}

Here is the Davis-Kahan-$\sin\left(\theta\right)$ theorem.
We won't prove it.

\begin{theorem}[Davis-Kahan]
\label{thm:davis-kahan}
Let $A,E\in\mathbb{R}^{n \times n}$ be symmetric matrices, and
consider $\tilde{A}=A+E$.
Let $S_1 \subset \mathbb{R}$ be an interval; and
denote by $\sigma_{S_1}\left(A\right)$ the eigenvalues of $A$ in $S_1$, and
by $V_1$ the eigenspace corresponding to those eigenvalues.
Ditto for $\sigma_{S_1}\left(\tilde{A}\right)$ and $\tilde{V}_1$.
Define the distance between the interval $S_1$ and the spectrum of $A$ 
outside of $S_1$ as
\[
\delta = \min \{  \| \lambda-s \| : \lambda\mbox{ is eigenvalue of }A, 
                                   \lambda \notin S_1, s \in S_1 \}   .
\]
Then the distance 
$d\left(V_1,\tilde{V}_1\right) 
   = \|\sin \theta \left( V_1,\tilde{V}_1 \right) \|$ 
between the two subspaces $V_1$ and $\tilde{V}_1$ can be bounded as 
\[
d\left(V_1,\tilde{V}_1\right) \le \frac{\|E\|}{\delta}  ,
\]
where $\|\cdot\|$ denotes the spectral or Frobenius norm.
\end{theorem}

What does this result mean?
For spectral clustering, let $L$ be the original (symmetric, and SPSD) 
hypothesized matrix, with $k$ disjoint clusters, and let $\tilde{L}$ be the 
perturbed observed matrix.
In addition, we want to choose the interval such that the first $k$ 
eigenvalues of both $L$ and $\tilde{L}$ are in it, and so let's choose the 
interval as follows.
Let $S_1 = [0,\lambda_k]$ (where we recall that the first $k$ eigenvalues of
the unperturbed matrix equal $0$); in this case, 
$ \delta = | \lambda_k - \lambda_{k+1} | $, i.e., $\delta$ equals the 
``spectral gap'' between the $k^{th}$ and the $(k+1)^{st}$ eigenvalue.

Thus, the above theorem says that the bound on the distance $d$ between the 
subspaces defined by the first $k$ eigenvectors of $L$ and $\tilde{L}$ is 
less if:
(1) the norm of the error matrix $\|E\|$ is smaller; and 
(2) the value of $\delta$, i.e., the spectral gap, is larger.
(In particular, note that we need the angle to be less than 90 degrees to 
get nontrivial results, which is the usual case; otherwise, rank is lost).

This result provides a useful qualitative guide, and there are some 
more refined versions of it, but note the following.
\begin{itemize}
\item
It is rarely the case that we see a nontrivial eigenvalue gap in real data.
\item
It is better to have methods that are robust to slow spectral decay.
Such methods exist, but they are more involved in terms of the linear algebra, 
and so many users of spectral graph methods avoid them.
We won't cover them here.
\item
This issue is analogous to what we saw with Cheeger's Inequality, were we 
saw that we got similar bounds on the objective function value for any 
vector whose Rayleigh quotient was close to the value of $\lambda_2$, but 
the actual vector might change a lot (since if there is a very small spectral 
gap, then permissible vectors might ``swing'' by 90 degrees).
\item
BTW, although this invalidates the hypotheses of Theorem~\ref{thm:davis-kahan}, 
the results of spectral algorithms might still be useful, basically since they 
are used as intermediate steps, i.e., features for some other task.
\end{itemize}
That being said, knowing this result is useful since it suggests and 
explains some of the eigenvalue heuristics that people do to make vanilla
spectral clustering work.

As an example of this, recall the row-wise reweighting we was last time.
As a general rule, eigenvectors of orthogonal matrices are robust, but not
otherwise in general.
Here, this manifests itself in whether or not the components of an 
eigenvector on a given component are ``bounded away from zero,'' meaning that
there is a nontrivial spectral gap.
For $L$ and $L_{rw}$, the eigenvectors are indicator vectors, so there is 
no need to worry about this, since they will be as robust as possible to 
perturbation.
But for $L_{sym}$, the eigenvector is $D^{1/2}\vec{1}_A$, and if there is 
substantial degree variability then this is a problem, i.e., for low-degree 
vertices their entries can be very small, and it is difficult to deal with 
them under perturbation.
So, the row-normalization is designed to robustify the algorithms.

This ``reweigh to robustify'' is an after-the-fact justification.
One could alternately note that all the results for degree-homogeneous 
Cheeger bounds go through to degree-heterogeneous cases, if one puts in 
factors of $d_{max}/d_{min}$ everywhere.
But this leads to much weaker bounds than if one considers conductance and 
incorporates this into the sweep cut.
I.e., from the perspective of optimization objectives, the reason to reweigh is 
to get tighter Cheeger's Inequality guarantees.

\subsection{Linear dimensionality reduction methods}

There are a wide range of methods that do the following: construct a graph 
from the original data; and then perform computations on the graph to do 
feature identification, clustering, classification, regression, etc. on the original 
data.
(We saw one example of this when we constructed a graph, computed its top $k$
eigenvectors, and then performed $k$-means on the original data in the space
thereby defined.)  
These methods are sometimes called \emph{non-linear dimensionality 
reduction methods} since the constructed  graphs can be interpreted as 
so-called kernels and since the resulting methods can be interpreted as 
kernel-based machine learning methods.
Thus, they indirectly boil down to computing the SVD---indirectly in that it is
in a feature space that is implicitly defined by the kernel.
This general approach is used for many other problems, and so we will 
describe it in some~detail.

To understand this, we will first need to understand a little bit about 
\emph{linear dimensionality reduction methods} (meaning, basically, those 
methods that directly boil down to the computing the SVD or truncated SVD of the 
input data) as well as kernel-based machine learning methods.
Both are large topics in its own right, and we will only touch the surface.

\subsubsection{PCA (Principal components analysis)}

Principal components analysis (PCA) is a common method for linear 
dimensionality that seeks to find a ``maximum variance subspace'' to describe
the data.
In more detail, say we are given $\{x_i\}_{i=1}^{n}$, with each 
$x_i\in\mathbb{R}^{m}$, and let's assume that the data have been centered in 
that $\sum_ix_i=0$.
Then, our goal is to find a subspace $P$, and an embedding 
$\vec{y}_i=P\vec{x}_i$, where $P^2=P$, s.t.
\[
\mbox{Var}(\vec{y}) = \frac{1}{n} \sum_i ||Px_i||^2
\]
is largest, \emph{i.e.}, maximize the projected variance, or where
\[
\mbox{Err}(\vec{y}) = \frac{1}{n} \sum_i ||x_i-Px_i||_2^2
\]
is smallest, \emph{i.e.}, minimize the reconstruction error.
Since Euclidean spaces are so structured, the solution to these two problems 
is identical, and is basically given by computing the SVD or truncated~SVD:
\begin{itemize}
\item
Let $C = \frac{1}{n}\sum_i x_i x_i^T$, \emph{i.e.}, $C \sim XX^T$.
\item
Define the variance as $\mbox{Var}(\vec{y}) = \mbox{Trace}(PCP)$
\item
Do the eigendecomposition to get 
$C = \sum_{i=1}^{m} \lambda_i \hat{e}_i\hat{e}_i^T$, 
where $\lambda_1 \geq \lambda_2 \geq \cdots \lambda_m \geq 0$.
\item
Let $P = \sum_{i=1}^{d} \hat{e}_i\hat{e}_i^T$, and then project onto this 
subspace spanning the top $d$ eigenfunctions of $C$.
\end{itemize}



\subsubsection{MDS (Multi-Dimensional Scaling)}

A different method (that boils down to taking advantage of the same 
structural result the SVD) is that of Multi-Dimensional Scaling (MDS), which
asks for the subspace that best preserves inter-point distances.
In more detail, say we are given $\{x_i\}_{i=1}^{n}$, with $x_i\in\mathbb{R}^{D}$, 
and let's assume that the data are centered in that $\sum_ix_i=0$.
Then, we have $\frac{n(n-1)}{2}$ pairwise distances, denoted $\Delta_{ij}$.
The goal is to find vectors $\vec{y}_i$ such that:
\[
|| \vec{y}_i - \vec{y}_j || \approx \Delta_{ij}
\]
We have the following lemma:
\begin{lemma}
If $\Delta_{ij}$ denotes the Euclidean distance of zero-mean vectors, then 
the inner products are
\[
G_{ij} = \frac{1}{2}\left( \sum_k\left( \Delta_{ik}^{2} + \Delta_{kj}^{2}\right) - \Delta_{ij}^2 -\sum_{kl} \Delta_{kl}^{2}  \right)
\]
\end{lemma}
Since the goal is to preserve dot products (which are a proxy for and in some 
cases related to distances), we will choose $\vec{y}_i$ to minimize
\[
\mbox{Err}(\vec{y}) = \sum_{ij} \left( G_{ij}-\vec{y}_i\cdot\vec{y}_j\right)^{2}
\]
The spectral decomposition of $G$ is given as
\[
G = \sum_{i=1}^{n} \lambda_i \hat{v}_i\hat{v}_i^T 
\]
where $\lambda_1 \geq \lambda_2 \geq \cdots \lambda_n \geq 0$.
In this case, the optimal approximation is given by
\[
y_{i\xi} = \sqrt{\lambda_{\xi}}v_{\xi i}
\]
for $\xi = 1,2,\ldots,d$, with $d \leq n$, 
which are simply scaled truncated eigenvectors.
Thus $G \sim X^T X$.

\subsubsection{Comparison of PCA and MDS}

At one level of granularity, PCA and MDS are ``the same,'' since they both
boil down to computing a low-rank approximation to the original data.
It is worth looking at them in a little more detail, since they come from 
different motivations and they generalize to non-linear situations in 
different ways.
In addition, there are a few points worth making as a comparison with 
some of the graph partitioning results we discussed.

To compare PCA and MDS:
\begin{itemize}
\item
$C_{ij} = \frac{1}{n}\sum_k x_{ki}x_{kj}$ is a 
$m \times m$ covariance matrix and takes roughly $O((n+d)m^2)$ time
to compute.
\item
$G_{ij} = \vec{x}_i\cdot\vec{x}_j$ is an $n \times n$ Gram matrix and
takes roughly $O((m+d)n^2)$ time to compute.
\end{itemize}

Here are several things to note:
\begin{itemize}
\item
PCA computes a low-dimensional representation that most faithfully 
preserves the covariance structure, in an ``averaged'' sense.
It minimizes the reconstruction error
\[
E_{PCA} = \sum_i || x_i = \sum_{\xi=1}^{m} (x_i \cdot e_{\xi})e_{\xi} ||_2^2  ,
\]
or equivalently it finds a subspace with minimum variance.
The basis for this subspace is given by the top $m$ eigenvectors of the 
$d \times d$ covariance matrix $C=\frac{1}{n}\sum_i x_i x_i^T$.
\item
MDS computes a low-dimensional representation of the high-dimensional data
that most faithfully preserve inner products, \emph{i.e.}, that minimizes
\[
E_{MDS} = \sum_{ij} \left( x_i \cdot x_j - \phi_i \cdot \phi_j \right)^2
\]
It does so by computing the Gram matrix of inner products 
$G_{ij} = x_i \cdot x_j$, so $G \approx X^TX$.
It the top $m$ eigenvectors of this are $\{v_{i}\}_{i=1}^{m}$ and
the eigenvalues are $\{\lambda_{i}\}_{i=1}^{m}$, then the 
embedding MDS returns is $\Phi_{i\xi} = \sqrt{\lambda_{\xi}}v_{\xi i}$.
\item
Although MDS is designed to preserve inner products, it is often motivated to 
preserve pairwise distances.
To see the connection, let 
\[
S_{ij} =  ||x_i-x_j||^2
\]
be the matrix of squared inter-point distances.
If the points are centered, then a Gram matrix consistent with these 
squared distances can be derived from the transformation 
\[
G = -\frac{1}{2} \left(I-uu^T\right) S \left(I-uu^T\right)
\]
where $u = \frac{1}{\sqrt{n}}( 1,\cdots,1)$.
\end{itemize}
Here are several additional things to note with respect to PCA and MDS and
kernel methods:
\begin{itemize}
\item
One can ``kernelize'' PCA by writing everything in terms of dot products.
The proof of this is to say that we can ``map'' the data $A$ to a feature 
space $\mathcal{F}$ with $\Phi(X)$.
Since $C = \frac{1}{n}\sum_{j=1}^{n}\phi(x_j)\phi(x_j)^T$ is a covariance
matrix, PCA can be computed from solving the eigenvalue problem:
Find a $\lambda > 0$ and a vector $v \ne 0$ s.t.
\begin{equation}
\lambda v = C v = \frac{1}{n} \sum_{j=1}^{n} (\phi(x_j)\cdot v) \phi(x_j)  .
\label{eqn:pca-ker1}
\end{equation}
So, all the eigenvectors $v_i$ with $\lambda_i$ must be in the span of the 
mapped data, \emph{i.e.}, $v \in \mbox{Span}\{\phi(x_1),\ldots,\phi(x_n)\}$, 
\emph{i.e.}, $v = \sum_{i=1}^{n} \alpha_i \phi(x_i)$ for some set of 
coefficients $\{\alpha_i\}_{i=1}^{n}$.
If we multiply~(\ref{eqn:pca-ker1}) on the left by $\phi(x_k)$, then we 
get
\[
\lambda(\phi(x_k)\cdot v) = (\phi(x_k)\cdot C v ) , \qquad k=1,\ldots,n   .
\]
If we then define
\begin{equation}
K_{ij} = (\phi(x_i),\phi(x_j)) = k(x_i,x_j) \in \mathbb{R}^{n \times n}  ,
\label{eqn:pca-ker2}
\end{equation}
then to compute the eigenvalues we only need
\[
\lambda \vec{\alpha} 
   = K \vec{\alpha} , \qquad \alpha = (\alpha_1,\ldots,\alpha_n)^T  .
\]
Note that we need to normalize $(\lambda_k,\alpha^k)$, and we can do so by
$\hat{K} = K - 1_n K - K 1_n - 1_n K 1_n$.
To extract features of a new pattern $\phi(x)$ onto $v^k$, we need
\begin{equation}
\label{eqn:pca-ker3}
(v^k\cdot\phi(x)) = \sum_{i=1}^{m} \alpha_i^k \phi(x_i)\cdot\phi(x)
                  = \sum_{i=1}^{m} \alpha_i^k k(x_i,x)  .
\end{equation}
So, the nonlinearities enter:
\begin{itemize}
\item
The calculation of the matrix elements in~(\ref{eqn:pca-ker2}).
\item
The evaluation of the expression~(\ref{eqn:pca-ker3}).
\end{itemize}
But, we can just compute eigenvalue problems, and there is no need to go 
explicitly to the high-dimensional space.
For more details on this, see ``An Introduction to Kernel-Based Learning
Algorithms,'' by by M{\"u}ller et al.  or ``Nonlinear component analysis 
as a kernel eigenvalue problem,'' by Sch{\"o}lkopf et al.
\item
Kernel PCA, at least for isotropic kernels $K$, where 
$K(x_i,x_j)=f(||x_i-x_j||)$, is a form of MDS and vice versa.
For more details on this, see ``On a Connection between Kernel PCA and 
Metric Multidimensional Scaling,'' by Williams and ``Dimensionality 
Reduction: A Short Tutorial,'' by Ghodsi.
To see this, recall that 
\begin{itemize}
\item
From the distances-squared, $\{\delta_{ij}\}_{ij}$, where 
$\delta_{ij} = ||x_i-x_j||^2_2 = (x_i - x_j)^T (x_i - x_j)$, we can 
construct a matrix $A$ with $A_{ij} = - \frac{1}{2} \delta_{ij}$.
\item
Then, we can let $B = HAH$, where $H$ is a ``centering'' matrix 
($H=I-\frac{1}{n} 1.1^T$).
This can be interpreted as centering, but really it is just a projection 
matrix (of a form not unlike we we have seen).
\item
Note that $B = HX(HX)^T$, (and $b_{ij} = (x_i-\bar{x})^T(x_j-\bar{x})$, with 
$\bar{x}=\frac{1}{n}\sum_ix_i$), and thus $B$ is SPSD.
\item
In the feature space, $\tilde{\delta}_{ij}$ is the Euclidean distance:
\begin{eqnarray*}
\tilde{\delta}_{ij} &=& (\phi(x_i) - \phi(x_j))^T (\phi(x_i) - \phi(x_j)) \\
                    &=& ||\phi(x_i) - \phi(x_j)||_2^2  \\
                    &=& 2(1-r(\delta_{ij}))   ,
\end{eqnarray*}
where the last line follows since with an isotropic kernel, where
$k(x_i,x_j) = r(\delta_{ij})$.
(If $K_{ij} = f(||x_i-x_j||)$, then $K_{ij} = r(\delta_{ij})$ ($r(0) = 1$).)
In this case, $A$ is such that $A_{ij} = r(\delta_{ij}) -1$, $A=K-1.1^T$, so 
(fact) $HAH=HKH$.
The centering matrix annihilates $11^T$, so $HAH=HKH$.
\end{itemize}
So, $K_{MDS} = -\frac{1}{2}(I-ee^T)A(I-ee^T)$, where $A$ is the matrix of 
squared distances.
\end{itemize}
So, the bottom line is that PCA and MDS take the data matrix and use SVD to 
derive embeddings from eigenvalues.
(In the linear case both PCA and MDS rely on SVD and can be constructed in 
$O(mn^2)$ time ($m > n$).)
They are very similar due to the linear structure and SVD/spectral theory.
If we start doing nonlinear learning methods or adding additional constraints, 
then these methods generalize in somewhat different ways.


\subsubsection{An aside on kernels and SPSD matrices}

The last few comments were about ``kernelizing'' PCA and MDS.
Here, we discuss this kernel issue somewhat more generally.

Recall that, given a collection $\mathcal{X}$ of data points, which are 
often but not necessarily elements of $\mathbb{R}^{m}$, techniques such as 
linear Support Vector Machines (SVMs), Gaussian Processes (GPs), Principle 
Component Analysis (PCA), and the related Singular Value Decomposition 
(SVD), identify and extract structure from $\mathcal{X}$ by computing linear 
functions, i.e., functions in the form of dot products, of the data.
(For example, in PCA the subspace spanned by the first $k$ eigenvectors is 
used to give a $k$ dimensional model of the data with minimal residual; thus, 
it provides a low-dimensional representation of the data.)

Said another way, these algorithms can be written in such a way that they 
only ``touch'' the data via the correlations between pairs of data points.
That is, even if these algorithms are often written in such as way that they
access the actual data vectors, they can be written in such a way that they 
only accesses the correlations between pairs of data vectors.
In principle, then, given an ``oracle'' for a different correlation matrix, 
one could run the same algorithm by providing correlations from the oracle, 
rather than the correlations from the original correlation matrix.

This is of interest essentially since it provides much greater flexibility 
in possible computations; or, said another way, it provides much greater 
flexibility in statistical modeling, without introducing too much additional
computational expense.
For example, in some cases, there is some sort of nonlinear structure in the 
data; or the data, e.g. text, may not support the basic linear operations of 
addition and scalar multiplication.
More commonly, one may simply be interested in working with more flexible
statistical models that depend on the data being analyzed, without making 
assumptions about the underlying geometry of the hypothesized data.

In these cases, a class of statistical learning algorithms known as 
\emph{kernel-based learning methods} have proved to be quite useful. 
These methods implicitly map the data into much higher-dimensional spaces, 
e.g., even up to certain $\infty$-dimensional Hilbert spaces, where 
information about their mutual positions (in the form of inner products) is 
used for constructing classification, regression, or clustering rules.
There are two points that are important here.
First, there is often an efficient method to compute inner products between 
very complex or even infinite dimensional vectors.
Second, while general $\infty$-dimensional Hilbert spaces are relatively 
poorly-structured objects, a certain class of $\infty$-dimensional Hilbert 
spaces known as Reproducing kernel Hilbert spaces (RKHSs) are 
sufficiently-heavily regularized that---informally---all of the ``nice''
behaviors of finite-dimensional Euclidean spaces still hold.
Thus, kernel-based algorithms provide a way to deal with nonlinear structure 
by reducing nonlinear algorithms to algorithms that are linear in some 
(potentially $\infty$-dimensional but heavily regularized) feature space 
$\mathcal{F}$ that is non-linearly related to the original input space.

The generality of this framework should be emphasized.
There are some kernels, e.g., Gaussian rbfs, polynomials, etc., that might 
be called \emph{a priori kernels}, since they take a general form that 
doesn't depend (too) heavily on the data; but there are other kernels that 
might be called \emph{data-dependent kernels} that depend very strongly 
on the data.
In particular, several of the methods to construct graphs from data that 
we will discuss next time, e.g., Isomap, local linear embedding, Laplacian 
eigenmap, etc., can be interpreted as providing data-dependent kernels.
These methods first induce some sort of local neighborhood structure on the 
data and then use this local structure to find a global embedding of the 
data into a lower dimensional space. 
The manner in which these different algorithms use the local information to 
construct the global embedding is quite different; but in general they can 
be interpreted as kernel PCA applied to specially-constructed Gram matrices.
Thus, while they are sometimes described in terms of finding non-linear
manifold structure, it is often more fruitful to think of them as 
constructing a data-dependent kernel, in which case they are useful or not 
depending on issues related to whether kernel methods are useful or whether 
mis-specified models are useful.

%% file: lect14.tex
\section{%
(03/10/2015): 
Some Practical Considerations (3 of 4):
Non-linear dimension reduction methods}

Reading for today.
\begin{compactitem}
\item
``Laplacian Eigenmaps for dimensionality reduction and data representation,'' in Neural Computation, by Belkin and Niyogi
\item
``Diffusion maps and coarse-graining: a unified framework for dimensionality reduction, graph partitioning, and data set parameterization,'' in IEEE-PAMI, by Lafon and Lee 
\end{compactitem}

Today, we will describe several related methods to identify structure in data.
The general idea is to do some sort of ``dimensionality reduction'' that is 
more general than just linear structure that is identified by a straightforward 
application of the SVD or truncated SVD to the input data.
The connection with what we have been discussing is that these procedures 
construct graphs from the data, and then they perform eigenanalysis on those 
graphs in order to construct ``low-dimensional embeddings'' of the data.
These (and many other related) methods are often called  \emph{non-linear 
dimension reduction methods}.
In some special cases they identify structure that is meaningfully non-linear; 
but it is best to think of them as constructing graphs to then construct 
data-dependent representations of the data that (like other kernel methods) 
is linear in some nonlinearly transformed version of the data.

\subsection{Some general comments}

The general framework for these methods is the following.
\begin{itemize}
\item 
Derive some (typically sparse) graph from the data, \emph{e.g.}, by connecting 
nearby data points with an $\epsilon$-NN or $k$-NN rule.
\item 
Derive a matrix from the graph (viz. adjacency matrix, Laplacian matrix).
\item 
Derive an embedding of the data into $\mathbb{R}^{d}$ using the eigenvectors 
of that matrix.
\end{itemize}

\emph{Many} algorithms fit this general outline.
Here are a few things worth noting about them.
\begin{itemize}
\item
They are not really algorithms (like we have been discussion) in the sense that 
there is a well-defined objective function that one is trying to optimize (exactly 
or approximately) and for which one is trying to prove running time or 
quality-of-approximation bounds.
\item
There exists theorems that say when each of these method ``works,'' but those 
theoretical results have assumptions that tend to be rather strong and/or 
unrealistic.
\item
Typically one has some sort of intuition and one shows that the algorithm works 
on some data in certain specialized cases.
It is often hard to generalize beyond those special cases, and so it is probably 
best to think of these as ``exploratory data analysis'' tools that construct 
data-dependent kernels.
\item
The intuition underlying these methods is often that the data ``live'' on a 
low-dimensional manifold.
Manifolds are very general structures; but in this context, it is best to think of them 
as being ``curved'' low-dimensional spaces.
The idealized story is that the processes generating the data have only a few 
degrees of freedom, that might not correspond to a linear subspace, and we want 
to reconstruct or find that low-dimensional manifold.
\item
The procedures used in constructing these data-driven kernels 
depend on relatively-simple algorithmic primitives:
shortest path computations;
least-squares approximation;
SDP optimization; and
eigenvalue decomposition.
Since these primitives are relatively simple and well-understood and since they 
can be run relatively quickly, non-linear dimensionality reduction methods that 
use them are often used to explore the data. 
\item
This approach is often of interest in semi-supervised learning, where there is a lot
of unlabeled data but very little labeled data, and where we want to use the 
unlabeled data to construct some sort of ``prior'' to help with predictions.
\end{itemize}
There are a large number of these and they have been reviewed elsewhere;
here we will only review those that will be related to the algorithmic and
statistical problems we will return to later.

\subsection{ISOMAP}

ISOMAP takes as input vectors $x_i \in \mathbb{R}^{D}, i=1,\ldots,n$, and
it gives as output vectors $y_i \in \mathbb{R}^{d}$, where $d \ll D$.
The stated goal/desire is to make near (resp. far) points stay close
(resp. far); and the idea to achieve this is to preserve geodesic distance 
along a submanifold.
The algorithm uses geodesic distances in an MDS computation:
\begin{enumerate}
\item 
\textbf{Step 1.}
Build the nearest neighbor graph, using $k$-NN or $\epsilon$-NN.
The choice here is a bit of an art.
Typically, one wants to preserve properties such as that the data are 
connected and/or that they are thought of as being a discretization of a 
submanifold.
Note that the $k$-NN here scales as $O(n^2D)$.
\item 
\textbf{Step 2.}
Look at the shortest path or geodesic distance between all pairs of points.
That is, compute geodesics.
Dijkstra's algorithm for shortest paths runs in 
$O(n^2 \log n + n^2 k )$ time.
\item 
\textbf{Step 3.}
Do Metric Multi-Dimensional scaling (MDS) based on $A$, the shortest path 
distance matrix.
The top $k$ eigenvectors of the Gram matrix then give the embedding.
They can be computed in $\approx O(n^d)$ time.
\end{enumerate}

\emph{Advantages}
\begin{itemize}
\item Runs in polynomial time.
\item There are no local minima.
\item It is non-iterative.
\item It can be used in an ``exploratory'' manner.     
\end{itemize}
\emph{Disadvantages}
\begin{itemize}
\item 
Very sensitive to the choice of $\epsilon$ and $k$, which is an art that is 
coupled in nontrivial ways with data pre-processing.
\item 
No immediate ``out of sample extension''' since so there is not obvious 
geometry to a graph, unless it is assumed about the original data.
\item 
Super-linear running time---computation with all the data points can be 
expensive, if the number of data points is large.
A solution is to choose ``landmark points,'' but for this to work one needs
to have already sampled at very high sampling density, which is often not 
realistic.
\end{itemize}
These strengths and weaknesses are not peculiar to ISOMAP; they are 
typical of other graph-based spectral dimensionality-reduction methods 
(those we turn to next as well as most others).

\subsection{Local Linear Embedding (LLE) }

For LLE, the input is vectors $x_i\in\mathbb{R}^{D},i=1,\ldots,n$; and the 
output is vectors $y_i\in\mathbb{R}^{d}$, with $d \ll D$.

\emph{Algorithm}
\begin{description}
\item [Step 1 : Construct the Adjacency Graph]  
There are two common variations:
\begin{enumerate}
\item  $\epsilon$ neighborhood 
\item  K Nearest neighbor Graph 
\end{enumerate}
Basically, this involves doing some sort of NN search; the metric of 
closeness or similarity used is based on prior knowledge; and the usual 
(implicit or explicit) working assumption is that the neighborhood in the 
graph $\approx$ the neighborhood of the underlying hypothesized manifold, 
at least in a ``local linear'' sense.
\item [Step 2 : Choosing  weights]  
That is, construct the graph.
Weights $W_{ij}$ must be chosen for all edges $ij$.
The idea is that each input point and its $k$-NN can be viewed as samples 
from a small approximately linear patch on a low-dimensional submanifold
and we can choose weights $W_{ij}$ to get small reconstruction error.
That is, weights are chosen based on the projection of each data point on 
the linear subspace generated by its neighbors. 
\begin{eqnarray*}
& \min \sum_i ||x_i - \sum_j W_{ij}x_j ||_2^2 &= \Phi(W)   \\
& \mbox{s.t. } W_{ij} = 0 \mbox{ if } (ij) \not\in E   \\
& \quad \sum_j W_{ij}=1, \forall i
\end{eqnarray*}
\item [Step 3 : Mapping to Embedded Co-ordinates] 
Compute output $y \in \mathbb{R}^d$ by solving the same equations, but now
with the $y_i$ as variables.
That is, let 
\[
\Psi(y)=\sum_{i} ||y_i - \sum_j W_{ij} y_j ||^2  ,
\]
for a fixed $W$, and then we want to solve
\begin{eqnarray*}
& \min \Psi(y)  \\
& \quad \mbox{s.t. } \sum_i y_i = 0 \\
& \quad \frac{1}{N} \sum_i y_i y_i^T = I
\end{eqnarray*}
To solve this minimization problem reduces to finding eigenvectors 
corresponding to the $d+1$ lowest eigenvalues of the the positive definite 
matrix $(I-W)'(I-W)=\Psi$.
\end{description}

Of course, the lowest eigenvalue is uninteresting for other reasons, so it is 
typically not included.
Since we are really computing an embedding, we could keep it if we wanted, 
but it would not be useful (it would assign uniform values to every node or 
values proportional to the degree to every node) to do downstream tasks 
people want to do.

\subsection{Laplacian Eigenmaps (LE)}

For LE, (which we will see has some similarities with LE), the input is 
vectors $x_i\in\mathbb{R}^{D},i=1,\ldots,n$; and the 
output is vectors $y_i\in\mathbb{R}^{d}$, with $d \ll D$.
The idea is to compute a low-dimensional representation that preserves
proximity relations (basically, a quadratic penalty on nearby points), mapping 
nearby points to nearby points, where ``nearness'' is encoded in $G$.

\emph{Algorithm}
\begin{description}
\item [Step 1 : Construct the Adjacency Graph]  
Again, there are two common variations:
\begin{enumerate}
\item  $\epsilon$-neighborhood 
\item  $k$-Nearest Neighbor Graph 
\end{enumerate}
\item [Step 2 : Choosing weights] 
We use the following rule to assign weights to neighbors:
\[ 
W_{ij}=\left\{ \begin{array}{c c}
e^{-\frac{||x_i-x_j||^2}{4t}} & \text{if vertices i \& j are connected by an edge} \\
0  & \text{otherwise} \end{array} \right.
\] 
Alternatively, we could simply set $W_{ij}=1$ for vertices $i$ and $j$ that
are connected by an edge---it should be obvious that this gives similar
results as the above rule under appropriate limits, basically since the 
exponential decay introduces a scale that is a ``soft'' version of this
``hard'' $0$-$1$ rule.
As a practical matter, there is usually a fair amount of ``cooking'' to get 
things to work, and this is one of the knobs to turn to cook things.
\item [Step 3 : Eigenmaps] 
Let 
\[
\Psi(y)=\sum_{i,j}\frac{w_{ij}||y_i-y_j||^2}{\sqrt{D_{ii} \cdot D_{jj}}}
\]
where $ D=\mbox{Diag}\{\sum_i w_{ij} : j=1 (1) n \}$, 
and we compute $ y \in \mathbb{R}^d $ such that
\begin{eqnarray*}
& \min \Psi(y) \\
& \quad \mbox{s.t. } \sum_i y_i = 0 \quad\mbox{centered} \\
& \quad \mbox{and } \frac{1}{N} \sum_i y_i y_i^T = I \quad\mbox{unit covariance}
\end{eqnarray*}
that is, s.t. that is minimized for each connected component of the graph.
This can be computed from the bottom $d+1$ eigenvectors of 
$\mathcal{L} = I-D^{-1/2}WD^{-1/2}$ , after dropping the bottom eigenvector 
(for the reason mentioned above).
\end{description}

LE has close connections to analysis on manifold, and understanding it will 
shed light on when it is appropriate to use and what its limitations are.
\begin{itemize}
\item
Laplacian in $\mathbb{R}^{d}$: $\Delta f = -\sum_i \frac{a\partial^2 f  }{ \partial x_i^2 }$
\item
Manifold Laplacian: change is measured along tangent space of the manifold.
\end{itemize}
The weighted graph $\approx$ discretized representation of the manifold.
There are a number of analogies, \emph{e.g.}, Stokes Theorem (which 
classically is a statement about the integration of differential forms which
generalizes popular theorems from vector calculus about integrating over
a boundary versus integration inside the region, and which thus generalizes
the fundamental theorem of calculus):
\begin{itemize}
\item
Manifold: $\int_{M} ||\nabla f||^2 = \int_{M} f \Delta f$
\item
Graph: $\sum_{ij} (f_i- f_j)^2 = f^T L f$
\end{itemize}
An extension of LE to so-called Diffusion Maps (which we will get to next time) 
will provide additional insight on these connections.

Note that the Laplacian is like a derivative, and so minimizing it will be 
something like minimizing the norm of a derivative.

\subsection{Interpretation as data-dependent kernels}

As we mentioned, these procedures can be viewed as constructing 
data-dependent kernels.
There are a number of technical issues, mostly having to do with the 
discrete-to-continuous transition, that we won't get into in this brief discussion.
This perspective provides light on why they work, when they work, and when 
they might not be expected to work.

\begin{itemize}
\item
ISOMAP.
Recall that for MDS, we have
$\delta_{ij}^{2} = (x_i-x_j)^T(x_i-x_j)=$ dissimilarities.
Then $A$ is a matrix s.t. $A_{ij} = -\delta_{ij}$, and 
$B=HAH$, with $H=I_n-\frac{1}{n}11^T$, a ``centering'' or ``projection'' 
matrix.
So, if the kernel $K(x_i,x_j)$ is stationary, i.e., if it is a function of 
$||x_i-x_j||^2 = \delta_{ij}^{2}$, as it is in the above construction,
then $K(x_i,x_j)=r(\delta_{ij})$, for some $r$ that scales s.t. $r(0)=1$.
Then $\tilde{\delta}_{ij}^{2}$ is the Euclidean distance in feature space, 
and if $A$ s.t. $A_{ij} = r(\delta_{ij})-1$, then $A=K-11^T$.
The ``centering matrix'' $H$ annihilates $11^T$, so $HAH=HKH$.
(See the Williams paper.)
So, 
\[
K_{ISOMAP} = HAH = H \mathcal{D} H 
                 = (I-11^T) \mathcal{D} (I-11^T)
\]
where $\mathcal{D}$ is the squared geodesic distance matrix.
\item
LE.
Recall that LE minimizes 
$\psi^T L \psi = \frac{1}{2}\sum_{ij} (psi_i-\psi_j)^2 W_{ij}$, 
and doing this involved computing eigenvectors of $L$ or $\mathcal{L}$, 
depending on the construction.
The point is that $L$ has close connections to diffusions on a graph---think
of it in terms of a continuous time Markov chain:
\[
\frac{\partial \psi(t)}{\partial t} = - L \psi(t)   .
\]
The solution is a Green's function or heat kernel, related to the matrix 
exponential of $L$:
\[
K_t = \exp(-Lt) = \sum_{\xi} \phi_{\xi}\phi_{\xi}^{T} e^{-\lambda_{\xi}t} ,
\]
where $\phi_{\xi}$, $\lambda_{\xi}$ are the eigenvectors/eigenvalues 
of $L$.
Then, $\psi(t) = K_t \psi(0)$ is the general solution, and
under assumptions one can show that
\[
K_L = \frac{1}{2}L^{+}
\]
is related to the ``commute time'' distance of diffusion:
\[
C_{ij} \sim L_{ii}^{+} + L_{jj}^{+} - L_{ij}^{+} = L_{ji}^{+}   .
\]
For the difference between the commute time in a continuous time Markov 
chain and the geodesic distance on the graph, think of the dumbbell 
example; we will get to this in more detail next time.
\item
LLE.
Recall that this says that we are approximating each point as a linear 
combination of neighbors.
Let $(W_{n})_{ij}, i,j\in[n]$ be the weight of a point $x_j$ in the 
expansion of $x_i$; then one can show that
\[
K_n(x_i,x_j) = \left( (I-W_n)^T(I-W_n) \right)_{ij}
\]
is PD on the domain $\mathcal{X}:x_i,\ldots,x_n$.
Also, one can show that if $\lambda$ is the largest eigenvalue of 
$(I-W)^T(I-W)=M$, then
\[
K_{LLE} = \left( (\lambda-1)I + W^T + W - W^TW \right)
\]
is a PSD matrix, and thus a kernel.
Note also that, under assumptions, we can view $M=(I-W^T(I-W)$ as 
$\mathcal{L}^2$.
\end{itemize}

\subsection{Connection to random walks on the graph: more on LE and diffusions}

Recall that, give a data set consisting of vectors, we can construct a graph 
$G=(V,E)$ in one of several ways.
Given that graph, consider the problem of mapping $G$ to a \emph{line} s.t.
connected points stay as close together as possible.
Let $\vec{y}=(y_1,\ldots,y_n)^T$ be the map, and by a ``good'' map we will
mean one that minimized $\sum_{ij} (y_i-y_j)^2W_{ij}$ under appropriate 
constraints, \emph{i.e.}, it will incur a big penalty if $W_{ij}$ is large
and $y_i$ and $y_j$ are far apart.
This has important connections with diffusions and random walks that we now discuss.

We start with the following claim:
\begin{claim}
$\frac{1}{2}\sum_{i,j} w_{ij} (y_i-y_j)^2=y^{T}L y$
\end{claim}
\begin{proof}
Recall that $W_{ij}$ is symmetric and that $D_{ii}= \sum_j W_{ij}$.
Then:
\begin{align*} \frac{1}{2}\sum_{i,j} w_{ij} (y_i-y_j)^2&=\frac{1}{2}\sum_{i,j} w_{ij}y_i^2+w_{ij} y_j^2-2w_{ij}y_iy_j\\
                                 &=\sum_{i} D_{ii} y_i^2-\sum_{ij} w_{ij} y_iy_j\\
                                 &=y^T L y \end{align*}
where $L=D-W$, W is the weight matrix and L is Laplacian, a
symmetric, positive semidefinite matrix that can be thought of as an
operator on functions defined on vertices of G.
\end{proof}

The following, which we have seen before, is a corollary:
\begin{claim}
$L$ is SPSD.
\end{claim}
\begin{proof}
Immediate, since $W_{ij}>0$.
\end{proof}

Recall that the solution to 
\begin{eqnarray*}
\arg \min y^T L y\\
 \textrm{s.t.} y^T Dy=1   ,
\end{eqnarray*} 
where $D$ gives a measure of the importance of vertices,
turns out to be given by solving a generalized eigenvector problem
$$
Ly=\lambda Dy   ,
$$ 
\emph{i.e.}, computing the bottom eigenvectors of that eigenproblem.
Actually, we usually solve, 
\begin{eqnarray*}
\arg \min y^T L y\\
 \textrm{s.t.} \quad yD\vec{1}=0\\
     y^T Dy=1   ,
\end{eqnarray*} 
the reason being that since $\vec{1}=(1,\ldots,1)$ is an eigenvector with 
eigenvalue $0$, it is typically removed.
As we saw above, the condition $y^T D\vec{1} = 0$ can be interpreted as 
removing a ``translation invariance'' in $y$ and so is removed.
Alternatively, it is uninteresting if these coordinates are going to be used simple 
to provide an embedding for other downstream tasks.
Also, the condition $y^T D y=1$ removes an arbitrary scaling factor in the
embedding, which is a more benign thing to do if we are using the coordinates 
as an embedding.

From the previous discussion on Cheeger's Inequality, we know that putting the
graph on a line (and, in that setting, sweeping to find a good partition) can reveal
when the graph doesn't ``look like'' a line graph.
Among other things, the distances on the line may be distorted a lot, relative to 
the original distances in the graph (even if they are preserved ``on average'').

Thus, more generally, the goal is to embed a graph in $\mathbb{R}^{n}$ s.t. 
distances are ``meaningful,'' where meaningful might mean the same as or very 
similar to distances in the graph.
Let $G=(V,E,W)$.
We know that if $W$ is symmetric, $W=W^T$, and point-wise positive, 
$W(x,y)\geq0$, then we can interpret pairwise similarities as probability mass 
transitions in a Markov chain on a graph.
Let $d(x) = \sum_z W(x,z) = $ the degree of node $x$.
Then let $P$ be the $n \times n$ matrix with entries
\[
P(x,w) = \frac{W(x,y)}{d(x)}
\]
which is the transition probability of going from $x$ to $y$ in one step, 
which equals the first order neighborhood of the graph.
Then, $P^{t}$ is higher order neighborhoods of the graph; this is sometimes 
interpreted as $\approx$ the ``intrinsic geometry'' of an underlying hypothesized 
manifold.

If the graph is connected, as it usually is due to data preprocessing, 
then
\[
\lim_{t \rightarrow \infty} P^{t}(x,y) = \phi_{0}(y)  ,
\]
where $\phi_{0}$ is the unique stationary distribution
\[
\phi_{0} = \frac{d(x)}{\sum_z d(z)}
\]
of the associated Markov chain.
(Note The Markov chain is reversible, as detailed-balance is satisfied:
$\phi_{0}(x)P_{1}(x,y) = \phi_{0}(y)P_1(y,x)$.)

Thus, graph $G$ defines a random walk.  
For some node $i$ the probability going from $i$ to $j$ is 
$P'_{ij}=\dfrac{w_{ij}}{d_i}$, where $d_i=\sum_{j} w_{ij}$. 
Consider if you are at node $i$ and you are move from $i$ in the following 
way:
$$
\begin{cases}
\textrm{move to a neighbor chosen u.a.r, w.p.=  } \dfrac{1}{d_i} & \textrm{w.p.}\frac{1}{2} \\
\textrm{stay at node i } & \textrm{w.p.} \frac{1}{2}
\end{cases}
$$
Then the transition matrix 
$P \in {\mathbb{R}}^{n \times m} = \dfrac{1}{2}I + \dfrac{1}{2} P'$. 
This is the so-called ``lazy'' random walk.

\textbf{Fact:} 
If $G$ is connected, then for any measure initial $v$ on the vertex, 
$\lim_{t\rightarrow \infty} P^{t} v=\dfrac{d_i}{\sum_{j} d_j} = \phi_0(i)$. 
This $\phi$ converges to the stationary distribution. 
$P$ is related to the normalized Laplacian.

If we look at the pre-asymptotic state, for $1 \ll t \ll \infty$, we could 
define similarity between vertex $x$ and $z$ in terms of the similarity 
between two density $P_t(x,\cdot)$ and $P_t(z,\cdot)$.
That is,
\begin{itemize}
\item
For $t \in (0,\infty)$, we want a metric between nodes s.t. $x$ and $z$ are 
close if $P_{t}(x,\cdot)$ and $P_{t}(z,\cdot)$ are close.
\item
There are various notions of closeness:
\emph{e.g.}, $\ell_1$ norm (see the Szummer and Jaakkola reference);
Kullback-Leibler divergences; $\ell_2$ distances; and so on.
\item
The $\ell_2$ distance is what we will focus on here (although we might 
revisit some of the others later). 
\end{itemize}

In this setting, the $\ell_2$ distance is defined as
\begin{eqnarray}
\nonumber
D_t^2(x,z) &=& ||P_t(x,\cdot)-P_t(z,\cdot)||^2_{1/\phi_0}  \\
           &=& \sum_{y} \frac{(P_t(x,y)-P_t(z,y))^2}{\phi_0(y)}   .
\label{eqn:dist1_xxx}
\end{eqnarray}
(So, in particular, the weights $\frac{1}{\phi_{0}(y)}$ penalize discrepancies 
on domains of low density more than high-density domains.)

This notion of distance between points will be small if they are connected
by many path.
The intuition is that if there are many paths, then there is some sort of 
geometry, redundancy, and we can do inference like with low-rank 
approximations.
(BTW, it is worth thinking about the difference between minimum spanning trees 
and random spanning trees, or degree and spectral measures of ranking like 
eigenvector centrality.)


%% file: lect15.tex
\section{%
(03/12/2015): 
Some Practical Considerations (4 of 4):
More on diffusions and semi-supervised graph construction}

Reading for today.
\begin{compactitem}
\item
``Transductive learning via spectral graph partitioning,'' in ICML, by Joachims
\item
``Semi-Supervised Learning Using Gaussian Fields and Harmonic Functions,'' in ICML, by Zhu, Ghahramani, and Lafferty 
\item
``Learning with local and global consistency,'' in NIPS, by Zhou et al.
\end{compactitem}

Today, we will continue talking about the connection between non-linear dimensionality reduction methods and diffusions and random walks.
We will also describe several methods for constructing graphs that are ``semi-supervised,'' in the sense that that there are labels associated with some of the data points, and we will use these labels in the process of constructing the graph.  
These will give rise to similar expressions that we saw with the unsupervised methods, although there will be some important differences.

\subsection{Introduction to diffusion-based distances in graph construction}

Recall from last time that we have our graph $G$, and if we run a random walk to the asymptotic state, then we converge to the leading non-trivial eigenvector.
Here, we are interested in looking at the 
pre-asymptotic state, i.e., at a time $t$ such that $1 \ll t \ll \infty$, and we want to define similarity between vertex $x$ and $z$, i.e., a metric between nodes $x$ and $z$, such that $x$ and $z$ are close if $P_{t}(x,\cdot)$ and $P_{t}(z,\cdot)$ are close.
Although there are other distance notions one could use, we will work with the $\ell_2$ distance.
In this setting, the $\ell_2$ distance is defined as
\begin{eqnarray}
\nonumber
D_t^2(x,z) &=& ||P_t(x,\cdot)-P_t(z,\cdot)||^2_{1/\phi_0}  \\
           &=& \sum_{y} \frac{(P_t(x,y)-P_t(z,y))^2}{\phi_0(y)}   .
\label{eqn:dist1}
\end{eqnarray}

Suppose the transition matrix $P$ has $q$ left and right eigenvectors
and eigenvalues $|\lambda_0|\geq|\lambda_1|\geq\ldots\geq|\lambda_n|\geq0$ 
s.t.
\begin{eqnarray*}
\phi_j^TP &=& \lambda_j\phi_j^T  \\
P\psi_j &=& \lambda \psi_j   ,
\end{eqnarray*}
where we note that $\lambda_0=1$ and $\psi_0=\vec{1}$.
Normalize s.t.
$\phi_k\psi_{\ell} = \delta_{k\ell}$
so that 
\begin{eqnarray*}
\| \phi_{\ell} \|^2_{1/\phi_0} &=& \sum_x \frac{ \phi_{\ell}^2(x) }{ \phi_0(x) } = 1 \\
\| \psi_{\ell} \|^2_{\phi_0} &=& \sum_x \psi_{\ell}^2(x) \phi_0(x) = 1 ,
\end{eqnarray*}
i.e., normalize the left (resp. right) eigenvector of $P$ w.r.t. $1/\phi_0$ 
(resp. $\phi_0$).
If $P_t(x,y)$ is the kernel of the $t^{th}$ iterate of $P$, then we have the 
following spectral decomposition: 
\begin{equation}
P_t(x,y)=\sum_j \lambda_j^t \psi_j(x)\phi_j(y)   .  
\label{eqn:spectdecomp1}
\end{equation}
(This is essentially a weighted PCA/SVD of $P^t$, where as usual the first $k$ 
terms provide the ``best'' rank-$k$ approximation, where the ``best'' is w.r.t.
$\|A\|^2 = \sum_x \sum_y \phi_0(x) A(x,y)^2 \phi_0(y)$.)
If we insert Eqn.~(\ref{eqn:spectdecomp1}) into Eqn.~(\ref{eqn:dist1}), then one 
can show that the $L_2$ distance is: 
\begin{eqnarray}
\label{eqn:diffdist1}
D_t^2(x,z) &=& \sum_{j=1}^{n-1} \lambda_j^{2t} (\psi_j(x)-\psi_j(z))^2   \\
\label{eqn:diffdist2}
           &\approx& \sum_{j=1}^{k} \lambda_j^{2t} (\psi_j(x)-\psi_j(z))^2   \\
\label{eqn:diffdist3}
           &\approx& \sum_{j=1}^{k}  (\psi_j(x)-\psi_j(z))^2  .
\end{eqnarray}
Eqn.~(\ref{eqn:diffdist1}) says that this provides a legitimate distance, and 
establishing this is on HW2.
The approximation of Eqn.~(\ref{eqn:diffdist2}) holds if the eigenvalues 
decay quickly, and  the approximation of Eqn.~(\ref{eqn:diffdist3}) holds if the 
large eigenvalues are all roughly the same size.
This latter expression is what is provided by LE.

\subsection{More on diffusion-based distances in graph construction}

Recall from last time that LE can be extended to include all of the eigenvectors
and also to weighting the embedding by the eigenvalues.
This gives rise to an embedding that is based on diffusions, sometimes it is called 
Diffusion Maps, that defines a distance in the Euclidean space that is related to 
a diffusion-based distance on the graph.
In particular, once we choose $k$ in some way, then here is the picture of the 
Diffusion Map:
\begin{center}
Map $\Psi_t: X \rightarrow \left(\begin{array}{c}
\lambda_1^t \psi_1(x)\\
\vdots\\
\lambda_k^t\psi_k(x)
\end{array}\right)$
\end{center}
and so
\begin{eqnarray*}
D_t^2(x,z) &\approx& \sum_{j=1}^{k} \lambda_j^{2t} (\psi_j(x)-\psi_j(z))^2  \\
           &\approx& ||\psi_t(x)-\psi_t(z)||^2  .
\end{eqnarray*}
This is a low-dimensional embedding, where the Euclidean distance in the 
embedded space corresponds to the ``diffusion distance'' in the original 
graph.
Here is the picture of the relationships:
\begin{center}

\[\begin{array}{ccc}
G &\leftrightarrow &R^n\\
| &\quad &| \\
\textrm{diffusion} & \leftrightarrow &||\cdot||_2
\end{array} \]
\end{center}

\textbf{Fact:} This defines a distance. If you think of a \texttt{diffusion} 
notion of distance in a graph G, this identically equals the Euclidean 
distance $||\cdot ||_2$ between $\Psi(i)$ and $\Psi(j)$. 
The diffusion notion of distance is related to the commute time between 
node $i$ and node $j$.
We will describe this next time when we talk about resistor networks.

In particular, Laplacian Eigenmaps chooses $k=k^*$, for some fixed $k^*$ 
and sets $t=0$. 
Under certain nice limits, this is close to the Laplace-Beltrami operator on 
the hypothesized manifold. 
But, more generally, these results will hold even if the graph is not drawn 
from such a nice setting.
\begin{itemize}
\item
None of this discussion assumes that the original vector data come from 
low-dimensional manifolds, although one does get a nice interpretation in 
that case.
\item
It is fair to ask what happens with spectral decomposition when there is 
not manifold, e.g., a star graph or a constant-degree expander.
Many of the ideas wetland about earlier in the semester would be relevant
in this case.
\end{itemize}

\subsection{A simple result connecting random walks to NCUT/conductance}

There are connections with graph partitioning, but understanding the 
connection with random walkers opens up a lot of other possibilities.
Here, we describe one particularly simple result: that a random walker, 
starting in the stationary distribution, goes between a set and its complement
with probability that depends on the NCUT of that set.
Although this result is simple, understanding it will be good since it will open 
up several other possibilities:
what if one doesn't run a random walk to the asymptotic state; 
what if one runs a random walk just a few steps starting from an arbitrary
distribution; 
when does the random walk provide regularization; 
and so on?

Here, we provide an interpretation for the NCUT objective (and thus related 
to normalized spectral clustering as well as conductance): when minimizing 
NCUT, we are looking for a suit through the graph s.t. the random walk 
seldom transitions from $A$ to $\bar{A}$, where $A \subset V$.
(This result says that conductance/NCUT is not actually looking for good 
cuts/partitions, but instead should be interpreted as providing bottlenecks
to diffusion.)

\begin{lemma}
Let $P=D^{-1}W$ be a random walk transition matrix.
Let $ \left( X_t \right)_{t \in \mathbb{Z}^{+}}$ be the random walk starting at 
$X_0 = \pi$, i.e., starting at the stationary distribution.
For disjoint subsets $A,B \subset V$, let 
$\mathbb{P}\left[ B | A \right] = \mathbb{P}\left[ X_1 \in B | X_0 \in A \right]$.
Then, 
\[
NCUT\left(A,\bar{A}\right) 
   = \mathbb{P}\left[ \bar{A} | A \right] + \mathbb{P}\left[ A | \bar{A} \right]  .
\]
\end{lemma}
\begin{proof}
Observe that 
\begin{eqnarray*}
\mathbb{P}\left[ X_0 \in A \mbox{ and } X_1 \in B \right]
   &=& \sum_{i \in A,j \in B} \mathbb{P}\left[ X_0 = i \mbox{ and } X_1 = j \right] \\
   &=& \sum_{i \in A,j \in B} \pi_i P_{ij}  \\
   &=& \sum_{i \in A,j \in B} \frac{ d_i }{ \mbox{Vol}(V) } \frac{ W_{ij} }{ d_i } \\
   &=& \frac{ 1 }{ \mbox{Vol}(V) } \sum_{i \in A,j \in B} W_{ij} 
\end{eqnarray*}
From this, we have that
\begin{eqnarray*}
\mathbb{P}\left[ X_1 \in B | X_0 \in A \right]
   &=& \frac{ \mathbb{P}\left[ X_0 \in A \mbox{ and } X_1 \in B \right] }{ \mathbb{P}\left[ X_0 \in A \right] } \\
   &=& \left( \frac{1}{\mbox{Vol}(V) } \sum_{i \in A, j\in B } W_{ij} \right) \left( \frac{ \mbox{Vol}(A) }{ \mbox{Vol}(V) } \right)^{-1} \\
   &=& \frac{ 1 }{ \mbox{Vol}(V) } \sum_{i \in A,j \in B} W_{ij} 
\end{eqnarray*}
The lemma follows from this and the definition of NCUT.
\end{proof}

\subsection{Overview of semi-supervised methods for graph construction}

Above, we constructed graphs/Laplacians in an unsupervised manner.
That is, there were just data points that were feature vectors without any classification/regression labels associated with them; and we constructed graphs from those unlabeled data points by using various NN rules and optimizing objectives that quantified the idea that nearby data points should be close or smooth in the graph.
The graphs were then used for problems such as data representation and unsupervised graph clustering that don't involve classification/regression labels, although in some cases they were also used for classification or regression problems.  

Here, we consider the situation in which some of the data have labels and we want to construct graphs to help make predictions for the unlabeled data.
In between the extremes of pure unsupervised learning and pure supervised learning, there are semi-supervised learning, transductive learning, and several other related classes of machine learning methods.
It is in this intermediate regime where using labels for graph construction is most interesting.

Before proceeding, here are a few things to note.
\begin{itemize}
\item
If unlabeled data and labeled data come from the same distribution, then there is no difference between unlabeled data and test data.
Thus, this transductive setup amounts to being able to see the test data (but not the labels) before doing training.
(Basically, the transductive learner can look at all of the data, including the test data with labels and as much training data as desired, to structure the hypothesis space.)
\item
People have  used labels to augment the graph construction process in an explicit manner.
Below, we will describe an example of how this is done implicitly in several cases.
\item
We will see that linear equations of a certain form that involve Laplacian matrices will arise.
This form is rather natural, and it has strong similarities with what we saw previously in the unsupervised setting.
In addition, it also has strong similarities with local spectral methods and locally-biased spectral methods that we will get to in a few weeks.
\end{itemize}

We will consider three approaches (those of Joachims, ZGL, and Zhao et al.).
Each of these general approaches considers constructing an augmented graph, with extra nodes, $s$ and $t$, that connect to the nodes that have labels.
Each can be understood within the following framework.
Let $L=D-A = B^TCB$, where $B$ is the unweighted edge-incidence matrix.
Recall, then, that the $s,t$-mincut problem is: 
\[
\min_{x \mbox{ s.t. } x_s=1,x_t=0} \|Bx \|_{C,1} 
   = \min_{x \mbox{ s.t. } x_s=1,x_t=0} \sum_{(u,v) \in E} C_{(u,v)}|x_u-x_v |  . 
\]
The $\ell_2$ minorant of this $s,t$-mincut problem is:
\[
\min_{x \mbox{ s.t. } x_s=1,x_t=0} \|Bx \|_{C,2} 
   = \min_{x \mbox{ s.t. } s=1,x_t=0} \left( \sum_{(u,v) \in E} C_{(u,v)} |x_u-x_v |^2 \right)^{1/2}  . 
\] 
This latter problem s equivalent to:
\[
\min_{x \mbox{ s.t. } x_s=1,x_t=0} \frac{1}{2}  \|Bx \|_{C,2}^{2} 
   = \min_{x \mbox{ s.t. } x_s=1,x_t=0} \sum_{(u,v) \in E} C_{(u,v)} |x_u-x_v |^2   
   = \min_{x \mbox{ s.t. } x_s=1,x_t=0} x^T L x   . 
\] 
The methods will construct Laplacian-based expressions by considering various types of $s$-$t$ mincut problems and then relaxing to the associated $\ell_2$ minorant.

\subsection{Three examples of semi-supervised graph construction methods}

The Joachims paper does a bunch of things, e.g., several engineering heuristics that are difficult to relate to a general setting but that no doubt help the results in practice, but the following is essentially what he does.
Given a graph and labels, he wants to find a vector $\vec{y}$ s.t. it satisfies the bicriteria:
\begin{enumerate}
\item
it minimizes $y^TLy$, and
\item
it has values
$$
\begin{cases}
 1 & \textrm{ for nodes in class $j$} \\
-1 & \textrm{ for nodes not in class $j$}
\end{cases}
$$
\end{enumerate}
What he does is essentially the following.  
Given $G=(V,E)$, he adds extra nodes to the node set:
$$
\begin{cases}
s & \textrm{ with the current class} \\
t & \textrm{ with the other class}
\end{cases}
$$
Then, there is the issue about what to add as weights when those nodes connect to nodes in $V$.
The two obvious weights are $1$ (i.e., uniform) or equal to the degree.
Of course, other choices are possible, and that is a parameter one could play with.
Here, we will consider the uniform case.
Motivated by problem with mincut (basically, what we described before, where it tends to cut off small pieces), Joachims instead considers NCUT and so arrives at the following problem: 
\[
 Y = \left( D_S + L \right) ^{-1} S ,
\]
where $D_S$ is a diagonal matrix with the row sums of $S$ on the diagonal, and $S$ is a vector corresponding to the class connected to the node with label $s$.
Note that since most of the rows of $S$ equal zero (since they are unlabeled nodes) and the rest have only 1 nonzero, $D_S$ is a diagonal indicator vector that is sparse.

ZGL is similar to Joachims, except that they strictly enforce labeling on the labeled samples.
Their setting is that they are interpreting this in terms of a Gaussian random field on the graph (which is like a NN method, except that nearest labels are captures in terms of random vectors on $G$).
They provide hard constraints on class labels.
In particular, they want to solve
\begin{eqnarray*}
\min_y & & \frac{1}{2} y^T L y \\
\mbox{s.t.} & & 
y_i = 
\left\{ \begin{array}{l l}
                    1 & \quad \text{node $i$ is labeled in class $j$}\\
                    0 & \quad \text{node $i$ is labeled in another class}\\
                    \mbox{free} & \quad \text{otherwise}
                 \end{array} 
         \right. 
\end{eqnarray*}
This essentially involves constructing a new graph from $G$ ,where the labels involve adding extra nodes, $s$ and $t$, where $s$ links to the current class and $t$ links to the other class.
For ZGL, weights $=\infty$ between $s$ and the current class as well as between $t$ and the other class.

Zhao, et al. extend this to account explicitly for the bicriteria that:
\begin{enumerate}
\item
\label{point:zhao1}
nearby points should have the same labels; and 
\item
\label{point:zhao2}
points on the same structure (clusters of manifold) have the same label.
\end{enumerate}
Note that Point~\ref{point:zhao1} is a ``local'' condition, in that it depends on nearby data points, while Point~\ref{point:zhao2} is a ``global'' condition, in that it depends on large-scale properties of clusters.
The point out that different semi-supervised methods take into account these two different properties and weight them in different ways.
Zhao et al. try to take both into account by ``iteratively'' spreading ZGL to get a classification function that has local and global consistency properties and is sufficiently smooth with respect to labeled and unlabeled data.

In more detail, here is the Zhao et al. algorithm.
\begin{enumerate}
\item
Form the affinity matrix $W$ from an rbf kernel.
\item
Construct $\tilde{W} = D^{-1/2}WD^{-1/2}$.
\item
Iterate $Y(t+1) = \alpha \tilde{W} Y + \left(1-\alpha\right)S $, where $\alpha \in (0,1)$ is a parameter, and where $S$ is a class label vector.
\item
Let $Y^{*}$ be the limit, and output it.
\end{enumerate} 
Note that $\tilde{Y}^{*} = \left(1-\alpha\right)\left(1-\alpha\tilde{W}\right)S$.
Alternatively, 
\[
Y = \left( L + \alpha D \right)^{-1} S = \left( D - \beta A \right)^{-1} S ,
\]
where $Y$ is the prediction vector, and where $S$ is the matrix of labels.
Here, we have chosen/defined $\alpha = \frac{1-\beta}{\beta}$ to relate the two forms.

Then the ``mincut graph'' has $G=(V,E)$, with nodes $s$ and $t$ such that
$$
\begin{cases}
s & \textrm{ connects to each sample labeled with class $j$ with weight $\alpha$} \\
t & \textrm{ connects to \emph{all} nodes in $G$ (except $s$) with weight $\alpha(d_i-s_i)$}
\end{cases}
$$
The $\ell_2$ minorant of this mincut graph has:
\[
\min_{x \mbox{ s.t. } x_s=1,x_t=0} \frac{1}{2}  \|Bx \|_{C,2}^{2}  
   = \min_{x \mbox{ s.t. } x_s=1,x_t=0} 
      \frac{1}{2}
      \left( \begin{array}{c} 1 \\ y \\ 0 \end{array} \right)^{T}
      \left( \begin{array}{ccc} \alpha e^Ts & -\alpha s^T & 0 \\
                         -\alpha s & \alpha D + L & \alpha(d-s) \\
                         0 & -\alpha(d-s)^T & \alpha e^T (d-s) \\
       \end{array}
\right)
\left( \begin{array}{c} 1 \\ y \\ 0 \end{array} \right)  .
\] 
In this case, $y$ solves $\left(\alpha D + L\right) y = \alpha s$.

We will see an equation of this form below in a  somewhat different context.
But for the semi-supervised learning context, we can interpret this as a class-specific smoothness constraint.
To do so, define
\[
A(Y) = \frac{1}{2}\left( \sum_{ij=1}^{n} W_{ij} \| \frac{1}{\sqrt{D_{ii}}}Y_i - \frac{1}{\sqrt{D_{jj}}} Y_j\|^2
                                    - \mu \sum_{j=1}^{n} \| Y_i - S_i \|^2 \right)
\]
to be the ``cost'' associated with the prediction $Y$.
The first term is a ``smoothness constraint,'' which is the sum of local variations, and it reflects that a good classification should not change much between nearby points.
The second term is a ``fitting constraint,'' and it says that a good classification function should not change much from the initial assignment provided by the labeled data.
The parameter $\mu$ givers the interaction between the two terms.
To see how this gives rise to the previous expression, observe that
\[
\frac{\partial Q(Y)}{\partial Y} |_{Y=Y^*}
   = Y^* - WY^* + \mu \left( Y^* - S \right) = 0 , 
\]
from which it follows that 
\[
Y^* - \frac{1}{1+\mu} WY^* - \frac{\mu}{1+\mu} S = 0  .
\]
If we define $\alpha = \frac{1}{1+\mu}$ and $\beta = \frac{\mu}{1+\mu}$, so that
$\alpha+\beta=1$, then we have that
\[
\left(I-\alpha W \right) Y^* = \beta S  ,
\]
and thus that
\[
Y^* = \beta \left( 1-\alpha W \right)^{-1} S  .
\]

%% file: lect16.tex
\section{%
(03/17/2015): 
Modeling graphs with electrical networks}

Reading for today.
\begin{compactitem}
\item
``Random Walks and Electric Networks,'' in arXiv, by Doyle and Snell
\end{compactitem}

\subsection{Electrical network approach to graphs}

So far, we have been adopting the usual approach to spectral graph theory:
understand graphs via the eigenvectors and eigenvalues of associated 
matrices.
For example, given a graph $G=(V,E)$, we defined an adjacency matrix $A$ 
and considered the eigensystem $A v = \lambda v$, and we also defined the
Laplacian matrix $L=D-A$ and considered the Laplacian quadratic form 
$x^TLx - \sum_{(ij) \in E} (x_i-x_j)^2$.
There are other ways to think about spectral graph methods that, while
related, are different in important ways.
In particular, one can draw from \emph{physical intuition} and define 
physical-based models from the graph $G$, and one can also consider more
directly vectors that are obtained from various \emph{diffusions and random 
walks} on $G$.
We will do the former today, and we will do the latter next time.

\subsection{A physical model for a graph}

In many physical systems, one has the idea that there is an equilibrium state and that
the system goes back to that equilibrium state when disturbed.
When the system is very near equilibrium, the force pushing it back to the equilibrium 
state is quadratic in the displacement from equilibrium, one can often define a 
potential energy that in linear in the displacement from equilibrium, and then the 
equilibrium state is the minimum of that potential energy function.

In this context, let's think about the edges of a graph $G=(V,E)$ as physical 
``springs,'' in which case the weights on the edges correspond to a spring constant 
$k$.
Then, the force, as a function of the displacement $x$  from equilibrium, is $F(x)=kx$, 
and the corresponding potential energy is $U(x) = \frac{1}{2}kx^2$.
In this case, i.e., if the graph is viewed as a spring network, then if we nail down 
some of the vertices and then let the rest settle to an equilibrium position, then we 
are interested in finding the minimum of the potential energy 
\[
\sum_{(ij)\in E} \left(x_i-x_j\right)^{2} = x^TLx   , 
\]
subject to the constraints on the nodes we have nailed down.
In this case, the energy is minimized when the non-fixed vertices have values equal to
\[
x_i = \frac{1}{d_i} \sum_{(ij) \in E} x_j  ,
\]
i.e., when the value on any node equals is the average of the values on its neighbors.
(This is the so-called \emph{harmonic property} which is very important, e.g., in 
harmonic analysis.)

As we have mentioned previously and will go into in more detail below, eigenvectors 
can be unstable things, and having some physical intuition can only help; so let's go 
a little deeper into these connections.

First, recall that the \emph{standard/weighted geodesic graph metric} defines a distance $d(a,b)$ between vertices $a$ and $b$ as the length of the minimum-length path, i.e., number of edges or the sum of weights over edges, on the minimum-length path connecting $a$ and $b$.  
(This is the ``usual'' notion of distance/metric on the nodes of a graph, but it will be different than distances/metrics implied by spectral methods and by what we will discuss today.)

Here, we will \emph{model a graph $G=(V,E)$ as an electrical circuit}.  
(By this, we mean a circuit that arises in electromagnetism and electrical engineering.)
This will allow us to use physical analogues, and it will allow us to get more robust proofs for several results.  
In addition, it allow us to define another notion of distance that is closer to diffusions.

As background, here are some physical facts from electromagnetism that we would like to mimic and that we would like our model to incorporate.
\begin{itemize}
\item
A basic \emph{direct current electrical circuit} consists of a \emph{battery} and one or more \emph{circuit elements} connected by wires.
Although there are other circuit elements that are possible, here we will only consider the use of resistors.
A battery consists of two distinct vertices, call them $\{a,b\}$, one of which is the source, the other of which is the sink.
(Although we use the same terms, ``source'' and ``sink,'' as we used with flow-based methods, the sources and since here will obey different rules.)
A \emph{resistor} between two points $a$ and $b$, i.e., between two nodes in $G$, has an associated (undirected and symmetric) quantity $r_{ab}$ called a \emph{resistance} (and an associated conductance $c_{ab} = \frac{1}{r_{ab}}$).
Also, there is a \emph{current} $Y_{ab}$ and a \emph{potential difference} $V_{ab}$ between nodes $a$ and $b$.
\end{itemize}

Initially, we can define the resistance between two nodes that are connected by an edge to depend (typically inversely) on the weight of that edge, but we want to extend the idea of resistance to a resistance between any two nodes.
To do so, an important notion is that of \emph{effective resistance}, which is the following.
Given a collection of resistors between nodes $a$ and $b$, they can be replaced with a single effective resistor with some other resistance.
Here is how the value of that effective resistance is determined.
\begin{itemize}
\item
If $a$ and $b$ have a node $c$ between them, i.e., the resistors are in series, and there are resistances $r_1=R_{ac}$ and $r_2 = R_{cb}$, then the effective resistance between $a$ and $b$ is given by $R_{ab}=r_1+r_2$.
\item
If $a$ and $b$ have no nodes between them but they are connected by two edges with resistances $r_1$ and $r_2$, i.e., the resistors are in parallel, then the effective resistance between $a$ and $b$ is given by $R_{ab} = \frac{1}{\frac{1}{r_1}+\frac{1}{r_2}}$.
\item
These rules can be applied recursively.
\end{itemize}
From this it should be clear that the number of paths as well as their lengths contribute to the effective resistance.
In particular, having $k$ parallel edges/paths leads to an effective resistance that is decreased by $\frac{1}{k}$; and adding the first additional edge between two nodes has a big impact on the effective resistance, but subsequent edges have less of an effect.
Note that this is vaguely similar to the way diffusions and random walks behave, and distances/metrics they might imply, as opposed to geodesic paths/distances defined above, but there is no formal connection (yet!).

Let a voltage source be connected between vertices $a$ and $b$, and let $Y>0$ be the net current out of source $a$ and into course $b$. 
Here we define two basic rules that our resistor networks must obey.

\begin{definition}
\label{def:kirchoff-current}
The \emph{Kirchhoff current law} states that the current $Y_{ij}$ between vertices $i$ and $j$ (where $Y_{ij}=-Y_{ji}$) satisfies 
\[
\sum_{j \in N(i)} Y_{ij} = 
\left\{ \begin{array}{l l}
                     Y & \quad i=a  \\
                    -Y & \quad i=b  \\
                    \mbox{free} & \quad \text{otherwise}
                 \end{array}
         \right. ,
\]
where $N(i)$ refers to the nodes that are neighbors of node $i$.
\end{definition}

\begin{definition}
\label{def:kirchoff-potential}
The \emph{Kirchhoff circuit/potential law} states that for every cycle $C$ in the network, 
\[
\sum_{(ij)\in C} Y_{ij} R_{ij} = 0.
\]
\end{definition}
From Definition~\ref{def:kirchoff-potential}, it follows that there is a so-called \emph{potential function} on the vertices/nodes of the graph. 
This is known as Ohm's Law.
\begin{definition}
\label{def:ohms-law}
\emph{Ohm's Law} states that, to any vertex $i$ in the vertex set of $G$, there is an associated potential, call it $V_i$, such that for all edges $(ij) \in E$ in the graph
\[
Y_{ij}R_{ij} = V_i - V_j  .
\]
\end{definition}

Given this potential function, we can define the effective resistance between any two nodes in $G$, i.e., between two nodes that are not necessarily connected by an edge.
\begin{definition}
Given two nodes, $a$ and $b$, in $G$, the effective resistance is $R_{ab} = \frac{V_a-V_b}{Y} $.
\end{definition}

\textbf{Fact.} 
Given a graph $G$ with edge resistances $R_{ij}$, and given some source-sink pair $(a,b)$, the effective resistance exists, it is unique, and (although we have defined it in terms of a current) it does \emph{not} depend on the net current.

\subsection{Some properties of resistor networks}

Although we have started with this physical motivation, there is a close connection between resistor networks and what we have been discussing so far this semester.

To see this, let's start with the following definition, which is a special case of the Moore-Penrose pseudoinverse.

\begin{definition}
The \emph{Laplacian pseudoinverse} is the unique matrix satisfying: 
\begin{enumerate}
\item
$L^{+}\vec{1}=0$; and
\item
For all $w \perp \vec{1}: \quad L^{+}w=v \mbox{ s.t. } Lv=w \mbox{ and } v\perp\vec{1} $ .
\end{enumerate}
\end{definition}

Given this, we have the following theorem.
Note that here we take the resistances on edges to be the inverse of the weights on those edges, which is probably the most common choice.

\begin{theorem}
Assume that the resistances of the edges of $G=(V,E)$ are given by $R_{ij} = \frac{1}{w_{ij}}$.
Then, the effective resistance between any two nodes $a$ and $b$ is given by:
\begin{eqnarray*}
R_{ab} &=& \left(e_a-e_b\right)^{T} L^{+} \left(e_a-e_b\right) \\
            &=& L_{aa}^{+} -2L_{ab}^{+} + L_{bb}^{+}  .
\end{eqnarray*}
\end{theorem}
\begin{proof}
The idea of the proof is that, given a graph, edge resistances, and net current, there always exists currents $Y$ and potentials $V$ satisfying Kirchhoff's current and potential laws; in addition, the vector of potentials is unique up to a constant, and the currents are unique.
I'll omit the details of this since it is part of HW2.
\end{proof}

Since the effective resistance between any two nodes is well-defined, we can define the total effective resistance of the graph.
(This is sometimes called the Kirchhoff index.)

\begin{definition}
The \emph{total effective resistance} is $R^{tot} =\sum_{ij=1}^{n} R_{ij}$.
\end{definition}

Before proceeding, think for a minute about why one might be interested in such a thing.
Below, we will show that the effective resistance is a distance; and so the total effective resistance is the sum of the distances between all pairs of points in the metric space.
Informally, this can be used to measure the total ``size'' or ``capacity'' of a graph. 
We used a similar thing (but for the geodesic distance) when we showed that expander graphs had a $\Theta\left(\log(n)\right)$ duality gap.
In that case, we did this, essentially, by exploiting the fact that there was a lot of flow to route and since most pairs of nodes were distance $\Theta\left(\log(n)\right)$ apart in the geodesic distance.

The quantity $R^{tot}$ can be expressed exactly in terms of the Laplacian eigenvalues (all of them, and not just the first one or first few).
Here is the theorem (that we won't prove).

\begin{theorem}
Let $\lambda_i$ be the Laplacian eigenvalues.
Then, $R^{tot} = n \sum_{i=1}^{n} \frac{1}{\lambda_i}$.
\end{theorem}

Of course, we can get a (weak) bound on $R^{tot}$ using just the leading nontrivial Laplacian eigenvalue.

\begin{corollary}
\[
\frac{n}{\lambda_2} \le R^{tot} \le \frac{n(n-1)}{\lambda_2}
\]
\end{corollary}

Next, we show that the effective resistance is a distance function. 
For this reason, it is sometimes called the resistance distance.

\begin{theorem}
The effective resistance $R$ is a metric.
\end{theorem}
\begin{proof}
We will establish the three properties of a metric.

First, from the above theorem, $R_{ij}=0 \iff i=j$.
The reason for this is since $e_i-e_j$ is in the null space of $L^{+}$ (which is the $\mbox{span}(\vec{1}$)) iff $i=j$.
Since the pseudoinverse of $L$ has eigenvalues $0,\lambda_2^{-1},\ldots,\lambda_n^{-1}$, it is PSD, and so $R_{ij} \ge 0$.

Second, since the pseudoinverse is symmetric, we have that $R_{ij}=R_{ji}$.

So, the only nontrivial thing is to show the triangle inequality holds.

To do so, we show two claims.
\begin{claim}
\label{claim:resist-metric-claim1}
Let $Y_{ab}$ be the vector 
$e_a-e_b= 
\left\{ \begin{array}{l l}
                     1 & \quad \mbox{ at } a  \\
                    -1 & \quad \mbox{ at } b  \\
                     0 & \quad \text{elsewhere,}
                 \end{array}
         \right.   
$
and let $V_{ab} = L^{+} Y_{ab}$.
Then, $V_{ab}(a) \ge V_{ab}(c) \ge V_{ab}(b)$, for all $c$.
\end{claim}
\begin{proof}
Recall that $V_{ab}$ is the induced potential when we have $1$ Amp going in $a$ and $1$ Amp coming out of $b$. For every vertex $c$, other than $a$ and $b$, the total flow is $0$, which means $\sum_{x\sim c}\frac{1}{R_{xc}}(V_{ab}(x)-V_{ab}(c))=0$, and it is easy to see $V_{ab}(c)=\frac{\sum_{x\sim c}C_{xc}V_{ab}(x)}{\sum_{x\sim c}C_{xc}}$ where $C_{xc}=\frac{1}{R_{xc}}$ is the conductance between $x,c$. $V_{ab}(c)$ has a value equal to the weighted average of values of $V_{ab}(x)$ at its neighbors.
We can use this to prove the claim by contradiction.
Assume that there exists a $c$ s.t. $V_{ab}(c) > V_{ab}(a) $.
If there are several such nodes, then let $c$ be the node s.t. $V_{ab}(c)$ is the largest.
In this case, $V_{ab}(c)$ is larger than the values at its neighbors.
This is a contradiction, since $V_c$ is a weighted average of the potentials at its neighbors. The proof of the other half of the claim is similar. (also $V_{ab}(a)\geq V_{ab}(b)$ as $V_{ab}(a)-V_{ab}(b)=R_{ab}\geq 0$)
\end{proof}

\begin{claim}
$R_{eff}(a,b) + R_{eff}(b,c) \ge R_{eff}(a,c) $
\end{claim}
\begin{proof}
Let $Y_{ab}$ and $Y_{bc}$ be the external current from sending one unit of current from $a \rightarrow b$ and $b \rightarrow c$, respectively.
Note that $Y_{ac} = Y_{ab} + Y_{bc}$.
Define the voltages
$V_{ab} = L^{+} Y_{ab}$, 
$V_{bc} = L^{+} Y_{bc}$, and
$V_{ac} = L^{+} Y_{ac}$.
By linearity, $V_{ac} = V_{ab} + V_{bc}$. 
Thus, it follows that
\[
R_{eff}(a,c) = Y_{ac}^TV_{ac}= Y_{ac}^TV_{ab} + Y_{ac}^{T}V_{bc} .
\]
By Claim~\ref{claim:resist-metric-claim1}, it follows that
\[
Y_{ac}^TV_{ab} = V_{ab}(a) - V_{ab}(c) \le V_{ab}(a) - V_{ab}(b) = R_{eff}(a,b) .
\]
Similarly, $Y_{ac}^T V_{bc} \le R_{eff}(b,c)$. 
This establishes the claim.
\end{proof}
The theorem follows from these two claims.
\end{proof}

Here are some things to note regarding the resistance distance.
\begin{itemize}
\item
$R_{eff}$ is non-increasing function of edge weights.
\item
$R_{eff}$ does not increase when edges are added.
\item
$R^{tot}$ strictly decreases when edges are added and weights are increased.
\end{itemize}
Note that these observations are essentially claims about the distance properties of two graphs, call them $G$ and $G^{\prime}$, when one graph is constructed from the other graph by making changes to one or more edges.

We have said that both geodesic distances and resistances distances are legitimate notions of distances between the nodes on a graph.
One might wonder about the relationship between them.
In the same way that there are different norms for vectors in $\mathbb{R}^{n}$, e.g., the $\ell_{1}$, $\ell_{2}$, and $\ell_{\infty}$, and those norms have characteristic sizes with respect to each other, so too we can talk about the relative sizes of different distances on nodes of a graph.
Here is a theorem relating the resistance distance with the geodesic distance.

\begin{theorem}
For $R_{eff}$ and the geodesic distance $d$:
\begin{enumerate}
\item
$R_{eff}(a,b) = d(a,b)$ iff there exists only one path between $a$ and $b$.
\item
$R_{eff}(a,b) < d(a,b)$ otherwise.
\end{enumerate}
\end{theorem}
\begin{proof}
If there is only one path $P$ between $a$ and $b$, then $Y_{ij}=Y$, for all $ij$ on this path (by Kirchhoff current law), and $V_i-V_j = Y R_{ij}$.
It follows that
\[
R_{ab} = \frac{V_a-V_b}{Y} = \sum_{(ij) \in P} \frac{V_i-V_i}{Y} = \sum_{(ij) \in P} V_{ij} = d_{ab}  .
\]
If a path between $a$ and $b$ is added, so that now there are multiple paths between $a$ and $b$, this new path might use part of the path $P$.
If it does, then call that part of the path $P_1$; consider the rest of $P$, and call the shorter of these $P_2$ and the larger $P_3$.

Observe that the current through each edge of $P_1$ is $Y$; and, in addition, that the current through each edge of $P_2$ and $P_3$ is the same for each edge in the path, call them $Y_2$ and $Y_3$, respectively.
Due to Kirchhoff current law and Kirchhoff circuit/potential law, we have that 
$Y_2+Y_3=Y$ and also that $Y_2,Y_3 > 0$, from which it follows that $Y_2 < Y$.
Finally, 
\begin{eqnarray*}
R_{ab} &=& \frac{V_a-V_b}{Y} \\
   &=& \sum_{(ij) \in P_1} \frac{V_i-V_j}{Y} + \sum_{(ij) \in P_2} \frac{V_i-V_j}{Y}  \\
   &=& \sum_{(ij) \in P_1} \frac{V_i-V_j}{Y} + \sum_{(ij) \in P_2} \frac{V_i-V_j}{Y_2}  \\
   &=& \sum_{(ij) \in P_1}R_{ij} + \sum_{(ij) \in P_2} R_{ij}  \\
   &=& d(a,b)
\end{eqnarray*}
The result follows since $R_{eff}$ doesn't increase when edges are added.
\end{proof}

In a graph that is a tree, there is a unique path between any two vertices, and so we have the following result.

\begin{claim}
The metrics $R_{eff}$ and $d$ are the same in a tree.
That is, on a tree, $R_{eff}(a,b)=d(a,b)$, for all nodes $a$ and $b$.
\end{claim}

\textbf{Fact.}
$R_{eff}$ can be used to bound several quantities of interest, in particular the commute time, the cover time, etc.
We won't go into detail on this.

Here is how $R_{eff}$ behaves in some simple examples.
\begin{itemize}
\item
Complete graph $K_n$.  
Of all graphs, this has the minimum $R_{eff}^{tot}$: $R_{eff}^{tot}(K_n)=n-1$.
\item
Path graph $P_n$.  
Among connected graphs, the path graph has the maximum $R_{eff}^{tot}$: 
$R_{eff}^{tot}(P_n) = \frac{1}{6}(n-1)n(n+1)$.
\item
Star $S_n$.
Among trees, this has the minimum $R_{eff}^{tot}$: $R_{eff}^{tot}(S_n) = (n-1)^2$.
\end{itemize}

\subsection{Extensions to infinite graphs}

All of what we have been describing so far is for finite graphs.
Many problems of interest have to do with infinite graphs.
Perhaps the most basic is whether random walks are recurrent.
In addition to being of interest in its own right, considering this question on infinite graphs should provide some intuition for how random walked based spectral methods perform on the finite graphs we have been considering.

\begin{definition}
A random walk is \emph{recurrent} if the walker passes through every point with probability $1$, 
or equivalently if the walker returns to the starting point with probability $1$.
Otherwise, the random walk is \emph{transient}.
\end{definition}

Note that---if we were to be precise---then we would have to define this for a single node, be precise about which of those two notions we are considering, etc.
It turns out that those two notions are equivalent and that a random walk is recurrent for one node iff it is recurrent for any node in the graphs.
We'll not go into these details here.

For irreducible, aperiodic random walks on finite graphs, this discussion is of less interest, since a random walk will eventually touch every node with probability proportional to its degree; but consider three of the simplest infinite graphs:
$\mathbb{Z}$, $\mathbb{Z}^2$, and $\mathbb{Z}^3$.
Informally, as the dimension increases, there are more neighbors for each node and more space to get lost in, and so it should be harder to return to the starting node.
Making this precise, i.e., proving whether a random walk on these graphs is recurrent is a standard problem, one version of which appears on HW2.

The basic idea for this that you need to use is to use something called Rayleigh's Monotonicity Law as well as the procedures of shorting and cutting.
Rayleigh's Monotonicity Law is a version of the result we described before, which says that $R_{eff}$ between two points $a$ and $b$ varies monotonically with individual resistances.
Then, given this, one can use this to do two things to a graph $G$:
\begin{itemize}
\item
\emph{Shorting} vertices $u$ and $v$: this is ``electrical vertex identification.''
\item
\emph{Cutting} edges between $u$ and $v$: this is ``electrical edge deletion.''
\end{itemize}
Both of these procedures involve constructing a new graph $G^{\prime}$ from the original graph $G$ (so that we can analyze $G^{\prime}$ and make claims about $G$).
Here are the things you need to know about shorting and cutting:
\begin{itemize}
\item
Shorting a network can only decrease $R_{eff}$.
\item
Cutting a network can only increase $R_{eff}$.
\end{itemize}
For $\mathbb{Z}^2$, if you short in ``Manhattan circles'' around the origin, then this only decreases $R_{eff}$, and you can show that $R_{eff}=\infty$ on the shorted graph, and thus $R_{eff}=\infty$ on the original $\mathbb{Z}^2$.
For $\mathbb{Z}^3$, if you cut in a rather complex way, then you can show that $R_{eff} < \infty$ on the cut graph,  meaning that $R_{eff} < \infty$ on the original $\mathbb{Z}^{3}$.
This, coupled with the following theorem, establish the result random walks on $\mathbb{Z}^2$ are recurrent, but random walks on $\mathbb {Z}^{3}$ are transient.

\begin{theorem}
A network is recurrent iff $R_{eff} = \infty$.
\end{theorem}

Using these ideas to prove the recurrence claims is left for HW2: getting the result for $\mathbb{Z}$ is straightforward; getting it for $\mathbb{Z}^2$ is more involved but should be possible; and getting it for $\mathbb{Z}^3$ is fairly tricky---look it up on the web, but it is left as extra credit.

%% file: lect17.tex
\section{%
(03/19/2015): 
Diffusions and Random Walks as Robust Eigenvectors}

Reading for today.
\begin{compactitem}
\item
``Implementing regularization implicitly via approximate eigenvector computation,'' in ICML, by Mahoney and Orecchia
\item
``Regularized Laplacian Estimation and Fast Eigenvector Approximation,'' in NIPS, by Perry and Mahoney
\end{compactitem}

Last time, we talked about electrical networks, and we saw that we could reproduce some of the things we have been doing with spectral methods with more physically intuitive techniques.
These methods are of interest since they are typically more robust than using eigenvectors and they often lead to simpler proofs.
Today, we will go into more detail about a similar idea, namely whether we can interpret random walks and diffusions as providing robust or regularized or stable analogues of eigenvectors.
Many of the most interesting recent results in spectral graph methods adopt this approach of using diffusions and random walks rather than eigenvectors.
We will only touch on the surface of this approach.

\subsection{Overview of this approach}

There are several advantages to thinking about diffusions and random walks as providing a robust alternative to eigenvectors.
\begin{itemize}
\item
New insight into spectral graph methods.
\item
Robustness/stability is a good thing in many situations.
\item
Extend global spectral methods to local spectral analogues.
\item
Design new algorithms, e.g., for Laplacian solvers.
\item
Explain why diffusion-based heuristics work as they do in social network, computer vision, machine learning, and many other applications.
\end{itemize}

Before getting into this, step back for a moment, and recall that spectral methods have many nice theoretical and practical properties.
\begin{itemize}
\item
\textbf{Practically.}
Efficient to implement; can exploit very efficient linear algebra routines; perform very well in practice, in many cases better than theory would suggest.
(This last claim means, e.g., that there is an intuition in areas such as computer vision and social network analysis that even if you could solve the best expansion/conductance problem, you wouldn't want to, basically since the approximate solution that spectral methods provide is ``better.'')
\item
\textbf{Theoretically.}
Connections between spectral and combinatorial ideas; and connections between Markov chains and probability theory that provides a geometric viewpoint.
\end{itemize}
Recently, there have been very fast algorithms that combine spectral and combinatorial ideas.
They rely on an optimization framework, e.g., solve max flow problems by relating them to these spectral-based optimization ideas.
These use diffusion-based ideas, which are a relatively new trend in spectral graph theory.

To understand better this new trend, recall that the classical view of spectral methods is based on Cheeger's Inequality and involves computing an eigenvector and performing sweep cuts to reveal sparse cuts/partitions.
The new trend is to replace eigenvectors with vectors obtained by running random walks.
This has been used in: 
\begin{itemize}
\item
fast algorithms for graph partitioning and related problems;
\item
local spectral graph partitioning algorithms; and
\item
analysis of real social and information networks.
\end{itemize}
There are several different types of random walks, e.g., Heat Kernel, PageRank, etc., and different walks are better in different situations.

So, one question is: Why and how do random walks arise naturally from an optimization framework?

One advantage of a random walk is that to compute an eigenvector in a very large graph, a vanilla application of the power method or other related iterative methods (especially black box linear algebra methods) might be too slow, and so instead one might run a random walk on the graph to get a quick approximation.

Let $W=AD^{-1}$ be the natural random walk matrix, and let $L=D-A$ be the Laplacian.
As we have discussed, it is well-known that the second eigenvector of the Laplacian can be computed by iterating $W$.
\begin{itemize}
\item
For ``any'' vector $y_0$ (or ``any'' vector $y_0$ s.t. $y_0D^{-1}\vec{1}=0$ or any random vector $y_0$ s.t. $y_0D^{-1}\vec{1}=0$), we can compute $D^{-1}W^{t}y_0$; and we can take the limit as $t \rightarrow \infty$ to get
\[
v_2(L) = \lim_{t \rightarrow} \frac{ D^{-1}W^{t} y_0 }{ \| W^{t}y_0 \|_{D^{-1}} }  ,
\]
where $v_2(L)$ is the leading nontrivial eigenvector of the Laplacian.
\item
If time is a precious resource, then one alternative is to avoid iterating to convergence, i.e., don't let $t \rightarrow \infty$ (which of course one never does in practice, but by this we mean don't iterate to anywhere near machine precision), but instead do some sort of ``early stopping.''
In that case, one does not obtain an eigenvector, but it is of interest to say something about the vector that is computed.
In many cases, this is useful, either as an approximate eigenvector or as a locally-biased analogue of the leading eigenvector.
This is very common in practice, and we will look at it in theory.
\end{itemize}

Another nice aspect of replacing an eigenvector with a random walk or by truncating the power iteration early is that the vectors that are thereby returned are more robust.
The idea should be familiar to statisticians and machine learners, although in a somewhat different form.
Say that there is a ``ground truth'' graph that we want to understand but that the measurement we make, i.e., the graph that we actually see and that we have available to compute with, is a  noisy version of this ground truth graph.
So, if we want to compute the leading nontrivial eigenvector of the unseen graph, then computing the leading nontrivial eigenvector of the observed graph is in general not a particularly good idea.
The reason is that it can be very sensitive to noise, e.g., mistakes or noise in the edges.
On the other hand, if we perform a random walk and keep the random walk vector, then that is a better estimate of the ground truth eigendirection.
(So, the idea is that eigenvectors are unstable but random walks are not unstable.)  

A different but related question is the following: why are random walks useful in the design of fast algorithms?
(After all, there is no ``ground truth'' model in this case---we are simply running an algorithm on the graph that is given, and we want to prove results about the algorithm applied to that graph.)
The reason is similar, but the motivation is different.
If we want to have a fast iterative algorithm, then we want to work with objects that are stable, basically so that we can track the progress of the algorithm.  
Working with vectors that are the output of random walks will be better in this sense.
Today, we will cover an optimization perspective on this.
(We won't cover the many applications of these ideas to graph partitioning and related algorithmic problems.)

\subsection{Regularization, robustness, and instability of linear optimization}

Again, take a step back.
What is regularization?
The usual way it is described (at least in machine learning and data analysis) is the following.

We have an optimization problem
\begin{eqnarray*}
\min_{x \in \mathcal{S}}  f(x)  ,
\end{eqnarray*}
where $f(x)$ is a (penalty) function, and where $\mathcal{S}$ is some constraint set.
This problem might not be particularly well-posed or well-conditioned, in the sense that the solution might change a lot if the input is changed a little.
In order to get a more well-behaved version of the optimization problem, e.g., one whose solution changes more gradually as problem parameters are varied, one might instead try to solve the problem
\begin{eqnarray*}
\min_{x \in \mathcal{S}}  f(x) + \lambda g(x)  ,
\end{eqnarray*}
where $\lambda \in \mathbb{R}^{+}$ is a parameter, and where $g(x)$ is (regularization) function.
The idea is that $g(x)$ is ``nice'' in some way, e.g., it is convex or smooth, and $\lambda$ governs the relative importance of the two terms, $f(x)$ and $g(x)$.
Depending on the specific situation, the advantage of solving the latter optimization problem is that one obtains a more stable optimum, a unique optimum, or smoothness conditions.
More generally, the benefits of including such a regularization function in ML and statistics is that:
one obtains increased stability; 
one obtains decreased sensitivity to noise; and 
one can avoid overfitting.

Here is an illustration of the instability of eigenvectors.
Say that we have a graph $G$ that is basically an expander except that it is connected to two small poorly-connected components. 
That is, each of the two components is well-connected internally but poorly-connected to the rest of $G$, e.g., connected by a single edge.
One can easily choose the edges/weights in $G$ so that the leading non-trivial eigenvector of $G$ has most of its mass, say a $1-\epsilon$ fraction of its mass, on the first small component.
In addition, one can then easily construct a perturbation, e.g., removing one edge from $G$ to construct a graph $G^{\prime}$ such that $G^{\prime}$ has a $1-\epsilon$ fraction of its mass on the second component.
That is, a small perturbation that consists of removing one edge can completely shift the eigenvector---not only its direction but also where in $\mathbb{R}^{n}$ its mass is supported.

Let's emphasize that last point.
Recalling our discussion of the Davis-Kahan theorem, as well as the distinction between the Rayleigh quotient objective and the actual partition found by performing a sweep cut, we know that if there is a small spectral gap, then eigenvectors can swing by 90 degrees.
Although the example just provided has aspects of that, this example here is even more sensitive: not only does the direction of the vector change in $\mathbb{R}^{n}$, but also the mass along the coordinate axes in $\mathbb{R}^{n}$ where the eigenvector is localized changes dramatically under a very minor perturbation to~$G$.

To understand this phenomenon better, here is the usual quadratic optimization formulation of the leading eigenvector problem.
For simplicity, let's consider a $d$-regular graph, in which case we get the following.
\begin{eqnarray}
\label{eqn:spectral-quadratic-formulation}
\mbox{Quadratic Formulation}: & \frac{1}{d} \min_{x\in\mathbb{R}^{n}} & x^TLx  \\
\nonumber
                & \mbox{s.t.} &  \|x\| = 1  \\
\nonumber
                & &  x \perp \vec{1}  .
\end{eqnarray}
This is an optimization over vectors $x\in\mathbb{R}^{n}$.
Alternatively, we can consider the following optimization problem over SPSD matrices.
\begin{eqnarray}
\label{eqn:spectral-sdp-formulation}
\mbox{SDP Formulation}: & \frac{1}{d} \min_{X\in\mathbb{R}^{n \times n}} & L \bullet X  \\
\nonumber
                & \mbox{s.t.} & I \bullet X = 1  \\
\nonumber
                & & J \bullet X = 0  \\
\nonumber
                & & X \succeq 0  ,
\end{eqnarray}
Recall that $ I \bullet X = \Trace{X} $ and that $J= 11^T$.
Recall also that if a matrix $X$ is rank-one and thus can be written as $X=xx^T$, then $L \bullet X = x^TLx$.

These two optimization problems, Problem~\ref{eqn:spectral-quadratic-formulation} and
Problem~\ref{eqn:spectral-sdp-formulation}, are equivalent, in that if $x^{*}$ is the vector solution to the former and $X^{*}$ is a solution of the latter, then $X^{*}=x^{*}{x^{*}}^{T}$.
In particular, note that although there is no constraint on the SDP formulation that the solution is rank-one, the solution turns out to be rank one.

Observe that this is a linear SDP, in that the objective and all the constraints are linear.
Linear SDPs, just as LPs, can be very unstable.
To see this in the simpler setting of LPs, consider a convex set $S \subset \mathbb{R}^{n}$ and a linear optimization problem:
\begin{equation*}
f(c) = \mbox{arg}\min_{x \in S} c^Tx  .
\end{equation*}
The optimal solution $f(c)$ might be very unstable to perturbations of $c$, in that we can have 
\[
\|c^{\prime}-c\| \le \delta \quad \mbox{and} \quad \|f(c^{\prime})-f(c) \| \gg \delta  .
\]
(With respect to our Linear SDP, think of the vector $x$ as the PSD variable $X$ and think of the vector $c$ as the Laplacian $L$.)
That is, we can change the input $c$ (or $L$) a little bit and the solution changes a lot.
One way to fix this is to introduce a regularization term $g(x)$ that is strongly convex.

So, consider the same convex set $S \subset \mathbb{R}^{n}$ and a regularized linear optimization problem
\begin{equation*}
f(c) = \mbox{arg}\min_{x \in S} c^Tx + \lambda g(x)  ,
\end{equation*}
where $\lambda\in\mathbb{R}^{+}$ is a parameter and where $g(x)$ is $\sigma$-strongly convex.
Since this is just an illustrative example, we won't define precisely the term $\sigma$-strongly convex, but we note that $\sigma$ is related to the derivative of $f(\cdot)$ and so the parameter $\sigma$ determines how strongly convex is the function $g(x)$.
Then, since $\sigma$ is related to the slope of the objective at $f(c)$, and since the slope of the new objective at $f(c) < \delta$, strong convexity ensures that we can find a new optimum $f(c^{\prime})$ at distance $< \frac{\delta}{\sigma}$.
So, we have 
\begin{equation}
\label{eqn:robustness-regularized-linear-optimiz}
\|c^{\prime}-c\| \le \delta \Rightarrow \|f(c^{\prime})-f(c) \| < \delta/\sigma  ,
\end{equation}
i.e., the strong convexity on $g(x)$ makes the problem stable that wasn't stable before.

\subsection{Structural characterization of a regularized SDP}

How does this translate to the eigenvector problem?
Well, recall that the leading eigenvector of the Laplacian solves the SDP, where $X$ appears linearly in the objective and constraints, as given in Problem~\ref{eqn:spectral-sdp-formulation}.
We will show that several different variants of random walks exactly optimize regularized versions of this SDP.
In particular they optimize problems of the form 
\begin{eqnarray}
\label{eqn:spectral-sdp-formulation-regularized}
\mbox{SDP Formulation}: & \frac{1}{d} \min_{X\in\mathbb{R}^{n \times n}} & L \bullet X + \lambda G(X)  \\
\nonumber
                & \mbox{s.t.} & I \bullet X = 1  \\
\nonumber
                & & J \bullet X = 0  \\
\nonumber
                & & X \succeq 0  ,
\end{eqnarray}
where $G(X)$ is an appropriate regularization function that depends on the specific form of the random walk and that (among other things) is strongly convex.

To give an interpretation of what we are doing, consider the eigenvector decomposition of $X$, where
\begin{equation}
X = \sum_{i} p_i v_i v_i^T , \quad \mbox{where} \quad
\left\{ \begin{array}{l l}
                    \forall i \quad p_i \ge 0  \\
                    \sum_i p_i = 1  \\
                    \forall i \quad v_i^T \vec{1} = 0
                 \end{array}
         \right.
\end{equation}
I've actually normalized things so that the eigenvalues sum to $1$.
If we do this, then the eigenvalues of $X$ define a probability distribution.
If we don't regularize in Problem~\ref{eqn:spectral-sdp-formulation-regularized}, i.e., if we set $\lambda=0$, then the optimal solution to Problem~\ref{eqn:spectral-sdp-formulation-regularized} puts all the weight on the second eigenvector (since $X^{*}=x^{*}{x^{*}}^{T}$).
If instead we regularize, then the regularization term ensures that the weight is spread out on all the eigenvectors, i.e., the optimal solution $X^{*}= \sum_i \alpha_i v_i v_i^T$, for some set of coefficients $\{\alpha_i\}_{i=1}^{n}$.
So, the solution is not rank-one, but it is more stable.

\textbf{Fact.}
If we take this optimization framework and put in ``reasonable'' choices for $G(X)$, then we can recover algorithms that are commonly used in the design of fast algorithms and elsewhere.
That the solution is not rank one makes sense from this perspective: if we iterate $t \rightarrow \infty$, then all the other eigendirections are washed out, and we are left with the leading direction; but if we only iterate to a finite $t$, then we still have admixing from the other eigendirections.

To see this in more detail, consider the following three types of random walks.
(Recall that $M=AD^{-1}$ and $W=\frac{1}{2}\left(I+M\right)$.)
\begin{itemize}
\item
\textbf{Heat Kernel.}
$H_t = \exp\left(-tL\right) = \sum_{k=0}^{\infty} \frac{\left(-t\right)^{k}}{k!} L^k = \sum_{i=1}^{n} e^{-\lambda_it} P_i$, where $P_i$ is a projection matrix onto that eigendirection.
\item
\textbf{PageRank.}
$R_{\gamma} = \gamma \left( I - \left( 1-\gamma \right) M \right)^{-1}$.
This follows since the PageRank vector is the solution to $$ \pi(\gamma,s) = \gamma s + (1-\gamma) M \pi(\gamma,s)  ,$$ which can be written as $$\pi(\gamma,s) = \gamma \sum_{t=0}^{\infty} \left(1-\gamma \right)^{t} M^t s = R_{\gamma} s  .$$
\item
\textbf{Truncated Lazy Random Walk.}
$W_{\alpha} = \alpha I + (1-\alpha) M$.
\end{itemize}
These are formal expressions describing the action of each of those three types of random walks, in the sense that the specified matrix maps the input to the output: to obtain the output vector, compute the matrix and multiply it by the input vector.
Clearly, of course, these random walks would not be implemented by computing these matrices explicitly; instead, one would iteratively apply a one-step version of them to the input vector.

Here are the three regularizers we will consider.
\begin{itemize}
\item
\textbf{von Neumann entropy.}
Here, $G_H = \Trace{ X \log X } = \sum_i p_i \log p_i$.
\item
\textbf{Log-determinant.}
Here, $G_D = -\log\det \left( X \right)$.
\item
\textbf{Matrix $p$-norm, $p >0$.}
Here, $G_p = \frac{1}{p} \| X \|_p^p = \frac{1}{p} \Trace{ X^p } = \frac{1}{p} \sum_i p_i^p$.
\end{itemize}
And here are the connections that we want to establish.
\begin{itemize}
\item
$G=G_H \Rightarrow^{entropy} X^{*} \sim H_t$, with $t=\lambda$.
\item
$G=G_D \Rightarrow^{logdet} X^{*} \sim R_{\gamma}$, with $\gamma \sim \lambda$.
\item
$G=G_p \Rightarrow^{p-norm} X^{*} \sim W_{\alpha}^{t}$, with $t \sim \lambda$.
\end{itemize}

Here is the basic structural theorem that will allow us to make precise this connection between random walks and regularized SDPs.
Note that its proof is a quite straightforward application of duality ideas.
\begin{theorem}
\label{thm:reg-rand-walk-struct}
Recall the regularized SDP of Problem~\ref{eqn:spectral-sdp-formulation-regularized}, and let $\lambda = 1/\eta$.
If $G$ is a connected, weighted, undirected graph, then let $L$ be the normalized Laplacian.
The the following are sufficient conditions for $X^{*}$ to be the solution of the regularized SDP.
\begin{enumerate}
\item
$X^{*} = \left( \nabla G \right)^{-1} \left( \eta \left( \lambda^{*} I - L \right) \right)$, for some $\lambda^{*}\in\mathbb{R}$.
\item
$I \bullet X^{*} = 1$.
\item
$X^{*} \succeq 0$.
\end{enumerate}
\end{theorem}
\begin{proof}
Write the Lagrangian $\mathcal{L}$ of the above SDP as
\[
\mathcal{L} = L \bullet X + \frac{1}{\eta}G(X) - \lambda \left( I \bullet X - 1 \right) - U \bullet X  ,
\]
where $\lambda\in\mathbb{R}$ and where $U \succeq 0$.
Then, the dual objective function is 
\[
h\left( \lambda,U \right) = \min_{X \succeq 0} \mathcal{L}(X,\lambda,U)  .
\]
Since $G(\cdot)$ is strictly convex, differentiable, and rotationally invariant, the gradient of $G$ over the positive semi-definite cone is invertible, and the RHS is minimized when
\[
X = \left( \nabla G \right)^{-1} \left( \eta \left(- L +  \lambda^{*} I + U \right) \right)  ,
\]
where $\lambda^{*}$ is chosen s.t. $I \bullet X^{*} = 1$.
Hence,
\begin{eqnarray*}
h \left( \lambda^{*},0 \right) 
   &=& L \bullet X^{*} + \frac{1}{\eta} G(X^{*}) - \lambda^{*} \left( I \cdot X^{*} - 1 \right)  \\
   &=& L \bullet X^{*} + \frac{1}{\eta} G(X^{*}) .
\end{eqnarray*}
By Weak Duality, this implies that $X^{*}$ is the optimal solution to the regularized SDP.
\end{proof}

\subsection{Deriving different random walks from Theorem~\ref{thm:reg-rand-walk-struct}}

To derive the Heat Kernel random walk from Theorem~\ref{thm:reg-rand-walk-struct}, let's do the following.
Since 
\[
G_H(X) = \Trace{X \log(X) } - \Trace{X}  , 
\]
it follows that $\left(\nabla G\right)(X)=\log(x)$ and thus that $\left( \nabla G \right)^{-1}(Y) = \exp(Y)$, from which it follows that 
\begin{eqnarray*}
X^{*} &=& \left( \nabla G \right)^{-1} \left( \eta \left ( \lambda I - L \right) \right) \\
         &=& \exp\left( \eta \left( \lambda I - L \right) \right)  \quad\mbox{for an appropriate choice of $\eta,\lambda$}\\
         &=& \exp \left( -\eta L \right) \exp \left( \eta \lambda \right) \\
         &=& \frac{ H_{\eta} }{ \Trace{ H_{\eta} } }  ,
\end{eqnarray*}
where the last line follows if we set $\lambda = \frac{-1}{\eta}\log\left( \Trace{ \exp \left( -\eta L \right) } \right)$

To derive the PageRank random walk from Theorem~\ref{thm:reg-rand-walk-struct}, we follow a similar derivation.
Since 
\[
G_D(X) = -\log\det\left(X\right)  , 
\]
it follows that $\left(\nabla G \right)(X) = -X^{-1}$ and thus that $\left( \nabla G \right)^{-1}(Y) = -Y^{-1}$, from which it follows that
\begin{eqnarray*}
X^{*} &=& \left( \nabla G \right)^{-1} \left( \eta \left ( \lambda I - L \right) \right) \\
         &=& -\left( \eta \left( \lambda I-L \right) \right)^{-1}  \quad\mbox{for an appropriate choice of $\eta,\lambda$}\\
         &=& \frac{ D^{-1/2} R_{\gamma} D^{-1/2} }{ \Trace{ R_{\gamma} } }  ,
\end{eqnarray*}
for $\eta,\lambda$ chosen appropriately.

Deriving the truncated iterated random walk  or other forms of diffusions is similar.

We will go into more details on the connection with PageRank next time, but for now we just state that the solution can be written in the form
\[
x^{*} = c \left( L_G - \alpha L_{K_n} \right)^{+} D s   ,
\]
for a constant $c$ and a parameter $\gamma$.  
That is, it is of the form of the solution to a linear equation, i.e., $L_G^{+}s$, except that there is a term that moderates the effect of the graph by adding the Laplacian of the complete graph. 
This is essentially a regularization term, although it is not usually described as such.
See the Gleich article for more details on this.

\subsection{Interpreting Heat Kernel random walks in terms of stability}

Here, we will relate the two previous results for the heat kernel.
Again, for simplicity, assume that $G$ is $d$-regular.
Recall that the Heat Kernel random walk is a continuous time Markov chain, modeling the diffusion of heat along the edges of $G$.
Transitions take place in continuous time $t$ with an exponential distribution:
\[
\frac{\partial \rho(t) }{\partial t} = -L \frac{\rho(t)}{d}
\Rightarrow 
\rho(t) = \exp\left( -\frac{t}{d}L \right) \rho(0)  .
\]
That is, this describes the way that the probability distribution changes from one step to the next and how it is related to $L$.
In particular, the Heat Kernel can be interpreted as a Poisson distribution over the number of steps of the natural random walk $W=AD^{-1}$, where we get the following:
\[
e^{ - \frac{t}{d} L } = e^{-t} \sum_{k=1}^{\infty} \frac{t^k}{k!} W^k   .
\]
What this means is: pick a number of steps from the Poisson distribution; and then perform that number of steps of the natural random walk.

So, if we have two graphs $G$ and $G^{\prime}$ and they are close, say in an $\ell_{\infty}$ norm sense, meaning that the edges only change a little, then we can show the following.
(Here, we will normalize the two graphs so that their respective eigenvalues sum to $1$.)
The statement analogous to Statement~\ref{eqn:robustness-regularized-linear-optimiz} is the following.
\[
\| G - G^{\prime} \|_{\infty}
\Rightarrow 
\| \frac{ H_{G}^{t} }{ I \bullet H_{G}^{t} } - \frac{ H_{G^{\prime}}^{t} }{ I \bullet H_{G^{\prime}}^{t} }   \|_{1}
\le 
t \delta  .
\]
Here, $\| \cdot \|_{1}$ is some other norm (the $\ell_1$ norm over the eigenvalues) that we won't describe in detail.
Observe that the bound on the RHS depends on how close the graphs are ($\delta$) as well as how long the Heat Kernel random walk runs ($t$).
If the graphs are far apart ($\delta$ is large), then the bound is weak.
If the random walk is run for a long time ($t \rightarrow \infty$), then the bound is also very weak.
But, if the walk is nor run too long, then we get a robustness result.
And, this follows from the strong convexity of the regularization term that the heat kernel is implicitly optimizing exactly.

\subsection{A statistical interpretation of this implicit regularization result}

Above, we provided three different senses in which early-stopped random walks can be interpreted as providing a robust or regularized notion of the leading eigenvectors of the Laplacian: e.g., in the sense that, in addition to approximating the Rayleigh quotient, they also exactly optimize a regularized version of the Rayleigh quotient.
Some people interpret regularization in terms of statistical priors, and so let's consider this next.

In particular, let's now give a statistical interpretation to this implicit regularization result.
By a ``statistical interpretation,'' I mean a derivation analogous to the manner in which $\ell_2$ or $\ell_1$ regularized $\ell_2$ regression can be interpreted in terms of a Gaussian or Laplace prior on the coefficients of the regression problem.
This basically provides a Bayesian interpretation of regularized linear regression.
The derivation below will show that the solutions to the Problem~\ref{eqn:spectral-sdp-formulation-regularized} that random walkers implicitly optimize can be interpreted as a regularized estimate of the pseudoinverse of the Laplacian, and so in some sense it provides a Bayesian interpretation of the implicit regularization provided by random~walks.

To start, let's describe the analogous results for vanilla linear regression.  
For some (statistics) students, this is well-known; but for other (non-statistics) students, it likely is not.  The basic idea should be clear; and we cover it here to establish notation and nomenclature.
Let's assume that we see $n$ predictor-response pairs in $\mathbb{R}^{p}\times\mathbb{R}$, call them $\{\left( x_i,y_i \right) \}_{i=1}^{n}$, and the goal is to find a parameter vector $\beta\in\mathbb{R}^{p}$ such that $\beta^{T}x_i \approx y_i$.
A common thing to do is to choose $\beta$ by minimizing the RSS (residual sum of squares), i.e., 
choosing 
\[
F(\beta) = RSS(\beta) = \sum_{i=1}^{n} \| y_i - \beta^Tx_i \|_2^2   .
\]
Alternatively, we could optimize a regularized version of this objective.
In particular, we have
\begin{eqnarray*}
\mbox{Ridge regression:} & & \min_{\beta} F(\beta) + \lambda \| \beta \|_2^2   \\
\mbox{Lasso regression:} & & \min_{\beta} F(\beta) + \lambda \| \beta \|_1    .
\end{eqnarray*}
To derive these two versions of regularized linear regression, let's model $y_i$ as independent random variables with distribution dependent on $\beta$ as follows:
\begin{equation}
\label{eqn:ls-eq0}
 y_i \sim N \left( \beta^T x, \sigma^2 \right)  ,
\end{equation}
i.e., each $y_i$ is a Gaussian random variable with mean $\beta^Tx_i$ and known variance $\sigma^2$.
This induces a conditional density for $y$ as follows:
\begin{equation}
\label{eqn:ls-eq1}
p\left( y | \beta \right) \sim \exp\{ \frac{-1}{2\sigma^2}F(\beta) \}  , 
\end{equation}
where the constant of proportionality depends only on $y$ and $\sigma$.
From this, we can derive the vanilla least-squares estimator.
But, we can also assume that $\beta$ is a random variable with distribution $p(\beta)$, which is known as a prior distribution, as follows:
\begin{equation}
\label{eqn:ls-eq2}
p(\beta) \sim \exp \{ -U(\beta) \}  ,
\end{equation}
where we adopt that functional form without loss of generality.
Since these two random variables are dependent, upon observing $y$, we have information on $\beta$, and this can be encoded in a posterior density, $p\left( \beta | y \right)$, which can be computed from Bayes' rule as follows:
\begin{eqnarray}
\nonumber
p\left( \beta | y \right) 
   &\sim& p\left(y | \beta \right) p(\beta) \\
   &\sim& \exp \{ \frac{-1}{2\sigma^2} F(\beta) - U(\beta) \}  .
\label{eqn:ls-eq3}
\end{eqnarray}
We can form the MAP, the maximum a posteriori, estimate of $\beta$ by solving 
\[
\max_{\beta} p\left( \beta | y \right) \mbox{ \textbf{iff} } \min_{\beta} - \log p\left( \beta | y \right)  .
\]
From this we can derive ridge regression and Lasso regression:
\begin{eqnarray*}
U(\beta) &=& \frac{\lambda}{2\sigma^2} \|\beta\|_2^2 \quad \Rightarrow \mbox{Ridge regression} \\
U(\beta) &=& \frac{\lambda}{2\sigma^2} \|\beta\|_1     \quad \Rightarrow \mbox{Lasso regression}
\end{eqnarray*}

To derive the analogous result for regularized eigenvectors, we will follow the analogous setup.
What we will do is the following.
Given a graph $G$, i.e., a ``sample'' Laplacian $L$, assume it is a random object drawn from a ``population'' Laplacian $\mathcal{L}$.
\begin{itemize}
\item
This induces a conditional density for $L$, call it $p\left( L | \mathcal{L} \right)$.
\item
Then, we can assume prior information about the population Laplacian $\mathcal{L}$ in the form of $p\left(\mathcal{L}\right)$.
\item
Then, given the observed $L$, we can estimate the population Laplacian by maximizing its posterior density $p\left(\mathcal{L} | L\right)$.
\end{itemize}
While this setup is analogous to the derivation for least-squares, there are also differences.
In particular, one important difference between the two approaches is that here there is one data point, i.e., the graph/Laplacian is a single data point, and so we need to invent a population from which it was drawn.
It's like treating the entire matrix $X$ and vector $y$ in the least-squares problem as a single data point, rather than $n$ data points, each of which was drawn from the same distribution.
(That's not a minor technicality: in many situations, including the algorithmic approach we adopted before, it is more natural to think of a graph as a single data point, rather than as a collection of data points, and a lot of statistical theory breaks down when we observe $N=1$ data point.)

In more detail, recall that a Laplacian is an SPSD matrix with a very particular structure, and let's construct/hypothesize a population from which it was drawn.
To do so, let's assume that nodes $n$ in the population and in the sample have the same degrees.
If $d=\left(d_1,\ldots, d_n \right)$ is the degree vector, and $D = \mbox{deg}\left( d_1,\ldots,d_n \right)$ is the diagonal degree matrix, then we can define the set 
\[
\chi = \{ X : X \succeq 0 , X D^{1/2}\vec{1} = 0 , \mbox{rank}(X) = n-1  \}  .
\]
So, the population Laplacian and sample Laplacian are both members of $\chi$.
To model $L$, let's use a scaled Wishart matrix with expectation $\mathcal{L}$.
(This distribution plays the role of the Gaussian distribution in the least-squares derivation.  Note that this is a plausible thing to assume, but other assumptions might be possible too.)
Let $m \ge n-1$ be a scale parameter (analogous to the variance), and suppose that $L$ is distributed over $\chi$ as $\frac{1}{m}\mbox{Wishart}\left(\mathcal{L},m\right)$.
Then $\mathbb{E}\left[ L | \mathcal{L} \right] = \mathcal{L}$ and $L$ has the conditional density
\begin{equation}
\label{eqn:lap-cond-likelihood}
p\left( L | \mathcal{L} \right) \sim \frac{1}{\det\left(\mathcal{L}\right)^{m/2}} \exp \{ \frac{-m}{2} \Trace{ L \mathcal{L}^{+} } \}  .
\end{equation}
This is analogous to Eqn~(\ref{eqn:ls-eq1}) above.
Next, we can say that $\mathcal{L}$ is a random object with prior density $p\left(\mathcal{L}\right)$, which without loss of generality we can take to be of the following form:
\[
p\left(\mathcal{L}\right) \sim \exp \{ -U\left(\mathcal{L}\right) \}  ,
\]
where $U$ is supported on a subset $\bar{\chi} \subseteq \chi$.
This is analogous to Eqn~(\ref{eqn:ls-eq2}) above.
Then, observing $L$, the posterior distribution for $\mathcal{L}$ is the following:
\begin{eqnarray*}
p\left(\mathcal{L}|L \right)
   &\sim& p \left( L | \mathcal{L} \right) p \left( \mathcal{L} \right) \\
   &\sim& \exp \{ \frac{m}{2}\Trace{ L\mathcal{L}^{+} } + \frac{m}{2} \log \det \left( \mathcal{L}^{+} \right) - U\left(\mathcal{L} \right) \}   ,
\end{eqnarray*}
with support determined by $\bar{\chi}$.
This is analogous to Eqn~(\ref{eqn:ls-eq3}) above.

If we denote by $\hat{\mathcal{L}}$ 
the MAP estimate of $\mathcal{L}$, then it follows that $\hat{\mathcal{L}}^{+}$ is the solution of the following optimization problem.
\begin{eqnarray}
\label{eqn:sdp-stat}
 & \min_{X} & \Trace{L \bullet X} + \frac{2}{m} U\left(X^{+}\right) - \log\det\left(X\right) \\
\nonumber
                & \mbox{s.t.} &  X \in \bar{\chi} \subseteq \chi .
\end{eqnarray}
If $\bar{\chi} = \{ X : \Trace{ X } =1 \} \cap \chi$, then
Problem~\ref{eqn:sdp-stat} is the same as 
Problem~\ref{eqn:spectral-sdp-formulation-regularized}, except for the factor of $\log\det\left(x\right)$.

This is almost the regularized SDP we had above.

Next, we present a prior that will be related to the PageRank procedure.
This will make the connection with the regularized SDP more precise.
In particular, we present a prior for the population Laplacian that permits us to exploit the above estimation framework to show that the MAP estimate is related to a PageRank computation.

The criteria for the prior are so-called neutrality and invariance conditions.
It is to be supported on $\chi$; and in particular, for any $X \in \chi$, it will have rank $n-1$ and satisfy $XD^{1/2}1=0$.
The prior will depend only on the eigenvalues of the Laplacian (or equivalently of the inverse Laplacian).
Let $\mathcal{L}^{+} = \tau O \Lambda O$ be the spectral decomposition of the inverse Laplacian, where $\tau$ is a scale parameter.
We will require that the distribution for $\lambda = \left( \lambda_1,\ldots, \lambda_n \right)$ be exchangeable (i.e., invariant under permutations) and neutral (i.e., $\lambda(v)$ is independent of $\frac{\lambda(u)}{1-\lambda(v)}$, for $u \ne v$, for all $v$).
The only non-degenerate possibility is that $\lambda$ is distributed as a Dirichlet distribution  as follows:
\begin{equation}
\label{eqn:dirichlet-prior}
p\left(\mathcal{L}\right) \sim p(\tau) \prod_{v=1}^{n-1} \lambda(v)^{\alpha-1} ,
\end{equation}
where $\alpha$ is a so-called shape parameter.
Then, we have the following lemma.

\begin{lemma}
\label{lem:reg-sdp-stat-prior-pagerank}
Given the conditional likelihood for $L$ given $\mathcal{L}$ in Eqn.~(\ref{eqn:lap-cond-likelihood}) and the prior density for $\mathcal{L}$ given in Eqn.~(\ref{eqn:dirichlet-prior});
if $\hat{\mathcal{L}}$ is the MAP estimate of $\mathcal{L}$, then 
\[
\frac{ \hat{\mathcal{L}}^{+} }{ \Trace{ \hat{\mathcal{L}}^{+} } }
\] 
solves the regularized SDP, with $G(X) = - -\log\det(X)$ and with the value of $\eta$ given in the proof~below.
\end{lemma}
\begin{proof}
For $\mathcal{L}$ in the support set of the posterior, we can define $\tau = \Trace{ \mathcal{L}^{+} }$ and $\Theta = \frac{1}{\tau} \mathcal{L}^{+}$, so that 
$\mbox{rank}(\Theta) = n-1$ and $\Trace{\Theta} = 1$.
Then, $p\left(\mathcal{L}\right) \sim \exp \{ -U\left( \mathcal{L} \right) \}$, where 
\[
U\left(\mathcal{L}\right) 
   = - \log \{ p(\tau) \det\left(\Theta\right)^{\alpha-1} \} 
   = -(\alpha-1) \log \det \left(\Theta\right) - \log \left( p(\tau) \right)  .
\]
Thus, 
\begin{eqnarray*}
p\left( \mathcal{L} | L \right) 
   &\sim& \exp \{ -\frac{m}{2}\Trace{ L\mathcal{L}^{+} } + \frac{m}{2} \log\det\left(\mathcal{L}^{+}\right) - U\left(\mathcal{L}\right) \} \\ 
   &\sim& \exp \{ \frac{-m\tau}{2}\Trace{L\Theta} + \frac{m+2(\alpha-1)}{2}\log\det\left(\Theta\right) + g(\tau) \}  ,
\end{eqnarray*}
where the second line follows since $\det\left(\mathcal{L}^{+}\right) = \tau^{n-1}\det\left(\Theta\right)$, and where $g(\tau) = \frac{m(n-1)}{2}\log(\tau) + \log p(\tau)$.
If $\hat{\mathcal{L}}$ maximizes the posterior likelihood, then define $\hat{\tau} = \Trace{\hat{\mathcal{L}^{+}}}$ and $\hat{\Theta} = \frac{1}{\tau}\hat{\mathcal{L}}^{+}$, and so $\hat{\Theta}$ must minimize
$\Trace{L\hat{\Theta}} - \frac{1}{\eta} \log\mbox{det}\left( \hat{\Theta} \right)$, where
\[
\eta = \frac{ m \hat{\tau} }{ m+2(\alpha-1) }
\]
This $\hat{\Theta}$ solves the regularized SDP with $G(x) = -\log\det\left(X\right)$.
\end{proof}

\textbf{Remark.}
Lemma~\ref{lem:reg-sdp-stat-prior-pagerank} provides a statistical interpretation of the regularized problem that is optimized by an approximate PageRank diffusion algorithm, in the sense that it gives a general statistical estimation procedure that leads to the Rayleigh quotient as well as statistical prior related to PageRank.
One can write down priors for the Heat Kernel and other random walks; see the two references if you are interested.
Note, however, that the prior for PageRank makes things particularly simple. 
The reason is that the extra term in Problem~\ref{eqn:sdp-stat}, i.e., the $\log\det\left(X\right)$ term, is of the same form as the regularization function that the approximate PageRank computation implicitly regularizes with respect to.
Thus, we can choose parameters to make this term cancel. 
Otherwise, there are extra terms floating around, and the statistical interpretation is more complex.

%% file: lect18.tex
\section{%
(03/31/2015): 
Local Spectral Methods (1 of 4):
Introduction and Overview}

Reading for today.
\begin{compactitem}
\item
``Spectral Ranking'', in arXiv, by Vigna
\item
``PageRank beyond the Web,'' in SIAM Review, by Gleich
\end{compactitem}

Last time, we showed that certain random walks and diffusion-based methods, \emph{when not run to the asymptotic limit}, exactly solve regularized versions of the Rayleigh quotient objective (in addition to approximating the Rayleigh quotient, in a manner that depends on the specific random walk and how the spectrum of $L$ decays).
There are two ways to think about these results.
\begin{itemize}
\item
One way to think about this is that one runs almost to the asymptotic state and then one gets a vector that is ``close'' to the leading eigenvector of $L$.
Note, however, that the statement of implicit regularization from last time does \emph{not} depend on the initial condition or how long the walk was run.
(The value of the regularization parameter, etc., does, but the form of the statement does not.)
Thus ...
\item
Another way to think about this is that one starts at any node, say a localized ``seed set'' of nodes, e.g., in which all of the initial probability mass is on one node or a small number of nodes that are nearby each other in the graph topology, and then one runs only a small number of steps of the random walk or diffusion.
In this case, it might be more natural/useful to try to quantify the idea that: if one starts the random walk on the small side of a bottleneck to mixing, and if one runs only a few steps of a random walk, then one might get stuck in that small set.
\end{itemize}

The latter is the basic idea of so-called \emph{local spectral methods}, which are a class of algorithms that have received a lot of attention recently. 
Basically, they try to extend the ideas of global spectral methods, where we compute eigenvectors, random walks, etc., that reveal structure about the entire graph, e.g., that find a partition that is quadratically-good in the sense of Cheeger's Inequality to the best conductance/expansion cut in the graph, to methods that reveal interesting structure in locally-biased parts of the graph.
Not only do these provide locally-biased versions of global spectral methods, but since spectral methods are often used to provide a ranking for the nodes in a graph and/or to solve other machine learning problems, these also can be used to provide a locally-biased or personalized version of a ranking function and/or to solve other machine learning problems in a locally-biased manner.

\subsection{Overview of local spectral methods and spectral ranking}

Here is a brief history of local and locally-biased spectral methods.
\begin{itemize}
\item
LS: developed a basic locally-biased mixing result in the context of mixing of Markov chains in convex bodies.
They basically show a partial converse to the easy direction of Cheeger's Inequality---namely, that if the conductance $\phi(G)$ of the graph $G$ is big, then \emph{every} random walk must converge quickly---and from this they also show that if the random walk fails to converge quickly, then by examining the probability distribution that arises after a few steps one can find a cut of small conductance.
\item
ST: used the LS result to get an algorithm for local spectral graph partitioning that used truncated random walks.
They used this to find good well-balanced graph partitions in nearly-linear time, which they then used as a subroutine in their efforts to develop nearly linear time solvers for Laplacian-based linear systems (a topic to which we will return briefly at the end of the semester).
\item
ACL/AC: improved the ST result by computing a personalized PageRank vector.
This improves the fast algorithms for Laplacian-based linear solvers, and it is of interest in its own right, so we will spend some time on it.
\item
C: showed that similar results can be obtained by doing heat kernel computations.
\item
AP: showed that similar results can be obtained with an evolving set method (that we won't discuss in detail).
\item
MOV: provided an optimization perspective on these local spectral methods.
That is, this is a locally-biased optimization objective that, if optimized exactly, leads to similar locally-biased Cheeger-like bounds.
\item
GM: characterized the connection between the strongly-local ACL and the weakly-local MOV in terms of $\ell_1$ regularization (i.e., a popular form of sparsity-inducing regularizations) of $\ell_2$ regression problems.
\end{itemize}

There are several reasons why one might be interested in these methods.
\begin{itemize}
\item
Develop faster algorithms.
This is of particular interest if we can compute locally-biased partitions without even touching all of the graph.
This is the basis for a lot of work on nearly linear time solvers for Laplacian-based linear systems.
\item
Improved statistical properties.
If we can compute locally-biased things, e.g., locally-biased partitions, without even touching the entire graph, then that certainly implies that we are robust to things that go on on the other side of the graph.
That is, we have essentially engineered some sort of regularization into the approximation algorithm; and it might be of interest to quantify this.
\item
Locally exploring graphs.
One might be interested in finding small clusters or partitions that are of interest in a small part of a graph, e.g., a given individual in a large social graph, in situations when those locally-biased clusters are not well-correlated with the leading or with any global eigenvector.
\end{itemize}
We will touch on all these themes over the next four classes.

For now, let $G=(V,E)$, and recall that $\mbox{Vol}(G) = \sum_{v \in T} d_v$ (so, in particular, we have that $\mbox{Vol}(G) = 2 |E| = 2m$.
Also, $A$ is the Adjacency Matrix, $W = D^{-1}A$, and $\mathcal{L} = I - D^{-1}A$ is the random walk normalized Laplacian.
For a vector $x\in\mathbb{R}^{n}$, let's define its support as 
\[
\mbox{Supp}(v) = \{ i \in V = [n] : v_i \ne 0 \}  .
\]
Then, here is the transition kernel for that vanilla random walk.
\[
\mathbb{P}\left[ x_{t+1}=j | x_t = i \right] = 
   \begin{cases}
    \frac{1}{d} \quad \text{if }  i \sim j \\
    0 \quad \text{otherwise}  .
   \end{cases}  
\]
If we write this as a transition matrix operating on a (row) vector, then we have that
\[
p(t) = s W^t  ,
\]
where $W$ is the transition matrix, and where $s=p(0)$ is the initial distribution, with $\|s\|_1 = 1$.
Then, $p = \vec{1}^T D/(\vec{1}^T D \vec{1})$ is the stationary distribution, i.e., 
\[
\lim_{t\rightarrow\infty} \mathbb{P}\left[ x_t = i \right] = \frac{d_i}{\mbox{Vol}(G)} ,
\]
independent of $s = p(0)$, as long as $G$ is connected and not bipartite.
(If it is bipartite, then let $W \rightarrow W_{LAZY} = \frac{1}{2} \left( I + D^{-1}A \right)$, and the same results holds.)

There are two common interpretations of this asymptotic random walk.
\begin{itemize}
\item
Interpretation 1: the limit of a random walk.
\item
Interpretation 2: a measure of the importance of a node.
\end{itemize}
With respect to the latter interpretation, think of an edge as denoting importance, and then what we want to find is the important nodes (often for directed graphs, but we aren't considering that here).
Indeed, one of the simplest \emph{centrality measures} in social graph analysis is the degree of a node.
For a range of reasons, e.g., since that is easy to spam, a refinement of that is to say that important nodes are those nodes that have links to important nodes.

This leads to  a large area known as \emph{spectral ranking} methods.
This area applies the theory of matrices or linear maps---basically eigenvectors and eigenvalues, but also related things like random walks---to matrices that represent some sort of relationship between entities.
This has a long history, most recently made well-known by the PageRank procedure (which is one version of it).
Here, we will follow Vigna's outline and description in his ``Spectral Ranking'' notes---his description is very nice since it provides the general picture in a general context, and then he shows that with several seemingly-minor tweaks, one can obtain a range of related spectral graph~methods.

\subsection{Basics of spectral ranking}

To start, take a step back, and let $M \in \mathbb{R}^{n \times n}$, where each column/row of $M$ represents some sort of entity, and $M_{ij}$ represents some sort of \emph{endorsement or approval} of entity $j$ from entity $i$.
(So far, one could potentially have negative entries, with the obvious interpretation, but this will often be removed later, basically since one has more structure if entries must be nonnegative.)
As Vigna describes, Seeley (in 1949) observed that one should define importance/approval recursively, since that will capture the idea that an entity is important/approved if other important entities think is is important/approved.
In this case, recursive could mean that
\[
r = rM ,
\]
i.e., that the index of the $i^{th}$ node equals the weighted sum of the indices of the entities that endorse the $i^{th}$ node.
This isn't always possible, and indeed Seeley considers nonnegative matrices that don't have any all-zeros rows, in which case uniqueness, etc., follow from the Perron-Frobenius ideas we discussed before.
This involves the left eigenvectors, as we have discussed; one could also look at the right eigenvectors, but the endorsement/approval interpretation fails to hold.

Later, Wei (in 1952) and Kendall (in 1955) were interested in \emph{ranking} sports teams, and they said essentially that better teams are those teams that beat better teams.
This involves looking at the rank induced by 
\[
\lim_{k\rightarrow\infty} M^k \vec{1}^T
\]
and then appealing to Perron-Frobenius theory.
The significant point here is three-fold.
\begin{itemize}
\item
The motivation is very different than Seeley's endorsement motivation.
\item
Using dominant eigenvectors on one side or the other dates back to mid 20th century, i.e., well before recent interest in the topic.
\item
The relevant notion of convergence in both of these motivating applications is that of \emph{convergence in rank} (where rank means the rank order of values of nodes in the leading eigenvector, and not anything to do with the linear-algebraic rank of the underlying matrix).
In particular, the actual values of the entires of the vector are not important.
This is very different than other notions of convergence when considering leading eigenvectors of Laplacians, e.g., the value of the Rayleigh quotient.
\end{itemize}

Here is the generalization.
Consider matrices $M$ with real and positive dominant eigenvalue $\lambda$ and its eigenvector, i.e., vector $r$ such that $ \lambda r = rM$, where let's say that the dimension of the eigenspace is one.

\begin{definition}
The \emph{left spectral ranking} associated with $M$ is, or is given by, the dominant left eigenvector.
\end{definition}

If the eigenspace does not have dimension one, then there is the usual ambiguity problem (which is sometimes simply assumed away, but which can be enforced by a reasonable rule---we will see a common way to do the latter in a few minutes), but if the eigenspace has dimension one, then we can talk of \emph{the} spectral ranking.
Note that it is defined only up to a constant: this is not a problem if all the coordinates have the same sign, but it introduces an ambiguity otherwise, and how this ambiguity is resolved can lead to different outcomes in ``boundary cases'' where it matters; see the Gleich article for examples and details of this.

Of course one could apply the same thing to $M^T$.  
The mathematics is similar, but the motivation is different.
In particular, Vigna argues that the endorsement motivation leads to the left eigenvectors, while the influence/better-than motivation leads to the right eigenvectors.

The next idea to introduce is that of ``damping.''
(As we will see, this will have a reasonable generative ``story'' associated with it, and it will have a reasonable statistical interpretation, but it is also important for a technical reason having to do with ensuring that the dimension of the eigenspace is one.)
Let $M$ be a zero-one matrix; then $\left( M^k \right)_{ij}$ is the number of directed paths from $i \rightarrow j$ in a directed graph defined by $M$.
In this case, an obvious idea of measuring the importance of $i$, i.e., to measure the number of paths going into $j$, since they represent recursive endorsements, which is given by
\[
\vec{1}\left( I + M + M^2 + \cdots \right) = I + \sum_{k=0}^{\infty} M^{k}   ,
\]
does \emph{not} work, since the convergence of the equation is not guaranteed, and it does not happen in~general.

If, instead, one can guarantee that the spectral radius of $M$ is less than one, i.e., that $\lambda_0 < 1$, then this infinite sum does converge.
One way to do this is to introduce a damping factor $\alpha$ to obtain
\[
\vec{1}\left( I + \alpha M + \alpha^2 M^2 + \cdots \right) = I + \sum_{k=0}^{\infty} \left( \alpha M\right)^{k}  .
\]
This infinite sum does converge, as the spectral radius of $\alpha M$ is strictly less than $1$, if $\alpha < \frac{1}{\lambda_0}$.
Katz (in 1953) proposed this.
Note that 
\[
\vec{1} \sum_{k=0}^{\infty} \left( \alpha M \right)^{k} = \vec{1} \left( I - \alpha M \right)^{-1} .
\]
So, in particular, we can compute this index by solving the \emph{linear system} as follows.
\[
x \left( I - \alpha M \right) = \vec{1} .
\]
This is a particularly well-structured system of linear equations, basically a system of equations where the constraint matrix can be written in terms of a Laplacian.
There has been a lot of work recently on developing fast solvers for systems of this form, and we will get back to this topic in a few weeks.

A generalization of this was given by Hubbel (in 1965), who said that one could define a status index $r$ by using a recursive equation $r = v+rM$, where $v$ is a ``boundary condition'' or ``exogenous contribution.''
This gives 
\[
r = v \left( I - M \right)^{-1} = v \sum_{k=0}^{\infty} M^k  .
\]
So, we are generalizing from the case where the exogenous contribution is $\vec{1}$ to an arbitrary vector $v$.
This converges if $\lambda_0 < 1$; otherwise, one could introduce an damping factor (as we will do below).

To see precisely how these are all related, let's consider the basic spectral ranking equation
\[
\lambda_0 r = r M  .
\]
If the eigenspace of $\lambda_0$ has dimension greater than $1$, then there is no clear choice for the ranking $r$.
One idea in this case is to perturb $M$ to satisfy this property, but we want to apply a ``structured perturbation'' in such a way that many of the other spectral properties are not damaged.
Here is the relevant theorem, which is due to Brauer (in 1952), and which we won't prove.

\begin{theorem}
Let $A \in \mathbb{C}^{n \times n}$, let 
\[
\lambda_0, \lambda_1, \ldots, \lambda_{n-1}
\] 
be the eigenvalues of $A$, and let $x\in\mathbb{C}^{n}$ be nonzero  vector such that $Ax^T = \lambda_0 x^T$.
Then, for all vectors $v\in\mathbb{C}^{n}$, the eigenvalues of $A+x^Tv$ are given by 
\[
\lambda_0 + vx^T, \lambda_1, \ldots, \lambda_{n-1}  .
\] 
\end{theorem}
That is, if we perturb the original matrix by an rank-one update, where the rank-one update is of the form of the outer product of an eigenvector of the matrix and an arbitrary vector, then one eigenvalue changes, while all the others stay the same.

In particular, this result can be used to split degenerate eigenvalues and to introduce a \emph{gap} into the spectrum of $M$.
To see this, let's consider a rank-one convex perturbation of our matrix $M$ by using a vector $v$ such that $vx^T = \lambda_0$ and by applying the theorem to $\alpha M$ and $\left(1-\alpha\right)x^Tv$.
If we do this then we get 
\[
\lambda_0 r = r \left( \alpha M + \left(1-\alpha\right) x^Tv \right) .
\]
Next, note that $\alpha M + \left(1-\alpha\right) x^Tv$ has the same dominant eigenvalues as $M$, but with algebraic multiplicity $1$, and all the other eigenvalues are multiplied by $\alpha \in (0,1)$.

This ensures a unique $r$.
The cost is that it introduces extra parameters ($\alpha$ is we set $v$ to be an all-ones vector, but the vector $v$ if we choose it more generally).
These parameters can be interpreted in various ways, as we will see.

An important consequence of this approach is that $r$ is defined only up to a constant, and so we can impose the constraint that $rx^T=\lambda_0$.
(Note that if $x = \vec{1}$, then this says that the sum of $r$'s coordinates is $\lambda_0$, which if all the coordinates have the same sign means that $\|r\|_1 = \lambda_0$.)
Then we get 
\[
\lambda_0 r = \alpha r M + \left(1-\alpha\right) \lambda_0 v .
\]
Thus, 
\[
r \left( \lambda_0 - \alpha M \right) = \left( 1-\alpha \right) \lambda_0 v  .
\]
From this it follows that 
\begin{eqnarray*}
r &=& \left(1-\alpha\right) \lambda_0 v \left( \lambda_0 - \alpha M \right)^{-1}  \\
  &=& \left(1-\alpha\right) v \left( 1 - \frac{\alpha}{\lambda_0} M \right)^{-1}  \\
  &=& \left(1-\alpha\right) v \sum_{k=0}^{\infty} \left( \frac{\alpha}{\lambda_0} M \right)^{k}  \\
  &=& \left(1-\lambda_0\beta\right) v \sum_{k=0}^{\infty} \left( \beta M \right)^{k}  ,
\end{eqnarray*}
which converges if $\alpha < 1$, i.e., if $\beta < \frac{1}{\lambda_0}$.

That is, this approach shows that the Katz-Hubbel index can be obtained as a rank-one perturbation of the original matrix.
In a little bit, we will get to what this rank-one perturbation ``means'' in different situations.

To review, we started with a matrix $M$ with possibly many eigenvectors associated with the dominant eigenvalue, and we introduced a structured perturbation get a specific eigenvector associated with $\lambda_0$, given the boundary condition~$v$.

The standard story is that if we start from a generic nonnegative matrix and normalize its rows to get a stochastic matrix, then we get a Markovian spectral ranking, which is the limit distribution of the random walk.
Here, we are slightly more general, as is captured by the following definition.

\begin{definition}
Given the matrix $M$, the sampled spectral ranking of $M$ with boundary condition $v$ and dumping factor $\alpha$ is
\[
r_0 = \left(1-\alpha\right) v \sum_{k=0}^{\infty} \left( \frac{\alpha}{\lambda_0} M \right)^{k}  ,
\]
for $|\alpha| < 1$.
\end{definition}

The interpretation of the boundary condition (from the sampled to the standard case) is the following.
\begin{itemize}
\item
In the damped case, the Markov chain is restarted to a fixed distribution $v$, and there is a single stationary distribution which is the limit of every starting distribution.
\item
In the standard case, $v$ is the starting distribution from which we capture the limit using an eigenprojection.
\end{itemize}
While in some sense equivalent, these two interpretations suggest different questions and interpretations, and we will consider both of these over the next few classes.

\subsection{A brief aside}

Here is an aside we will get back to over the next few classes.

Consider a vanilla random walk, where at each time step, we follow the graph with some probability $\beta$, and we randomly jump to any uniformly-chosen node in the graph with probability $1-\beta$.
This Markov chain has stationary distribution 
\[
p_{\beta} = \frac{1-\beta}{n}\vec{1} + \beta p_{\beta} W .
\]
This is often known as PageRank, which has received a lot of attention in web ranking, but from the above discussion it should be clear that it is one form of the general case of spectral ranking methods.
We can also ask for a ``personalized'' version of this, by which we informally mean a ranking of the nodes that in some sense is conditioned on or personalized for a given node or a given seed set of nodes.
We can get this by using a personalized PageRank, by randomly jumping (not to any node chosen uniformly at random but) to a ``seed set'' $s$ of nodes.
This PPR is the unique solution to 
\[
p_{\beta}(s) = \left(1-\beta\right)s + \beta p_{\beta}(s) W  , 
\]
i.e., it is of the same form as the expression above, except that an all-ones vector has been replaced by a seed set or boundary condition vector $V$.
This latter expression solves the linear equation
\[
\left( I - \beta W \right) p_{\beta} (s) = \left( 1-\beta \right) s  .
\]
We can write this expression as an infinite sum as 
\[
p_{\beta}(s) = \left( 1-\beta \right) s + \beta \sum_{t=0}^{\infty} \beta^t \left(sW\right)^t
\]
Thus, note that the following formulations of PageRank (as well as spectral ranking more generally) are equivalent.
\begin{itemize}
\item
$\left(I-\beta W \right) x = \left( 1-\beta \right) s$
\item
$\left( \gamma D + L \right) z = \gamma s$, where $\beta=\frac{1}{1+\gamma}$ and $x = Dz$.
\end{itemize}
As noted above, this latter expression is of the form of Laplacian-based linear equations.
It is of the same form that arises in those semi-supervised learning examples that we discusses.
We will talk toward the end of the term about how to solve equations of this form more generally.

%% file: lect19.tex
\section{%
(04/02/2015):
Local Spectral Methods (2 of 4):
Computing spectral ranking with the push procedure}

Reading for today.
\begin{compactitem}
\item
``The Push Algorithm for Spectral Ranking'', in arXiv, by Boldi and Vigna
\item
``Local Graph Partitioning using PageRank Vectors,'' in FOCS, by Andersen, Chung, and Lang
\end{compactitem}

Last time, we talked about spectral ranking methods, and we observed that they can be computed as eigenvectors of certain matrices or as the solution to systems of linear equations of certain constraint matrices.
These computations can be performed with a black box solver, but they can also be done with specialized solvers that take advantage of the special structure of these matrices.
(For example, a vanilla spectral ranking method, e.g., one with a preference vector $\vec{v}$ that is an all-ones vector $\vec{1}$, has a large eigenvalue gap, and this means that one can obtain a good solution with just a few steps of a simple iterative method.)
As you can imagine, this is a large topic.
Here, we will focus in particular on how to solve for these spectral rankings in the particular case when the preference vector $\vec{v}$ is has small support, i.e., when it has its mass localized on a small seed set of nodes.
This is a particularly important use case, and the methods developed for it are useful much more generally.

In particular, in this class, we will describe how to compute personalized/local spectral rankings with a procedure called the \emph{push procedure}.
This has several interesting properties, in general, but in particular when computing locally-biased or personalized spectral rankings in a large graph.
As with the last class, we will follow the Boldi-Vigna's notes on ``The Push Algorithm for Spectral Ranking'' on the topic.

\subsection{Background on the method}

To start, let $M \in \mathbb{R}^{n \times n}$ be a nonnegative matrix; and WLOG assume that $\| M \|_1 =1$, i.e., assume that $M$ is stochastic.
(Actually, this assumes that $M$ is sub-stochastic and that at least one row of $M$ sums to $1$; this distinction won't matter for what we do, but it could matter in other cases, e.g., when dealing with directed graphs, etc.)
Then, let $v$ be a nonnegative vector s.t. $\|v\|_1=1$, i.e., assume that $v$ is a distribution, and let $\alpha \in [0,1)$.

Then, recall that we defined the spectral ranking of $M$ with preference vector $v$ and damping factor $\alpha$ to be
\begin{eqnarray*}
r &=& \left(1-\alpha\right) v \left(I-\alpha M\right)^{-1} \\
  &=& \left(1-\alpha\right) v \sum_{k=0}^{\infty} \alpha^k M^k  .
\end{eqnarray*}
For this, note that $r$ need not be a distribution unless $M$ is stochastic (although it is usually applied in that case).
Note also that the linear operator is defined for all $\alpha \in \left[0,\frac{1}{\rho(M)} \right)$, but estimating $\rho(M)$ can be difficult.
(We'll see an interesting example next week of what arises when we push this parameter to the limit.)

From this perspective, while it is difficult to ``guess'' what is the spectral ranking $r$ associated with a preference vector $v$, it is ``easy'' to solve the inverse problem:
given $r$, solve 
$$
v = \frac{1}{1-\alpha} r \left(I-\alpha M \right) .
$$
(While this equation is always true, the resulting preference vector might not be a distribution; otherwise, one could obtain any spectral ranking using a suitable preference vector.)

While this is obvious, it has important consequences for computing spectral rankings and approximate spectral rankings that we will now describe.

Consider the indicator vector $\mathcal{X}_x(z) = [ x=z]$.
If we want to obtain the vector $\left(1-\alpha\right)\mathcal{X}_x$ as a spectral ranking, then the associated preference vector $v$ has the simple form:
\begin{eqnarray}
\label{eqn:pref-vect}
v &=& \frac{1}{1-\alpha}\left(1-\alpha\right) \mathcal{X}_x \left( I-\alpha M\right) \\
\nonumber
  &=& \mathcal{X}_x - \alpha \sum_{x \rightarrow y} M_{xy} \mathcal{X}_y \\
\nonumber
  &=& \mathcal{X}_x - \alpha \frac{1}{d(x)} \sum_{x \rightarrow y} \mathcal{X}_y  , 
\end{eqnarray}
where note that if $M_{xy} = \frac{1}{d(x)}$ is the natural random walk, then we can take it out of the summation.

This is true in general.
The important point for us here is that if $v$ is highly concentrated, e.g., if it is an indicator vector of $s$ small set of nodes, and if $\alpha$ is not too close to $1$, then most of the updates done by the linear solver or by an iterative method either don't update much or update below a threshold level.
Motivated by this, we will discuss the so-called \emph{push algorithm}---this uses a particular form of updates to reduce computational burden.
This was used in work of Jeh and Widom and also of Berkhin on personalized page rank, and it was used with this name by ACL.
Although they all apply it to PageRank, it applies for the steady state of Markov chains with restart, and it is basically an algorithm for spectral ranking with damping (see the Boldi, Lonati, Santini, Vigna paper).

\subsection{The basic push procedure}

The basic idea of this approach is that, rather that computing an \emph{exact} PPR vector by iterating the corresponding Markov chain (e.g., with vanilla matrix-vector multiplies) until it converges, it is also possible to consider computing an \emph{approximate} PPR vector much more efficiently.
Recall that the PPR vector is the unique solution to 
\[
\pi_{\alpha}(s) = \left(1-\alpha\right)s + \alpha \pi_{\alpha}(s) W  ,
\]
and it can be written as an infinite sum
\[
\pi_{\alpha}(s) = \left(1-\alpha\right)s + \left(1-\alpha\right) \sum_{t=1}^{\infty} \alpha^t \left(s W\right)^t .
\]
With this notation, we can define the following notion of approximation of a PPR vector.
\begin{definition}
An \emph{$\epsilon$-approximate PageRank vector} $\pi_{\alpha}(s)$ is any PageRank vector $\pi_{\alpha}(s-r)$, where $r$ is nonnegative and $r(v) \le \epsilon d_v$, for all $v \in V$.
\end{definition}

\textbf{Fact.} 
The approximation error of an $\epsilon$-approximate PageRank vector on any set of nodes $S$ can be bounded in terms of the $\mbox{Vol}(S)$ and $\epsilon$.
Here is a basic lemma from ACL.
\begin{lemma}
For all $\epsilon$-approximate PageRank vectors $\pi_{\alpha}(s-r)$, and for all $S \subset V$, we have
\[
\pi_{\alpha}(s) 1^T_S \ge \pi_{\alpha}(s-r) 1^T_S \ge \pi_{\alpha}(s) 1^T_S - \epsilon \mbox{Vol}(S)  .
\]
\end{lemma}

Here is an algorithm to compute an $\epsilon$-approximate PageRank vector; let's call this algorithm \textsc{ApproxPR($s,\alpha,\epsilon$)}.
\begin{enumerate}
\item
Let $p=\vec{0}$ and $r = \vec{s}$.
\item
While $r_u \ge \epsilon d_u$ for some vertex $u$,
\begin{itemize}
\item
Pick and $u$ such that $r_u \ge \epsilon d_n$
\item
Apply \textsc{Push($u$)}
\end{itemize}
\item
Return the vectors $p$ and $r$
\end{enumerate}
And here is the \textsc{Push($u$)} algorithm that is called by the \textsc{ApproxPR($s,\alpha,\epsilon$)} algorithm.
\begin{enumerate}
\item
Let $p^{\prime} = p$ and $r^{\prime} = r$, except for the following updates:
\begin{compactitem}
\item
$p_u^{\prime} = p_u + \alpha r_u$
\item
$r_u^{\prime} = \left( 1-\alpha \right) \frac{r_u}{2}$
\item
$r_v^{\prime} = r_v + \left(1-\alpha\right) \frac{r_u}{2d_u}$, for all vertices $v$ such that $(u,v) \in E$.
\end{compactitem}
\end{enumerate}
Note that we haven't specified the order in which the pushes are executed, i.e., in which the \textsc{Push($u$)} algorithm is called by the \textsc{ApproxPR($s,\alpha,\epsilon$)} algorithm, and so they can be done in different ways, leading to slightly different algorithms.

Here is the theorem that ACL establishes about this procedure.
\begin{theorem}
Algorithm \textsc{ApproxPR($s,\alpha,\epsilon$)} has the following properties.
\begin{itemize}
\item
For all starting vertices with $\|s\|_1 \le 1$ and for all $\epsilon \in (0,1]$, the algorithm returns an $\epsilon$-approximate $p$ for $p_{\alpha}(s)$.
\item
The support of $p$ satisfies 
\[
\mbox{Vol(Supp}(p)) \le \frac{2}{(1+\alpha)\epsilon}  .
\]
\item
The running time of the algorithm is $O\left( \frac{1}{\alpha\epsilon} \right)$.
\end{itemize}
\end{theorem}
The idea of the proof---outlined roughly above---is that Algorithm \textsc{Push($u$)}  preserves the approximate PageRank condition $p = \pi \left( s-r \right)$; and the stopping condition ensures that it is an $\epsilon$ approximation.
The running time follows since $\|r\|_1=1$ and it decreases by $\alpha\epsilon d_n$ at each time \textsc{Push($u$)}  is called.

\subsection{More discussion of the basic push procedure}

The basic idea of the method is to keep track of two vectors, $p$ which is the vector of current approximations and $r$ which is the vector of residuals, such that the following global invariant is satisfied at every step of the algorithm.
\begin{equation}
\label{eqn:invariant}
p  + \left(1-\alpha\right) r \left(I-\alpha M\right)^{-1} = \left(1-\alpha\right) v \left(I-\alpha M\right)^{-1}  .
\end{equation}
Initially, $p=0$ and $r=v$, and so this invariant is trivially satisfied.
Subsequently, at each step the push method increases $p$ and decreases $r$ to keep the invariant satisfied.

To do so, it iteratively ``pushes'' some node $x$.
A \emph{push on $x$} adds $\left(1-\alpha\right) r_x \mathcal{X}_x$ to the vector $p$.
To keep the invariant try, the method must update $r$.
To do so, think of $r$ as a preference vector, in which we are just trying to solve Eqn.~(\ref{eqn:pref-vect}).
By linearity, if we subtract 
\[
r_x \left( \mathcal{X}_x - \alpha \sum_{x \rightarrow y} M_{xy} \mathcal{X}_y \right)
\]
from $r$, then the value $\left( 1-\alpha \right) r \left( I-\alpha M \right)^{-1}$ will decrease by $\left(1-\alpha\right) r_x \mathcal{X}_x$, thus preserving the invariant.

This is a good choice since
\begin{compactitem}
\item
We zero out the $x^{th}$ entry of $r$.
\item
We add a small positive quantities to a small set of entries (small set, if the graph is sparse, which is the use case we are considering).
This increases the $\ell_1$ norm of $p$ by $\left( 1-\alpha\right)r_x$, and it decreases $r$ by at least the same amount.
\end{compactitem}

Since we don't create negative entries in this process, it always holds that 
\[
\|p\|_1 + \|r\|_1 \le 1 ,
\]
and thus we can keep track of two norms at each update.
The error in the estimate is then given by the following.
\begin{eqnarray*}
\| \left(1-\alpha \right) r \left( I - \alpha M \right)^{-1} \|_1
  &=& \left(1-\alpha\right) \| r \sum_{k \ge 0} \alpha^k M^k \|_1   \\
  &\le& \left(1-\alpha\right) \|r \|_1 \sum_{k \ge 0} \alpha^k \| M^k \|_1   \\
  &\le& \| r \|_1
\end{eqnarray*}
So, in particular, we can control the absolute error of the algorithm by controlling the $\ell_1$ error of the residual.

\subsection{A different interpretation of the same process}

Here is a different interpretation of the push method.
Some might find useful---if so, good, and if not, then just ignore the following.

Just as various random walks can be related to diffusion of heat/mass, one can think of PageRank as related to the diffusion of a substance where some fraction of that substance gets stuck in place at each time step---e.g., diffusing paint, where the point of paint is that part of the paint dries in place and then stops diffusing/flowing.
(Think here of doing the updates to push asynchronously.)
At each time step, $1-\alpha$ fraction of the paint dries in place, and $\alpha$ fraction of the paint does a lazy random walk.
So, we need to keep track of two quantities: the amount of wet paint (that still moves at the next step) and the amount of dry paint (that is probability mass that won't move again).
If we let 
\begin{itemize}
\item
$p : V \rightarrow \mathbb{R}^{n}$ be a vector that says how much pain is stuck at each vertex.
\item
$r : V \rightarrow \mathbb{R}^{n}$ be a vector saying how much wet paint remains at each vertex.
\end{itemize} 
Then, at $t=0$, $r^0 = \mathcal{X}_u$, and these vectors evolve as
\begin{eqnarray*}
p^{t+1} &=& p^t + \left(1-\alpha\right) r^t \\
r^t &=& \alpha r \hat{W}  .
\end{eqnarray*}
(Note that I have swapped sides where the vector is multiplying, so I think I have inconsistencies here to fix.  Also, I think I refer to $\alpha$ and $1-\alpha$ inconsistently, so that must be fixed too)

Given this, if we let $p^{\infty}$ be where paint is dried at the end of the process, then 
\begin{eqnarray*}
p^{\infty} 
  &=& \left(1-\alpha\right) \sum_{t \ge 0} r^t \\
  &=& \left( 1-\alpha \right) \sum_{t \ge 0} \alpha^t r^0 \hat{W}^t  \\
  &=& \sum_{t \ge 0} \alpha^t \mathcal{X}_u \hat{W}^t
\end{eqnarray*}
This is simply PPR, with a different scaling and $\alpha$.
Recall here that $\hat{W} = \frac{1}{2}I+\frac{1}{2}W$.
Since $\left(I+X\right)^{-1} = \sum_{i=0}^{\infty} X^{i}$, if the spectral radius of $X$ is less than $1$, 
it follows that 
\begin{eqnarray*}
p^{\infty} &=& \left(1-\alpha\right) \mathcal{X}_u \left(I - \alpha \hat{W} \right)^{-1} \\
  &=& \left(1-\alpha \right) \mathcal{X}_u \left( \left(1-\frac{\alpha}{2} \right) I - \frac{\alpha}{2}W \right)^{-1} \\
  &=& \gamma \mathcal{X}_u \left( I - \left( 1-\gamma \right) W \right)^{-1}   ,
\end{eqnarray*}
where the parameter $\gamma$ is defined in terms of the parameter $\alpha$.

Note that in this we don't need to update the above time process but we can ignore time and do it in an asynchronous manner.
That is, we can compute this by solving a linear equation or running a random walk in which one keeps track of two vectors that has this interpretation in terms of diffusing paint.
But, Jeh-Widom and Berkhin note that rather than doing it with these equations, one can instead choose a vertex, say that a fraction $\alpha$ of paint at that vertex is dry, and then push the wet paint to the neighbors according to the above rule.
(That is, we can ignore time and do it asynchronously.)
But this just gives us the push process.
To see this, let $\pi_{p,r}$ be the vector of ``dried paint'' that we eventually compute.  
Then
\begin{eqnarray*}
\pi_{p,r} &=& p + \left(1-\alpha\right) \sum_{t \ge 0} r \alpha^t W^t \\
   &=& p + \left( 1-\alpha \right) r \left( I -\alpha W \right)^{-1} 
\end{eqnarray*}
In this case, the updates we wrote above, written another way, are the following: pick a vertex $u$ and create $p^{\prime}$ and $r^{\prime}$ as follows.
\begin{itemize}
\item
$p^{\prime}(u) = p(u) + \left(1-\alpha\right) r(u)$
\item
$r^{\prime}(u) = 0$
\item
$r^{\prime}(v) = r(v) + \frac{\alpha}{d(u)} r(u)$, for all neighbors $v$ of $u$.
\end{itemize}
Then, it can be shown that
\[
\pi_{p^{\prime},r^{\prime}} = \pi_{p,r}  ,
\]
which is the invariant that we noted before.

So, the idea of computing the approximate PageRank vectors is the following.
\begin{itemize}
\item
Pick a vertex where $r(u)$ is large or largest and distribute the paint according to that rule.
\item
Choose a threshold $\epsilon$ and don't bother to process a vertex if $r(u) \le \epsilon d(u)$.
\end{itemize}
Then, it can be shown that
\begin{lemma}
The process will stop within $\frac{1}{\epsilon\left(1-\alpha\right)}$ iterations.
\end{lemma}
(Note there is an inconsistency with $\alpha$ and $1-\alpha$, i.e., whether $\alpha$ is the teleporting or non-teleporting probability, that needs to be fixed.)

\subsection{Using this to find sets of low-conductance}

This provides a way to ``explore'' a large graph.
Two things to note about this.
\begin{itemize}
\item
This way is very different that DFS or BFS, which are not so good if the diameter is very small, as it is in many real-world graphs.
\item
The sets of nodes that are found will be different than with a geodesic metric.
In particular, diffusions might get stuck in good conductance sets.
\end{itemize}

Following up on that second point, ACL showed how to use approximate PPR to find sets of low conductance, if they start from a random vector in a set.
Importantly, since they do it with the Approximate PPR vector, the number of nodes that is touched is proportional to the size of the output set.
So, in particular, if there is a set of small conductance, then the algorithm will run very quickly.

Here is the basic idea.
ACL first did it, and a simpler approach/analysis can be found in AC07.
\begin{itemize}
\item
Given $\pi_v$, construct $q_v(u) = \frac{ \pi_v(u) }{ d(u) }$.
\item
Number the vertices, WLOG, such that $q_v(1) \ge q_v(2) \ge \cdots q_v(n)$, and let $S_k = \{ 1,\ldots,k \}$.
\end{itemize}
Then, it can be shown that starting from a random vertex in the set of low conductance, then one of the sets has low conductance.

Computationally, these and related diffusion-based methods use only local information in the graph.
If one can look at the entire graph, then one can get information about global eigenvectors, which approximate sparse cuts.
But, if the graphs are huge, then one might be interested in the following question: 
\begin{itemize}
\item
Given a graph $G=(V,E)$ and a node $v \in V$, find a set of nodes $V$ such that the Cheeger ratio $h_S$ is small.
\end{itemize}
Recall that $h_S = \frac{|E(S,\bar{S})|}{\min \{ \mbox{Vol}(S),\mbox{Vol}(\bar{S}) \} }$.
Before, we were interested in good global clusters, and thus we defined $h_G = \min_{S \subset V} h_S$ and tried to optimize it, e.g., by computing global eigenvectors; but here we are not interested in global eigenvectors and global clusters.

Note that it is not immediately obvious that diffusions and PageRank are related to this notion of local clustering, but we will see that they are.
In particular, we are interested in the following.
\begin{itemize}
\item
Given a seed node $s$, we want to find a small set of nodes that is near $s$ and that is well connected internally but that is less well connected with the rest of the graph.
\end{itemize}
The intuition is that if a diffusion is started in such a cluster, then it is unlikely to leave the cluster since the cluster is relatively poorly-connected with the rest of the graph.
This is often true, but one must be careful, since a diffusion that starts, e.g., at the boundary of a set $S$ can leave $S$ at the first step.
So, the precise statement is going to be rather complicated.

Here is a lemma for a basic diffusion process of a graph $G$.

\begin{lemma}
Let $C \subset V$ be a set with Cheeger ratio $h_C$, and let $t_0$ be a time parameter.
Define $C_{t_0}$ to be a subset of $C$ such that for all $v \in C_{t_0}$ and all $t \le t_0$, we have that $\vec{1}_v^T W^t \vec{1}_{\bar{C}}^T \le t_0 h_C$.
Then, the volume of $C_{t_0}$ is such that 
\[
\mbox{Vol}\left( C_{t_0} \right) \ge \frac{1}{2} \mbox{Vol}\left( C \right)  .
\]
\end{lemma}

This is a complicated lemma, but it implies that for vertex set $C \subset V$ with small Cheeger ratio, that there are many nodes $v \in C$ such that the probability of a random walk starting at $v$ leaves $C$ is~low.

There is a similar result for the PageRank diffusion process.
The statement is as follows.

\begin{lemma}
Let $C \subset V$ have Cheeger ratio $h_C$, and let $\alpha \in (0,1]$.
Then, there is a subset $C_{\alpha} \subseteq C$ with volume $\mbox{Vol}\left( C_{\alpha} \right) \ge \frac{1}{2} \mbox{Vol}\left( C \right)$ such that for all $v \in C_{\alpha}$, the PageRank vector $\pi_{\alpha}(1_V)$ satisfies
\[
\pi_{\alpha}\left(1_V\right) \vec{1}_C^T \ge 1- \frac{h_C}{\alpha}  .
\]
\end{lemma}

This is also a complicated lemma, but it implies that for any vertex set $C \subseteq V$ with small $h_C$, that there are many nodes $v \in C$ for which PPR, using $v$ as a start node, is small on nodes outside of~$C$.

These and related results show that there is a relationship or correlation between the Cheeger ratio and diffusion processes, even for ``short'' random walks do not reach the asymptotic limit.
In particular, for a set $C \subset V$ with small $h_C$, it is relatively hard for a diffusion process started within  $C$ to leave $C$.

We can use these ideas to get \emph{local spectral clustering algorithms}.
To do this, we need a method for creating vertex cuts from a PR or random walk vector.
Here is the way.
\begin{itemize}
\item
Say that $\pi_{\alpha}(s)$ is a PPR vector.
\item
Then, create a set $C_S$ by doing a sweep cut over $\pi_{\alpha}(s)$.
\end{itemize}
That is, do the usual sweep cut, except on the vector returned by the algorithm, rather than the leading nontrivial eigenvector of the Laplacian.

Here is one such lemma that can be proved about such a local spectral algorithm.

\begin{lemma}
If $\pi_{\alpha}(s)$ is a PPR vector with $\|s\|_1 \le 1$ and there exists $S \subseteq V$ and a constant $\delta$ such that $\pi_{\alpha}(s) \vec{1}_S - \frac{\mbox{Vol}(S)}{\mbox{Vol}(G)} > \delta$, then 
\[
h_{C_s} < \left( \frac{12\alpha\log\left( 4\sqrt{\mbox{Vol}(S)}/\delta \right)}{\delta} \right)^{1/2} .
\]
\end{lemma}

This too is a complicated lemma, but the point is that if there exists a set $S$ where the PPR is much larger than the stationary distribution $\frac{d}{\mbox{Vol}(G)}$, then a sweep cut over $\pi_{\alpha}(S)$ produces a set $C_S$ with low Cheeger ratio: $O\left( \sqrt{ \alpha \log \left( \mbox{Vol}(S) \right) } \right)$.

This result is for PPR; that there is a similar result for \emph{approximate} PPR, which has a more complicated statement still.

Two things to note.
\begin{itemize}
\item
While the statements of these theoretical results is quite complicated, these methods do very well in practice in many applications.
(We won't discuss this much.)
\item
These algorithms were originally used as primitives for the Laplacian solvers, a topic to which we will return in a few weeks.
\end{itemize}

While one can prove results about the output of these algorithms, one might also be interested in what these algorithms optimize.
Clearly, they \emph{approximately} optimize something related to a local version of the global spectral partitioning objective; so the question is what do they optimize \emph{exactly}.
This is the topic we will turn to next---it will turn out that there are interesting connections here to with implicit regularization ideas.

%% file: lect20.tex
\section{%
(04/07/2015): 
Local Spectral Methods (3 of 4):
An optimization perspective on local spectral methods}

Reading for today.
\begin{compactitem}
\item
``A Local Spectral Method for Graphs: with Applications to Improving Graph Partitions and Exploring Data Graphs Locally,'' in JMLR, by Mahoney, Orecchia, and Vishnoi
\end{compactitem}

Last time, we considered local spectral methods that involve short random walks started at a small set of localized seed nodes.
Several things are worth noting about this.
\begin{itemize}
\item
The basic idea is that these random walks tend to get trapped in good conductance clusters, if there is a good conductance cluster around the seed node.
A similar statement holds for approximate localized random walks, e.g., the ACL push procedure---meaning, in particular, that one can implement them ``quickly,'' e.g., with the push algorithm, without even touching all of the nodes in $G$.
\item
The exact statement of the theorems that can be proven about how these procedures can be used to find good locally-biased clusters is quite technically complicated---since, e.g., one could step outside of the initial set of nodes if one starts near the boundary---certainly the statement is much more complicated than that for the vanilla global spectral method.
\item
The global spectral method is on the one hand a fairly straightforward algorithm (compute an eigenvector or some other related vector and then perform a sweep cut with it) and on the other hand a fairly straightforward objective (optimize the Rayleigh quotient variance subject to a few reasonable constraints).
\item
Global spectral methods often do very well in practice.
\item
Local spectral methods often do very well in practice.
\item
A natural question is: what objective do local spectral methods optimize---exactly, not approximately? Or, relatedly, can one construct an objective that is quickly-solvable and that also comes with similar locally-biased Cheeger-like guarantees? 
\end{itemize}
To this end, today we will present a local spectral ansatz that will have several appealing properties:
\begin{itemize}
\item
It can be computed fairly quickly, as a PPR.
\item
It comes with locally-biased Cheeger-like guarantees.
\item
It has the same form as several of the semi-supervised objectives we discussed.
\item
Its solution touches all the nodes of the input graph (and thus it is not as quick to compute as the push procedure which does not).
\item
The strongly local spectral methods that don't touch all of the nodes of the input graph are basically $\ell_1$-regularized variants of it.
\end{itemize}

\subsection{A locally-biased spectral ansatz}

Here is what we would like to do.
\begin{itemize}
\item
We would like to introduce an ansatz for an objective function for locally-biased spectral graph partitioning.
This objective function should be a locally-biased version of the usual global spectral partitioning objective; its optimum should be relatively-quickly computable; it should be useful to highlight locally-interesting properties in large data graphs; and it should have some connection to the local spectral algorithms that we have been discussing.
\end{itemize}

In addition to having an optimization formulation of locally-biased spectral partitioning methods, there are at least two reasons one would be interested in such an objective.
\begin{itemize}
\item
A small sparse cut might be poorly correlated with the second (or even all) global eigenvectors of $L$, and so it might be invisible to global spectral methods.
\item
We might have exogenous information about a specific region of a large graph in which we are most interested, and so we might want a method that finds clusters near that region, e.g., to do exploratory data analysis.
\end{itemize}

Here is the approach we will take.
\begin{itemize}
\item
We will start with the usual global spectral partitioning objective function and add to it a certain locality constraint.
\item
This program will be a non-convex problem (as is the global spectral partitioning problem), but its solution will be computable as a linear equation that is a generalization of the PR spectral ranking method.
\item
In addition, we will show that it can be used to find locally biased partitions near an input seed node, it has connections with the ACL push-based local spectral method, etc.
\end{itemize}

Let's set notation.
The Laplacian is $L=D-A$; and the normalized Laplacian is $\mathcal{L}=D^{-1/2}LD^{-1/2}$.
The degree-weighted inner product is given by $x^TDy = \sum_{i=1}^{n} x_iy_id_i$.
In this case, the weighted complete graph is given by 
\[
A_{K_n} = \frac{1}{\mbox{Vol}(G)}D11^TD ,
\]
in which case $D_{K_n}=D_G$ and thus
\[
L_{K_n}=D_{K_n} - A_{K_n} = D_G - \frac{1}{\mbox{Vol}(G)}D11^TD   .
\]

Given this notation, see the left panel of Figure~\ref{fig:spectral} for the usual spectral program $\mathsf{Spectral}(G)$, and see the right panel of Figure~\ref{fig:spectral} for \textsf{LocalSpectral}$(G,s,\kappa)$, a locally-biased spectral program.
(Later, we'll call the \textsf{LocalSpectral} objective ``MOV,'' to compare and contrast it with the ``ACL'' push~procedure.)

\begin{center}
\begin{figure}
\begin{minipage}{0.5\textwidth}
\begin{alignat*}{4}
  &\text{min} & x^T  L_{G} x \\
                     &\text{s.t.} & x^T D_{G} x = 1  \\
                     &            & (x^T D_{G} 1)^2 = 0  \\
                  &            & x \in \mathbb{R}^V
\end{alignat*}
\end{minipage}
\begin{minipage}{0.5\textwidth}
\begin{alignat*}{4}
 &\text{min} & x^T  L_{G} x               \\
                     &\text{s.t.} & x^T  D_{G} x = 1           \\
                   &            & (x^T D_{G} 1)^2 = 0  \\
                     &            &(x^T D_{G} s) ^2 \geq \kappa   \\
                  &            & x \in \mathbb{R}^V
\end{alignat*}
\end{minipage}
\caption{Global and local spectral optimization programs.
Left: The usual spectral program $\mathsf{Spectral}(G)$.
Right: A locally-biased spectral program
\textsf{LocalSpectral}$(G,s,\kappa)$.
In both cases, the optimization variable is the vector $x \in \mathbb{R}^{n}$.
}
\label{fig:spectral}
\end{figure}
\end{center}

In the above, we assume WLOG that 
\[
s \text{ is such that } 
   \left\{ 
      \begin{array}{l l}
                    s^TDs=1 &  \\
                    s^TD1=0 & 
      \end{array}                
   \right.     .
\]
This ``WLOG'' just says that one can subtract off the part of $s$ along the all-ones vector; we could have parametrized the problem to include this component and gotten similar results to what we will present below, had we not done this.
Note that $s$ can actually be any vector (that isn't in the span of the all-ones vector); but it is convenient to think of it as an indicator vector of a small ``seed set'' of nodes $S \subset V$.  

The constraint $\left(x^TDs\right)^2\ge\kappa$ says that the projection of the solution $x$ is at least $\sqrt{\kappa}$ in absolute value, where $\kappa\in(0,1)$.
Here is the interpretation of this constraint.
\begin{itemize}
\item
The vector $x$ must be in a spherical cap centered at $s$ with angle at most $\mbox{arccos}\left(\sqrt{\kappa}\right)$ from $s$.
\item
Higher values of $\kappa$ correspond to finding a vector that is more well-correlated with the seed vector.
While the technical details are very different than with strongly local spectral methods such as ACL, informally one should think of this as corresponding to shorter random walks or, relatedly, higher values of the teleportation parameter that teleports the walk back to the original seed set of nodes.
\item
If $\kappa=0$, then there is no correlation constraint, in which case we recover \textsf{Spectral}($G$).
\end{itemize}

\subsection{A geometric notion of correlation}

Although \textsf{LocalSpectral} is just an objective function and no geometry is explicitly imposed, there is a geometric interpretation of this in terms of a geometric notion of correlation between cuts in $G$.
Let's make explicit the geometric notion of correlation between cuts (or, equivalently, between partitions, or sets of nodes) that is used by \textsf{LocalSpectral}.

Given a cut $(T,\bar{T})$ in a graph $G=(V,E)$, a natural vector in $\mathbb{R}^{n}$ to associate with it is its indicator/characteristic vector, in which case the correlation between a cut $(T,\bar{T})$ and another cut $(U,\bar{U})$ can be captured by the inner product of the characteristic vectors of the two cuts.
Since we are working on the space orthogonal to the degree-weighted all-ones vector, we'll do this after we remove from the characteristic vector its projection along the all-ones vector.
In that case, again, a notion of correlation is related to the inner product of two such vectors for two cuts.
More precisely, given a set of nodes $T \subseteq V$, or equivalently a cut
$(T,\bar{T})$, one can define the unit vector $s_{T}$~as
\[
s_T \defeq  \sqrt{\frac{\vol(T)\vol(\bar{T})}{2m}} \; \left(\frac{1_T}{\vol(T)} - \frac{1_{\bar{T}}}{\vol(\bar{T})}\right)   ,
\]
in which case
\[
s_{T}(i) = \left\{ \begin{array}{ll}
                      \sqrt{\frac{\vol(T)\vol(\bar{T})}{2m}} \cdot \frac{1}{\vol(T)}  & \mbox{if $i \in T $} \\
                      - \sqrt{\frac{\vol(T)\vol(\bar{T})}{2m}} \cdot \frac{1}{\vol(\bar{
T})}    & \mbox{if $i \in \bar{T}$}   
                   \end{array}
           \right.    .
\]   

Several observations are immediate from this definition.
\begin{itemize}
\item
One can replace $s_{T}$ by $s_{\bar{T}}$ and the correlation remains the same with any other set, and so this is well-defined.
Also, $s_T = - s_{\bar{T}}$; but since here we only consider quadratic functions of $s_T,$ we can consider both $s_T$ and $s_{\bar{T}}$ to be representative vectors for the cut $(T, \bar{T}).$ 
\item
Defined this way, it immediately follows that $ s_T^T D_G 1 = 0$ and that $ s_T^T D_G s_T = 1$.
Thus, $s_T \in \mathcal{S}_{D}$ for $T \subseteq V$, where we denote by $\mathcal{S}_{D}$ the set of vectors $\{x \in \mathbb{R}^V: x^T D_G 1 = 0\}$; and $s_T$ can be seen as an appropriately normalized version of the vector consisting of the uniform distribution over $T$ minus the uniform
distribution over $\bar{T}$.
\item
One can introduce the following measure of correlation between two sets of nodes, or equivalently between two cuts, say a cut $(T, \bar{T})$ and a cut $(U, \bar{U})$:
\[
K(T,U) \defeq ( s_T D_G s_U )^2    .
\]
Then it is easy to show that:
$K(T,U) \in [0,1]$;
$K(T,U) = 1$ if and only if $T=U$ or $\bar{T}=U$;
$K(T,U) = K(\bar{T}, U)$; and
$K(T,U) = K(T, \bar{U})$.
\item
Although we have described this notion of geometric correlation in terms of vectors of the form $s_T \in \mathcal{S}_{D}$ that represent partitions $(T,\bar{T})$, this correlation is clearly well-defined for other vectors $s \in \mathcal{S}_{D}$ for which there is not such a simple interpretation in terms of cuts.
\end{itemize}

Below we will show that the solution to \textsf{LocalSpectral} can be characterized in terms of a PPR vector.
If we were interested in objectives that had solutions of different forms, e.g., the form of a heat kernel, then this would correspond to an objective function with a different constraint, and this would then imply a different form of correlation.

\subsection{Solution of \textsf{LocalSpectral}}

Here is the basic theorem characterizing the form of the solution of \textsf{LocalSpectral}.

\begin{theorem}[Solution Characterization]
\label{thm:pagerank}
Let $s \in \mathbb{R}^{n}$ be a seed vector such that $s^T D_G 1 =0$, $s^T D_G s = 1$, and $s^T D_G v_2 \neq 0$, where $v_{2}$ is the second generalized eigenvector of $L_G$ with respect to $D_G$.
In addition, let $1> \kappa \geq 0$ be a correlation parameter, and let $x^{\star}$ be an optimal solution to \textsf{LocalSpectral}$(G,s,\kappa)$.
Then, there exists some $\gamma \in (-\infty, \lambda_{2}(G))$ and a $c \in [0, \infty]$ such that
\begin{equation}
\label{eqn:xstar}
 x^{\star} = c(L_{G}-\gamma D_G)^{+} D_G s.
\end{equation}
\end{theorem}

Before presenting the proof of this theorem, here are several things to note.
\begin{itemize}
\item
$s$ and $\kappa$ are the parameters of the program; $c$ is a normalization factor that rescales the norm of the solution vector to be $1$ (and that can be computed in linear time, given the solution vector); and $\gamma$ is implicitly defined by $\kappa$, $G$, and~$s$.
\item
The correct setting of $\gamma$ ensures that $(s^T D_{G} x^\star)^2 = \kappa,$ i.e., that $x^\star$ is found exactly on the boundary of the feasible region.
\item
$x^\star$ and $\gamma$ change as $\kappa$ changes.
In particular, as $\kappa$ goes to $1$, $\gamma$ tends to $-\infty$ and $x^\star$ approaches $s$; conversely, as $\kappa$ goes to $0$,  $\gamma$ goes to $\lambda_2(G)$ and $x^\star$ tends towards $v_2$, the global eigenvector.
\item
For a fixed choice of $G$, $s$, and $\kappa$, an $\epsilon$-approximate solution to \textsf{LocalSpectral} can be computed in time
$\tilde{O}\left(\frac{m}{\sqrt{\lambda_{2}(G)}} \cdot \log(\frac{1}{\epsilon })\right)$
using the Conjugate Gradient Method; or in time
$\tilde{O}\left(m \log(\frac{1}{\epsilon })\right)$ using the Spielman-Teng
linear-equation solver (that we will discuss in a few weeks), where the $\tilde{O}$ notation hides $\log\log(n)$ factors.
This is true for a fixed value of $\gamma$, and the correct setting of $\gamma$ can be found by binary search.

While that is theoretically true, and while there is a lot of work recently on developing practically-fast nearly-linear-time Laplacian-based solvers, this approach might not be appropriate in certain applications.
For example, in many applications, one has precomputed an eigenvector decomposition of $L_G$, and then one can use those vectors and obtain an approximate solution with a small number of inner products.
This can often be much faster in~practice.
\end{itemize}

In particular, solving \textsf{LocalSpectral}  is \emph{not} ``fast'' in the sense of the original local spectral methods, i.e., in that the running time of those methods depends on the size of the output and doesn't depend on the size of the graph.
But the running time to solve \textsf{LocalSpectral} \emph{is} fast, in that its solution depends essentially on computing a leading eigenvector of a Laplacian $L$ and/or can be solved with ``nearly linear time'' solvers that we will discuss in a few weeks.

While Eqn.~(\ref{eqn:xstar}) is written in the form of a linear equation, there is a close connection between the solution vector $x^\star$ and the Personalized PageRank (PPR) spectral ranking procedure.
\begin{itemize}
\item
Given a vector $s \in \mathbb{R}^{n}$ and a \emph{teleportation} constant $\alpha> 0$, the PPR vector can be written as 
\[
\mbox{pr}_{\alpha,s}=\left(L_{G}+\frac{1-\alpha}{\alpha}D_{G}\right)^{-1}D_{G}s  .
\]
By setting $\gamma = -\frac{1-\alpha}{\alpha}$, one can see that the optimal solution to \textsf{LocalSpectral} is proved to be a generalization PPR.
\item
In particular, this means that for high values of the correlation parameter $\kappa$ for which the corresponding $\gamma$ satisfies $\gamma < 0$, the optimal solution to \textsf{LocalSpectral} takes the form of a PPR vector.
On the other hand, when $\gamma \geq 0,$ the optimal solution to \textsf{LocalSpectral} provides a smooth way of transitioning from the PPR vector to the global second eigenvector~$v_2$.
\item
Another way to interpret this is to say that for values of $\kappa$ such that $\gamma <0$, then one could compute the solution to \textsf{LocalSpectral} with a random walk or by solving a linear equation, while for values of $\kappa$ for which $\gamma>0$, one can only compute the solution by solving a linear equation and not by performing a random walk.
\end{itemize}
About the last point, we have talked about how random walks compute regularized or robust versions of the leading nontrivial eigenvector of $L$---it would be interesting to characterize an algorithmic/statistical tradeoff here, e.g., if/how in this context certain classes of random walk based algorithms are less powerful algorithmically than related classes of linear equation based algorithms but that they implicitly compute regularized solutions more quickly for the parameter values for which they are able to compute solutions.

\subsection{Proof of Theorem~\ref{thm:pagerank}}

Here is an outline of the proof, which essentially involves ``lifting'' a rank-one constraint to obtain an SDP in order to get strong duality to apply.
\begin{itemize}
\item
Although \textsf{LocalSpectral} is not a convex optimization problem, it can be relaxed to an SDP that is convex.
\item
From strong duality and complementary slackness, the solution to the SDP is rank one.
\item
Thus, the vector making up the rank-one component of this rank-one solution is the solution to \textsf{LocalSpectral}.
\item
The form of this vector is of the form of a PPR.
\end{itemize}

Here are some more details.
Consider the primal $\textsf{SDP}_p$ and dual $\textsf{SDP}_d$ SDPs, given in the left panel and right panel, respectively, of Figure~\ref{fig:sdp}.

\begin{center}
\begin{figure}
\begin{minipage}{0.5\textwidth}
\begin{alignat*}{4}
\quad&  &\text{minimize} \quad &&  L_{G} \circ  X\\
  &  &\text{s.t.} \quad &&  L_{K_{n}} \circ  X = 1\ \\
  &  & &&  ( D_{G} s)( D_{G} s)^T \circ  X \geq \kappa \\
  & & & & X \succeq 0
\end{alignat*}
\end{minipage}
\begin{minipage}{.25\textwidth}
\begin{alignat*}{4}
\quad&  &\text{maximize} \quad && \alpha + \kappa \beta\\
  &  &\text{s.t.} \quad &&  L_{G}  \succeq  \alpha  L_{K_{n}} + \beta ( D_{G} {s})( D_{G} {s})^T \ \\
  &  & &&  \beta \geq 0 \\
  & & & & \alpha \in \mathbb{R}
\end{alignat*}
\end{minipage}
\caption{Left:  Primal SDP relaxation of \textsf{LocalSpectral}$(G,s, \kappa)$: $\textsf{SDP}_{p}(G,s,\kappa)$.  
For this primal, the optimization variable is $X \in \mathbb{R}^{n \times n}$ such that $X$ is SPSD.
Right: Dual SDP relaxation of \textsf{LocalSpectral}$(G,s, \kappa)$: $\textsf{SDP}_{d}(G,s,\kappa)$.
For this dual, the optimization variables are $\alpha,\beta\in\mathbb{R}$.}
\label{fig:sdp}
\end{figure}
\end{center}

Here are a sequence of claims.

\begin{claim}
The primal SDP, $\textsf{SDP}_p$ is a relaxation of \textsf{LocalSpectral}.
\end{claim}
\begin{proof}
Consider $x\in\mathbb{R}^{n}$, a feasible vector for \textsf{LocalSpectral}.
Then, the SPSD matrix $X=xx^T$ is feasible for \textsf{SDP}$_p$.
\end{proof}

\begin{claim}
Strong duality holds between $\textsf{SDP}_p$ and $\textsf{SDP}_d$.
\end{claim}
\begin{proof}
The program $\textsf{SDP}_p$ is convex, and so it suffices to check that Slater's constraint qualification conditions hold for $\textsf{SDP}_p$.
To do so, consider $X=ss^T$.
Then, 
\[
\left(D_Gs\right)\left(D_Gs\right)^T \circ ss^T = \left(s^TD_Gs\right)^2 = 1 > \kappa  .
\]
\end{proof}

\begin{claim}
The following feasibility and complementary slackness conditions are sufficient for a primal-dual pair $X^{*}$, $\alpha^{*}$, $\beta^{*}$ to be an optimal solution.
The feasibility conditions are:
\begin{eqnarray}
\nonumber
  L_{K_{n}} \circ  X^\star &=& 1 , \label{F1} \\
\nonumber
 ( D_{G} {s})( D_{G} {s})^T \circ  X^\star &\geq& \kappa ,  \label{F2}  \\
  L_{G}- \alpha^\star  L_{K_{n}}  - \beta^\star ( D_{G} {s})( D_{G} {s})^T &\succeq& 0 , \mbox{ and} \label{F3} \\
\nonumber
 \beta^\star &\geq& 0  \label{F4}  ,
\end{eqnarray}
and the complementary slackness conditions are:
\begin{eqnarray}
\nonumber
 \alpha^\star(  L_{K_{n}}  \circ  X^\star - 1) &=& 0 , \label{C1} \\
 \beta^\star ( ( D_{G} {s})( D_{G} {s})^T \circ  X^\star - \kappa) &=& 0 \label{C2} , \mbox{ and} \\
 X^\star \circ  ( L_{G}- \alpha^\star  L_{K_{n}}  - \beta^\star ( D_{G} {s})( D_{G} {s})^T ) &=& 0 \label{C3}  .
\end{eqnarray}
\end{claim}
\begin{proof}
This follows from the convexity of $\textsf{SDP}_p$ and Slater's condition.
\end{proof}

\begin{claim}
The feasibility and complementary slackness conditions, coupled with the assumptions of the theorem, imply that $X^{*}$ is rank one and that $\beta^{*} \ge 0$.
\end{claim}
\begin{proof}
If we plug $v_{2}$ in Eqn.~\eqref{F3}, then we obtain that
$ v_{2}^{T}L_{G}v_{2} - \alpha^{\star} -\beta^{\star}  (v_{2}^T D_{G} s)^{2} \geq 0.$

But $ v_{2}^{T}L_{G}v_{2}=\lambda_{2}(G)$ and $\beta^{\star} \geq 0.$ Hence, $\lambda_{2}(G) \geq \alpha^{\star}.$
Suppose $\alpha^\star = \lambda_2(G).$ As $s^T D_{G} v_2 \neq 0,$ it must be the case that $\beta^\star = 0.$ Hence, by Equation~\eqref{C3}, we must have $X^\star \circ L(G) = \lambda_2(G),$ which implies that $X^\star = v_2v_2^T,$ {\em i.e.},  the optimum for \textsf{LocalSpectral} is the global eigenvector $v_2$. This corresponds to a choice of $\gamma = \lambda_2(G)$ and $c$ tending to infinity.

Otherwise, we may assume that $\alpha^\star< \lambda_2(G).$ Hence, since $G$ is connected and $\alpha^{\star} <\lambda_{2}(G),$ $L_{G}-\alpha^{\star}L_{K_{n}}$ has rank exactly $n-1$ and kernel parallel to the vector $1.$

From the complementary slackness condition \eqref{C3} we can deduce that the image of $X^{\star}$ is in the kernel of $ L_{G}- \alpha^\star  L_{K_{n}}  - \beta^\star ( D_{G} {s})( D_{G} {s})^T.$

If $\beta^\star > 0,$ we have that $ \beta^\star ( D_{G} {s})( D_{G} {s})^T$ is a rank one matrix and, since $s^T D_{G} 1 = 0,$ it reduces the rank of $L_{G}-\alpha^{\star}L_{K_{n}}$ by one precisely. If $\beta^{\star}=0$ then $X^{\star}$ must be $0$ which is not possible  if $\textsf{SDP}_{p}(G,s,\kappa)$ is feasible.

Hence, the rank of $ L_{G}- \alpha^\star  L_{K_{n}}  - \beta^\star ( D_{G} {s})( D_{G} {s})^T$ must be exactly $n-2.$ As we may assume that $1$ is in the kernel of $X^\star$, $X^{\star}$ must be of rank one.
This proves the claim.
\end{proof}

\textbf{Remark.}
It would be nice to have a cleaner proof of this that is more intuitive and that doesn't rely on ``boundary condition'' arguments as much.

Now we complete the proof of the theorem.
From the claim it follows that, $X^{\star}=x^{\star}x^{\star T}$ where $x^{\star}$ satisfies the equation
$$  
(L_{G}- \alpha^\star  L_{K_{n}}  - \beta^\star ( D_{G} {s})( D_{G} {s})^T)x^{\star}=0.
$$
From the second complementary slackness condition,
Equation~\eqref{C2},
and the fact that $\beta^{\star}>0,$ we obtain that
$ (x^{\star})^T D_{G} s = \pm \sqrt{\kappa}.$
Thus,
$x^{\star} =\pm \beta^{\star} \sqrt{\kappa} (L_{G}-\alpha^{\star}L_{K_{n}})^{+}D_{G}s,$ as required.

\subsection{Additional comments on the \textsf{LocalSpectral} optimization program}

Here, we provide some additional discussion for this locally-biased spectral partitioning objective.
Recall that the proof we provided for Cheeger's Inequality showed that in some sense the usual global spectral methods ``embed'' the input graph $G$ into a complete graph; we would like to say something similar here.

To do so, observe that the dual of \textsf{LocalSpectral} is given by the following.
\begin{eqnarray*}
\label{prog:spectral-local-d1}
                     &\text{maximize} & \alpha + \beta \kappa                    \\
                     &\text{s.t.} & L_{G} \succeq \alpha L_{K_n} + \beta \Omega_T \\
                     &            & \beta \ge 0    ,
\end{eqnarray*}
where $\Omega_T=D_Gs_Ts_T^TD_G$.
Alternatively, by subtracting the second constraint of \textsf{LocalSpectral} from the first constraint, it follows that
$$
x^T\left(L_{K_n}-L_{K_n}s_Ts_T^TL_{K_n}\right)x \le 1-\kappa  .
$$
Then it can be shown that
$$
L_{K_n}-L_{K_n}s_Ts_T^TL_{K_n}
   = \frac{L_{K_{T}}}{\vol(\bar{T})} + \frac{L_{K_{\bar{T}}}}{\vol(T)}  ,
$$
where $L_{K_{T}}$ is the $D_G$-weighted complete graph on the vertex set $T$.
Thus, \textsf{LocalSpectral} is equivalent~to
\begin{eqnarray*}
\label{prog:spectral-local-p2}
                              &\text{minimize} & x^T  L_{G} x                \\
                              &\text{s.t.} & x^T  L_{K_n} x = 1      \\
                              &            & x^T\left( \frac{L_{K_{T}}}{\vol(\bar{T})} + \frac{L_{K_{\bar{T}}}}{\vol(T)} \right)x \le 1-\kappa   .
\end{eqnarray*}
The dual of this program is given by the following.
\begin{eqnarray*}
\label{prog:spectral-local-d2A}
                             &\text{maximize} & \alpha - \beta(1-\kappa)   \\
\label{prog:spectral-local-d2B}
                             &\text{s.t.} & L_{G} \succeq \alpha L_{K_n} - \beta\left( \frac{L_{K_{T}}}{\vol(\bar{T})} + \frac{L_{K_{\bar{T}}}}{\vol(T)} \right)  \\
\label{prog:spectral-local-d2C}
                             &            & \beta \ge 0      .
\end{eqnarray*}
Thus, from the perspective of this dual, \textsf{LocalSpectral} can be viewed as ``embedding'' a combination of a complete graph $K_n$ and a weighted combination of complete graphs on the sets $T$ and $\bar{T}$, i.e., $K_T$ and $K_{\bar{T}}$.
Depending on the value of $\beta$, the latter terms clearly discourage cuts
that substantially cut into $T$ or $\bar{T}$, thus encouraging partitions
that are well-correlated with the input cut $(T,\bar{T})$.

If we can establish a precise connection between the optimization-based \textsf{LocalSpectral} procedure and operational diffusion-based procedures such as the ACL push procedure, then this would provide additional insight as to ``why'' the short local random walks get stuck in small seed sets of nodes.
This will be one of the topics for next time.

%% file: lect21.tex
\section{%
(04/09/2015): 
Local Spectral Methods (4 of 4):
Strongly and weakly locally-biased graph partitioning}

Reading for today.
\begin{compactitem}
\item
``Anti-differentiating Approximation Algorithms: A case study with Min-cuts, Spectral, and Flow,'' in ICML, by Gleich and Mahoney
\item
``Think Locally, Act Locally: The Detection of Small, Medium-Sized, and Large Communities in Large Networks,'' in PRE, by Jeub, Balachandran, Porter, Mucha, and Mahoney
\end{compactitem}

Last time we introduced an objective function (\textsf{LocalSpectral}) that looked like the usual global spectral partitioning problem, except that it had a locality constraint, and we showed that its solution is of the form of a PPR vector.
Today, we will do two things.
\begin{itemize}
\item
We will introduce a locally-biased graph partitioning problem, we show that the solution to \textsf{LocalSpectral} can be used to compute approximate solutions to that problem.
\item
We describe the relationship between this problem and what the strongly-local spectral methods, e.g., the ACL push method, compute.
\end{itemize}

\subsection{Locally-biased graph partitioning}

We start with a definition.
\begin{definition}[Locally-biased graph partitioning problem.]
\label{def:locally-biased-partitioning}
Given a graph $G=(V,E)$, an input node $u \in V$, a number $k \in \mathbb{Z}^{+}$, find a set of nodes $T \subset V$ s.t.
$$
\phi(u,k) = \min_{T \subset V : u \in T, \mbox{Vol}(T) \le k} \phi(T)   ,
$$
i.e., find the best conductance set of nodes of volume not greater than $k$ that contains the node $u$.
\end{definition}
That is, rather than look for the best conductance cluster in the entire graph (which we considered before), look instead for the best conductance cluster that contains a specified seed node and that is not too large.

Before proceeding, let's state a version of Cheeger's Inequality that applies not just to the leading nontrivial eigenvector of $L$ but instead to \emph{any} ``test vector.''
\begin{theorem}
\label{thm:cheeger2}
Let $x\in\mathbb{R}^{n}$ s.t. $x^TD \vec{1} = 0$.
Then there exists a $t \in [n]$ such that 
$S \equiv \mbox{SweepCut}_t(x) \equiv \{ i : x_i \ge t \}  $
satisfies
$ \frac{x^TLx}{x^TDx} \ge \frac{\phi(S)^2}{8} $.
\end{theorem}

\textbf{Remark.}
This form of Cheeger's Inequality provides additional flexibility in at least two ways.
First, if one has computed an approximate Fiedler vector, e.g., by running a random walk many steps but not quite to the asymptotic state, then one can appeal to this result to show that Cheeger-like guarantees hold for that vector, i.e., one can obtain a ``quadratically-good'' approximation to the global conductance objective function using that vector.
Alternatively, one can apply this to \emph{any} vector, e.g., a vector obtained by running a random walk just a few steps from a localized seed node.
This latter flexibility makes this form of Cheeger's Inequality very useful for establishing bounds with both strongly and weakly local spectral methods.

Let's also recall the objective with which we are working; we call it \textsf{LocalSpectral}$(G,s,\kappa)$ or \textsf{LocalSpectral}.
Here it is.
\begin{alignat*}{4}
  &\text{min} & x^T  L_{G} x \\
                     &\text{s.t.} & x^T D_{G} x = 1  \\
                     &            & (x^T D_{G} 1)^2 = 0  \\
                     & & (x^T D_{G} s )^{2} \geq \kappa  \\
                  &            & x \in \mathbb{R}^n
\end{alignat*}

Let's start with our first result, which says that \textsf{LocalSpectral} is a relaxation of the intractable combinatorial problem that is the locally-biased version of the global spectral partitioning problem (in a manner analogous to how the global spectral partitioning problem is a relaxation of the intractable problem of finding the best conductance partition in the entire graph).
More precisely, we can choose the seed set $s$ and correlation parameter $\kappa$ such that $\textsf{LocalSpectral}(G,s,\kappa)$ is a relaxation of the problem defined in Definition~\ref{def:locally-biased-partitioning}.

\begin{theorem}
\label{thm:relaxation}
For $u \in V$, \textsf{LocalSpectral}$(G,v_{\{u\}},1/k)$ is a relaxation
of the problem of finding a minimum conductance cut $T$ in $G$ which
contains the vertex $u$ and is of volume at most~$k$.
In particular, $\lambda(G,v_{\{u\}},1/k) \leq \phi(u,k)$.
\end{theorem}
\begin{proof}
If we let $x=v_{T}$ in
\textsf{LocalSpectral}$(G,v_{\{u\}},1/k)$,
then $v_{T}^{T}L_{G}v_{T}=  \phi(T)$, $v_{T}^{T}D_{G}1=0$, and
$v_{T}^{T}D_{G}v_{T}=1$.
Moreover, we have that
$$ 
(v_{T}^{T}D_{G}v_{\{u\}})^{2} = \frac{d_{u}(2m-\vol(T))}{\vol(T) (2m-d_{u})} \geq 1/k  ,
$$
which establishes the lemma.
\end{proof}

Next, let's apply sweep cut rounding to get locally-biased cuts that are quadratically good, thus establishing a locally-biased analogue of the hard direction of Cheeger's Inequality for this problem.
In particular, we can apply Theorem~\ref{thm:cheeger2} to the optimal solution for $\textsf{LocalSpectral}(G,v_{\{u\}},1/k)$ and obtain a cut $T$ whose conductance is quadratically close to the optimal value $\lambda(G,v_{\{u\}},1/k)$.
By Theorem~\ref{thm:relaxation}, this implies that $\phi(T) \leq O(\sqrt{\phi(u,k)})$, which essentially establishes the following theorem.

\begin{theorem}[Finding a Cut]
\label{thm:cut}
Given an unweighted graph $G=(V,E)$, a vertex $u \in V$ and a positive
integer $k$, we can find a cut in $G$ of conductance at most
$O(\sqrt{ \phi(u,k)})$ by computing a sweep cut of the optimal vector for
$\textsf{LocalSpectral}(G, v_{\{u\}},1/k)$.
\end{theorem}

\textbf{Remark.}
What this theorem states is that we can perform a sweep cut over the vector that is the solution to $\textsf{LocalSpectral}(G,v_{\{u\}},1/k)$ in order to obtain a locally-biased partition; and that this partition comes with quality-of-approximation guarantees analogous to that provided for the global problem $\textsf{Spectral}(G)$ by Cheeger's inequality.

We can also use the optimal value of \textsf{LocalSpectral} to provide lower bounds on the conductance value of other cuts, as a function of how well-correlated they are with the input seed vector $s$.
In particular, if the seed vector corresponds to a cut $U$, then we get lower bounds on the conductance of other cuts $T$ in terms of the correlation between $U$ and $T$.

\begin{theorem}[Cut Improvement]
\label{thm:improve}
Let $G$ be a  graph and $s \in \mathbb{R}^{n}$ be such that $s^{T}D_{G} 1=0,$ where $D_{G}$ is the degree matrix of $G$.
In addition, let $\kappa \geq 0$ be a correlation parameter.
Then, for all sets $T \subseteq V$ such that $\kappa' \defeq  (s^{T}D_{G}v_{T})^{2}$, we have that
\[
\phi(T) \geq \left\{ \begin{array}{ll}
                       \lambda(G,s,\kappa)
                       & \mbox{if $\kappa \leq \kappa'$} \\
                       \frac{\kappa'}{\kappa} \cdot \lambda(G,s,\kappa)
                       & \mbox{if $\kappa' \leq \kappa$.}
                    \end{array}
            \right.
\]
In particular, if $s=s_{U}$ for some $U \subseteq V,$ then note that $\kappa'= K(U,T).$
\end{theorem}
\begin{proof}
It follows from 
the results that we established in the last class 
that $\lambda(G,s,\kappa)$ is the same as the optimal value of $\textsf{SDP}_{p}(G,s,\kappa)$ which, by strong duality, is the same as the optimal value of $\textsf{SDP}_{d}(G,s,\kappa)$.
Let $\alpha^{\star},\beta^{\star}$ be the optimal dual values to $\textsf{SDP}_{d}(G, s,\kappa).$
Then, from the dual feasibility constraint
$ L_{G}- \alpha^\star  L_{K_{n}}  - \beta^\star ( D_G {s})( D_G {s})^T \succeq 0 ,$
it follows that
$$ 
s_{T}^{T}L_{G}s_{T} - \alpha^{\star}s_{T}^{T}L_{K_{n}}s_{T}-\beta^{\star} (s^T D_G s_{T})^{2} \geq 0.
$$
Notice that since $ s_{T}^T D_G 1 =0$, it follows that $ s_{T}^{T}L_{K_{n}}s_{T}=s_{T}^{T}D_Gs_{T}=1$.
Further, since $s_{T}^{T}L_{G}s_{T}=\phi(T),$ we obtain, if $\kappa \leq \kappa',$ that
$$
\phi(T) \geq \alpha^{\star} + \beta^{\star} ( s^T D_G s_{T} )^{2} \geq  \alpha^{\star} + \beta^{\star}\kappa = \lambda (G,s,\kappa).
$$
If on the other hand, $\kappa' \leq \kappa,$ then
$$
\phi(T) \geq \alpha^{\star} + \beta^{\star} (s^T D_G s_{T}) ^{2} \geq \alpha^{\star} + \beta^{\star}\kappa \geq   \frac{ \kappa'}{\kappa} \cdot (\alpha^{\star} + \beta^{\star}\kappa) =\frac{ \kappa'}{\kappa} \cdot  \lambda (G,s,\kappa)  .
$$
Finally, observe that if $s=s_{U}$ for some $U \subseteq V,$ then $ (s_{U}^T D_G s_{T} )^{2} =K(U,T).$
Note that strong duality was used here.
\end{proof}

\textbf{Remark.}
We call this result a ``cut improvement'' result since it is the spectral analogue of the flow-based ``cut improvement'' algorithms we mentioned when doing flow-based graph partitioning.
\begin{itemize}
\item
These flow-based cut improvement algorithms were originally used as a post-processing algorithm to improve partitions found by other algorithms.
For example, GGT, LR (Lang-Rao), and AL (which we mentioned before).
\item
They provide guarantees of the form: for any cut $\left(C,\bar{C}\right)$ that is $\epsilon$-correlated with the input cut, the cut output by the cut improvement algorithm has conductance $\le$ some function of the conductance of $\left(C,\bar{C}\right)$ and $\epsilon$.
\item
Theorem~\ref{thm:improve} shows that, while the cut value  output by this spectral-based ``improvement'' algorithm might \emph{not} be improved, relative to the input, as they are often guaranteed to do with flow-based cut-improvement algorithms, they do not decrease in quality too much, and in addition one can make claims about the cut quality of ``nearby'' cuts.
\item
Although we don't have time to discuss it, these two operations can be viewed as building blocks or ``primitives'' that can be combined in various ways to develop algorithms for other problems, e.g., finding minimum conductance cuts.
\end{itemize}

\subsection{Relationship between strongly and weakly local spectral methods}

So far, we have described two different ways to think about local spectral algorithms.
\begin{itemize}
\item
\textbf{Operational.}
This approach provides an algorithm, and one can prove locally-biased Cheeger-like guarantees.
The exact statement of these results is quite complex, but the running time of these methods is extremely fast since they don't even need to touch all the nodes of a big~graph.
\item
\textbf{Optimization.}
This approach provides a well-defined optimization objective, and one can prove locally-biased Cheeger-like guarantees.
The exact statement of these results is much simpler, but the running time is only moderately fast, since it involves computing eigenvectors or linear equations on sparse graphs, and this involves at least touching all the nodes of a big~graph.
\end{itemize}

An obvious question here is the following.
\begin{itemize}
\item
Shat is the precise relationship between these two approaches?
\end{itemize}
We'll answer this question by considering the weakly-local \textsf{LocalSpectral} optimization problem (that we'll call MOV below) and the PPR-based local spectral algorithm due to ACL (that we'll call ACL below).
What we'll show is roughly the following.
\begin{itemize}
\item
We'll show roughly that if MOV optimizes an $\ell_2$ based penalty, then ACL optimizes an $\ell_1$-regularized version of that $\ell_2$ penalty.
\end{itemize}
That's interesting since $\ell_1$ regularization is often introduced to enforce or encourage sparsity.
Of course, there is no $\ell_1$ regularization in the statement of the strongly local spectral methods like ACL, but clearly they enforce some sort of sparsity, since they don't even touch most of the nodes of a large graph.
Thus, this result can be interpreted as providing an implicit regularization characterization of a fast approximation algorithm.

\subsection{Setup for implicit $\ell_1$ regularization in strongly local spectral methods}

Recall that $L=D-A=B^TCB$, where $B$ is the unweighted edge-incidence matrix.
Then 
\[
\|Bx\|_{C,1} = \sum_{(ij)\in E} C_{(ij)} |x_i-x_j| = \mbox{cut}(S)  ,
\]
where $S=\{i:x_i=1\}$.
In addition, we can obtain a spectral problem by changing $\|\cdot\|_1 \rightarrow \|\cdot\|_2$ to get 
\[
\|Bx\|_{C,2}^{2} = \sum_{(ij)\in E} C_{(ij)} \left(x_i-x_j\right)^2
\]

Let's consider a specific $(s,t)$-cut problem that is inspired by the AL FlowImprove procedure.
To do so, fix a set of vertices (like we did when we did the semi-supervised eigenvector construction), and define a \emph{new} graph that we will call the ``localized cut graph.''
Basically, this new graph will be the original graph augmented with two additional nodes, call them $s$ and $t$, that are connected by weights to the nodes of the original graph.
Here is the definition.

\begin{definition}[localized cut graph]
Let $G = (V,E)$ be a graph, let $S$ be a set of vertices, possibly empty, let $\bar{S}$ be the complement set, and let $\alpha$ be a non-negative constant.
Then the \emph{localized cut graph} is the weighted, undirected graph with adjacency matrix:
\[ 
A_S = \left[
            \begin{array}{ccc}
            0 & \alpha d_S^T & 0 \\
          \alpha d_S & A & \alpha d_{\bar{S}} \\
          0 & \alpha d_{\bar{S}}^T & 0 
             \end{array}
\right] 
\]
where $d_S = D e_S$ is a degree vector localized on the set $S$, $A$ is the adjacency matrix of the original graph $G$, and $\alpha \ge 0$ is a non-negative weight.
Note that the first vertex is $s$ and the last vertex is $t$.
\end{definition}

We'll use the $\alpha$ and $S$ parameter to denote the matrices for the localized cut graph.
For example, the \emph{incidence matrix} $B(S)$  of the localized cut graph, which depends on the set $S$, is given by the~following.
\[
B(S) = 
\left[ 
\begin{array}{ccc} 
e & -I_S & 0 \\ 
0 & B & 0 \\ 
0 & -I_{\bar{S}} & e 
\end{array} 
\right]  , 
\]
where, recall, the variable $I_S$ are the columns of the identity matrix corresponding to vertices in $S$.
The edge-weights of the localized cut graph are given by the diagonal matrix $C(\alpha)$, which depends on the value $\alpha$.

Given this, recall that the $1$-norm formulation of the LP for the min-$s,t$-cut problem, i.e., the minimum weighted $s,t$ cut in the flow graph, is given by the following.
\begin{alignat*}{4}
 &\text{min} & \|Bx\|_{C(\alpha),1}              \\
                      &\text{s.t.} &x_s=1 , x_t=0  , x \ge 0   .
\end{alignat*}

Here is a theorem that shows that PageRank implicitly solves a $2$-norm variation of the $1$-norm formulation of the $s,t$-cut problem. 

\begin{theorem}
\label{thm:antiderivative1}
Let $B(S)$ be the incidence matrix for the localized cut graph, and $C(\alpha)$ be the edge-weight matrix.
The PageRank vector $z$ that solves
\[ (\alpha D + L) z = \alpha v \]
with $v = d_{S}/\vol(S)$
is a renormalized solution of the 2-norm cut computation:
\begin{eqnarray}
\label{eq:pr-cut}
\text{min} & \|B(S)x\|_{C(\alpha),2}              \\
\nonumber
\text{s.t.} & x_s=1  , x_t=0   .
\end{eqnarray}
Specifically, if $x(\alpha,S)$ is the solution of Prob.~\eqref{eq:pr-cut}, then
\[ 
x(\alpha,S) =
\left[ 
\begin{array}{c} 
 1 \\ \vol(S) z \\ 0  
\end{array} 
\right]   . 
\]
\end{theorem}
\begin{proof}
The key idea is that the 2-norm problem corresponds with a quadratic objective, which PageRank solves.  
The quadratic objective for the 2-norm approximate cut is:
\begin{eqnarray*}  
 \| B(S) x \|_{C(\alpha),2}^2 
&=& x^T B(S)^T C(\alpha) B(S) x  \\ 
&=& x^T 
\left[ 
\begin{array}{ccc} 
\alpha \mbox{vol}(S) & -\alpha d_S^T & 0 \\ 
-\alpha d_S & L + \alpha D & -\alpha d_{\bar{S}} \\ 
0 & -\alpha d_{\bar{S}} & \alpha \mbox{vol}(\bar{S}) 
\end{array} 
\right] 
x. 
\end{eqnarray*} 
If we apply the constraints that $x_s = 1$ and $x_t = 0$ and let $x_G$ be the free set of variables, then we arrive at the unconstrained objective:
\begin{eqnarray*} 
& &
\hspace{-70mm}
\left[
\begin{array}{ccc}
1 & x_G^T & 0
\end{array} 
\right]
\left[
\begin{array}{ccc} 
\alpha \mbox{vol}(S) & -\alpha d_S^T & 0 \\ 
-\alpha d_S & L + \alpha D & -\alpha d_{\bar{S}} \\ 
0 & -\alpha d_{\bar{S}} & \alpha \mbox{vol}(\bar{S}) 
\end{array} 
\right]
\left[
\begin{array}{c}
1 \\ x_G \\ 0
\end{array}
\right]  \\
\hspace{+50mm}
&=& 
x_G^T (L + \alpha D) x^{}_G - 2 \alpha x_G^T d^{}_{S} + \alpha \mbox{vol}(S). 
 \end{eqnarray*} 

Here, the solution $x_G$ solves the linear system
\[ (\alpha D + L) x_G = \alpha d_{S}. \]
The vector $x_G = \vol(S) z$, where $z$ is the solution of the PageRank problem defined in the theorem, which concludes the proof.
\end{proof}

Theorem~\ref{thm:antiderivative1} essentially says that for each PR problem, there is a related cut/flow problem that ``gives rise'' to it.
One can also establish the reverse relationship that extracts a cut/flow problem from \emph{any} PageRank problem.

To show this, first note that the proof of Theorem~\ref{thm:antiderivative1} works since the edges we added had weights proportional to the degree of the node, and hence the increase to the degree of the nodes was proportional to their current degree.
This causes the diagonal of the Laplacian matrix of the localized cut graph to become $\alpha D + D$.
This idea forms the basis of our subsequent analysis.
For a general PageRank problem, however, we require a slightly more general definition of the localized cut graph, which we call a \emph{PageRank cut graph}.
Here is the definition.

\begin{definition}  
Let $G = (V,E)$ be a graph, and let $s \ge 0$ be a vector such that $d - s \ge 0$.
Let $s$ connect to each node in $G$ with weights given by the vector $\alpha s$, and let $t$ connect to each node in $G$ with weights given by $\alpha (d - s)$.
Then the \emph{PageRank cut graph} is the weighted, undirected graph with adjacency matrix:
\[ 
A(s) = 
\left[
\begin{array}{ccc}
 0 & \alpha s^T & 0 \\ 
 \alpha s & A & \alpha (d - s) \\ 
 0 & \alpha (d - s)^T & 0 
 \end{array}
\right] . 
\]
\end{definition}

We use $B(s)$ to refer to the incidence matrix of this PageRank cut graph.
Note that if $s=d_S$, then this is simply the original construction.

With this, we state the following theorem,
which is a sort
of converse to Theorem~\ref{thm:antiderivative1}.
The proof is similar to that of Theorem~\ref{thm:antiderivative1} and so it is omitted.

\begin{theorem}
\label{thm:antiderivative1converse}
Consider any PageRank problem that fits the framework of 
\[
(I - \beta P^T) x = (1-\beta) v  .
\]
The PageRank vector $z$ that solves
\[ ( \alpha D + L) z = \alpha v \]
is a renormalized solution of the 2-norm cut computation:
\begin{eqnarray} 
\label{eq:pr-unif}
\text{min} & \| B(s) x \|_{C(\alpha),2} \\
\nonumber
 & x_s = 1, x_t = 0
 \end{eqnarray}
with $s = v$.  
Specifically, if $x(\alpha,S)$ is the solution of the 2-norm cut, then
\[ 
x(\alpha,s) 
= 
\left[
\begin{array}{c} 1 \\  z \\ 0 \end{array}
\right] . 
\]
\end{theorem}

Two things are worth noting about this result.
\begin{itemize}
\item
A corollary of this result is the following: if $s = e$, then the solution of a 2-norm cut is a reweighted, renormalized solution of PageRank with $v = e/n$.
That is, as a corollary of this approach, the \emph{standard} PageRank problem with $v = e/n$ gives rise to a cut problem where $s$ connects to each node with weight $\alpha$ and $t$ connects to each node $v$ with weight $\alpha(d_v - 1)$.
\item
This also holds for the semi-supervised learning results we discussed.
In particular, e.g., the procedure of Zhou et al. for semi-supervised learning on graphs solves the following:
\[ ( I - \beta D^{-1/2}AD^{-1/2})^{-1} Y . \]
(The other procedures solve a very similar problem.)
This is exactly a PageRank equation for a degree-based scaling of the labels, and thus the construction from Theorem~\ref{thm:antiderivative1converse} is directly applicable.
\end{itemize}

\subsection{Implicit $\ell_1$ regularization in strongly local spectral methods}

In light of these results, let's now move onto the ACL procedure.
We will show a connection between it and an $\ell_1$ regularized version of an $\ell_2$ objective,  as established in Theorem~\ref{thm:antiderivative1converse}.
In particular, we will show that the ACL procedure for \emph{approximating} a PPR vector \emph{exactly} computes a hybrid $1$-norm $2$-norm variant of the min-cut problem.
The balance between these two terms (the $\ell_2$ term from Problem~\ref{eq:pr-unif} and an additional $\ell_1$ term) has the effect of producing sparse PageRank solutions that also have sparse truncated residuals, and it also provides an interesting connection with $\ell_1$-regularized
$\ell_2$-regression problems.

We start by reviewing the ACL method and describing it in such a way to make these connections easier to establish.

Consider the problem $(I - \beta A D^{-1}) x = (1-\beta) v$, where $v = e_i$ is localized onto a single node.
In addition to the PageRank parameter $\beta$, the procedure has two parameters: $\tau > 0$ is a accuracy parameter that determines when to stop, and $0 < \rho \le 1$ is an additional approximation term that we introduce.
As $\tau \to 0$, the computed solution $x$ goes to the PPR vector that is non-zero everywhere.
The value of $\rho$ has been $1/2$ in most previous implementations of the procedure; and here we present a modified procedure that makes the effect of $\rho$ explicit.
\begin{enumerate}
\item $x^{(1)} = 0, r^{(1)} = (1-\beta) e_i $, $k = 1$
\item \emph{while} any $r_j  > \tau d_j$ \qquad \emph{(where $d_j$ is the degree of node $j$)}
\item \hspace{1em} $x^{(k+1)} = x^{(k)} +  (r_j - \tau d_j \rho) e_j$
\item \hspace{1em} $r^{(k+1)}_i =
                \begin{cases}
                                \tau d_j \rho & i = j \\
                                r^{(k)}_i + \beta (r_j - \tau d_j \rho) / d_j & i \sim j \\
                                r^{(k)}_i & \text{otherwise}
                \end{cases}$
\item \hspace{1em} $k \leftarrow k + 1$
\end{enumerate}

As we have noted previously, one of the important properties of this procedure is that the algorithm maintains the invariant $r = (1-\beta) v - (I - \beta A D^{-1}) x $ throughout.
For any $0 \le \rho \le 1$, this algorithm converges because the sum of entries in the residual always decreases monotonically.
At the solution we will have \[ 0 \le r \le \tau d, \] which provides an $\infty$-norm style worst-case \emph{approximation} guarantee to the exact PageRank solution.

Consider the following theorem.
In the same way that Theorem~\ref{thm:antiderivative1converse} establishes that a PageRank vector can be interpreted as optimizing an $\ell_2$ objective involving the edge-incidence matrix, the following theorem establishes that, in the case that $\rho = 1$,  the ACL procedure to approximate this vector can be interpreted as solving an $\ell_1$-regularized $\ell_2$ objective.
That is, in addition to \emph{approximating} the solution to the objective function that is optimized by the PPR, this algorithm also \emph{exactly} computes the solution to an $\ell_1$ regularized version of the same objective.

\begin{theorem}
\label{thm:implicitL1Reg}
Fix a subset of vertices $S$. 
Let $x$ be the output from the ACL procedure with $\rho = 1$, $0 < \beta < 1$, $v = d_{S}/\vol(S)$, and $\tau$ fixed.  
Set $\alpha = \frac{1-\beta}{\beta}$, $\kappa = \tau \vol(S) / \beta$, and let $z_G$ be the solution on graph vertices of the sparsity-regularized cut problem:
\begin{eqnarray}
\label{eq:acl-pr}
\text{min} & \frac{1}{2}\| B(s) z \|_{C(\alpha),2}^{2} + \kappa \| Dz \|_{1} \\
\nonumber
\text{s.t.} & z_s = 1, z_t = 0, z \ge 0  ,
\end{eqnarray}
where $z = \left[\begin{array}{c} 1 \\ z_G \\ 0 \end{array}\right]$ as above.
Then $x = D z_G/\mbox{vol}(S)$.
\end{theorem}
\begin{proof}   
If we expand the objective function and apply the constraint $z_s = 1, z_t = 0$, then Prob.~\eqref{eq:acl-pr} becomes:
\begin{eqnarray}
\text{min} & \frac{1}{2} z_G^T (\alpha D + L) z_G - \alpha z_G^T d_{S} + \alpha^2 \mbox{vol}(S) + \kappa d^Tz_G \\
\nonumber
\text{s.t.} & z_G \ge 0
\end{eqnarray}
Consider the optimality conditions of this quadratic problem (where $s$ are the Lagrange multipliers):
\[ 
\begin{aligned}
 0 & = (\alpha D + L) z_G  - \alpha d_{\bar{S}} + \kappa d - s \\
s & \ge 0 \\
z_G & \ge 0 \\
z_G^T s&  = 0.
\end{aligned} 
\]
These are both necessary and sufficient because $(\alpha D + L)$ is positive definite.
In addition, and for the same reason, the solution is unique.

In the remainder of the proof, we demonstrate that vector $x$ produced by the ACL method satisfies these conditions. 
To do so, we first translate the optimality conditions to the equivalent PageRank normalization:
\[ 
\begin{aligned}
0 & = (I - \beta A D^{-1}) D z_G/\mbox{vol}(S) -    (1-\beta) d_{S} / \mbox{vol}(S) + \beta \kappa / \mbox{vol}(S) d - \beta s / \mbox{vol}(S) \\
s & \ge 0 \qquad z_G  \ge 0  \qquad z_G^T s  = 0.
\end{aligned} \]

When the ACL procedure finishes with $\beta$, $\rho$, and $\tau$ as in the theorem, the vectors $x$ and $r$ satisfy:
\[ 
\begin{aligned}
   r & = (1-\beta) v - (I - \beta A D^{-1}) x \\
   x & \ge 0 \\
   0  &  \le r \le \tau d = \beta \kappa / \mbox{vol}(S) d. 
\end{aligned} 
\]

Thus, if we set $s$ such that $ \beta s / \mbox{vol}(S) =  \beta \kappa / \mbox{vol}(S) d - r $, then we satisfy the first condition with $x =  D z_G/\mbox{vol}(S)$. 
All of these transformations preserve $x \ge 0$ and $z_G \ge 0$. Also, because $\tau d \ge r$, we also have $s \ge 0$. 
What remains to be shown is $z_G^T s = 0$.

Here, we show $x^T (\tau d - r) = 0$, which is equivalent to the condition $z_G^T s = 0$ because the non-zero structure of the vectors is identical. 
Orthogonal non-zero structure suffices because $z_G s = 0$ is equivalent to either $x_i = 0$ or $\tau d_i - r_i = 0$ (or both) for all $i$.  
If $x_i \not= 0$, then at some point in the execution, the vertex $i$ was chosen at the step $r_j > \tau d_j$. 
In that iteration, we set $r_i = \tau d_i$.  
If any other step increments $r_i$, we must revisit this step and set $r_i = \tau d_i$ again. 
Then at a solution, $x_i \not= 0$ requires $r_i = \tau d_i$. 
For such a component, $s_i = 0$, using the definition above. 
For $x_i = 0$, the value of $s_i$ is irrelevant, and thus, we have $x^T (\tau d - r) = 0$.
\end{proof}

\textbf{Remark.}
Finally, a comment about $\rho$, which is set to $1$ in this theorem but equals $1/2$ in most prior uses of the ACL push method.
The proof of Theorem~\ref{thm:implicitL1Reg} makes the role of $\rho$ clear.
If $\rho < 1$, then the output from ACL is \emph{not} equivalent to the solution of Prob.~\eqref{eq:acl-pr}, i.e., the renormalized solution will \emph{not} satisfy $z_G^T s = 0$; but setting $\rho < 1$, however, \emph{will} compute a solution much more rapidly.
It is a nice open problem to get a clean statement of implicit regularization when $\rho < 1$.

%% file: lect22.tex
\section{%
(04/14/2015): 
Some Statistical Inference Issues (1 of 3):
Introduction and Overview}

Reading for today.
\begin{compactitem}
\item
``Towards a theoretical foundation for Laplacian-based manifold methods,'' in JCSS, by Belkin and Niyogi
\end{compactitem}

\subsection{Overview of some statistical inference issues}

So far, most of what we have been doing on spectral methods has focused on various sorts of algorithms---often but not necessarily worst-case algorithms.
That is, there has been a bias toward algorithms that are more rather than less well-motivated statistically---but there hasn't been a lot statistical emphasis \emph{per se}. 
Instead, most of the statistical arguments have been informal and by analogy, e.g., if the data are nice, then one should obtain some sort of smoothness, and Laplacians achieve that in a certain sense; or diffusions on graphs should look like diffusions on low-dimensional spaces or a complete graph; or diffusions are robust analogues of eigenvectors, which we illustrated in several ways; and so on.

Now, we will spend a few classes trying to make this statistical connection a little more precise.
As you can imagine, this is a large area, and we will only be able to scratch the surface, but we will try to give an idea of the space, as well as some of the gotchas of naively applying existing statistical or algorithmic methods here---so think of this as pointing to lots of interesting open questions to do statistically-principled large-scale computing, rather than the final word on the topic.

From a statistical perspective, many of the issues that arise are somewhat different than much of what we have been considering.
\begin{itemize}
\item
Computation is much less important (but perhaps it should be much more so).
\item
Typically, one has some sort of model (usually explicit, but sometimes implicit, as we saw with the statistical characterization of the implicit regularization of diffusion-based methods), and one wants to compute something that is optimal for that model.
\item
In this case, one might want to show things like convergence or consistence (basically, that what is being computed on the empirical data converges to the answer that is expected, as the number of data points $n \rightarrow \infty$).
\end{itemize}
For spectral methods, at a high level, there are basically two types of reference states or classes of models that are commonly-used: one is with respect to some sort of very low-dimensional space; and the other is with respect to some sort of random graph model.
\begin{itemize}
\item
\textbf{Low-dimensional spaces.}
In this simplest case, this is a line; more generally, this is a low-dimensional linear subspace; and, more generally, this is a low-dimensional manifold.
Informally, one should think of low-dimensional manifolds in this context as basically low-dimensional spaces that are curved a bit; or, relatedly, that the data are low-dimensional, but perhaps not in the original representation.
This manifold perspective provides added descriptive flexibility, and it permits one to take advantage of connections between the geometry of continuous spaces and graphs (which in very special cases are a discretization of those continuous places).
\item
\textbf{Random graphs.}
In the simples case, this is simply the $G_{nm}$ or $G_{np}$ Erdos-Renyi (ER) random graph.
More generally, one is interested in finding clusters, and so one works with the stochastic blockmodel (which can be thought of as a bunch of ER graphs pasted together).
Of course, there are many other extensions of basic random graph models, e.g., to include degree variability, latent factors, multiple cluster types, etc.
\end{itemize}
These two places provide two simple reference states for statistical claims about spectral graph methods; and the types of guarantees one obtains are somewhat different, depending on which of these reference states is assumed.
Interestingly, (and, perhaps, not surprisingly) these two places have a direct connection with the two complementary  places (line graphs and expanders) that spectral methods implicitly embed the data.

In both of these cases, one looks for theorems of the form: ``If the data are drawn from this place and things are extremely nice (e.g., lots of data and not too much noise) then good things happen (e.g., finding the leading vector, recovering hypothesized clusters, etc.) if you run a spectral method.
We will cover several examples of this.
A real challenge arises when you have realistic noise and sparsity properties in the data, and this is a topic of ongoing research.

As just alluded to, another issue that arises is that one needs to specify not only the hypothesized statistical model (some type of low-dimensional manifold or some type of random graph model here) but also one needs to specify exactly what is the problem one wants to solve. 
Here are several~examples.
\begin{itemize}
\item
One can ask to recover the objective function value of the objective you write down.
\item
One can ask to recover the leading nontrivial eigenvector of the data.
\item
One can ask to converge to the Laplacian of the hypothesized model.
\item
One can ask to find clusters that are present in the hypothesized model.
\end{itemize}
The first bullet above is most like what we have been discussing so far.
In most cases, however, people want to use the solution to that objective for something else, and the other bullets are examples of that.
Typically in these cases one is asking for a lot more than the objective function value, e.g., one wants to recover the ``certificate'' or actual solution vector achieving the optimum, or some function of it like the clusters that are found by sweeping along it, and so one needs stronger assumptions.
Importantly, many of the convergence and statistical issues are quite different, depending on the exact problem being considered.

\begin{itemize}
\item
\textbf{Today}, we will assume that the data points are drawn from a low-dimensional manifold and that from the empirical point cloud of data we construct an empirical graph Laplacian; and we will ask how this empirical Laplacian relates to the Laplacian operator on the manifold.
\item
\textbf{Next time}, we will ask whether spectral clustering is consistent in the sense that it converges to something meaningful and $n \rightarrow \infty$, and we will provide sufficient conditions for this (and we will see that the seemingly-minor details of the differences between unnormalized spectral clustering and normalized spectral clustering lead to very different statistical results).
\item
\textbf{On the day after that}, we will consider results for spectral clustering in the stochastic blockmodel, for both vanilla situations as well as for situations in which the data are very~sparse.
\end{itemize}

\subsection{Introduction to manifold issues}

Manifold-based ML is an area that has received a lot of attention recently, but for what we will discuss today one should think back to the discussion we had of Laplacian Eigenmaps.
At root, this method defines a set of features that can then be used for various tasks such as data set parametrization, clustering, classification, etc.
Often the features are useful, but sometimes they are not; here are several examples of when the features developed by LE and related methods are often less than useful.
\begin{itemize}
\item
\textbf{Global eigenvectors are localized.}
In this case, ``slowly-varying'' functions (by the usual precise definition) are not so slowly-varying (in a sense that most people would find intuitive).
\item
\textbf{Global eigenvectors are not useful.}
This may arise if one is interested in a small local part of the graph and if information of interest is not well-correlated with the leading or with any eigenvector.
\item
\textbf{Data are not meaningfully low-dimensional.}
Even if one believes that there is some sort of hypothesized curved low-dimensional space, there may not be a small number of eigenvectors that capture most of this information.
(This does \emph{not} necessarily mean that the data are ``high rank,'' since it is possible that the spectrum decays, just very slowly.)
This is more common for very sparse and noisy data, which are of course very common.
\end{itemize}
Note that the locally-biased learning methods we described, e.g., the \textsf{LocalSpectral} procedure, the PPR procedure, etc., was motivated by one common situation when the global methods such as LE and related methods had challenges.

While it may be fine to have a ``feature generation machine,'' most people prefer some sort of theoretical justification that says when a method works in some idealized situation.  
To that end, many of the methods like LE assume that the data are drawn from some sort of low-dimensional manifold.
Today, we will talk about one statistical aspect of that having to do with converging to the manifold.

To start, here is a simple version of the ``manifold story'' for a classification problem.
Consider a $2$-class classification problem with classes $C_1$ and $C_2$, where the data elements are drawn from some space $\mathcal{X}$, whose elements are to be classified.
A statistical or probabilistic model typically includes the following two ingredients:
a probability density $p(x)$ on $\mathcal{X}$; and
class densities $\{ p \left( C_i | x \in \mathcal{X} \right) \}$,  for $i\in\{1,2\}$.
Importantly, if there are unlabeled data, then the unlabeled data don't tell us much about the conditional class distributions, as we can't identify classes without labels, but the unlabeled data can help us to improve our estimate of the probability distribution $p(x)$.
That is, the unlabeled data tell us about $p(x)$, and the labeled data tell us about $\{ p \left( C_i | x \in \mathcal{X} \right) \}$.

If we say that the data come from a low-dimensional manifold $\mathcal{X}$, then a natural geometric object to consider is the Laplace-Beltrami operator on $\mathcal{X}$.
In particular, let $\mathcal{M} \subset \mathbb{R}^{n}$ be an $n$-dimensional compact manifold isometrically embedded in $\mathbb{R}^{k}$.
(Think of this as an $n$-dimensional ``surface'' in $\mathbb{R}^{k}$.)
The Riemannian structure on $\mathcal{M}$ induces a volume form that allows us to integrate functions defined on $\mathcal{M}$.
The square-integrable functions form a Hilbert space $\mathcal{L}^2\left(\mathcal{M}\right)$.
Let $C^{\infty}\left(\mathcal{M}\right)$ be the space of infinitely-differentiable functions on $\mathcal{M}$.
Then, the Laplace-Beltrami  operator is a second order differentiable operator
$\Delta_{\mathcal{M}} : C^{\infty}\left(\mathcal{M}\right) \rightarrow C^{\infty}\left(\mathcal{M}\right)$.
We will define this in more detail below; for now, just note that if the manifold is $\mathbb{R}^{n}$, then the Laplace-Beltrami  operator is $\Delta = - \frac{\partial^2}{\partial x_1^2}$.
There are two important properties of the Laplace-Beltrami  operator.

\begin{itemize}
\item
\textbf{It provides a basis for $\mathcal{L}^2\left(\mathcal{M}\right)$.}
In general, $\Delta$ is a PSD self-adjoint operator (w.r.t. the $\mathcal{L}^2$ inner product) on twice differentiable functions.
In addition, if $\mathcal{M}$ is a \emph{compact} manifold, then $\Delta$ has a discrete spectrum, the smallest eigenvalue of $\Delta$ equals $0$ and the associated eigenfunction is the constant eigenfunction, and the eigenfunctions of $\Delta$ provide an orthonormal basis for the Hilbert space $\mathcal{L}^2\left(\mathcal{M}\right)$.
In that case, any function $f \in \mathcal{L}^2\left(\mathcal{M}\right)$ can be written as $f(x) = \sum_{i=1}^{\infty} a_i e_i(x)$, where $e_i$ are the eigenfunctions of $\Delta$, i.e., where $\Delta e_i = \lambda_i e_i$.

In this case, then the simplest model for the classification problem is that the class membership is a square-integrable function, call it $m: \mathcal{M}\rightarrow\{-1,+1\}$, in which case the classification problem can be interpreted as interpolating a function on the manifold.
Then we can choose the coefficients to get an optimal fit, $m(x) = \sum_{i=1}^{n} a_i e_i$, in the same way as we might approximate a signal with a Fourier series.
(In fact, if $\mathcal{M}$ is a unit circle, call it $S^1$, then $\Delta_{S_1} f(\theta) = -\frac{d^2f(\theta)}{d \theta^2}$, and the eigenfunctions are sinusoids with eigenvalues $\{1^2,2^2,\ldots\}$. and we get the usual Fourier series.)
\item
\textbf{It provides a smoothness functional.}
Recall that a simple measure of the degree of smoothness for a function $f$ on the unit circle $S^1$ is
$$
S(f) = \int_{S^1} | f(\theta)^{\prime}|^2 d\theta  .
$$
In particular, $f$ is smooth iff this is close to zero.
If we take this expression and integrate by parts, then we get
$$
S(f) = \int_{S^1} f^{\prime}(\theta) d\theta
     = \int_{S^1} f \Delta f d\theta
     = \langle \Delta f,f \rangle_{\mathcal{L}^2\left(S^1\right)}  .
$$
More generally, if $f : \mathcal{M} \rightarrow \mathbb{R}$, then it follows that
$$
S(f) = \int_{\mathcal{M}} | \nabla f |^2 d\mu
     = \int_{\mathcal{M}} f \Delta f d\mu
     = \langle \Delta f,f \rangle_{\mathcal{L}^2\left(\mathcal{M}\right)}  .
$$
So, in particular, the smoothness of the eigenfunction is controlled by the eigenvalue, i.e., 
$$
S(e_i) = \langle \Delta e_i , e_i \rangle_{\mathbb{L}^2\left(\mathcal{M}\right)} = \lambda_i  ,
$$
and for arbitrary $f$ that can be expressed as $f = \sum_i \alpha_i e_i$, we have that
$$
S(f) = \langle \Delta f , f \rangle 
     = \left\langle \sum_i \alpha_i \Delta e_i , \sum_i \alpha_i e_i \right\rangle
     = \sum_i \lambda_i \alpha_i^2    .
$$
(So, in particular, approximating a function $f$ by its first $k$ eigenfunctions is a way to control the smoothness of the eigenfunctions; and the linear subspace where the smoothness functions is finite is a RKHS.)
\end{itemize}

This has strong connections with a range of RKHS problems.
(Recall that a RKHS is a Hilbert space of functions where the evaluation functionals, the functionals that evaluate functions at a point, are bounded linear functionals.)
Since the Laplace-Beltrami operator on $\mathcal{M}$ can be used to provide a basis for $\mathcal{L}^2\left(\mathcal{M}\right)$, we can take various classes of functions that are defined on the manifold and solve problems of the form
\begin{equation}
\min_{f \in H} \sum_i \left( y_i - f(x_i) \right)^2 + \lambda G(f)  ,
\label{eqn:reg-opt}
\end{equation}
where $H :\mathcal{M} \rightarrow \mathbb{R}$.
In general, the first term is the empirical risk, and the second term is a stabilizer or regularization term.
As an example, one could choose 
$G(f) = \int_{\mathcal{M}} \langle \nabla f, \nabla f \rangle = \sum_i \alpha_i^2 \lambda_i $ (since $f = \sum_i \alpha_i e_i(x)$), and $H = \{ f = \sum_i \alpha_i e_i | G(f) < \infty \}$, 
in which case one gets an optimization problem that is quadratic in the $\alpha$ variables.

As an aside that is relevant to what we discussed last week with the heat kernel, let's go through the construction of a RKHS that is invariantly defined on the manifold $\mathcal{M}$.
To do so, let's fix an infinite sequence of non-negative numbers $\{ \mu_i | i \in \mathbb{Z}^{+} \}$ s.t. $\sum_i \mu_i < \infty$ (as we will consider in the examples below).
Then, define the following linear space of continuous functions
$$
H = \left\{ f = \sum_i \alpha_i f_i  |  \sum_i \frac{\alpha_i^2}{\mu_i} < \infty \right\} . 
$$
Then, we can define the inner product as: 
for $f = \sum_i \alpha_i f_i$ and $g = \sum_i \beta_i g_i$, we have $\left\langle f,g \right\rangle = \sum_i \frac{\alpha_i \beta_i}{ \mu_i} $.
Then, $H$ is a RKHS with the following kernel: $K(p,q) = \sum_i \mu_i e_i(p)e_i(q)$.
Then, given this, we can solve regularized optimization problems of the form given in Eqn.~(\ref{eqn:reg-opt}) above.
In addition, we can get other choices of kernels by using different choices of $\mu$ vectors.
For example, if we let $\mu_i = e^{-t\lambda_i}$, where $\lambda_i$ are the eigenvalues of $\Delta$, then we get the heat kernel corresponding to heat diffusion on the manifold; if we let $\mu_i=0$, for all $i > i^{*}$, then we are solving an optimization problem in a finite dimensional space; and so on.

All of this discussion has been for data drawn from an hypothesized manifold $\mathcal{M}$.
Since we are interested in a smoothness measure for functions for a graph, then if we think of the graph as a model for the manifold, then we want the value of a function not to change too much between points.
In that case, we get
$$
S_G(f) = \sum_{i ~ j} W_{ij} \left( f_i - f_j \right)  ,
$$
and it can be shown that
$$
S_G(f) = fLf^T = \langle f,Lf \rangle_G = \sum_{i=1}^{n} \lambda_i \langle f,e_i \rangle_G  .
$$
Of course, this is the discrete object with which we have been working all along.
Viewed from the manifold perspective, this corresponds to the discrete analogue of the integration by parts we performed above.
In addition, we can use all of this to consider questions having to do with ``regularization on manifolds and graphs,'' as we have allude to in the past.

To make this connection somewhat more precise, recall that for a RKHS, there exists a kernel $K : X \times X \rightarrow \mathbb{R}$ such that $f(x) = \left< f(\cdot),K(x,\cdot) \right>_H$.
For us today, the domain $X$ could be a manifold $\mathcal{M}$ (in which case we are interested in kernels $K: \mathcal{M} \times \mathcal{M} \rightarrow\mathbb{R}$), or it could be points from $\mathbb{R}^{n}$ (say, on the nodes of the graph that was constructed from the empirical original data by a nearest neighbor rule, in which case we are interested in kernels $K:\mathbb{R}^{n}\times\mathbb{R}^{n}\rightarrow\mathbb{R}$).
We haven't said anything precise yet about how these two relate, so now let's turn to that and ask about connections between kernels constructed from these two different places, as $n \rightarrow\infty$.

\subsection{Convergence of Laplacians, setup and background}

Now let's look at questions of convergence.

If $H : \mathcal{M} \rightarrow{R}$ is a RKHS invariantly defined on $\mathcal{M}$, then the key goal is to minimize regularized risk functionals of the form 
$$
E_{\lambda} = \min_{f \in H} \mathbb{E}\left[ \left( y-f(x)\right)^2 \right] + \lambda \|f\|_H^2  .
$$
In principle, we can do this---if we had an infinite amount of data available \emph{and} the true manifold is known.
Instead, we minimize the empirical risk which is of the form
$$
\hat{E}_{\lambda,n} = \min_{f \in H} \frac{1}{n} \sum \left( y_i - f(x_i) \right)^2 + \lambda \|f \|_H  .
$$
The big question is: how far is $\hat{E}_{\lambda,n}$ from $E_{\lambda}$.

The point here is the following: assuming the manifold is known or can be estimated from the data, then making this connection is a relatively-straightforward application of Hoeffding bounds and regularization/stability ideas.
But:
\begin{itemize}
\item
In theory, establishing convergence to the hypothesized manifold is challenging.
We will get to this below.
\item
In practice, testing the hypothesis that the data are drawn from a manifold in some meaningful sense of the word is harder still.
(For some reason, this question is not asked in this area.
It's worth thinking about what would be test statistics to validate or invalidate the manifold hypothesis, e.g., is that the best conductance clusters are not well balanced sufficient to invalidate it?)
\end{itemize}

So, the goal here is to describe conditions under which the point cloud in $\mathcal{X}$ of the sample points converges to the Laplace-Beltrami  operator on the underlying hypothesized manifold $\mathcal{M}$.
From this perspective, the primary data are points in $\mathcal{X}$, that is assumed to be drawn from an underlying manifold, with uniform or nonuniform density, and we want to make the claim that the Adjacency Matrix or Laplacian Matrix of the empirical data converges to that of the manifold. 
(That is, the data are not a graph, as will arise with the discussion of the stochastic block model.)
In particular, the graph is and empirical object, and if we view spectral graph algorithms as applying to that empirical object then they are stochastically justified when they can relate to the underlying processes generating the data.

What we will describe today is the following.
\begin{itemize}
\item
For data drawn from a uniform distribution on a manifold $\mathcal{M}$, the graph Laplacian converges to the Laplace-Beltrami  operator, as $n \rightarrow \infty$ and the kernel bandwidth is chosen appropriately (where the convergence is uniform over points on the manifold and for a class of functions).
\item
The same argument applies for arbitrary probability distributions, except that one converges to a weighted Laplacian; and in this case the weights can be removed to obtain convergence to the normalized Laplacian.
(Reweighting can be done in other ways to converge to other quantities of interest, but we won't discuss that in detail.)
\end{itemize}

Consider a compact smooth manifold $\mathcal{M}$ isometrically embedded in $\mathbb{R}^{n}$.
The embedding induces a measure corresponding to volume form $\mu$ on the manifold (e.g., the volume form for a closed curve, i.e., an embedding of the circle, measures the usual curve length in $\mathbb{R}^{n}$).
The Laplace-Beltrami operator $\Delta_{\mathcal{M}}$ is the key geometric object associated to a Riemannian manifold.
Given $\rho \in \mathcal{M}$, the tangent space $T_{\rho}\mathcal{M}$ can be identified with the affine space to tangent vectors to $\mathcal{M}$ at $\rho$.
(This vector space has a natural inner product induced by embedding $\mathcal{M} \subset \mathbb{R}^{n}$.)
So, given a differentiable function $f:\mathcal{M} \rightarrow \mathcal{R}$, let $\nabla_{\mathcal{M}}f$ be the gradient vector on $\mathcal{M}$ (where $\nabla_{\mathcal{M}}f(p)$ points in the direction of fastest ascent of $f$ at $\rho$.
Here is the definition.

\begin{definition}
The \emph{Laplace-Beltrami operator} $\Delta_{\mathcal{M}}$ is the divergence of the gradient, i.e., 
$$
\Delta_{\mathcal{M}}f = - \mbox{div} \left( \nabla_{\mathcal{M}} f \right)  .
$$ 
Alternatively, $\Delta_{\mathcal{M}}$ can be defined as the unique operator s.t., for all two differentiable functions $f$ and $h$, 
$$
\int_{\mathcal{M}} h(x) \Delta_{\mathcal{M}} f(x) d \mu(x) 
   = \int_{\mathcal{M}} \left< \nabla_{\mathcal{M}} h(x), \nabla_{\mathcal{M}} f(x) \right> d \mu  ,
$$
where the inner product is on the tangent space and $\mu$ is the uniform measure.
\end{definition}

In $\mathbb{R}^{n}$, we have $\Delta f = - \sum_i \frac{\partial^2 f}{\partial x_i^2}$.
More generally, on a $k$-dimensional manifold $\mathcal{M}$, in a local coordinate system $\left( x_1, \ldots, x_n \right)$, with a metric tensor $g_{ij}$, if $g^{ij}$ are the components of the inverse of the metric tensor, then the Laplace-Beltrami operator applied to a function $f$ is
$$
\Delta_{\mathcal{M}} f 
   = \frac{1}{\sqrt{\mbox{det}(g)}}  
      \sum_j \frac{\partial}{\partial x^j} \left( \sqrt{\mbox{det}(g)} \sum_i g^{ij} \frac{\partial f}{ \partial x_i} \right)
$$

(If the manifold has nonuniform measure $\nu$, given by $d \nu(x) = P(x) d \mu (x)$, for some function $P(x)$ and with $d\mu$ being the canonical measure corresponding to the volume form, then we have the more general notion of a \emph{weighted manifold Laplacian}:
$\Delta_{\mathcal{M},\mu} = \Delta_P f = \frac{1}{P(x)} \mbox{div} \left( P(x) \nabla_{\mathcal{M}} f \right) $.)

The question is how to reconstruct $\Delta_{\mathcal{M}}$, given a finite sample of data points from the manifold?
Here are the basic objects (in addition to $\Delta_{\mathcal{M}}$) that are used to answer this question.
\begin{itemize}
\item
\textbf{Empirical Graph Laplacian.}
Given a sample of $n$ points $x_i, \ldots, x_n$ from $\mathcal{M}$, we can construct a weighted graph with weights $W_{ij} = e^{-\|x_i-x_j\|^2/4t}$, and then 
\[
\left( L_n^t\right)_{ij} 
= 
\left\{ \begin{array}{l l}
                    -W_{ij} & \quad \text{if $i \ne j$}\\
                    \sum_k W_{ik}      & \quad \text{if $i=j$}
                 \end{array}
         \right.  .
\]
Call $L_n^t$ the graph Laplacian matrix.
We can think of $L_n^t$ as an operation of functions on the $n$ empirical data points:
\[
L_n^t f(x_i) = f(x_i) \sum_j e^{-\|x_i-x_j\|^2/(4t)} - \sum_j f(x_j) e^{-\|x_i-x_j\|^2/(4t)}  ,
\]
but this operator operates only on the empirical data, i.e., it says nothing about other points from $\mathcal{M}$ or the ambient space in which $\mathcal{M}$ is embedded.
\item
\textbf{Point Cloud Laplace operator.}
This formulation extends the previous results to any function on the ambient space.
Denote this by $\underbar{L}_n^t$ to get
\[
\underbar{L}_n^t f(x) = f(x) \frac{1}{n} \sum_j e^{-\|x-x_j\|^2/(4t)} - \frac{1}{n}\sum_j f(x_j) e^{-\|x-x_j\|^2/(4t)}
\]
(So, in particular, when evaluated on the empirical data points, we have that $\underbar{L}_n^t f(x_i) = \frac{1}{n} L_n^T f(x_i)$.)
Call $\underbar{L}_n^t$ the Laplacian associated to the point cloud $x_1,\ldots,x_n$.
\item
\textbf{Functional approximation to the Laplace-Beltrami operator.}
Given a measure $\nu$ on $\mathcal{M}$, we can construct an operator 
\[
\underbar{L}^tf(x) = f(x) \int_{\mathcal{M}} e^{-\|x-y\|^2/(4t)} d\nu(y) - \int_{\mathcal{M}} f(y) e^{-\|x-y\|^2/(4t)} d\nu(y) .
\]
Observe that $\underbar{L}_n^t$ is just a special form of $\underbar{L}^t$, corresponding to the Dirac measure supported on $x_1,\ldots,x_n$.
\end{itemize}

\subsection{Convergence of Laplacians, main result and discussion}

The main result they describe is to establish a connection between the graph Laplacian associated to a point cloud (which is an extension of the graph Laplacian from the empirical data points to the ambient space) and the Laplace-Beltrami operator on the underlying manifold $\mathcal{M}$.
Here is the main results.
\begin{theorem}
Let $x_1,\ldots,x_n$ be data points sampled from a uniform distribution on the manifold $\mathcal{M} \subset \mathbb{R}^{n}$.
Choose $t_n = n^{-1/(k+2+\alpha)}$, for $\alpha > 0$, and let $f \in C^{\infty}\left(\mathcal{M}\right)$.
Then
\[
\lim_{n \rightarrow \infty} \frac{1}{t_n(4 \pi t_n)^{k/2}} \underbar{L}_n^{t_n} f(x) 
   = \frac{1}{\mbox{Vol}\left(\mathcal{M}\right)} \Delta_{\mathcal{M}} f(x) ,
\]
where the limit is taken in probability and $\mbox{Vol}\left(\mathcal{M}\right)$ is the volume of the manifold with respect to the canonical measure.
\end{theorem}

We are not going to go through the proof in detail, but we will outline some key ideas used in the proof.
Before doing that, here are some things to note.
\begin{itemize}
\item
This theorem assert pointwise convergence of $\underbar{L}_n^t f(p)$ to $\Delta_{\mathcal{M}} f(p)$, for a fixed function $f$ and a fixed point $p$.
\item
Uniformity over all $p \in \mathcal{M}$ follows almost immediately from the compactness of $\mathcal{M}$.
\item
Uniform convergence over a class of function, e.g., functions $C^{k}\left(\mathcal{M}\right)$ with bounded $k^{th}$ derivative, follows with more effort.
\item
One can consider a more general probability distribution $P$ on $\mathcal{M}$ according to which data points are sampled---we will get back to an example of this below. 
\end{itemize}

For the proof, the easier part is to show that $\underbar{L}_n^t \rightarrow \underbar{L}^t$, as $n \rightarrow \infty$, if points are samples uniformly: this uses some basic concentration results.
The harder part is to connect $\underbar{L}^t$ and $\Delta_{\mathcal{M}}$: what must be shown is that when $t \rightarrow 0$, then $L^t $ appropriately scaled converges to $\Delta_{\mathcal{M}}$.

Here are the basic proof ideas, which exploit heavily connections with the heat equation on $\mathcal{M}$.

For simplicity, consider first $\mathbb{R}^{n}$, where we have the following theorem.
\begin{theorem}[Solution to heat equation on $\mathbb{R}^{k}$]
Let $f(x)$ be a sufficiently differentiable bounded function.
Then 
\[
H^t f = \left(4 \pi t \right)^{-k/2} \int_{\mathbb{R}^{k}} e^{-\frac{\|x-y\|^2}{4t}} f(y) dy  ,
\]
and
\[
f(x) = \lim_{t \rightarrow 0} H^t f(x) = \left(4 \pi t\right)^{-k/2} \int_{\mathbb{R}^{k}} e^{-\frac{\|x-y\|^2}{4t}} f(y)dy  ,
\]
and the function $u(x,t) = H^t f$ satisfies the heat equation
\[
\frac{\partial}{\partial t} u(x,t) + \Delta u(x,t) = 0 
\]
with initial condition $u(x,0) = f(x)$.
\end{theorem}

This result for the heat equation is the key result for approximating the Laplace operator.
\begin{eqnarray*}
\Delta f(x) &=& -\frac{\partial}{\partial t} u(x,t) |_{t=0} \\ 
   &=& - \frac{\partial}{\partial t} H^t f(x) |_{t=0}  \\
   &=& \lim_{t \rightarrow 0} \frac{1}{t} \left( f(x) - H^t f(x) \right)  .
\end{eqnarray*}
By this last result, we have a scheme for approximating the Laplace operator.
To do so, recall that the heat kernel is the Gaussian that integrates to $1$, and so
\[
\Delta f(x) = \lim_{t \rightarrow 0} -\frac{1}{t} \left( \left(4 \pi t\right)^{-k/2} \int_{\mathbb{R}^{k}} e^{\frac{-\|x-y\|^2}{4t}}f(y)dy - f(x) \left(4 \pi t\right)^{-k/2} \int_{\mathbb{R}^{k}} e^{\frac{-\|x-y\|^2}{4t}} dy  .
 \right)
\]
It can be shown that this can be approximated by the point cloud $x_1,\ldots,x_n$ by computing the empirical version as
\begin{eqnarray*}
\hat{\Delta} f(x) 
   &=& \frac{1}{t} \frac{\left(4 \pi t\right)^{-k/2}}{n} \left( f(x)\sum_i e^{\frac{-\|x-x_i\|^2}{4t}}  - \sum_i e^{\frac{-\|x-x_i\|^2}{4t}}  f(x_i) \right) \\
   &=& \frac{1}{t\left(4\pi t\right)^{k/2}} \underbar{L}_n^t f(x)   .
\end{eqnarray*}
It is relatively straightforward to extend this to a convergence result for $\mathbb{R}^{k}$.
To extend it to a convergence result for arbitrary manifolds $\mathcal{M}$, two issues arise:
\begin{itemize}
\item
With very few exceptions, we don't know the exact form of the heat kernel $H_{\mathcal{M}}^t(x,y)$.
(It has the nice form of a Gaussian for $\mathcal{M}=\mathbb{R}^{k}$.)
\item
Even asymptotic forms of the heat kernel requires knowing the geodesic distance between points in the point cloud, but we can only observe distance in the ambient space.
\end{itemize}
See their paper for how they deal with these two issues; this involves methods from differential geometry that are very nice but that are not directly relevant to what we are doing.

Next, what about sampling with respect to nonuniform probability distributions?
Using the above proof, we can establish that we converge to a weighted Laplacian.
If this is not of interest, then once can instead normalize differently and get one of two results.
\begin{itemize}
\item
The weighted scaling factors can be removed by using a different normalization of the weights of the point cloud.
This different normalization basically amounts to considering the normalized Laplacian.
See below.
\item
With yet a different normalization, we can recover the Laplace-Beltrami  operator on the manifold.
The significance of this is that it is possible to separate geometric aspects of the manifold from the probability distribution on it.
This is of interest to harmonic analysts, and it underlies extension of the Diffusion Maps beyond the Laplacian Eigenmaps.
\end{itemize}

As for the first point, if we have a compact Riemannian manifold $\mathcal{M}$ and a probability distribution $P:\mathcal{M} \rightarrow \mathbb{R}^{+}$ according to which points are drawn in an i.i.d. fashion.
Assume that $a \le P(x) \le b$, for all $x \in \mathcal{M}$.
Then, define the point cloud Laplacian operator as 
\[\underbar{L}_n^t f(x) = \frac{1}{n} \sum_{i=1}^{n} W(x_i,x_j) \left( f(x) - f(x_i) \right)
\]
If $W(x,x_i) = e^{\frac{\|x-x_i\|}{4t}}$, then this corresponds to the operator we described above.
In order to normalized the weights, let 
\[
W(x,x_i) = \frac{1}{t} \frac{G_t(x,x_i)}{\sqrt{\hat{d}_t(x)}\sqrt{\hat{d}_t(x_i)}}  ,
\]
where 
\begin{eqnarray*}
G_t(x,x_i) &=& \frac{1}{\left( 4 \pi t\right)^{k/2}  }e^{-\frac{\|x-x_i\|^2}{4t} }   ,    \\
\hat{d}_t(x) &=& \frac{1}{n} \sum_{j \ne i} G_t(x,x_j)  ,  \quad\mbox{and} \\
\hat{d}_t(x_i) &=& \frac{1}{n-1} \sum_{j \ne i} G_t(x_i,x_j)  ,
\end{eqnarray*}
where the latter two quantities are empirical estimates of the degree function $d_t(x)$, where 
\[
d_t(x) = \int_{\mathcal{M}} G_t (x,y) P(y) \mbox{Vol}(y)  .
\]
Note that we get a degree function---which is a continuous function defined on $\mathcal{M}$.
This function bears some resemblance to the diagonal degree matrix of a graph, and it can be thought of as a multiplication operator, but it has very different properties than an integral operator like the heat kernel.
We will see this same function next time, and this will be important for when we get consistency with normalized versus unnormalized spectral clustering.

%% file: lect23.tex
\section{%
(04/16/2015):
Some Statistical Inference Issues (2 of 3):
Convergence and consistency questions}

Reading for today.
\begin{compactitem}
\item
``Consistency of spectral clustering,'' in Annals of Statistics, by von Luxburg, Belkin, and Bousquet
\end{compactitem}

Last time, we talked about whether the Laplacian constructed from point clouds converged to the Laplace-Beltrami operator on the manifold from which the data were drawn, under the assumption that the unseen hypothesized data points are drawn from a probability distribution that is supported on a low-dimensional Riemannian manifold.
While potentially interesting, that result is a little unsatisfactory for a number of reasons, basically since one typically does not test the hypothesis that the underlying manifold even exists, and since the result doesn't imply anything statistical about cluster quality or prediction quality or some other inferential goal.
For example, if one is going to use the Laplacian for spectral clustering, then probably a more interesting question is to ask whether the actual clusters that are identified make any sense, e.g., do they converge, are they consistent, etc.
So, let's consider these questions.
Today and next time, we will do this in two different ways.
\begin{itemize}
\item
Today, we will address the question of the consistency of spectral clustering when there are data points drawn from some space $\mathcal{X}$ and we have similarity/dissimilarity information about the points.
We will follow the paper ``Consistency of spectral clustering,'' by von Luxburg, Belkin, and Bousquet.
\item
Next time, we will ask similar questions but for a slightly different data model, i.e., when the data are from very simple random graph models.
As we will see, some of the issues will be similar to what we discuss today, but some of the issues will be different.
\end{itemize}

I'll start today with some general discussion on: algorithmic versus statistical approaches; similarity and dissimilarity functions; and embedding data in Hilbert versus Banach spaces.
Although I covered this in class briefly, for completeness I'll go into more detail here.

\subsection{Some general discussion on algorithmic versus statistical approaches}

When discussing statistical issues, we need to say something about our model of the data generation mechanism, and we will discuss one such model here.
This is quite different than the algorithmic perspective, and there are a few points that would be helpful to clarify.

To do so, let's take a step back and ask: how are the data or training points generated?
Here are two possible answers.
\begin{itemize}
\item
\textbf{Deterministic setting.}
Here, someone just provides us with a fixed set of objects (consisting, e.g, of a set of vectors or a single graph) and we have to work with this particular set of data.
This setting is more like the algorithmic approach we have been adopting when we prove worst-case bounds.
\item
\textbf{Probabilistic setting.}
Here, we can consider the objects as a random sample generated from some unknown probability distribution $P$.
For example, this $P$ could be on (Euclidean or Hilbert or Banach or some other) space $\mathcal{X}$.
Alternatively, this $P$ could be over random graphs or stochastic blockmodels.
\end{itemize}

There are many differences between these two approaches.
One is the question of what counts as ``full knowledge.''  
A related question has to do with the objective that is of interest.
\begin{itemize}
\item
In the deterministic setting, the data at hand count as full knowledge, since they are all there is.
Thus, when one runs computations, one wants to make statements about the data at hand, e.g., how close in quality is the output of an approximation algorithm to the output of a more expensive exact computation.
\item
In the probabilistic setting, complete or full knowledge is to know $P$ exactly, and the finite sample contains only noisy information about $P$.
Thus, when we run computations, we are only secondarily interested in the data at hand, since we are more interested in $P$, or relatedly in what we can say if we draw another noisy sample from $P$ tomorrow.
\end{itemize}
Sometimes, people think of the deterministic setting as the probabilistic setting, in which the data space equals the sample space and when one has sampled all the data.
Sometimes this perspective is useful, and sometimes it is not.

In either setting, one simple problem of potential interest (that we have been discussing) is clustering: given a training data $(x_i)_{i=1,\ldots,n}$, where $x_i$ correspond to some features/patterns but for which there are no labels available, the goal is to find some sort of meaningful clusters.
Another problem of potential interest is classification: given training points $(x_i,y_i)_{i=1,\ldots,n}$, where $x_i$ correspond to some features/patterns and $y_i$ correspond to labels, the goal is to infer a rule to assign a correct $y$ to a new $x$.
It is often said that, in some sense, in the supervised case, \emph{what} we want to achieve is well-understood, and we just need to specify \emph{how} to achieve it; while in the latter case both \emph{what} we want to achieve as well as \emph{how} we want to achieve it is not well-specified.
This is a popular view from statistics and ML; and, while it has some truth to it, it hides several things.
\begin{itemize}
\item
In both cases, one specifies---implicitly or explicitly---an objective and tries to optimize it. 
In particular, while the vague idea that we want to predict labels is reasonable, one obtains very different objectives, and thus very different algorithmic and statistical properties, depending on how sensitive one is to, e.g., false positives versus false negatives. 
Deciding on the precise form of this can be as much of an art as deciding on an unsupervised clustering objective.
 \item
The objective to be optimized could depend on just the data at hand, or it could depend on some unseen hypothesized data (i.e., drawn from $P$).
In the supervised case, that might be obvious; but even in the unsupervised case, one typically is not interested in the output per se, but instead in using it for some downstream task (that is often not specified).
 \end{itemize}
All that being said, it is clearly easier to validate the supervised case.
But we have also seen that the computations in the supervised case often boil down to computations that are identical to computations that arise in the unsupervised case.
For example, in both cases locally-biased spectral ranking methods arise, but they arise for somewhat different reasons, and thus they are used in somewhat different ways.

From the probabilistic perspective, due to randomness in the generation of the training set, it is common to study ML algorithms from this statistical or probabilistic point of view and to model the data as coming from a probability space.
For example, in the supervised case, the unseen data are often modeled by a probability space of the form
\[
\left(\left( \mathcal{X}\times\mathcal{Y} \right), \sigma\left( \mathcal{B}_{\mathcal{X}}\times\mathcal{B}_{\mathcal{Y}}  \right), P \right)
\]
 where $\mathcal{X}$ is the feature/pattern space and $\mathcal{Y}$ is the label space, $\mathcal{B}_{\mathcal{X}}$ and $\mathcal{B}_{\mathcal{Y}}$ are $\sigma$-algebras on $\mathcal{X}$ and $\mathcal{Y}$, and $P$ is a joint probability distribution on patterns and labels.
(Don't worry about the $\sigma$-algebra and measure theoretic issues if you aren't familiar with them, but note that $P$ is the main object of interest, and this is what we were talking about last time with labeled versus unlabeled data.)
The typical assumption in this case is that $P$ is unknown, but that one can sample $\mathcal{X}\times\mathcal{Y}$ from $P$.
On the other hand, in the unsupervised case, there is no $\mathcal{Y}$, and so in that case the unseen data are more often modeled by a probability space of the form
\[
\left( \mathcal{X}, \mathcal{B}_{\mathcal{X}}, P \right)  ,
\]
in which case the data training points $(x_i)_{i=1,\ldots,n}$ are drawn from $P$.

From the probabilistic perspective, one is less interested in the objective function quality on the data at hand, and instead one is often interested in finite-sample performance issues and/or asymptotic convergence issues.
For example, here are some questions of interest.
\begin{itemize}
\item
Does the classification constructed by a given algorithm on a finite sample converge to a limit classifier at $n \rightarrow \infty$?
\item
If it converges, is the limit classifier the best possible; and if not, how suboptimal is it?
\item
How fast does convergence take place, as a function of increasing $n$?
\item
Can we estimate the difference between finite sample classifier and the optimal classifier, given only the sample?
\end{itemize}

Today, we will look at the convergence of spectral clustering from this probabilistic perspective.
But first, let's go into a little more detail about similarities and dissimilarities.

\subsection{Some general discussion on similarities and dissimilarities}

When applying all sorts of algorithms, and spectral algorithms in particular, MLers work with some notion either of similarity or dissimilarity.
For example, spectral clustering uses an adjacency matrix, which is a sort of similarity function.
Informally, a dissimilarity function is a notion that is somewhat like a distance measure; and a similarity/affinity function measures similarities and is sometimes thought about as a kernel matrix.
Some of those intuitions map to what we have been discussing, e.g., metrics and metric spaces, but in some cases there are differences.

Let's start first with dissimilarity/distance functions.
In ML, people are often a little less precise than say in TCS; and---as used in ML---dissimilarity functions satisfy some or most or all of the following, but typically at least the first two.
\begin{itemize}
\item
(D1) $d(x,x)=0$
\item
(D2) $d(x,y) \ge 0$
\item
(D3) $d(x,y) = d(y,x)$
\item
(D4) $d(x,y)=0 \Rightarrow x=y$
\item
(D5) $d(x,y) + d(y,z) \ge d(x,z)$
\end{itemize}

Here are some things to note about dissimilarity and metric functions.
\begin{itemize}
\item
Being more precise, a \emph{metric} satisfies all of these conditions; and a \emph{semi-metric} satisfies all of these except for (D4).
\item
MLers are often interested in dissimilarity measures that do \emph{not} satisfy (D3), e.g., the Kullback-Leibler ``distance.''
\item
There is also interest  in cases where (D4) is not satisfied.
In particular, the so-called cut metric---which we used for flow-based graph partitioning---was a semi-metric.
\item
Condition (D4) says that if different points have distance equal to zero, then this implies that they are really the same point.
Clearly, if this is not satisfied, then one should expect an algorithm should have difficulty discriminating points (in clustering, classification, etc. problems) which have distance zero.
\end{itemize}

Here are some commonly used methods to transform non-metric dissimilarity functions into proper metric functions.
\begin{itemize}
\item
If $d$ is a distance function and $x_0\in\mathcal{X}$ is arbitrary, then $\tilde{d}(x,y) = | d(x,x_0) - d(y,x_0) |$ is a semi-metric on $\mathcal{X}$.
\item
If $\left(\mathcal{X},d\right)$ is a finite dissimilarity space with $d$ symmetric and definite, then 
\[
\tilde{d} = \left\{ \begin{array}{l l}
                       d(x,y)+c & \quad \text{if $x \ne y$} \\
                       0        & \quad \text{if $x = y$}
                    \end{array}
            \right.   ,
\]
with $c \ge \max_{p,q,r \in \mathcal{X}} | d(p,q)+d(p,r)+d(r,q) | $, is a metric.
\item
If $D$ is a dissimilarity matrix, then there exists constants $h$ and $k$ such that the matrix with elements $\tilde{d}_{ij} = \left( d_{ij}^{2} + h \right)^{1/2}$, for $i \ne j$, and also $\bar{d}_{ij} = d_{ij} + k$, for $i \ne j$, are Euclidean.
\item
If $d$ is a metric, so are $d+c$, $d^{1/r}$, $\frac{d}{d+c}$, for $c \ge 0$ and $r \ge 1$.
If $w: \mathbb{R}\rightarrow\mathbb{R}$ is monotonically increasing function s.t. $w(x)=0 \iff x=0$ and $w(x+y) \le w(x) + w(y)$; then if $d(\cdot,\cdot)$ is a metric, then $w(d(\cdot,\cdot))$ is a metric.
\end{itemize}

Next, let's go to similarity functions.
As used in ML, similarity functions satisfy some subset of the following.
\begin{itemize}
\item
(S1) $s(x,x) > 0$
\item
(S2) $s(x,y)=s(y,x)$
\item
(S3) $s(x,y) \ge 0$
\item
(S4) $\sum_{ij=1}^{n} c_ic_j s(x_i,x_j) \ge 0$, for all $n \in \mathbb{N},  c_i\in\mathbb{R},   x_i\in\mathcal{X}$ PSD.
\end{itemize}

Here are things to note about these similarity functions.
\begin{itemize}
\item
The non-negativity is actually \emph{not} satisfied by two examples of similarity functions that are commonly used: correlation coefficients and scalar products
\item
One can transform a bounded similarity function to a nonnegative similarity function by adding an offset: $s(x,y) = s(x,y) +c$ for come $c$.
\item
If $S$ is PSD, then it is a kernel.
This is a rather strong requirement that is mainly satisfied by scalar products in Hilbert spaces.
\end{itemize}

It is common to transform \emph{similarities to dissimilarities}.  Here are two ways to do that.
\begin{itemize}
\item
If the similarity is a scalar product in a Euclidean space (i.e., PD), then one can compute the metric
\[
d(x,y)^2 = \left<x-y,x-y\right> = \left<x,x\right> - 2\left(x,y\right> + \left<y,y\right>  .
\]
\item
If the similarity function is normalized, i.e., $0 \le s(x,y) \le 1$, and $s(x,x)=1$, for all $x,y$, then $d=1-s$ is a distance.
\end{itemize}
It is also common to transform \emph{dissimilarities to similarities}. Here are two ways to do that.
\begin{itemize}
\item
If the distance is Euclidean, then one can compute a PD similarity 
\[
s(x,y) = \frac{1}{2}\left( d(x,0)^2 + d(y,0)^2 - d(x,y)^2 \right)  ,
\] 
where $0 \in \mathcal{X}$ is an arbitrary origin.
\item
If $d$ is a dissimilarity, then a nonnegative decreasing function of $d$ is a similarity, e.g., 
$s(x,y) = \exp\left( -d(x,y)^2/t \right)$, for $t\in\mathbb{R}$, and also $s(x,y) = \frac{1}{1-d(x,y)}$.
\end{itemize}

These and related transformations are often used at the data preprocessing step, often in a somewhat ad hoc manner.
Note, though, that the use of any one of them implies something about what one thinks the data ``looks like'' as well as about how algorithms will perform on the data.

\subsection{Some general discussion on embedding data in Hilbert and Banach spaces}

Here, we discuss embedding data (in the form of similarity or dissimilarity functions) into Hilbert and Banach spaces.
To do so, we start with an informal definition (informal since the precise notion of dissimilarity is a little vague, as discussed above).
\begin{definition}
A space $\left(\mathcal{X},d\right)$ is a dissimilarity space or a metric space, depending on whether $d$ is a dissimilarity function or a metric function.
\end{definition}

An important question for distance/metric functions, i.e., real metrics that satisfy the above conditions, is the following: when can a given metric space $\left(\mathcal{X},d\right)$ be embedded \emph{isometrically} in Euclidean space $\mathcal{H}$ (or, slightly more generally, Hilbert space $\mathcal{H}$).
That is, the goal is to find a mapping $\phi : \mathcal{X} \rightarrow \mathcal{H}$ such that $d(x,y) = \|\phi(x)-\phi(y)\|$, for all $x,y\in\mathcal{X}$.
(While this was something we relaxed before, e.g., when we looked at flow-based algorithms and looked at relaxations where there were distortions but they were not too too large, e.g., $O(\log n)$, asking for isometric embeddings is more common in functional analysis.)
To answer this question, note that distance in Euclidean vector space satisfies (D1)--(D5), and so a necessary condition for the above is the (D1)--(D5) be satisfied.
The well-known Schoenberg theorem characterizes which metric spaces can be isometrically embedded in Hilbert~space.
\begin{theorem}
A metric space $\left(\mathcal{X},d\right)$ can be embedded isometrically into Hilbert space iff $-d^2$ is conditionally positive definite, i.e., iff
\[
-\sum_{ij=1}^{\ell} c_ic_j d^2(x_i,x_j) \ge 0
\]
for all $\ell\in\mathbb{N},  x_i,x_j\in\mathcal{X},  c_i,c_j\in\mathbb{R}$, with $\sum_ic_i = 0$.
\end{theorem}

Informally, this says that Euclidean spaces and Hilbert spaces are not ``big enough'' for arbitrary metric spaces.
(We saw this before when we showed that constant degree expanders do not embed well in Euclidean spaces.)
More generally, though, isometric embeddings into certain Banach spaces can be achieved for arbitrary metric spaces.
(More on this later.)
For completeness, we have the following definition.

\begin{definition}
Let $X$ be a vector space over $\mathcal{C}$.
Then $X$ is a \emph{normed linear space} if for all $f\in X$, there exists a number, $\|f\|\in\mathbb{R}$, called the norm of $f$ s.t.:
(1) $\|f\| \ge 0$; 
(2) $\|f\| = 0 \mbox{ iff } f=0$;
(3) $\|cf\| = |c|\|f\|$, for all scalar $c$;
(4) $\|f+g\| \le \|f\|+\|g\|$.
A \emph{Banach space} is a complete normed linear space.
A \emph{Hilbert space} is a Banach space, whose norm is determined by an inner product.
\end{definition}
This is a large area, most of which is off topic for us.
If you are not familiar with it, just note that RKHSs are particularly nice Hilbert spaces that are sufficiently heavily regularized that the nice properties of $\mathbb{R}^{n}$, for $n<\infty$, still hold; general infinite-dimensional Hilbert spaces are more general and less well-behaved; and general Banach spaces are even more general and less well-behaved.
Since it is determined by an inner product, the norm for a Hilbert space is essentially an $\ell_2$ norm; and so, if you are familiar with the $\ell_1$ or $\ell_{\infty}$ norms and how they differ from the $\ell_2$ norm, then that might help provide very rough intuition on how Banach spaces can be more general than Hilbert~spaces.

\subsection{Overview of consistency of normalized and unnormalized Laplacian spectral methods}

Today, we will look at the convergence of spectral clustering from this probabilistic perspective.
Following the von Luxburg, Belkin, and Bousquet paper, we will address the following two questions.
\begin{itemize}
\item
Q1: Does spectral clustering converge to some limit clustering if more and more data points are sampled and as $n\rightarrow\infty$?
\item
Q2: If it does converge, then is the limit clustering a useful partition of the input space from which the data are drawn?
\end{itemize}

One reason for focusing on these questions is that it can be quite difficult to determine what is a cluster and what is a good cluster, and so as a more modest goal one can ask for ``consistency,'' i.e., that the clustering constructed on a finite sample drawn from some distribution converges to a fixed limit clustering of the whole data space when $n \rightarrow \infty$.
Clearly, this notion is particularly relevant in the probabilistic setting, since then we obtain a partitioning of the underlying space $\mathcal{X}$ from which the data are drawn.

Informally, this will provide an ``explanation'' for why spectral clustering works.
Importantly, though,  this consistency ``explanation'' will be very different than the ``explanations'' that have been offered in the deterministic or algorithmic setting, where the data at hand represent full knowledge.
In particular, when just viewing the data at hand, we have provided the following informal explanation of why spectral clustering works.
\begin{itemize}
\item
Spectral clustering works since it wants to find clusters s.t. the probability of random walks staying within a cluster is higher and the probability of going to the complement is smaller.
\item
Spectral clustering works since it approximates via Cheeger's Inequality the intractable expansion/conductance objective.
\end{itemize}
In both of those cases, we are providing an explanation in terms of the data at hand; i.e., while we might have an underlying space $\mathcal{X}$ in the back of our mind, they are statements about the data at hand, or actually the graph constructed from the data at hand.

The answer to the above two questions (Q1 and Q2) will be basically the following.
\begin{itemize}
\item
Spectral clustering with the normalized Laplacian is consistent under very general conditions.
For the normalized Laplacian, when it can be applied, then the corresponding clustering does converge to a limit.
\item
Spectral clustering with the non-normalized Laplacian is not consistent, except under very specific conditions.
These conditions have to do with, e.g., variability in the degree distribution, and these conditions often do \emph{not} hold in practice.
\item
In either case, if the method converges, then the limit does have intuitively appealing properties and splits the space $\mathcal{X}$ up into two pieces that are reasonable; but for the non-normalized Laplacian one will obtain a trivial limit if the strong conditions are not satisfied.
\end{itemize}

As with last class, we won't go through all the details, and instead the goal will be to show some of the issues that arise and tools that are used if one wants to establish statistical results in this area; and also to show you how things can ``break down'' in non-ideal situations.

To talk about convergence/consistency of spectral clustering, we need to make statements about eigenvectors, and for this we need to use the spectral theory of bounded linear operators, i.e., methods from functional analysis.
In particular, the information we will need will be somewhat different than what we needed in the last class when we talked about the convergence of the Laplacian to the hypothesized Laplace-Beltrami operator, but there will be some similarities.
Today, we are going to view the data points as coming from some Hilbert or Banach space, call in $\mathcal{X}$, and from these data points we will construct an empirical Laplacian.
(Next time, we will consider graphs that are directly constructed via random graph processes and stochastic block models.)
The main step today will be to establish the convergence of the eigenvalues and eigenvectors of random graph Laplacian matrices for growing sample sizes.
This boils down to questions of convergence of random Laplacian matrices constructed from sample point sets.

(Note that although there has been a lot of work in random matrix theory on the convergence of random matrices with i.i.d. entries or random matrices with fixed sample size, e.g., covariance matrices, this work isn't directly relevant here, basically since the random Laplacian matrix grows with the sample size $n$ and since the entries of the random Laplacian matrix are not independent.  Thus, more direct proof methods need to be used here.)

Assume we have a data space $\mathcal{X} = \{ x_1. \ldots, x_n \}$ and a pairwise similarity $k: \mathcal{X}\times\mathcal{X} \rightarrow \mathbb{R}$, which is usually symmetric and nonnegative.
For any fixed data set of $n$ points, define the following:
\begin{itemize}
\item
the Laplacian $L_n = D_n-K_n$, 
\item
the normalized Laplacian $L_n^{\prime} = D_n^{-1/2}L_nD_n^{-1/2}$, and 
\item
the random walk Laplacian $L_n^{\prime\prime} = D_n^{-1}L_n$.
\end{itemize}
(Although it is different than what we used before, the notation of the von Luxburg, Belkin, and Bousquet paper is what we will use here.)
Note that here we assume that $d_i > 0$, for all $i$.
We are interested in computing the leading eigenvector or several of the leading eigenvectors of one of these matrices and then clustering with them.

To see the kind of convergence result one could hope for, consider the second eigenvector $\left(v_1,\ldots,v_n\right)^T$ of $L_n$, and let's interpret is as a function $f_n$ on the discrete space $\mathcal{X}_n = \{ X_1,\ldots,X_n \}$ by defining the function $f_n(X_i)=v_i$.
(This is the view we have been adopting all along.)
Then, we can perform clustering by performing a sweep cut, or we can cluster based on whether the value of $f_n$ is above or below a certain threshold.
Then, in the limit $n\rightarrow\infty$, we would like $f_n\rightarrow f$, where $f$ is a function on the entire space $\mathcal{X}$, such that we can threshold $f$ to partition $\mathcal{X}$.

To do this, we can do the following.
\begin{compactenum}
\item
Choose this space to be $C\left(\mathcal{X}\right)$, the space of continuous functions of $\mathcal{X}$.
\item
Construct a function $d \in C\left(\mathcal{X}\right)$, a degree function, that is the ``limit'' as $n\rightarrow\infty$ of the discrete degree vector $\left(d_1,\ldots,d_n\right)$.
\item
Construct linear operators $U$, $U^{\prime}$, and $U^{\prime\prime}$ on $C\left(\mathcal{X}\right)$ that are the limits of the discrete operators $L_n$, $L_n^{\prime}$, and $L_n^{\prime\prime}$.
\item
Prove that certain eigenfunctions of the discrete operates ``converge'' to the eigenfunctions of the limit operators.
\item
Use the eigenfunctions of the limit operator to construct  a partition for the entire space $\mathcal{X}$.
\end{compactenum}
We won't get into details about the convergence properties here, but below we will highlight a few interesting aspects of the limiting process.
The main result they show is that in the case of normalized spectral clustering, the limit behaves well, and things converge to a sensible partition of the entire space; while in the case of unnormalized spectral clustering, the convergence properties are much worse (for reasons that are interesting that we will describe).

\subsection{Details of consistency of normalized and unnormalized Laplacian spectral methods}

Here is an overview of the two main results in more detail.

\textbf{Result 1.}
(Convergence of normalized spectral clustering.)
Under mild assumptions, \emph{if the first $r$ eigenvalues of the limit operator $U^{\prime}$ satisfy $\lambda_i \ne 1$ and have multiplicity one}, then 
\begin{compactitem}
\item
the same hold for the first $r$ eigenvalues of $L_n^{\prime}$, as $n\rightarrow\infty$;
\item 
the first $r$ eigenvalues of $L_n^{\prime}$ converge to the first $r$ eigenvalues of $U^{\prime}$; 
\item
the corresponding eigenvectors converge; and 
\item
the clusters found from the first $r$ eigenvectors on finite samples converge to a limit clustering of the entire data space.
\end{compactitem}

\textbf{Result 2.}
(Convergence of unnormalized spectral clustering.)
Under mild assumptions, \emph{if the first $r$ eigenvalues of the limit operator $U$ do not lie in the range of the degree function $d$ and have multiplicity one}, then 
\begin{compactitem}
\item
the same hold for the first $r$ eigenvalues of $L_n$, as $n\rightarrow\infty$;
\item 
the first $r$ eigenvalues of $L_n$ converge to the first $r$ eigenvalues of $U$; 
\item
the corresponding eigenvectors converge; and 
\item
the clusters found from the first $r$ eigenvectors on finite samples converge to a limit clustering of the entire data space.
\end{compactitem}

Although both of these results have a similar structure (``if the inputs are nice, then one obtains good clusters''), the ``niceness'' assumptions are very different: for normalized spectral clustering, it is the rather innocuous assumption that $\lambda_i \ne 1$, while for unnormalized spectral clustering it is the much stronger assumption that $\lambda_i \in \mbox{range}(d)$.
This assumption is necessary, as it is needed to ensure that the eigenvalue $\lambda_i$ is isolated in the spectrum of the limit operator.
This is a requirement to be able to apply perturbation theory to the convergence of eigenvectors.
In particular, here is another result.

\textbf{Result 3.}
(The condition $\lambda \notin \mbox{range}(d)$ is necessary.)
\begin{compactitem}
\item
There exist similarity functions such that there exist no nonzero eigenvectors outside of $\mbox{range}(d)$.
\item
In this case, the sequence of second eigenvalues of $\frac{1}{n}L_n$ converge to $\min d(x)$, and the corresponding eigenvectors do \emph{not} yield a sensible clustering of the entire data space.
\item
For a wide class of similarity functions, there exist only finitely many eigenvalues $r_0$ outside of $\mbox{range}(d)$, and the same problems arise if one clusters with $r > r_0$ eigenfunctions.
\item
The condition $\lambda \notin \mbox{range}(d)$ refers to the limit and cannot be verified on a finite sample.
\end{compactitem}
That is, unnormalized spectral clustering can fail completely, and one cannot detect it with a finite~sample.

The reason for the difference between the first results is the following.
\begin{itemize}
\item
In the case of normalized spectral clustering, the limit operator $U^{\prime}$ has the form $U^{\prime}=I-T$, where $T$ is a compact linear operator.
Thus, the spectrum of $U^{\prime}$ is well-behaved, and all the eigenvalues $\lambda \ne 1$ are isolated and have finite multiplicity.
\item
In the case of unnormalized spectral clustering, the limit operator $U$ has the form $U=M-S$, where $M$ is a multiplication operator, and $S$ is a compact integral operator.
Thus, the spectrum of $U$ is not as nice as that of $U^{\prime}$, since it contains the interval $\mbox{range}(d)$, and the eigenvalues will be isolated only if $\lambda_i \ne \mbox{range}(d)$.
\end{itemize}

Let's get into more detail about how these differences arise.
To do so, let's make the following assumptions about the data.
\begin{itemize}
\item
The data space $\mathcal{X}$ is a compact metric space, $\mathcal{B}$ is the Borel $\sigma$-algebra on $\mathcal{X}$, and $P$ is a probability measure on $\left(\mathcal{X},\mathcal{B}\right)$.
We draw a sample of points $\left(X_i\right)_{i\in\mathbb{N}}$ i.i.d. from $P$.
The similarity function $k: X \times X \rightarrow \mathcal{R}$ is symmetric, continuous, and there exists an $\ell > 0$ such that $k(x,y) > \ell$, for all $x,y\in\mathcal{X}$.
(The assumption that $f$ is bounded away from $0$ is needed due to the division in the normalized Laplacian.)
\end{itemize}

For $f:\mathcal{X}\rightarrow\mathbb{R}$, we can denote the range of $f$ by $\mbox{range}(f)$.
Then, if $\mathcal{X}$ is connected and $f$ is continuous then $\mbox{range}(f) = \left[ \inf_x f(x), \sup_x f(x) \right]$.
Then we can define the following.
\begin{definition}
The \emph{restriction operator} $\rho_n : C\left(\mathcal{X}\right)\rightarrow\mathbb{R}^{n}$ denotes the random operator which maps a function to its values on the first $n$ data points, i.e., 
\[
\rho_n(f) = \left( f(X_1),\ldots,f(X_n) \right)^T  .
\]
\end{definition}

Here are some facts from spectral and perturbation theory of linear operators that are needed.

Let $E$ be a real-valued Banach space, and let $T: E \rightarrow E$ be a bounded linear operator.
Then, an \emph{eigenvalue} of $T$ is defined to be a real or complex number $\lambda$ such that 
\[
Tf=\lambda f, \mbox{ for some } f \in E  .
\]
Note that $\lambda$ is an eigenvalue of $T$ iff the operator $T-\lambda$ has a nontrivial kernel (recall that if $L:V \rightarrow W$ then $\mbox{ker}(L) = \{ v \in V : L(v) = 0 \}$) or equivalently if $T-\lambda$ is \emph{not} injective (recall that $f:A \rightarrow B$ is injective iff $\forall a,b \in A$ we have that $f(a)=f(b) \Rightarrow a=b$, i.e., different elements of the domain do not get mapped to the same element).
Then, the \emph{resolvent} of $T$ is defined to be 
\[
\rho(T) = \{ \lambda \in \mathbb{R} : \left(\lambda-T\right)^{-1} \mbox{ exists and is bounded} \}  ,
\]
and the \emph{spectrum} of $T$ id defined to be 
\[
\sigma(T) = \mathbb{R} \setminus \rho(T)  .
\]
This holds very generally, and it is the way the spectrum is generalized in functional analysis.

(Note that if $E$ is finite dimensional, then every non-invertible operator is not injective; and so $\lambda \in \sigma(T) \Rightarrow \lambda \mbox{ is an eigenvalue of } T$.
If $E$ is infinite dimensional, this can fail; basically, one can have operators that are injective but that have no bounded inverse, in which case the spectrum can contain more than just eigenvalues.)

We can say that a point $\sigma_{iso} \subset \sigma(T)$ is \emph{isolated} if there exists an open neighborhood $\xi \subset \mathbb{C}$ of $\sigma_{iso}$ such that $\sigma(T) \cap \xi = \left( \sigma_{iso} \right)$.
If the spectrum $\sigma(T)$ of a bounded operator $T$ in a Banach space $E$ consists of isolated parts, then for each isolated part of the spectrum, a \emph{spectral projection} $P_{iso}$ can be \emph{defined} operationally as a path integral over the complex plane of a path $\Gamma$ that encloses $\sigma_{iso}$ and that separates it from the rest of $\sigma(T)$, i.e., for $\sigma_{iso}\in\sigma(T)$, the corresponding spectral projection is 
\[
P_{iso} = \frac{1}{2 \pi i} \int_{\Gamma} \left( T- \lambda I \right)^{-1} d\lambda , 
\]
where $\Gamma$ is a closed Jordan curve in the complex plane separating $\sigma_{iso}$ from the rest of the spectrum.
If $\lambda$ is an isolated eigenvalue of $\sigma(T)$, then the dimension of the range of the spectral projection $P_{\lambda}$ is defined to be the \emph{algebraic multiplicity} of $\lambda$, (for a finite dimensional Banach space, this is the multiplicity of the root $\lambda$ of the characteristic polynomial, as we saw before), and the \emph{geometric multiplicity} is the dimension of the eigenspace of $\lambda$.

One can split up the spectrum into two parts: the \emph{discrete spectrum} $\sigma_d$(T) is the part of $\sigma(T)$ that consists of isolated eigenvalues of $T$ with finite algebraic multiplicity; and the \emph{essential spectrum} is $\sigma_{ess}(T) = \sigma(T) \setminus \sigma_d(T) $.
It is a fact that the essential spectrum cannot be changed by a finite-dimensional or compact perturbation of an operator, i.e., for a bounded operator $T$ and a compact operator $V$, it holds that $\sigma_{ess}(T+V) = \sigma_{ess}(T)$.
The important point here is that one can define spectral projections only for isolated parts of the spectrum of an operator and that these isolated parts of the spectrum are the only parts to which perturbation theory can be applied.

Given this, one has perturbation results for compact operators.
We aren't going to state these precisely, but the following is an informal statement.
\begin{itemize}
\item
Let $\left( E , \|\cdot\|_E \right)$ be a Banach space, and $\left(T_n\right)_n$ and $T$ bounded linear operators on $E$ with $T_n\rightarrow T$.
Let $\lambda \in \sigma(T)$ be an isolated eigenvalue with finite multiplicity $m$, and let $\xi \subset \mathbb{C}$ be an open neighborhood of $\lambda$ such that $\sigma(T) \cap \xi = \{\lambda\}$.
Then, 
\begin{compactitem}
\item
eigenvalues converge, 
\item
spectral projections converge, and
\item 
if $\lambda$ is a simple eigenvalue, then the corresponding eigenvector converges.
\end{compactitem}
\end{itemize}

We aren't going to go through the details of their convergence argument, but we will discuss the following issues.

The technical difficulty with proving convergence of normalized/unnormalized spectral clustering, e.g., the convergence of $\left(v_n\right)_{n \in \mathbb{N}}$ or of $\left(L_n^{\prime}\right)_{n\in\mathbb{N}}$, is that for different sample sized $n$, the vectors $v_n$ have different lengths and the matrices $L_n^{\prime}$ have different dimensions, and so they ``live'' in different spaces for different values of $n$.
For this reason, one can't apply the usual notions of convergence.
Instead, one must show that there exists functions $f\in C\left(\mathcal{X}\right)$ such that $\| v_n - \rho_n f \| \rightarrow 0$, i.e., such that the eigenvector $v_n$ and the restriction of $f$ to the sample converge.
Relatedly, one relates the Laplacians to some other operator such that they are all defined on the same space.
In particular, one can define a sequence $(U_n)$ of operators that are related to the matrices $(L_n)$; but each operator $(U_n)$ is defined on the space $C(\mathcal{X})$ of continuous functions on $\mathcal{X}$, independent of $n$.

All this involves constructing various functions and operators on $C\left(\mathcal{X}\right)$.
There are basically two \emph{types} of operators, integral operators and multiplication operators, and they will enter in somewhat different ways (that will be responsible for the difference in the convergence properties between normalized and unnormalized spectral clustering).
So, here are some basic facts about integral operators and multiplication operators.

\begin{definition}
Let $\left(\mathcal{X},\mathcal{B},\mu\right)$ be a probability space, and let $k \in L_2\left( \mathcal{X}\times\mathcal{X},\mathcal{B}\times\mathcal{B},\mu\times\mu \right)$.
Then, the function $S: L_2\left(\mathcal{X},\mathcal{B},\mu\right) \rightarrow L_2\left(\mathcal{X},\mathcal{B},\mu\right)$ defined as 
\[
Sf(x) : \int_{\mathcal{X}} k(x,y) f(y) d\mu(y)
\]
is an \emph{integral operator} with kernel $k$.
\end{definition}
If $\mathcal{X}$ is compact and $k$ is continuous, then (among other things) the integral operator $S$ is compact.

\begin{definition}
Let $\left(\mathcal{X},\mathcal{B},\mu\right)$ be a probability space, and let $d \in L_{\infty}\left(\mathcal{X},\mathcal{B},\mu\right)$.
Then a \emph{multiplication operator} $M_d : L_2\left(\mathcal{X},\mathcal{B},\mu\right) \rightarrow L_2\left(\mathcal{X},\mathcal{B},\mu\right)$ is 
\[
M_d f = fd .
\]
\end{definition}
This is a bounded linear operator; but if $d$ is non-constant, then the operator $M_d$ is \emph{not} compact.

Given the above two different types of operators, let's introduce specific operators on $C\left(\mathcal{X}\right)$ corresponding to matrices we are interested in.
(In general, we will proceed by identifying vectors $\left( v_1,\ldots,v_n \right)^{T} \in \mathbb{R}^{n}$ with functions $f \in C\left(\mathcal{X}\right)$ such that $f(v_i) = v_i$ and extending linear operators on $\mathbb{R}^{n}$ to deal with such functions rather than vectors.)
Start with the unnormalized Laplacian: $L_n = D_n - K_n$, where $D = \mbox{diag}(d_i)$, where $d_i = \sum_{ij} K(x_i,x_j)$.

We want to relate the degree vector $\left( d_1,\ldots,d_n \right)^T$ to a function on $C\left(\mathcal{X}\right)$.
To do so, define the true and empirical degree functions:
\begin{eqnarray*}
d(x)   &=& \int k(x,y)   dP(y) \in C\left(\mathcal{X}\right)  \\
d_n(x) &=& \int k(x,y) dP_n(y) \in C\left(\mathcal{X}\right)  
\end{eqnarray*}
(Note that $d_n \rightarrow d$ as $n\rightarrow\infty$ by a LLN.)
By definition, $d_n(x_i) = \frac{1}{n}d_i$, and so the empirical degree function agrees with the degrees of the points $X_i$, up to the scaling $\frac{1}{n}$.

Next, we want to find an operator acting on $C\left(\mathcal{X}\right) $ that behaves similarly to the matrix $D_n$ on $\mathbb{R}^{n}$.
Applying $D_n$ to a vector $f = \left(f_1,\ldots,f_n\right)^{T} \in \mathbb{R}^{n}$ gives $\left(D_nf\right)_i = d_if_i$, i.e., each element is multiplied by $d_i$.
So, in particular, we can interpret $\frac{1}{n}D_n$ as a multiplication operator.
Thus, we can define the true and empirical multiplication operators:
\begin{eqnarray*}
M_d     : C\left(\mathcal{X}\right) \rightarrow C\left(\mathcal{X}\right) & & \quad M_d    f(x) =   d(x) f(x) \\
M_{d_n} : C\left(\mathcal{X}\right) \rightarrow C\left(\mathcal{X}\right) & & \quad M_{d_n}f(x) = d_n(x) f(x) 
\end{eqnarray*}

Next, we will look at the matrix $K_n$.
Applying it to a vector $f \in \mathbb{R}^{n}$ gives $\left( K_n f \right)_i = \sum_j K(x_i,x_j) f_j $.
Thus, we can define the empirical and true integral operator:
\begin{eqnarray*}
S_{n} : C\left(\mathcal{X}\right) \rightarrow C\left(\mathcal{X}\right) & & \quad S_{n}f(x) = \int k(x,y) f(y) dP_n(y) \\
S     : C\left(\mathcal{X}\right) \rightarrow C\left(\mathcal{X}\right) & & \quad S_n  f(x) = \int k(x,y) f(y) dP(y) 
\end{eqnarray*}

With these definitions, we can define the \emph{empirical unnormalized graph Laplacian}, 
$U_n : C\left(\mathcal{X}\right)\rightarrow C\left(\mathcal{X}\right)$, and 
the \emph{true unnormalized graph Laplacian},
 $U : C\left(\mathcal{X}\right)\rightarrow C\left(\mathcal{X}\right)$ as
\begin{eqnarray*}
U_nf(x) &=& M_{d_n}   f(x) - S_nf(x) = \int k(x,y) \left( f(x)-f(y) \right) dP_n(y) \\
  Uf(x) &=& M_d f(x) - Sf(x) = \int k(x,y) \left( f(x)-f(y) \right) dP(y)
\end{eqnarray*}

For the normalized Laplacian, we can proceed as follows.
Recall that $v$ is an eigenvector of $L_n^{\prime}$ with eigenvalue $v$ iff $v$ is an eigenvector of $H_n^{\prime} = D^{-1/2} K_n D^{-1/2}$ with eigenvalue $1-\lambda$.
So, consider $H_n^{\prime}$, defined as follows.
The matrix $H_n^{\prime}$ operates on a vector $f = \left( f_1,\ldots, f_n \right)^{T}$ as $\left( H_n^{\prime} f \right)_{i} = \sum_j \frac{ K(x_i,x_j) }{ \sqrt{ d_id_j} } $.
Thus, we can define the normalized empirical and true similarity functions
\begin{eqnarray*}
h_n(x,y) &=& k(x,y) / \sqrt{ d_n(x) d_n(y) }  \\
  h(x,y) &=& k(x,y) / \sqrt{   d(x)   d(y) }  
\end{eqnarray*}
and introduce two integral operators
\begin{eqnarray*}
T_{n} : C\left(\mathcal{X}\right) \rightarrow C\left(\mathcal{X}\right) & & \quad T_{n}f(x) = \int h_n(x,y) f(y) dP_n(y) \\
T     : C\left(\mathcal{X}\right) \rightarrow C\left(\mathcal{X}\right) & & \quad T  f(x) = \int   h(x,y) f(y) dP(y) 
\end{eqnarray*}

Note that for these operators the scaling factors $\frac{1}{n}$ which are hidden in $P_n$ and $d_n$ cancel each other.
Said another way, the matrix $H_n^{\prime}$ already has $\frac{1}{n}$ scaling factor---as opposed to the matrix $K_n$ in the unnormalized case.
So, contrary to the unnormalized case, we do not have to scale matrices $H_n^{\prime}$ and $H_n$ with the $\frac{1}{n}$ factor.

All of the above is machinery that enables us to transfer the problem of convergence of Laplacian matrices to problems of convergence of sequences of operators on $C\left(\mathcal{X}\right)$.

Given the above, they establish a lemma which, informally, says that under the general assumptions:
\begin{compactitem}
\item
the functions $d_n$ and $d$ are continuous, bounded from below by $\ell > 0$, and bounded from above by $\|k\|_{\infty}$, 
\item
all the operators are bounded, 
\item
all the integral operators are compact, 
\item
all the operator norms can be controlled.
\end{compactitem}
The hard work is to show that the empirical quantities converge to the true quantities; this is done with the perturbation result above (where, recall, the perturbation theory can be applied only to isolated parts of the spectrum).
In particular:
\begin{itemize}
\item
In the normalized case, this is true if $\lambda \ne 1$ is an eigenvalue of $U^{\prime}$ that is of interest.
The reason is that $U^{\prime} = I - T^{\prime}$ is a compact operator.
\item
In the unnormalized case, this is true if $\lambda \notin \mbox{range}(d)$ is an eigenvalue of $U$ that is of interest.
The reason is that $U = M_{d}-S$ is \emph{not} a compact operator, unless $M_d$ is a multiple of the~identity.
\end{itemize}
So, the key difference is the condition under which eigenvalues of the limit operator are isolated in the spectrum: for the normalized case, this is true if $\lambda \neq 1$, while for the non normalized case, this is true if $\lambda \notin \mbox{range}(d)$.

In addition to the ``positive'' results above, a ``negative'' result of the form given in the following lemma can be established.
\begin{lemma}[Clustering fails if $\lambda\notin\mbox{range}(d)$ is violated.]
Assume that  $\sigma(U) - \{0\} \cup\mbox{range}(d)$ with eigenvalue $0$ having multiplicity $1$, and that the probability distribution $P$ on $\mathcal{X}$ has no point masses.
Then the sequence of second eigenvectors of $\frac{1}{n}L_n$ converges to $\min_{x\in\mathcal{X}}d(x)$.
The corresponding eigenfunction will approximate the characteristic function of some $x\in\mathcal{X}$, with $d(x)=\min_{x\in\mathcal{X}}d(x)$ or a linear combination of such functions.
\end{lemma}
That is, in this case, the corresponding eigenfunction does \emph{not} contain any useful information for clustering (and one can't even check if $\lambda\in\mbox{range}(d)$ with a finite sample of data points).

While the analysis here has been somewhat abstract, the important point here is that this is \emph{not} a pathological situation: a very simple example of this failure is given in the paper; and this phenomenon will arise whenever there is substantial degree heterogeneity, which is very common in practice.

%% file: lect24.tex
\section{%
(04/21/2015):
Some Statistical Inference Issues (3 of 3):
Stochastic blockmodels}

Reading for today.
\begin{compactitem}
\item
``Spectral clustering and the high-dimensional stochastic blockmodel,'' in The Annals of Statistics, by Rohe, Chatterjee, and Yu
\item
``Regularized Spectral Clustering under the Degree-Corrected Stochastic Blockmodel,'' in NIPS, by Qin and Rohe
\end{compactitem}

Today, we will finish up talking about statistical inference issues by discussing them in the context of stochastic blockmodels.
These are different models of data generation than we discussed in the last few classes, and they illustrate somewhat different issues.

\subsection{Introduction to stochastic block modeling}

As opposed to working with expansion or conductance---or some other ``edge counting'' objective like cut value, modularity, etc.---the \emph{stochastic block model (SBM)} is an example of a so-called probabilistic or generative model.
Generative models are a popular way to encode assumptions about the way that latent/unknown parameters interact to create edges $(ij)$
Then, they assign a probability value for each edges $(ij)$ in a network.
There are several advantages to this approach.
\begin{itemize}
\item
It makes the assumptions about the world/data explicit.
This is as opposed to encoding them into an objective and/or approximation algorithm---we saw several examples of reverse engineering the implicit properties of approximation algorithms.
\item
The parameters can sometimes be interpreted with respect to hypotheses about the network structure.
\item
It allows us to use likelihood scores, to compare different parameterizations or different models.
\item
It allows us to estimate missing structures based on partial observations of graph structure.
\end{itemize}
There are also several disadvantages to this approach.
The most obvious is the following.
\begin{itemize}
\item
One must fit the model to the data, and fitting the model can be complicated and/or computationally expensive.
\item
As a result of this, various approximation algorithms are used to fit the parameters.
This in turn leads to the question of what is the effect of those approximations versus what is the effect of the original hypothesized model? (I.e., we are back in the other case of reverse engineering the implicit statistical properties underlying approximation algorithms, except here it is in the approximation algorithm to estimate the parameters of a generative model.)
This problem is particularly acute for sparse and noisy data, as is common.
\end{itemize}

Like other generative models, SBMs define a probability distribution over graphs, $\mathbb{P}\left[ G | \Theta \right]$, where $\Theta$ is a set of parameters that govern probabilities under the model.
Given a specific $\Theta$, we can then draw or \emph{generate} a graph $G$ from the distribution by flipping appropriately-biased coins.
Note that \emph{inference} is the reverse task: given a graph $G$, either just given to us or generated synthetically by a model, we want to recover the model, i.e., we want to find the specific values of $\Theta$ that generated~it.

The simpled version of a SBM is specified by the following.
\begin{itemize}
\item
A positive integer $k$, a scalar value denoting the the number of blocks.
\item
A vector $\vec{z}\in\mathbb{R}^{n}$, where $z_i$ gives the group index of vertex $i$.
\item
A matrix $M\in\mathbb{R}^{k \times k}$, a stochastic block matrix, where $M_{ij}$ gives the probability that a vertex of type $i$ links to a vertex of type $j$.
\end{itemize}
Then, one generates edge $(ij)$ with probability $M_{z_iz_j}$.
That is, edges are not identically distributed, but they are conditionally independent, i.e., conditioned on their types, all edges are independent, and for a given pair of types $(ij)$, edges are i.i.d.

Observe that the SBM has a relatively large number of parameters, ${k \choose 2}$, even after we have chosen the labeling on the vertices.
This has plusses and minuses.
\begin{itemize}
\item
Plus: it allows one the flexibility to model lots of possible structures and reproduce lots of quantities of interest.
\item
Minus: it means that there is a lot of flexibility, thus making the possibility of overfitting more likely.
\end{itemize}

Here are some simple examples of SBMs.
\begin{itemize}
\item
If $k=1$ and $M_{ij}=p$, for all $i,j$, then we recover the vanilla ER model.
\item
Assortative networks, if $M_{ii} > M_{ij}$, for $i \ne j$.
\item
Disassortative networks, if $M_{ii} < M_{ij}$, for $i \ne j$.
\end{itemize}

\subsection{Warming up with the simplest SBM}

To illustrate some of the points we will make in a simple context, consider the ER model.  
\begin{itemize}
\item
If, say, $p = \frac{1}{2}$ and the graph $G$ has more than a handful of nodes, then it will be very easy to estimate $p$, i.e., to estimate the parameter vector $\Theta$ of this simple SBM, basically since measure concentration will occur very quickly and the empirical estimate of $p$ we obtain by counting the number of edges will be very close to its expected value, i.e., to $p$.
More generally, if $n$ is large and $p \gtrsim \frac{\log(n)}{n}$, then measure will still concentrate, i.e., the empirical and expected values of $p$ will be close, and we will be able to estimate $p$ well.
(This is related to the well-known observation that if $p \gtrsim\frac{\log(n)}{n}$, then $G_{np}$ and $G_{nm}$ are very similar, for appropriately chosen values of $p$ and $m$.)
\item
If, on the other hand, say, $p=\frac{3}{n}$, then this is not true.
In this regime, measure has \emph{not} concentrated for most statistics of interest: the graph is not even fully connected; the giant component has nodes of degree almost $O\left(\log(n)\right)$; and the giant component has small sets of nodes of size $\Theta\left( \log(n) \right)$ that have conductance $O\left( \frac{1}{\log(n)} \right)$.
(Contrast all of these the a $3$-regular random graph, which: is fully connected, is degree-homogeneous, and is a very good~expander.)
\end{itemize}
In these cases when measure concentration fails to occur, e.g., due to exogenously-specified degree heterogeneity or due to extreme sparsity, then one will have difficulty with recovering parameters of hypothesized models.
More generally, similar problems arise, and the challenge will be to show that one can reconstruct the model under as broad a range of parameters as possible.

\subsection{A result for a spectral algorithm for the simplest nontrivial SBM}

Let's go into detail on the following simple SBM (which is the simplest aside from ER).
\begin{itemize}
\item
Choose a partition of the vertics, call them $V^1$ and $V^2$, and WLOG let $V^1 = \{ 1,\ldots,\frac{n}{2} \}$ and $V^2 = \{ \frac{n}{2}+1,\ldots,n \}$.
\item
Then, choose probabilities $p > q$ and place edges between vertices $i$ and $j$ with probability
\[
\mathbb{P}\left[ (ij) \in E \right]
   = \left\{ \begin{array}{l l}
                    q  & \quad \text{if $i \in V^1$ and $j \in V^2$ of $i \in V^2$ and $j \in V^1$ }\\
                    p  & \quad \text{otherwise}
             \end{array}
     \right.  ,
\]
\end{itemize}
In addition to being the ``second simplest'' SBM, this is also a simple example of a \emph{planted partition model}, which is commonly studied in TCS and related areas.

Here is a fact:
\[
\mathbb{E}\left[ \text{number of edges crossing bw } V^1 \text{ and } V^2 \right] = q |V^1| |V^2| .
\]
In addition, if $p$ is sufficiently larger than $q$, then every other partition has more edges.
This is the basis of recovering the model.
Of course, if $p$ is only slightly but not sufficiently larger than $q$, then there might be fluctuational effects such that it is difficult to find this from the empirical graph.
This is analogous to having difficulty with recovering $p$ from very sparse ER, as we discussed.

Within the SBM framework, the most important inferential task is recovering cluster membership of nodes from a single observation of a graph (i.e., the two clusters in this simple planted partition form of the SBM).
There are a variety of procedures to do this, and here we will describe spectral~methods.

In particular, we will follow a simple analysis motivated by McSherry's analysis, as described by Spielman, that will provide a ``positive'' result for sufficiently dense matrices where $p$ and $q$ are sufficiently far apart.
Then, we will discuss this model more generally, with an emphasis on how to deal with very low-degree nodes that lead to measure concentration problems.
In particular, we will focus on a form of regularized spectral clustering, as done by Qin and Rohe in their paper ``Regularized spectral clustering under the degree-corrected stochastic blockmodel.''
This has connections with what we have done with the Laplacian over the last few weeks.

To start, let $M$ be the \emph{population adjacency matrix}, i.e., the hypothesized matrix, as described above.
That is, 
\[
M = \left( \begin{array}{cc} p\vec{1}\vec{1}^{T} & q\vec{1}\vec{1}^{T} \\
                             q\vec{1}\vec{1}^{T} & p\vec{1}\vec{1}^{T}
           \end{array}
    \right)
\]
Then, let $A$ be the \emph{empirical adjacency matrix}, i.e., the actual matrix that is generated by flipping coins and on which we will perform computations.
This is generated as follows: let $A_{ij}=1$ w.p. $M_{ij}$ and s.t. $A_{ij}=A_{ji}$.
So, the basic goal is going to be to recover clusters in $M$ by looking at information in $A$.

Let's look at the eigenvectors. 
First, since $M\vec{1} = \frac{n}{2}(p+q)\vec{1}$, we have
\begin{eqnarray*}
\mu_1 &=& \frac{n}{2}(p+q) \\
w_1   &=& \vec{1}  ,
\end{eqnarray*}
where $\mu_1$ and $w_1$ are the leading eigenvalue and eigenvector, respectively.
Then, since the second eigenvector (of $M$) is constant on each cluster,  we have that $Mw_2 = \mu_2 w_2$, where
\begin{eqnarray*}
\mu_2 &=& \frac{n}{2}(p-q) \\
w_2   &=& \begin{cases}
               \frac{1}{\sqrt{n}} \text{ if } i \in V^1 \\
              -\frac{1}{\sqrt{n}} \text{ if } i \in V^2 \\
          \end{cases}  .
\end{eqnarray*}

In that case, here is a simple algorithm for finding the planted bisection.
\begin{enumerate}
\item
Compute $v_2$, the eigenvector of second largest eigenvalue of $A$.
\item
Set $S = \{ i : v_2(i) \geq 0 \}$
\item
Guess that $S$ is one side of the bisection and that $\bar{S}$ is the other side.
\end{enumerate}
We will show that under not unreasonable assumptions on $p$, $q$, and $S$, then by running this algorithm one gets the hypothesized cluster mostly right.

Why is this?

The basic idea is that $A$ is a perturbed version of $M$, and so by perturbation theory the eigenvectors of $A$ should look like the eigenvectors of $M$.

Let's define $R = A-M$.
We are going to view $R$ as a random matrix that depends on the noise/randomness in the coin flipping process.
Since matrix perturbation theory bounds depend on (among other things) the norm of the perturbation, the goal is to bound the probability that $\|R\|_2$ is large.
There are several methods from random matrix theory that give results of this general form, and one or the other is appropriate, depending on the exact statement that one wants to prove.
For example, if you are familiar with Wigner's semi-circle law, it is of this general form.
More recently, Furedi-Komlos got another version; as did Krivelevich and Vu; and Vu.
Here we state a result due to Vu.

\begin{theorem}
With probability tending to one, if $p \ge c \frac{\log^4(n)}{n}$, for a constant $c$, then 
\[
\|R\|_2 \le 3 \sqrt{pn}  .
\]
\end{theorem}

The key question in theorems like this is the value of $p$.
Here, one has that $p \gtrapprox \frac{\log(n)}{n}$, meaning that one can get pretty sparse (relative to $p=1$) but not extremely sparse (relative to $p=\frac{1}{n}$ or $p=\frac{3}{n}$).
If one wants stronger results (e.g., not just mis-classifying only a constant fraction of the vertices, which we will do below, but instead that one predicts correctly for all but a small fraction of the vertices), then one needs $p$ to be larger and the graph to be denser.
As with the ER example, the reason for this is that we need to establish concentration of appropriate estimators.

Let's go onto perturbation theory for eigenvectors.
Let $\alpha_1 \ge \alpha_2 \ge \cdots \alpha_n$ be the eigenvalues of $A$, and let $\mu_1 > \mu_2 > \mu_3 = \cdots \mu_n = 0$ be the eigenvalues of $M$.

Here is a fact from matrix perturbation theory that we mentioned before: for all $i$, 
\[
|\alpha_i - \mu_i | \le \|A-M\|_2 = \|R\|_2  .
\]

The following two claims are easy to establish.

\begin{claim}
If $\|R\|_2 < \frac{n}{4}(p-q)$, then 
\[
\frac{n}{4}(p-q) < \alpha_2 < \frac{3n}{4}(p-q)
\]
\end{claim}

\begin{claim}
If, in addition, $q > \frac{p}{3}$, then $\frac{3n}{4}(p-q) < \alpha_1$.
\end{claim}

From these results, we have a separation, and so we can view $\alpha_2$ as a perturbation of $\mu_2$.
The question is: can we view $v_2$ as a perturbation of $w_2$?
The answer is Yes.
Here is a statement of this result.

\begin{theorem}
Let $A$, $M$ be symmetric matrices, and let $R=M-A$.
Let $\alpha_1 \ge \cdots \ge \alpha_n$ be the eigenvectors of $A$, with $v_1,\cdots,v_n$ the corresponding eigenvectors.
Let $\mu_1 \ge \cdots \ge \mu_n$ be the eigenvectors of $M$, with $w_1,\cdots,w_n$ the corresponding eigenvectors.
Let $\theta_i$ be the angle between $v_i$ and $w_i$.
Then, 
\begin{eqnarray*}
\sin \theta_i \le \frac{ 2 \| R \|_2 }{ \min_{j \ne i} |\alpha_i - \alpha_j | } \\
\sin \theta_i \le \frac{ 2 \| R \|_2 }{ \min_{j \ne i} |\mu_i - \mu_j | } 
\end{eqnarray*}
\end{theorem}
\begin{proof}
WLOG, we can assume $\mu_i=0$, since the matrices $M-\mu_i I$ and $A-\alpha_i I$ have the same eigenvectors as $M$ and $A$, and $M-\mu_i I$ has the $i$th eigenvalue being $0$.
Since the theorem is vacuous if $\mu_i$ has multiplicities, we can assume unit multiplicity, and that $w_i$ is a unit vector in the null space of $M$.
Due to the assumption that $\mu_i = 0$, we have that $|\alpha_i| \le \|R\|_2$.

Then, expand $v_i$ in an eigenbasis of $M$: $v_i = \sum_j c_j w_j$, where $c_j = w_j^T v_i$.
Let $\delta = \min_j |\mu_j|$.
Then observe that 
\[
\|Mv_i\|_2^2 = \sum_j c_j^2 \mu_j^2 \ge \sum_{j \ne i} c_j^2 \delta^2 = \delta^2 \sum_{j \ne i} c_j^2 = \delta^2 \left( 1-c_i^2 \right) = \delta^2 \sin^2 \theta_i
\]
and also that 
\[
\| Mv_i \| \le \| Av_i \| + \| Rv_i \| = \alpha_i + \| R v_i \| \le 2 \| R \|_2 .
\]
So, from this it follows that $\sin \theta_i \le \frac{ 2 \|R\|_2 }{ \delta } $ .
\end{proof}

This is essentially a version of the Davis-Kahan result we saw before.
Note that it says that the amount by which eigenvectors are perturbed depends on how close are other eigenvalues, which is what we would expect.

Next, we use this for partitioning the simple SBM.
We want to show that not too many vertices are mis-classified.

\begin{theorem}
Given the two-class SBM defined above, assume that $p \ge c \frac{\log^4(n)}{n}$ and that $q > p/3$.
If one runs the spectral algorithm described above, then at most a constant fraction of the vertices are misclassified.
\end{theorem}
\begin{proof}
Consider the vector $\vec{\delta} = v_2 - w_2$.
For all $i \in V$ that are misclassified by $v_2$, we have that $ | \delta(i) | \ge \frac{1}{\sqrt{n}} $.
So, if $v_2$ misclassified $k$ vertices, then $\|\delta\| \ge \sqrt{k/n}$.
Since $u$ and $v$ are unit vectors, we have the crude bound that $\|\delta\| \le \sqrt{2} \sin \theta_2$.

Next, we can combine this with the perturbation theory result above.
Since $q > p/3$, we have that $\min_{j \ne 2} | \mu_2 - \mu_i | = \frac{n}{2}(p-q)$; and since $p \ge c \frac{\log^4(n)}{n}$, we have that $\|R\| \le 3 \sqrt{pn}$.
Then, 
\[
\sin \theta_2 \le \frac{ 3 \sqrt{pn} }{ \frac{n}{2}(p-q) } = \frac{ 6 \sqrt{p} }{ \sqrt{n}(p-q) }.
\]
So, the number $k$ of mis-classified vertices satisfies $\sqrt{\frac{k}{n}} \le \frac{ 6 \sqrt{p} }{ \sqrt{n}(p-q) } $, and thus $ k \le \frac{ 36 p }{ (p-q)^2 } $.
\end{proof}

So, in particular, if $p$ and $q$ are both constant, then we expect to misclassify at most a constant fraction of the vertices.
E.g.,  if $p=\frac{1}{2}$ and $q = p - \frac{12}{\sqrt{n}}$, then $\frac{ 36p }{(p-q)^2 } = \frac{n}{8}$, and so only a constant fraction of the vertices are misclassified.

This analysis is a very simple result, and it has been extended in various ways.
\begin{itemize}
\item
The Ng et al. algorithm we discussed before computes $k$ vectors and then does $k$ means, making similar gap assumptions.
\item
Extensions to have more than two blocks, blocks that are not the same size, etc.
\item
Extensions to include degree variability, as well as homophily and other empirically-observed properties of networks.
\end{itemize}

The general form of the analysis we have described goes through to these cases, under the following types of assumptions.
\begin{itemize}
\item
The matrix is dense enough.
Depending on the types of recovery guarantees that are hoped for, this could mean that $\Omega(n)$ of the edges are present for each node, or perhaps $\Omega(\mbox{polylog}(n))$ edges for each node.
\item
The degree heterogeneity is not too severe.
Depending on the precise algorithm that is run, this can manifest itself by placing an upper bound on the degree of the highest degree node and/or placing a lower bound on the degree of the lowest degree node.
\item
The number of clusters is fixed, say as a function of $n$, and each of the clusters is not too small, say a constant fraction of the nodes.
\end{itemize}
Importantly, \emph{none} of these simplifying assumptions are true for most ``real world'' graphs.
As such, there has been a lot of recent work focusing on dealing with these issues and making algorithms for SBMs work under broader assumptions.
Next, we will consider one such~extension.

\subsection{Regularized spectral clustering for SBMs}

Here, we will consider a version of the degree-corrected SBM, and we will consider doing a form of \emph{regularized spectral clustering (RSC)} for it.

Recall the definition of the basic SBM.
\begin{definition}
Given nodes $V=[n]$, let $z:[n]\rightarrow[k]$ be a partition of the $n$ nodes into $k$ blocks, i.e., $z_i$ is the block membership of the $i^{th}$ node.
Let $B \in [0,1]^{k \times k}$.
Then, under the SBD, we have that the probability of an edge between $i$ and $j$ is 
\[
P_{ij} = B_{z_iz_j}, \text{ for all } i,j \in \{1,\ldots,n\}  .
\]
\end{definition}
In particular, this means that, given $z$, the edges are independent.

Many real-world graphs have substantial degree heterogeneity, and thus it is common to in corporate this into generative models.
Here is the extension of the SBM to the \emph{Degree-corrected stochastic block model (DC-SBM)}, which introduces additional parameters $\theta_i$, for $i \in[n]$, to control the node~degree.
\begin{definition}
Given the same setup as for the SBM, specify also additional parameters $\theta_i$, for $i\in[n]$.
Then, under the DC-SBM, the probability of an edge between $i$ and $j$ is
\[
P_{ij} = \theta_i \theta_j B_{z_iz_j} ,
\]
where $\theta_i \theta_j B_{z_iz_j} \in [0,1] $, for all $i,j \in [n]$.
\end{definition}

Note: to make the DC-SBM identifiable (i.e., so that it is possible in principle to learn the true model parameters, say given an infinite number of observations, which is clearly a condition that is needed for inference), one can impose the constraint that $\sum_i \theta_i \delta_{z_i,r} = 1$, for each block $r$.
(This condition says that $\sum_i \theta_i = 1$ within each block.)
In this case $B_{st}$, for $s \ne t$, is the expected number of links between block $s$ and block $t$; and $B_{st}$, for $s = t$, is the expected number of links within block $s$.

Let's say that $A \in \{0,1\}^{n \times n}$ is the adjacency matrix; $L = D^{-1/2}A D^{-1/2}$.
In addition, let $\mathcal{A} = \mathbb{E}\left[ A \right] $ be the population matrix, under the DC-SBM.
Then, one can express $\mathcal{A}$ as $\mathcal{A} = \Theta Z B Z^T \Theta$, where $\Theta \in \mathbb{R}^{n \times n} = \mbox{diag}(\theta_i)$, and where $Z \in \{0,1\}^{n \times k}$ is a membership matrix with $Z_{it} = 1$ iff node $i$ is in block $t$, i.e., if $z_i = t$.

We are going to be interested in very sparse matrices, for which the minimum node degree is very small, in which case a vanilla algorithm will fail to recover the SBM blocks. 
Thus, we will need to introduce a regularized version of the Laplacian.
Here is the~definition.

\begin{definition}
Let $\tau > 0$.
The \emph{regularized graph Laplacian} is $L_{\tau} = D_{\tau}^{-1/2} A D_{\tau}^{-1/2} \in \mathbb{R}^{n \times n}$, with $D_{\tau} = D + \tau I$, for $\tau > 0$.
\end{definition}

This is defined for the empirical data; but given this, we can define the corresponding population~quantities:
\begin{eqnarray*}
\mathcal{D}_{ii} &=& \sum_j \mathcal{A}_{ij}  \\
\mathcal{D}_{\tau} &=& \mathcal{D} + \tau I \\
\mathcal{L} &=& \mathcal{D}^{-1/2} \mathcal{A} \mathcal{D}^{-1/2} \\
\mathcal{L}_{\tau} &=& \mathcal{D}_{\tau}^{-1/2} \mathcal{A} \mathcal{D}_{\tau}^{-1/2} 
\end{eqnarray*}

Two things to note.
\begin{itemize}
\item
Under the DC-SBM, if the model is identifiable, then one should be able to determine the partition from $\mathcal{A}$ (which we don't have direct access to, given the empirical data).
\item
One also wants to determine the partition from the empirical data $A$, under broader assumptions than before, in particular under smaller minimum degree.
\end{itemize}

Here is a description of the basic algorithm of Qin and Rohe.
Basically, it is the Ng et al. algorithm that we described before, except that we apply it to the regularized graph Laplacian, i.e., it involves finding the leading eigenvectors of $L_{\tau}$ and then clustering in the low dimensional space.

Given as input an Adjacency Matrix $A$, the number of clusters $k$, and the regularizer $\tau \ge 0$.
\begin{enumerate}
\item
Compute $L_{\tau}$.
\item
Compute the matrix $X_{\tau} = \left[ X_1^{\tau} ,\ldots, X_k^{\tau}\right] \in \mathbb{R}^{n \times k}$, the orthogonal matrix consisting of the $k$ largest eigenvectors of $L_{\tau}$.
\item
Compute the matrix $X_{\tau}^{*} \in \mathbb{R}^{n \times k}$ by normalizing each row of $X_{\tau}$ to have unit length, i.e., project each row of $X_{\tau}$ onto the unit sphere in $\mathbb{R}^{k}$, i.e., $X_{ij}^{*,\tau} = X_{ij}^{\tau} / \sum_j X_{ij}^{\tau,2}$.
\item
Run $k$ means on the rows of $X^{*}_{\tau}$ to create $k$ non-overlapping clusters $V_1,\ldots,V_k$.
\item
Output $V_1,\ldots,V_k$; node $i$ is assigned to cluster $r$ if the $i^{th}$ tow of $X^{*}_{\tau}$ is assigned to $V$.
\end{enumerate}
There are a number of empirical/theoretical tradeoffs in determining the best value for $\tau$, but one can think of $\tau$ as being the average node degree.

There are several things one can show here.

First, one can show that $L_{\tau}$ is close to $\mathcal{L}_{\tau}$.

\begin{theorem}
Let $G$ be the random graph with $\mathbb{P}\left[ \mbox{edge bw } ij \right] = P_{ij}$.
Let $\delta = \min_i \mathcal{D}_{ii}$ be the minimum expected degree of $G$.
If $\delta+\tau > O\left(\log(n)\right)$, then with constant probability
\[
\| L_{\tau} - \mathcal{L}_{\tau} \| \le O(1) \sqrt{ \frac{ \log(n) }{ \delta + \tau } } .
\]
\end{theorem}

\textbf{Remark.}
Previous results required that the minimum degree $\delta \ge O(\log(n))$, so this result generalizes these to allow $\delta$ to be much smaller, assuming the regularization parameter $\tau$ is large enough.
Importantly, typical real networks do \emph{not} satisfy the condition that $\delta \ge O(\log(n))$, and RSC is most interesting when this condition fails.
So, we can apply this result in here to graph with small node degrees.

\textbf{Remark.}
The form of $L_{\tau}$ is similar to many of the results we have discussed, and one can imagine implementing RSC (and obtaining this theorem as well as those given below) by computing approximations such as what we have discussed.
So far as I know, that has not been done.

Second, one can bound the difference between the empirical and population eigenvectors.
For this, one needs an additional concept.
\begin{itemize}
\item
Given an $n \times k$ matrix $A$, the \emph{statistical leverage scores} of $A$ are the diagonal elements of the projection matrix onto the span of $A$.
\end{itemize}
In particular, if the $n \times k$ matrix $U$ is an orthogonal matrix for the column span of $A$, then the leverage scores of $A$ are the Euclidean norms of the \emph{rows} of $U$.  
For a ``tall'' matrix $A$, the $i^{th}$ leverage score has an interpretation in terms of the leverage or influence that the $i^{th}$ row of an $A$ has on the least-squares fit problem defined by $A$.
In the following, we will use an extension of the leverage scores, defined relative to the best rank-$k$ approximation the the matrix.

\begin{theorem}
\label{thm:eigenvectors}
Let $X_{\tau}$ and $\mathcal{X}_{\tau}$ be in $\mathbb{R}^{n \times k}$ contain the top $k$ eigenvectors of $L_{\tau}$ and $\mathcal{L}_{\tau}$, respectively.
Let 
\[
\xi = \min_i \{ \min \{ \| X_{\tau}^{i} \|_2 , \| \mathcal{X}_{\tau}^{i} \|_2 \} \}  .
\]
Let $X_{\tau}^{*}$ and $\mathcal{X}_{\tau}^{*}$ be the row normalized versions of $X_{\tau}$ and $\mathcal{X}_{\tau}$.
Assume that $\sqrt{ \frac{ k \log(n) }{ \delta + \tau } } \le O( \lambda_k)$ and $\delta + \tau > O(\log(n))$.
Then, with constant probability, 
\begin{eqnarray*}
\| X_{\tau} - \mathcal{X}_{\tau} O \|_F &\le& O \left( \frac{1}{\lambda_k} \sqrt{k \log(n)}{\delta+\tau} \right)  \\
\| X_{\tau}^{*} - \mathcal{X}_{\tau}^{*} O \|_F &\le& O \left( \frac{1}{\xi\lambda_k} \sqrt{k \log(n)}{\delta+\tau} \right)  ,
\end{eqnarray*}
where $O$ is a rotation matrix.
\end{theorem}

Note that the smallest leverage score enters the second expression but not the first expression.
That is, it does not enter the bounds on the empirical quantities, but it does enter into the bounds for the population quantities.

We can use these results to derive misclassification rate for RSC.
The basic idea for the misclassification rate is to run $k$-means on the rows of $X_{\tau}^{*}$ and also on the rows of $\mathcal{X}_{\tau}^{*}$.
Then, one can say that a node on the empirical data is clustered correctly if it is closer to the centroid of the corresponding cluster on the population data.
This basic idea needs to be modified to take into account the fact that if any $\lambda_i$ are equal, then only the subspace spanned by the eigenvectors is identifiable, so we consider this up to a rotation $O$.

\begin{definition}
If $C_i O$ is closer to $\mathcal{C}_i $ than any other $\mathcal{C}_j$, then we say that the node is correctly clustered; and we define the misclassified nodes to be 
\[
\mathcal{M} = \left\{ i : \exists j \ne i \mbox{ s.t. } \| C_i O^T -\mathcal{C}_i \|_2 > \| \mathcal{C}_i O^T - C_j \right\} .
\]
\end{definition}

Third, one can bound the misclassification rate of the RCS classifier with the following theorem.

\begin{theorem}
\label{thm:misclassification}
With constant probability, the misclassification rate is
\[
\frac{|\mathcal{M}|}{n} \le c \frac{ k \log(n) }{ n \xi^2 (\delta+\tau) \lambda_k^2  }.
\]
\end{theorem}

Here too the smallest leverage score determines the overall quality.

\textbf{Remark.}
This is the first result that explicitly relates leverage scores to the statistical performance of a spectral clustering algorithm.
This is a large topic, but to get a slightly better sense of it, recall that the leverage scores of $\mathcal{L}_{\tau}$ are $\| \mathcal{X}_{\tau}^{i} \|_2^2 = \frac{ \theta_i^{\tau} }{ \sum_j \theta_j^{\tau} \delta_{z_jz_i} }$.
So, in particular, if a node $i$ has a small expected degree, then $\theta_i^{\tau}$ is small and $\| \mathcal{X}_{\tau}^{i} \|_2 $ is small.
Since $\xi$ appears in the denominator of the above theorems, this leads to a worse bound for the statistical claims in these theorems.
In particular, the problem arises due to projecting $X_{\tau}^{i}$ onto the unit sphere, i.e., while large-leverage nodes don't cause a problem, errors for small-leverage rows can be amplified---this didn't arise when we were just making claims about the empirical data, e.g., the first claim of Theorem~\ref{thm:eigenvectors}, but when considering statistical performance, e.g., the second claim of Theorem~\ref{thm:eigenvectors} or the claim of Theorem~\ref{thm:misclassification}, for nodes with small leverage score it amplifies noisy measurements.

%% file: lect25.tex
\section{%
(04/23/2015):
Laplacian solvers (1 of 2)}

Reading for today.
\begin{compactitem}
\item
``Effective Resistances, Statistical Leverage, and Applications to Linear Equation Solving,'' in arXiv, by Drineas and Mahoney
\item
``A fast solver for a class of linear systems,'' in CACM, by Koutis, Miller, and Peng
\item
``Spectral Sparsification of Graphs: Theory and Algorithms,'' in CACM, by Batson, Spielman, Srivastava, and Teng 
\end{compactitem}
 
(Note: the lecture notes for this class and the next are taken from the lecture notes for the final two classes of the class I taught on Randomized Linear Algebra in Fall 2013.)

\subsection{Overview}

We have seen problems that can be written in the form of a system of linear equations with Laplacian constraint matrices, i.e.,  
\[
Lx = b .
\]
For example, we saw this with the various semi-supervised learning methods as well as with the MOV weakly-local spectral method.
In some cases, this arises in slightly modified form, e.g., as an augmented/modified graph and/or if there are additional projections (e.g., the Zhou et al paper on ``Learning with labeled and unlabeled data on a directed graph,'' that is related to the other semi-supervised methods we discussed, does this explicitly).
Today and next time we will discuss how to solve linear equations of this form.

\subsection{Basic statement and outline}

While perhaps not obvious, solving linear equations of this form is a useful \emph{algorithmic primitive}---like divide-and-conquer and other such primitives---much more generally, and thus there has been a lot of work on it in recent years.

Here is a more precise statement of the use of this problem as a primitive.
\begin{definition}
The \emph{Laplacian Primitive} concerns systems of linear equations defined by Laplacian constraint matrices:
\begin{itemize}
\item
\textsc{INPUT}: a Laplacian $L \in \mathbb{R}^{n \times n}$, a vector $b\in\mathbb{R}^{n}$ such that $\sum_{i=1}^{n} b_i = 0$, and a number $\epsilon>0$.
\item
\textsc{OUTPUT}: a vector $\tilde{x}_{opt} \in \mathbb{R}^{n}$ such that $\|\tilde{x}_{opt} - L^{\dagger} b \|_{L} \le \epsilon \| L^{\dagger} b\|_{L}$, where for a vector $z\in\mathbb{R}^{n}$ the $L$-norm is given by $\|z\|_L=\sqrt{z^TLz}$.
\end{itemize}
\end{definition}

While we will focus on linear equations with Laplacian constraint matrices, most of the results in this area hold for a slightly broader class of problems.
In particular, they hold for any linear system $Ax=b$, where $A$ is an SDD (symmetric diagonally dominant) matrix (i.e., that the diagonal entry of each row is larger, or not smaller, than the sum of the absolute values of the off-diagonal entries in that row).
The reason for this is that SDD systems are linear-time reducible to Laplacian linear systems via a construction that only doubles the number of nonzero entries in the matrix.

As mentioned, the main reason for the interest in this topic is that, given a fast, e.g., nearly linear time algorithm, for the Laplacian Primitive, defined above, one can obtain a fast algorithm for all sorts of other basic graph problems.
Here are several examples of such problems.
\begin{itemize}
\item
Approximate Fiedler vectors.
\item
Electrical flows.
\item
Effective resistance computations.
\item
Semi-supervised learning for labeled data.
\item
Cover time of random walks.
\item
Max flow and min cut and other combinatorial problems.
\end{itemize}
Some of these problems we have discussed.
While it might not be surprising that problems like effective resistance 
computations and semi-supervised learning for labeled data can be solved with this primitive, it should be surprising that max flow and min cut and other combinatorial problems can be solved with this primitive.
We won't have time to discuss this in detail, but some of the theoretically fastest algorithms for these problems are based on using this primitive.

Here is a statement of the basic result that led to interest in this area.
\begin{theorem}[ST]
There is a randomized algorithm for the Laplacian Primitive that runs in expected time $O\left( m \log^{O(1)} (n) \log \left(1/\epsilon\right)  \right)$, where $n$ is the number of nodes in $L$, $m$ is the number of nonzero entries in $L$, and $\epsilon$ is the precision parameter.
\end{theorem}
Although the basic algorithm of ST had something like the $50^{th}$ power in the exponent of the logarithm, it was a substantial theoretical breakthrough, and since then it has been improved by KMP to only a single log, leading to algorithms that are practical or almost practical.
Also, although we won't discuss it in detail, many of the local and locally-biased spectral methods we have discussed arose out of this line of work in an effort to develop and/or improve this basic result.

At a high level, the basic algorithm is as follows.
\begin{enumerate}
\item
Compute a sketch of the input by sparsifying the input graph.
\item
Use the sketch to construct a solution, e.g., by solving the subproblem with any black box solver or by using the sketch as a preconditioner for an iterative algorithm on the original~problem.
\end{enumerate}
Thus, the basic idea of these methods is very simple; but to get the methods to work in the allotted time, and in particular to work in nearly-linear time, is very complicated.

Today and next time, we will discuss these methods, including a simple but slow method in more detail and a fast but complicated method in less detail.
\begin{itemize}
\item
\textbf{Today.} 
We will describe a simple, non-iterative, but slow algorithm.
This algorithm provides a very simple version of the two steps of the basic algorithm described above; and, while slow, this algorithm highlights several basic ideas of the more sophisticated versions of these~methods.
\item
\textbf{Next time.}
We will describe a fast algorithm provides a much more sophisticated implementation of the two steps of this basic algorithm.
Importantly, it makes nontrivial use of combinatorial ideas and couples the linear algebra with combinatorial preconditioning in interesting~ways.
\end{itemize}

\subsection{A simple slow algorithm that highlights the basic ideas}

Here, we describe in more detail a very simple algorithm to solve Laplacian-based linear systems.
It will be good to understand before we get to the fast but more complicated versions of the algorithm.

Recall that $L = D-W = B^TWB$ is our Laplacian, where $B$ is the $m \times n$ edge-incidence matrix, and where $W$ is an $m \times m$ edge weight matrix.
In particular, note that $m > n$ (assume the graph is connected to avoid trivial cases), and so the matrix $B$ is a \emph{tall} matrix.

Here is a restatement of the above problem.
\begin{definition}
Given as input a Laplacian matrix $L \in \mathbb{R}^{n \times n}$, a vector $b \in \mathbb{R}^{n}$, compute 
\[
\mbox{argmin}_{x\in\mathbb{R}^{n}} \| Lx - b \|_2  .
\]
The minimal $\ell_2$ norm $x_{opt}$ is given by $x_{opt} = L^{\dagger}b$, where $L^{\dagger}$ is the Moore-Penrose generalized inverse of $L$.
\end{definition}
We have reformulated this as a regression since it makes the proof below, which is based on RLA (Randomized Linear Algebra) methods, cleaner.

The reader familiar with linear algebra might be concerned about the Moore-Penrose generalized inverse since, e.g., it is typically not well-behaved with respect to perturbations in the data matrix.
Here, the situation is particularly simple: although $L$ is rank-deficient, (1) it is invertible if we work with vectors $b\perp\vec{1}$, and (2) because this null space is particular simple, the pathologies that typically arise with the Moore-Penrose generalized inverse do \emph{not} arise here.
So, it isn't too far off to think of this as the inverse.  

Here is a simple algorithm to solve this problem.
This algorithm takes as input $L$, $b$, and $\epsilon$; and it returns as output a vector $\tilde{x}_{opt}$.
\begin{enumerate}
\item
Form $B$ and $W$, define $\Phi = W^{1/2}B \in \mathbb{R}^{m \times n}$, let $U_{\Phi}\in\mathbb{R}^{m \times n}$ be an orthogonal matrix spanning the column space of $\Phi$, and let $\left( U_{\Phi} \right)_{(i)}$ denote the $i^{th}$ row of $U_{\Phi}$.
\item
Let $p_i$, for $i \in [n]$ such that $\sum_{i=1}^{n}p_i=1$ be given by 
\begin{equation}
\label{eqn:lev-scores}
p_i \ge \beta\frac{ \| \left( U_{\Phi} \right)_{(i)}\|_2^2  }{ \| U_{\Phi} \|_F^2 } = \frac{\beta}{n} \| \left( U_{\Phi} \right)_{(i)} \|_2^2
\end{equation}
for some value of $\beta\in(0,1]$.
(Think of $\beta=1$, which is a legitimate choice, but the additional flexibility of allowing $\beta\in(0,1)$ will be important in the next class.)
\end{enumerate}

A key aspect of this algorithm is that the sketch is formed by choosing elements of the Laplacian with the probabilities in Eqn.~(\ref{eqn:lev-scores}); these quantities are known as the statistical leverage scores, and they are of central interest in RLA.
Here is a definition of these scores more generally.

\begin{definition}
Given a matrix $A\in\mathbb{R}^{m \times n}$, where $m > n$, the $i^{th}$ leverage score is
\[
\left(P_{A}\right)_{ii} = \left(U_AU_A^T\right)_{ii} = \| \left( U_A \right)_{ii} \|_2^2 ,
\]
i.e., it is equal to the diagonal element of the projection matrix onto the column span of $A$.
\end{definition}

Here is a definition of a seemingly-unrelated notion that we talked about before.

\begin{definition}
Given $G=(V,E)$, a connected, weighted, undirected graph with $n$ nodes, $m$ edges, and corresponding weights $w_e \ge 0$, for all $e \in E$, let $L=B^TWB$.
Then, the effective resistance $R_{e}$ across edge $e \in E$ are given by the diagonal elements of the matrix $R=BL^{\dagger}B$.
\end{definition}

Here is a lemma relating these two quantities.

\begin{lemma}
Let $\Phi = W^{1/2} B$ denote the scaled edge-incidence matrix.
If $\ell_i$ is the leverage score of the $i^{th}$ row of $\Phi$, then $\frac{\ell_i}{w_i}$ is the effective resistance of the $i^{th}$ edge.
\end{lemma}
\begin{proof}
Consider the matrix 
\[
P = W^{1/2}B \left( B^T W B \right)^{\dagger} B^T W^{1/2} \in \mathbb{R}^{m \times m} ,
\]
and notice that $P = W^{1/2}R W^{1/2}$ is a rescaled version of $R = B L^{\dagger} B$, whose diagonal elements are the effective resistances.
Since $\Phi = W^{1/2}B$, it follows that 
\[
P = \Phi \left( \Phi^T \Phi \right)^{\dagger} \Phi^T .
\]
Let $U_{\Phi}$ be an orthogonal matrix spanning the columns of $\Phi$.
Then, $P=U_{\Phi}U_{\Phi}^T$, and so 
\[
P_{ii} = \left( U_{\Phi} U_{\Phi}^T \right)_{ii} = \| \left( U_{\Phi} \right)_{(i)} \|_2^2 ,
\]
which establishes the lemma.
\end{proof}

So, informally, we sparsify the graph by biasing our random sampling toward edges that are ``important'' or ``influential'' in the sense that they have large statistical leverage or effective resistance, and then we use the sparsified graph to solve the subproblem.

Here is the main theorem for this algorithm.
\begin{theorem}
With constant probability, $\| x_{opt} - \tilde{x}_{opt} \|_L \le \epsilon \| x_{opt} \|_L$.
\end{theorem}
\begin{proof}
The main idea of the proof is that we are forming a sketch of the Laplacian by randomly sampling elements, which corresponds to randomly sampling rows of the edge-incidence matrix, and that we need to ensure that the corresponding sketch of the edge-incidence matrix is a so-called subspace-preserving embedding.
If that holds, then the eigenvalues of the edge-incidence matrix and it's sketch are close, and thus the eigenvalues of the Laplacian are close, and thus the original Laplacian and the sparsified Laplacian are ``close,'' in the sense that the quadratic form of one is close to the quadratic form of the other.

Here are the details.

By definition, 
\[
\| x_{opt} - \tilde{x}_{opt} \|_L^2 
   = \left(x_{opt}-\tilde{x}_{opt} \right)^{T} L \left(x_{opt}-\tilde{x}_{opt} \right) .
\]
Recall that $L = B^TWB$, that $x_{opt} = L^{\dagger}b$, and that $\tilde{x}_{opt} = \tilde{L}^{\dagger}b$.
So, 
\begin{eqnarray*}
\| x_{opt} - \tilde{x}_{opt} \|_L^2 
  &=& \left(x_{opt}-\tilde{x}_{opt} \right)^{T} B^TWB \left(x_{opt}-\tilde{x}_{opt} \right)  \\
  &=& \| W^{1/2} B \left( x_{opt}-\tilde{x}_{opt} \right) \|_2^2 
\end{eqnarray*}
Let $\Phi \in \mathbb{R}^{m \times n}$ be defined as $\Phi = W^{1/2}B$, and let its SVD be $\Phi = U_{\Phi} \Sigma_{\Phi} V_{\Phi}^T$.
Then
\[
L = \Phi^T\Phi = V_{\Phi} \Sigma_{\Phi}^{2} V_{\Phi}^T
\] 
and \[
x_{opt} = L^{\dagger} b = V_{\Phi} \Sigma_{\Phi}^{-2} V_{\Phi}^T b .
\]
In addition
\[
\tilde{L} = \Phi^T S^T S \Phi = \left( S\Phi \right)^{T} \left( S\Phi \right) 
\]
and also
\[
\tilde{x}_{opt} 
   = \tilde{L}^{\dagger}b 
   = \left( S\Phi \right)^{\dagger} \left( S\Phi \right)^{T\dagger} b 
   = \left( SU_{\Phi}\Sigma_{\Phi}V_{\Phi}^{T}\right)^{\dagger} \left( SU_{\Phi}\Sigma_{\Phi}V_{\Phi}^{T}\right)^{T\dagger} b  
\]
By combining these expressions, we get that 
\begin{eqnarray*}
\| x_{opt} - \tilde{x}_{opt} \|_L^2
  &=& \| \Phi \left( x_{opt} - \tilde{x}_{opt} \right) \|_2^2 \\
  &=& \| U_{\Phi}\Sigma_{\Phi}V_{\Phi}^{T} \left( V_{\Phi} \Sigma_{\Phi}^{-2} V_{\Phi}^T - \left( SU_{\Phi}\Sigma_{\Phi}V_{\Phi}^{T}\right)^{\dagger} \left( SU_{\Phi}\Sigma_{\Phi}V_{\Phi}^{T}\right)^{T\dagger} \right) b \|_2^2 \\
  &=& \| \Sigma_{\Phi}^{-1} V_{\Phi}^T b - \Sigma_{\Phi} \left( SU_{\Phi}\Sigma_{\Phi}V_{\Phi}^{T}\right)^{\dagger} \left( SU_{\Phi}\Sigma_{\Phi}V_{\Phi}^{T}\right)^{T\dagger} V_{\Phi} b \|_2^2
\end{eqnarray*}

Next, we note the following:
\[
\mathbb{E}\left[ \| U_{\Phi}^T S^T S U_{\Phi} - I  \|_2 \right] \le \sqrt{\epsilon}  ,
\]
where of course the expectation can be removed by standard methods.
This follows from a result of Rudelson-Vershynin, and it can also be obtained as a matrix concentration bound.
This is a key result in RLA, and it holds since we are sampling $O\left(\frac{n}{\epsilon} \log \left( \frac{n}{\epsilon} \right) \right)$ rows from $U$ according to the leverage score sampling probabilities.

From standard matrix perturbation theory, it thus follows that 
\[
\left| \sigma_i \left( U_{\Phi}^TS^TSU_{\Phi} \right) - 1 \right|
   = \left| \sigma_i^2\left(SU_{\Phi}\right)-1 \right| \le \sqrt{\epsilon}  .
\]

So, in particular, the matrix $SU_{\Phi}$ has the same rank as the matrix $U_{\Phi}$.
(This is a so-called subspace embedding, which is a key result in RLA; next time we will interpret it in terms of graphic inequalities that we discussed before.)

In the rest of the proof, let's condition on this random event being true.

Since $SU_{\Phi}$ is full rank, it follows that
\[
\left( SU_{\Phi}\Sigma_{\Phi}\right)^{\dagger} = \Sigma_{\Phi}^{-1} \left( SU_{\Phi} \right)^{\dagger} .
\]
So, we have that 
\begin{eqnarray*}
\| x_{opt} - \tilde{x}_{opt} \|_L^2
  &=& \| \Sigma_{\Phi}^{-1} V_{\Phi}^T b - \left( SU_{\Phi}\right)^{\dagger}\left(SU_{\Phi}\right)^{T\dagger} \Sigma_{\Phi}^{-1}V_{\Phi}^T b  \|_2^2  \\
  &=& \| \Sigma_{\Phi}^{-1} V_{\Phi}^T b - V_{\Omega}\Sigma_{\Omega}^{-2}V_{\Omega}^{T} \Sigma_{\Phi}^{-1}V_{\Phi}^T b  \|_2^2 ,
\end{eqnarray*}
where the second line follows if we define $\Omega = S U_{\Phi}$ and let its SVD be 
\[
\Omega = SU_{\Phi} = U_{\Omega} \Sigma_{\Omega} V_{\Omega}^T   .
\]
Then, let $\Sigma_{\Omega}^{-1} = I+E$, for a diagonal error matrix $E$, and use that $V_{\Omega}^TV_{\Omega} = V_{\Omega}V_{\Omega}^T = I$ to write 
\begin{eqnarray*}
\|x_{opt}-\tilde{x}_{opt} \|_L^2
   &=&   \| \Sigma_{\Phi}^{-1}V_{\Phi}^Tb - V_{\Omega}\left(I+E\right) V_{\Omega}^{T} \Sigma_{\Phi}^{-1} V_{\Phi}^{T} b \|_2^2 \\
   &=&   \| V_{\Omega}E V_{\Omega}^T \Sigma_{\Phi}^{-1} V_{\Phi}^{T} b \|_2^2 \\
   &=&   \| E V_{\Omega}^T \Sigma_{\Phi}^{-1} V_{\Phi}^{T} b \|_2^2 \\
   &\le& \| E V_{\Omega}^T\|_2^2 \| \Sigma_{\Phi}^{-1} V_{\Phi}^{T} b \|_2^2 \\
   &=&   \| E \|_2^2 \| \Sigma_{\Phi}^{-1} V_{\Phi}^{T} b \|_2^2
\end{eqnarray*}
But, since we want to bound $\|E\|$, note that 
\[
\left|E_{ii}\right| = \left| \sigma_i^{-2}\left(\Omega\right) - 1 \right|
                    = \left| \sigma_i^{-1}\left(SU_{\Phi}\right) - 1 \right| .
\]
So, 
\[
\|E\|_2 = \max_i \left| \sigma_i^{-2}\left(SU_{\Phi}\right) - 1 \right| \le \sqrt{\epsilon} .
\]
So, 
\[
\| x_{opt}-\tilde{x}_{opt} \|_L^2 \le \epsilon \| \Sigma_{\Phi}^{-1}V_{\Phi}^T b \|_2^2 .
\]
In addition, we can derive that 
\begin{eqnarray*}
\| x_{opt} \|_L^2 
   &=& x_{opt}^T L x_{opt} \\
   &=& \left( W^{1/2}B x_{opt} \right)^T \left( W^{1/2}B x_{opt} \right) \\
   &=& \| \Phi x_{opt} \|_2^2 \\
   &=& \| U_{\Phi} \Sigma_{\Phi} V_{\Phi}^T V_{\Phi} \Sigma_{\Phi}^{-2} V_{\Phi}^T b \|_2^2 \\
   &=& \| \Sigma_{\Phi}^{-1} V_{\Phi}^{T} b \|_2^2  .
\end{eqnarray*}
So, it follows that 
\[
\|x_{opt}-\tilde{x}_{opt} \|_L^2 \le \epsilon \| x_{opt} \|_L^2 ,
\] 
which establishes the main result.
\end{proof}

Before concluding, here is where we stand.
This is a very simple algorithm that highlights the basic ideas of Laplacian-based solvers, but it is not fast.
To make it fast, two things need to be done.
\begin{itemize}
\item
We need to compute or approximate the leverage scores quickly.
This step is very nontrivial.
The original algorithm of ST (that had the $\log^{50}(n)$ term) involved using local random walks (such as what we discussed before, and in fact the ACL algorithm was developed to improve this step, relative to the original ST result) to construct well-balanced partitions in nearly-linear time.
Then, it was shown that one could use effective resistances; this was discovered by SS independently of the RLA-based method outlined above, but it was also noted that one could call the nearly linear time solver to approximate them.
Then, it was shown that one could relate it to spanning trees to construct combinatorial preconditioners.
If this step was done very carefully, then one obtains an algorithm that runs in nearly linear time.
In particular, though, one needs to go beyond the linear algebra to map closely to the combinatorial properties of graphs, and in particular find low-stretch spanning trees.
\item
Instead of solving the subproblem on the sketch, we need to use the sketch to create a preconditioner for the original problem and then solve a preconditioned version of the original problem.
This step is relatively straightforward, although it involves applying an iterative algorithm that is less common than popular CG-based methods.
\end{itemize}
We will go through both of these in more detail next time.

%% file: lect26.tex
\section{%
(04/28/2015):
Laplacian solvers (2 of 2)}

Reading for today.
\begin{compactitem}
\item
Same as last class.
\end{compactitem}

Last time, we talked about a very simple solver for Laplacian-based systems of linear equations, i.e., systems of linear equations of the form $Ax=b$, where the constraint matrix $A$ is the Laplacian of a graph.
This is not fully-general---Laplacians are SPSD matrices of a particular form---but equations of this form arise in many applications, certain other SPSD problems such as those based on SDD matrices can be reduced to this, and there has been a lot of work recently on this topic since it is a primitive for many other problems.
The solver from last time is very simple, and it highlights the key ideas used in fast solvers, but it is very slow.
Today, we will describe how to take those basic ideas and, by coupling them with certain graph theoretic tools in various ways, obtain a ``fast'' nearly linear time solver for Laplacian-based systems of linear equations.

In particular, today will be based on the Batson-Spielman-Srivastava-Teng and the Koutis-Miller-Peng articles.

\subsection{Review from last time and general comments}

Let's start with a review of what we covered last time.

Here is a very simple algorithm.
Given as input the Laplacian $L$ of a graph $G=(V,E)$ and a right hand side vector $b$, do the following.
\begin{itemize}
\item
Construct a sketch of $G$ by sampling elements of $G$, i.e., rows of the edge-node incidence matrix, with probability proportional to the leverage scores of that row, i.e., the effective resistances of that edge.
\item
Use the sketch to construct a solution, e.g., by solving the subproblem with a black box or using it as a preconditioner to solve the original problem with an iterative method. 
\end{itemize}

The basic result we proved last time is the following.
\begin{theorem}
Given a graph $G$ with Laplacian $L$, let $x_{opt}$ be the optimal solution of $Lx=b$; then the above algorithm returns a vector $\tilde{x}_{opt}$ such that, with constant probability, 
\begin{equation}
\label{eqn:solution-approximation}
\| x_{opt} - \tilde{x}_{opt} \|_{L} \le \epsilon \| x_{opt} \|_{L}  .
\end{equation}
\end{theorem}
The proof of this result boiled down to showing that, by sampling with respect to a judiciously-chosen set of nonuniform importance sampling probabilities, then one obtains a data-dependent subspace embedding of the edge-incidence matrix.
Technically, the main thing to establish was that, if $U$ is an $m \times n$ orthogonal matrix spanning the column space of the weighted edge-incidence matrix, in which case $I = I_n = U^TU$, then 
\begin{equation}
\label{eqn:subspace-embedding}
\| I - \left(SU\right)^{T}\left(SU\right) \|_2 \le \epsilon ,
\end{equation}
where $S$ is a sampling matrix that represents the effect of sampling elements from $L$.

The sampling probabilities that are used to create the sketch are weighted versions of the statistical leverage scores of the edge-incidence matrix, and thus they also are equal to the effective resistance of the corresponding edge in the graph.
Importantly, although we didn't describe it in detail, the theory that provides bounds of the form of Eqn.~(\ref{eqn:subspace-embedding}) is robust to the exact form of the importance sampling probabilities, e.g., bounds of the same form hold if any other probabilities are used that are ``close'' (in a sense that we will discuss) to the statistical leverage scores.

The running time of this simple strawman algorithm consists of two parts, both of which the fast algorithms we will discuss today improve upon.
\begin{itemize}
\item
Compute the leverage scores, exactly or approximately.
A naive computation of the leverage scores takes $O(mn^2)$ time, e.g., with a black box QR decomposition routine.
Since they are related to the effective resistances, one can---theoretically at least compute them with any one of a variety of fast nearly linear time solvers (although one has a chicken-and-egg problem, since the solver itself needs those quantities).
Alternatively, since one does not need the exact leverage scores, one could hope to approximate them in some way---below, we will discuss how this can be done with low-stretch spanning trees.
\item
Solve the subproblem, exactly or approximately.
A naive computation of the solution to the subproblem can be done in $O(n^3)$ time with standard direct methods, or it can be done with an iterative algorithm that requires a number of matrix-vector multiplications that depends on the condition number of $L$ (which in general could be large, e.g., $\Omega(n)$) times $m$, the number of nonzero elements of $L$.
Below, we will see how this can be improved with sophisticated versions of certain preconditioned iterative algorithms.
\end{itemize}

More generally, here are several issues that arise.
\begin{itemize}
\item
Does one use exact or approximate leverage scores?
Approximate leverage scores are sufficient for the worst-case theory, and we will see that this can be accomplished by using LSSTs, i.e., combinatorial techniques. 
\item
How good a sketch is necessary?
Last time, we sampled $\Theta\left( \frac{ n \log(n) }{\epsilon^2} \right)$ elements from $L$ to obtain a $1\pm\epsilon$ subspace embedding, i.e., to satisfy Eqn.~(\ref{eqn:subspace-embedding}), and this leads to an $\epsilon$-approximate solution of the form of Eqn~(\ref{eqn:solution-approximation}).
For an iterative method, this might be overkill, and it might suffice to satisfy Eqn.~(\ref{eqn:subspace-embedding}) for, say, $\epsilon = \frac{1}{2}$.
\item
What is the dependence on $\epsilon$?
Last time, we sampled and then solved the subproblem, and thus the complexity with respect to $\epsilon$ is given by the usual random sampling results.
In particular, since the complexity is a low-degree polynomial in $\frac{1}{\epsilon}$, it will be essentially impossible to obtain a high-precision solution, e.g., with $\epsilon = 10^{-16}$, as is of interest in certain~applications.
\item
What is the dependence on the condition number $\kappa(L)$?
In general, the condition number can be very large, and this will manifest itself in a large number of iterations (certainly in worst case, but also actually quite commonly).
By working with a preconditioned iterative algorithm, one should aim for a condition number of the preconditioned problem that is quite small, e.g., if not constant then $\log(n)$ or less.
In general, there will be a tradeoff between the quality of the preconditioner and the number of iterations needed to solve the preconditioned~problem.
\item
Should one solve the subproblem directly or use it to construct a preconditioned to the original problem?
Several of the results just outlined suggest that an appropriate iterative methods should be used and this is what leads to the best results. 
\end{itemize}

\textbf{Remark.}
Although we are not going to describe it in detail, we should note that the LSSTs will essentially allow us to approximate the large leverage scores, but they won't have anything to say about the small leverage scores.
We saw (in a different context) when we were discussing statistical inference issues that controlling the small leverage scores can be important (for proving statistical claims about unseen data, but not for claims on the empirical data).
Likely similar issues arise here, and likely this issue can be mitigated by using implicitly regularized Laplacians, e.g., as as implicitly computed by certain spectral ranking methods we discussed, but as far as I know no one has explicitly addressed these questions.

\subsection{Solving linear equations with direct and iterative methods}

Let's start with the second step of the above two-level algorithm, i.e., how to use the sketch from the first step to construct an approximate solution, and in particular how to use it to construct a preconditioner for an iterative algorithm.
Then, later we will get back to the first step of how to construct the sketch.

As you probably know, there are a wide range of methods to solve linear systems of the form $Ax=b$, but they fall into two broad categories.
\begin{itemize}
\item
\textbf{Direct methods.}
These include Gaussian elimination, which runs in $O(n^3)$ time; and Strassen-like algorithms, which run in $O(n^{\omega})$ time, where $\omega = 2.87 \ldots 2.37$.
Both require storing the full set of in general $O(n^2)$ entries.
Faster algorithms exist if $A$ is structured.
For example, if $A$ is $n \times n$ PSD with $m$ nonzero, then conjugate gradients, used as a direct solver, takes $O(mn)$ time, which if $m = O(n)$ is just $O(n^2)$.
That is, in this case, the time it takes it proportional to the time it takes just to write down the inverse.
Alternatively, if $A$ is the adjacency matrix of a path graph or any tree, then the running time is $O(m)$; and so on.
\item
\textbf{Iterative methods.}
These methods don't compute an exact answer, but they do compute an $\epsilon$-approximate solution, where $\epsilon$ depends on the structural properties of $A$ and the number of iterations, and where $\epsilon$ can be made smaller with additional iterations.
In general, iterations are performed by doing matrix-vector multiplications.
Advantages of iterative methods include that one only needs to store $A$, these algorithms are sometimes very simple, and they are often faster than running a direct solver.
Disadvantages include that one doesn't obtain an exact answer, it can be hard to predict the number of iterations, and the running time depends on the eigenvalues of $A$, e.g., the condition number $\kappa(A) = \frac{\lambda_{max}(A)}{\lambda_{min}(A)}$.
Examples include the Richardson iteration, various Conjugate Gradient like algorithms, and the Chebyshev~iteration.
\end{itemize}

Since the running time of iterative algorithms depend on the properties of $A$, so-called \emph{preconditioning methods} are a class of methods to transform a given input problem into another problem such that the modified problem has the same or a related solution to the original problem; and such that the modified problem can be solved with an iterative method more quickly.

For example, to solve $Ax=b$, with $A\in\mathbb{R}^{n \times n}$ and with $m=\textbf{nnz}(A)$, if we define $\kappa(A) = \frac{\lambda_{max}(A)}{\lambda_{min}(A)}$, where $\lambda_{max}$ and $\lambda_{min}$ are the maximum and minimum non-zero eigenvalues of $A$, to be the condition number of $A$, then CG runs in 
\[
O\left( \sqrt{\kappa{(A)}}\log\left(1/\epsilon\right) \right)
\]
iterations (each of which involves a matrix-vector multiplication taking $O(m)$ time) to compute and $\epsilon$-accurate solution to $Ax=b$.
By an $\epsilon$-accurate approximation, here we mean the same notion that we used above, i.e., that $$\|\tilde{x}_{opt}-A^{\dagger}b\|_{A}\le\epsilon\|A^{\dagger}b\|_{A},$$ where the so-called $A$-norm is given by $\|y\|_{A}=\sqrt{y^TAy}$.
This $A$-norm is related to the usual Euclidean norm as follows: $\|y\|_{A}\le\kappa(A)\|y\|_2$ and $\|y\|_2\le\kappa(A)\|y\|_{A}$.
While the $A$-norm is perhaps unfamiliar, in the context of iterative algorithms it is not too dissimilar to the usual Euclidean norm, in that, given an $\epsilon$-approximation for the former, we can obtain an $\epsilon$-approximation for the latter with $O\left( \log\left( \kappa(A)/\epsilon \right) \right)$ extra iterations.

In this context, preconditioning typically means solving 
\[
B^{-1}Ax = B^{-1}b ,
\]
where $B$ is chosen such that $\kappa\left(B^{-1}A\right)$ is small; and it is easy to solve problems of the form $Bz=c$. 
The two extreme cases are $B=I$, in which case it is easy to compute and apply but doesn't help solve the original problem, and $B=A^{-1}$, which means that the iterative algorithm would converge after zero steps but which is difficult to compute.
The running time of the preconditioned problem involves $$O\left( \sqrt{\kappa\left( B^{-1}A \right)} \log\left(1/\epsilon \right) \right)$$ matrix vector multiplications.
The quantity $\kappa\left( B^{-1}A \right)$ is sometimes known as the \emph{relative condition number of $A$ with respect to $B$}---in general, finding a $B$ that makes it smaller takes more initial time but leads to fewer iterations.
(This was the basis for the comment above that there is a tradeoff in choosing the quality of the preconditioner, and it is true more generally.)

These ideas apply more generally, but we consider applying them here to Laplacians.
So, in particular, given a graph $G$ and its Laplacian $L_{G}$, one way to precondition it is to look for a different graph $H$ such that $L_H \approx L_G$.
For example, one could use the sparsified graph that we computed with the algorithm from last class.
That sparsified graph is actually an $\epsilon$-good preconditioned, but it is too expensive to compute.
To understand how we can go beyond the linear algebra and exploit graph theoretic ideas to get good approximations to them more quickly, let's discuss different ways in which two graphs can be close to one another.

\subsection{Different ways two graphs can be close}

We have talked formally and informally about different ways graphs can be close, e.g., we used the idea of similar Laplacian quadratic forms when talking about Cheeger's Inequality and the quality of spectral partitioning methods.
We will be interested in that notion, but we will also be interested in other notions, so let's now discuss this topic in more detail.
\begin{itemize}
\item
\textbf{Cut similarity.}
One way to quantify the idea that two graphs are close is to say that they are similar in terms of their cuts or partitions.
The standard result in this area is due to BZ, who developed the notion of \emph{cut similarity} to develop fast algorithms for min cut and max flow and other related combinatorial problems.
This notion of similarity says that two graphs are close if the weights of the cuts, i.e., the sum of edges or edge weights crossing a partition, are close for all cuts.
To define it, recall that, given a graph $G=(V,E,W)$ and a set $S \subset V$, we can define $\mbox{cut}_G = \sum_{u \in S, v \in \bar{S} } W_{(uv)}$.
Here is the definition.
\begin{definition}
Given two graphs, $G=(V,E,W)$ and $\tilde{G} = (V,\tilde{E},\tilde{W})$, on the same vertex set, we say that $G$ and $\tilde{G}$ are \emph{$\sigma$-cut-similar} if, for all $S\subseteq V$, we have that
\[
\frac{1}{\sigma} \mbox{cut}_{\tilde{G}}(S) \le \mbox{cut}_{G}(S) \le \sigma \mbox{cut}_{\tilde{G}}(S)  .
\]
\end{definition}
As an example of a result in this area, the following theorem shows that every graph is cut-similar to a graph with average degree $O(\log(n))$ and that one can compute that cut-similar graph quickly.
\begin{theorem}[BK]
For all $\epsilon > 0$, every graph $G=(V,E,W)$ has a $\left(1+\epsilon\right)$-cut-similar graph $\tilde{G}=(V,\tilde{E},\tilde{V})$ such that $\tilde{E} \subseteq E$ and  $|\tilde{E}| = O\left( n \log(n/\epsilon^2) \right)$.
In addition, the graph $\tilde{G}$ can be computed in  $O\left( m \log^3 (n) + m \log (n/\epsilon^2) \right)$ time.
\end{theorem}
\item
\textbf{Spectral similarity.}
ST introduced the idea of spectral similarity in the context of nearly linear time solvers.
One can view this in two complementary ways.
\begin{itemize}
\item
As a generalization of cut similarity.
\item
As a special case of subspace embeddings, as used in RLA.
\end{itemize}
We will do the former here, but we will point out the latter at an appropriate point.

Given $G=(V,E,W)$, recall that $L:\mathbb{R}^{n}\rightarrow\mathbb{R}$ is a quadratic form associated with $G$ such that 
\[
L(x) = \sum_{(uv) \in E} W_{(uv)} \left( x_u - x_v \right)^2  .
\]
If $S \subset V$ and if $x$ is an indicator/characteristic vector for the set $S$, i.e., it equals $1$ on nodes $u \in S$, and it equals $0$ on nodes $v \in S$, then for those indicator vectors $x$, we have that $L(x) = \mbox{cut}_{G}(x)$.
We can also ask about the values it takes for other vectors $x$.
So, let's define the following.
\begin{definition}
Given two graphs, $G=(V,E,W)$ and $\tilde{G}=(V,\tilde{E},\tilde{W})$, on the same vertex set, we say that $G$ and $\tilde{G}$ are \emph{$\sigma$-spectrally similar} if, for all $x\in\mathbb{R}^{n}$, we have that 
\[
\frac{1}{\sigma} L_{\tilde{G}}(x) \le L_{G}(x) \le \sigma L_{\tilde{G}}(x) .
\]
\end{definition}
That is, two graphs are spectrally similar if their Laplacian quadratic forms are close.

In addition to being a generalization of cut similarity, this also corresponds to a special case of subspace embeddings, restricted from general matrices to edge-incidence matrices and their associated~Laplacians.
\begin{itemize}
\item
To see this, recall that subspace embeddings preserve the geometry of the subspace and that this is quantified by saying that all the singular values of the sampled/sketched version of the edge-incidence matrix are close to $1$, i.e., close to those of the edge-incidence matrix of the original un-sampled graph.
Then, by considering the Laplacian, rather than the edge-incidence matrix, the singular values of the original and sketched Laplacian are also close, up to a quadratic of the approximation factor on the edge-incidence matrix.
\end{itemize}

Here are several other things to note about spectral embeddings.
\begin{itemize}
\item
Two graphs can be cut-similar but not spectrally-similar.
For example, consider $G$ to be an $n$-vertex path and $\tilde{G}$ to be an $n$-vertex cycle.
They are $2$-cut similar but are only $n$-spectrally similar.
\item
Spectral similarity is identical to the notion of relative condition number in NLA that we mentioned above.
Recall, given $A$ and $B$, then $ A \preceq B$ iff $x^TAx \le x^TBx$, for all $x\in\mathbb{R}^{n}$.
Then, $A$ and $B$, if they are Laplacians, are spectrally similar if $\frac{1}{\sigma}B \preceq A \preceq \sigma B$.
In this case, they have similar eigenvalues, since: from the Courant-Fischer results, if $\lambda_1,\ldots,\lambda_n$ are the eigenvalues of $A$ and $\tilde{\lambda}_1,\ldots,\tilde{\lambda}_n$ are the eigenvalues of $B$, then for all $i$ we have that $\frac{1}{\sigma} \tilde{\lambda}_i \le \lambda_i \le \sigma \tilde{\lambda_i}$.
\item
More generally, spectral similarity means that the two graphs will share many spectral or linear algebraic properties, e.g., effective resistances, resistance distances, etc.
\end{itemize}
\item
\textbf{Distance similarity.}
If one assigns a length to every edge $e \in E$, then these lengths induce a shortest path distance between every $u,v \in V$.
Thus, given a graph $G=(V,E,W)$, we can let $d: V \times V \rightarrow \mathbb{R}^{+}$ be the shortest path distance.
Given this, we can define the following notion of similarity.
\begin{definition}
Given two graphs, $G=(V,E,W)$ and $\tilde{G}=(V,\tilde{E},\tilde{W})$, on the same vertex set, we say that $G$ and $\tilde{G}$ are \emph{$\sigma$-distance similar} if, for all pairs of vertices $u,v \in V$, we have that 
\[
\frac{1}{\sigma} \tilde{d}(u,v) \le d(u,v) \le \sigma \tilde{d}(u,v)  .
\]
\end{definition}
Note that if $\tilde{G}$ is a subgraph if $G$, then $d_{G}(u,v) \le d_{\tilde{G}}(u,v)$, since shortest-path distances can only increase.
(Importantly, this does \emph{not} necessarily hold if the edges of the subgraph are re-weighted, as they were done in the simple algorithm from the last class, when the subgraph is constructed; we will get back to this later.)
In this case, a spanner is a subgraph such that distances in the other direction are not changed too much.
\begin{definition}
Given a graph $G=(V,E,W)$, a \emph{$t$-spanner} is a subgraph of $G$ such that for all $u,v \in V$, we have that $d_{\tilde{G}}(u,v) \le t d_{G}(u,v)$.
\end{definition}
There has been a range of work in TCS on spanners (e.g., it is known that every graph has a $2t+1$ spanner with $O\left( n^{1+1/\epsilon} \right)$ edges) that isn't directly relevant to what we are doing. 

We will be most interested in spanners that are trees or nearly trees.

\begin{definition}
Given a graph $G=(V,E,W)$, a \emph{spanning tree} is a tree includes all vertices in $G$ and is a subgraph of $G$.
\end{definition}
There are various related notions that have been studied in different contexts: for example, minimum spanning trees, random spanning trees, and low-stretch spanning trees (LSSTs).
Again, to understand some of the differences, think of a path versus a cycle.
For today, we will be interested in LSSTs.
The most extreme form of a sparse spanner is a LSST, which has only $n-1$ edges but which approximates pairwise distances up to small, e.g., hopefully polylog,~factors.
\end{itemize}

\subsection{Sparsified graphs}

Here is an aside with some more details about sparsified graphs, which is of interest since this is the first step of our Laplacian-based linear equation solver algorithm.
Let's define the following, which is a slight variant of the above.
\begin{definition}
Given a graph $G$, a \emph{$(\sigma,d)$-spectral sparsifier} of $G$ is a graph $\tilde{G}$ such that
\begin{compactenum}
\item
$\tilde{G}$ is $\sigma$-spectrally similar to $G$.
\item
The edges of $\tilde{G}$ are reweighed versions of the edges of $G$.
\item
$\tilde{G}$ has $\le d|V|$ edges.
\end{compactenum}
\end{definition}
 
\textbf{Fact.}
Expanders can be thought of as sparse versions of the complete graph; and, if edges are weighted appropriately, they are spectral sparsifiers of the complete graph.
This holds true more generally for other graphs.
Here are examples of such results.
\begin{itemize}
\item
SS showed that every graph has a $\left(1+\epsilon, O(\left( \frac{\log(n)}{\epsilon^2} \right) \right)$ spectral sparsifier.
This was shown by SS with an effective resistance argument; and it follows from what we discussed last time: last time, we showed that sampling with respect to the leverage scores gives a subspace embedding, which preserves the geometry of the subspace, which preserves the Laplacian quadratic form, which implies the spectral sparsification claim.
\item
BSS showed that every $n$ node graph $G$ has a $\left(  \frac{\sqrt{d}+1}{\sqrt{d}-1},d\right)$-spectral sparsifier (which in general is more expensive to compute than running a nearly linear time solver).
In particular, $G$ has a $\left( 1+2\epsilon, \frac{4}{\epsilon^2} \right)$-spectral sparsifier, for every $\epsilon\in(0,1)$.
\end{itemize}

Finally, there are several ways to speed up the computation of graph sparsification algorithms, relative to the strawman that we presented in the last class.
\begin{itemize}
\item
Given the relationship between the leverage scores and the effective resistances and that the effective resistances can be computed with a nearly linear time solver, one can use the ST or KMP solver to speed up the computation of graph sparsifiers.
\item
One can use local spectral methods, e.g., diffusion-based methods from ST or the push algorithm of ACL, to compute well-balanced partitions in nearly linear time and from them obtain spectral sparsifiers.
\item
Union of random spanning trees.
It is known that, e.g., the union of two random spanning trees is $O(\log(n))$-cut similar to $G$; that the union of $O\left( \log^2(n)/\epsilon^2 \right)$ reweighed random spanning trees is a $1+\epsilon$-cut sparsifier; and so on.  
This suggests looking at spanning trees and other related combinatorial quantities that can be quickly computed to speed up the computation of graph sparsifiers. 
We turn to this next.
\end{itemize}

\subsection{Back to Laplacian-based linear systems}

KMP considered the use of combinatorial preconditions, an idea that traces back to Vaidya.
They coupled this with a fact that has been used extensively in RLA: that only approximate leverage scores are actually needed in the sampling process to create a sparse sketch of $L$.
In particular, they compute upper estimates of the leverage scores or effective resistance of each edge, and they compute these estimates on a modified graph, in which each upper estimate is sufficiently good.
The modified graph is rather simple: take a LSST and increase its weights.
Although the sampling probabilities obtained from the LSST are strictly greater than the effective resistances, they are not too much greater in aggregate.
This, coupled with a rather complicated iterative preconditioning scheme, coupled with careful accounting with careful data structures, will lead to a solver that runs in $O\left( m \log(n)\log(1/\epsilon) \right)$ time, up to $\log\log(n)$ factors.
We will discuss each of these briefly in turn.

\paragraph{Use of approximate leverage scores.}
Recall from last class that an important step in the algorithm was to use nonuniform importance sampling probabilities.  
In particular, if we sampled edges from the edge-incidence matrix with probabilities $\{p_i\}_{i=1}^{m}$, where each $p_i = \ell_i$, where $\ell_i$ is the effective resistance or statistical leverage score of the weighted edge-incidence matrix, then we showed that if we sampled $r=O\left( n\log(n)/\epsilon\right)$ edges, then it follows that 
\[
\| I - \left(SU_{\Phi}\right)^{T}\left(SU_{\Phi}\right) \|_2 \le \epsilon  ,
\]
from which we were able to obtain a good relative-error solution.

Using probabilities exactly equal to the leverage scores is overkill, and the same result holds if we use any probabilities $p_i^{\prime}$ that are ``close'' to $p_i$ in the following sense: if 
\[
p_i^{\prime} \ge \beta \ell_i   ,
\]
for $\beta\in(0,1]$ and $\sum_{i=1}^{m} p_i^{\prime}=1$, then the same result follows if we sample $r=O\left( n\log(n)/(\beta\epsilon)\right)$ edges, i.e., if we oversample by a factor of $1/\beta$.
The key point here is that it is essential not to underestimate the high-leverage edges too much.
It is, however, acceptable if we overestimate and thus oversample some low-leverage edges, as long as we don't do it too much.

In particular, let's say that we have the leverage scores $\{\ell_1\}_{i=1}^{m}$ and overestimation factors $\{\gamma_i\}_{i=1}^{m}$, where each $\gamma_i \ge 1$.
From this, we can consider the probabilities 
\[
p_i^{\prime\prime} = \frac{\gamma_i\ell_i}{\sum_{i=1}^{m}\gamma_i\ell_i}   .
\]
If $\sum_{i=1}^{m}\gamma_i\ell_i$ is not too large, say $O\left(n\log(n)\right)$ or some other factor that is only slightly larger than $n$, then dividing by it (to normalize $\{\gamma_i\ell_i\}_{i=1}^{m}$ to unity to be a probability distribution) does not decrease the probabilities for the high-leverage components too much, and so we can use the probabilities $p_i^{\prime\prime}$ with an extra amount of oversampling that equals $\frac{1}{\beta} = \sum_{i=1}^{m}\gamma_i\ell_i$.

\paragraph{Use of LSSTs as combinatorial preconditioners.}
Here, the idea is to use a LSST, i.e., use a particular form of a ``combinatorial preconditioning,'' to replace $\ell_i=\ell_{(uv)}$ with the stretch of the edge $(uv)$ in the LSST.
Vaidya was the first to suggest the use of spanning trees of $L$ as building blocks as the base for preconditioning matrix $B$.
The idea is then that the linear system, if the constraint matrix is the Laplacian of a tree, can be solved in $O(n)$ time with Gaussian elimination.
(Adding a few edges back into the tree gives a preconditioned that is only better, and it is still easy to solve.)
Boman-Hendrickson used a LSST as a stand-along preconditioner.
ST used a preconditioner that is a LSST plus a small number of extra edges.
KMP had additional extensions that we describe here.

Two question arise with this approach.
\begin{itemize}
\item
Q1: What is the appropriate base tree?
\item
Q2: Which off-tree edges should added into the preconditioner?
\end{itemize}
One idea is to use a tree that concentrates that maximum possible weight from the total weight of the edges in $L$.
This is what Vaidya did; and, while it led to good result, the results weren't good enough for what we are discussing here.
(In particular, note that it doesn't discriminate between different trees in unweighted graphs, and it won't provide a bias toward the middle edge of a dumbbell graph.)
Another idea is to use a tree that concentrates mass on high leverage/influence edges, i.e., edges with the highest leverage in the edge-incidence matrix or effective resistance in the corresponding Laplacian.  

The key idea to make this work is that of \emph{stretch}.
To define this, recall that for every edge $(u,v)\in E$ in the original graph Laplacian $L$, there is a unique ``detour'' path between $u$ and $v$ in the tree $T$.
\begin{definition}
The \emph{stretch} of the edge with respect to $T$ equals the distortion caused by this detour.
\end{definition}
In the unweighted case, this stretch is simply the length of the tree path, i.e., of the path between nodes $u$ and $v$ that were connected by an edge in $G$ in the tree $T$.
Given this, we can define the~following.
\begin{definition}
The \emph{total stretch} of a graph $G$ and its Laplacian $L$ with respect to a tree $T$ is the sum of the stretches of all off-tree edges.
Then, a \emph{low-stretch spanning tree (LSST)} $T$ is a tree such that the total stretch is low.
\end{definition}
Informally, a LSST is one such that it provides a good ``on average'' detours for edges of the graph, i.e., there can be a few pairs of nodes that are stretched a lot, but there can't be too many such~pairs.

There are many algorithms for LSSTs.
For example, here is a result that is particularly relevant for us.
\begin{theorem}
Every graph $G$ has a spanning tree $T$ with total stretch $\tilde{O}\left( m \log(n) \right)$, and this tree can be found in $\tilde{O}\left( m \log(n) \right)$ time.
\end{theorem}

In particular, we can use the stretches of pairs of nodes in the tree $T$ in place of the leverage scores or effective resistances as importance sampling probabilities: they are larger than the leverage scores, and there might be a few that much larger, but the total sum is not much larger than the total sum of the leverage scores (which equals $n-1$).

\paragraph{Paying careful attention to data structures, bookkeeping, and recursive preconditions.}
Basically, to get everything to work in the allotted time, one needs the preconditioner $B$ that is extremely good approximation to $L$ and that can be computed in linear time.
What we did in the last class was to compute a ``one step'' preconditioner, and likely any such ``one step'' preconditioned won't be substantially easier to compute that solving the equation; and so KMP consider recursion in the construction of their preconditioner.
\begin{itemize}
\item
In a recursive preconditioning method, the system in the preconditioned $B$ is not solved exactly but only approximately, via a recursive invocation of the same iterative method.
So, one must find a preconditioned for $B$, a preconditioned for it, and so on.
This gives s multilevel hierarchy of progressively smaller graphs.
To make the total work small, i.e., $O(kn)$, for some constant $k$, one needs the graphs in the hierarchy to get small sufficiently fast.
It is sufficient that the graph on the $(i+1)^{th}$ level is smaller than the graph on the $i^{th}$ level by a factor of $\frac{1}{2k}$.
However, one must converge within $O(kn)$.
So, one can use CG/Chebyshev, which need $O(k)$ iterations to converge, when $B$ is a $k^2$-approximation of $L$ (as opposed to $O(k^2)$ iterations which are needed for something like a Richardson's iteration).
\end{itemize}

So, a LSST is a good base; and a LSST also tells us which off-tree edges, i.e., which additional edges from $G$ that are not in $T$, should go into the preconditioner.
\begin{itemize}
\item
This leads to an $\tilde{O}\left( m \log^2(n)\log(1/\epsilon) \right)$ algorithm.
\end{itemize}
If one keeps sampling based on the same tree and does some other more complicated and careful stuff, then one obtains a hierarchical graph and is able to remove the the second log factor to yield a potentially practical solver.
\begin{itemize}
\item
This leads to an $\tilde{O}\left( m \log (n)\log(1/\epsilon) \right)$ algorithm.
\end{itemize}
See the BSST and KMP papers for all the details.